\documentclass[11pt]{article}
\usepackage[margin=0.75in]{geometry}
\usepackage{color}
\usepackage{graphicx}
\usepackage{amsmath}
\usepackage{amssymb}
\usepackage{amsthm}
\usepackage{subfigure}
\usepackage{rotating}
\usepackage{verbatim}
\usepackage{hyperref}
\usepackage{bm}
\usepackage[round,authoryear]{natbib}
\usepackage{bbm}
\usepackage{xr}
\usepackage{array}
\usepackage{tikz}
\usepackage{natbib}
\usepackage{mathrsfs}
\usepackage{enumerate}
\usepackage{centernot}
\newcommand{\sigalg}{\mathscr{A}} % revision with red colors
\newcommand{\revision}[1]{{\color{black}{#1}}} % revision with red colors

\DeclareFontFamily{U}{mathx}{\hyphenchar\font45}
\DeclareFontShape{U}{mathx}{m}{n}{<-> mathx10}{}
\DeclareSymbolFont{mathx}{U}{mathx}{m}{n}
\DeclareMathAccent{\widebar}{0}{mathx}{"73}
\usepackage{float}
\newtheorem{theorem}{Theorem}[section] 
\newtheorem{lemma}[theorem]{Lemma} 
 
\newtheorem{proposition}[theorem]{Proposition} 
\newtheorem{definition}[theorem]{Definition} 
\newtheorem{assump}[theorem]{Assumption}
\newtheorem{cond}[theorem]{Condition}
\newtheorem{indhyp*}{Induction Hypothesis} 
\newtheorem{remark}{Remark}[theorem]
\newtheorem{premark}{Remark}[section]
\newtheorem*{fact}{Fact}

\floatstyle{ruled} 
\newfloat{algorithm}{tbp}{loa} 
\usepackage{algorithm}
\usepackage{algorithmic}

\newtheorem{innercustomlem}{Lemma}
\newenvironment{custlem}[1]
  {\renewcommand\theinnercustomlem{#1}\innercustomlem}
  {\endinnercustomlem}

\newcommand{\myrom}[1]{(\romannumeral #1)}
\newcommand{\myRom}[1]{\MakeUppercase{\romannumeral #1}}
\newcommand{\bs}{\boldsymbol}

\def\independenT#1#2{\mathrel{\rlap{$#1#2$}\mkern2mu{#1#2}}} 
\newcommand\independent{\protect\mathpalette{\protect\independenT}{\perp}} % independence symbol

\newcommand{\deq}{\stackrel{d}{=}}
\newcommand{\iid}{\stackrel{\emph{i.i.d.}}{\sim}}

\newcommand{\nsubp}{_{i=1}^{n}}  % 

\newcommand{\EE}[1]{\mathbb{E}\left[{#1}\right]} % Expectation
\newcommand{\Ep}[2]{\mathbb{E}_{#1}\left[{#2}\right]} % Expectation with respect to some law
\newcommand{\Ec}[2]{\mathbb{E}\left[{#1}\,|\,{#2}\right]} % Conditional expectation
\newcommand{\Epc}[3]{\mathbb{E}_{#1}\left[{#2}\,|\,{#3}\right]} % Conditional expectation

\newcommand{\Ecmid}[2]{\mathbb{E}\left[{#1} \,\middle|\, {#2}\right]} % Conditional expectation 
\newcommand{\ec}[2]{\mathbb{E}[{#1}\,|\,{#2}]} % Conditional expectation
\newcommand{\ee}[1]{\mathbb{E}[{#1}]} % Expectation
\newcommand{\PP}[1]{\mathbb{P}\left({#1}\right)} % Probability
\newcommand{\Pp}[2]{\mathbb{P}_{#1}\left({#2}\right)} % Probability with respect to some law
\newcommand{\Pc}[2]{\mathbb{P}\left({#1}\,|\, {#2}\right)} % Conditional probability
\newcommand{\Pcmid}[2]{\mathbb{P}\left({#1} \, \middle|\, {#2}\right)} % Conditional expectation 
\newcommand{\Var}[1]{\mathrm{Var}\left({#1}\right)} % Variance
\newcommand{\Varc}[2]{\mathrm{Var}\left({#1}\,|\,{#2}\right)} % Conditional variance
\newcommand{\Varp}[2]{\mathrm{Var}_{#1}\left({#2}\right)} % Variance with respect to some law

\newcommand{\Varpc}[3]{\mathrm{Var}_{#1}\left({#2}\,|\,{#3}\right)} % Conditional expectation
\newcommand{\Ppc}[3]{\mathbb{P}_{#1}\left({#2}\,|\,{#3}\right)} % Conditional expectation

\newcommand{\Varcmid}[2]{\mathrm{Var}\left({#1}\,\middle|\, {#2}\right)} % Conditional expectation 

\newcommand{\varc}[2]{\mathrm{Var}({#1}\,|\,{#2})} % Conditional variance
\newcommand{\Cov}[2]{\mathrm{Cov}\left({#1},{#2}\right)} % Covariance
\newcommand{\Covc}[3]{\mathrm{Cov}\left({#1},{#2}\,|\,{#3}\right)} % Conditional covariance
\newcommand{\Covp}[3]{\mathrm{Cov}_{#1}\left({#2},{#3}\right)} % Covariance with respect to some law

\newcommand{\covc}[3]{\mathrm{Cov}({#1},{#2}\,|\,{#3})} % Conditional covariance
\DeclareMathOperator{\Tr}{Tr}

\newcommand{\Xk}{(\tilde{X},Z)}
\newcommand{\Xtil}{\tilde{X}}
\newcommand{\Wtil}{\tilde{W}}

\newcommand{\Xj}{X}
\newcommand{\xj}{x}
\newcommand{\Xkj}{\tilde{X}}
\newcommand{\noj}{\text{-}j}
\newcommand{\noi}{\text{-}i}
\newcommand{\Xnoj}{Z}
\newcommand{\xnoj}{z} 
\newcommand{\Xinoj}{Z_{i}}

\newcommand{\jnull}{Y\independent X \mid\Xnoj} % Y indep Xj | X_\noj

\newcommand{\eps}{\epsilon}
\newcommand{\condmean}{\Ec{Y}{X,Z}}
\newcommand{\condmeanj}{\Ec{Y}{\Xnoj}}

\newcommand{\ANOVA}{\EE{\Varc{\condmean}{\Xnoj}}}
\newcommand{\calI}{\mathcal{I}}
\newcommand{\Ij}{\calI}

\newcommand{\mustar}{\mu^{\star}}
\newcommand{\muhat}{\hat{\mu}}
\newcommand{\gstar}{g^{\star}}
\newcommand{\hstar}{h^{\star}}

\newcommand{\thetamu}{f(\mu)}

\newcommand{\thetamus}{f(\mustar)}

\newcommand{\muW}{\mu(W)}
\newcommand{\muX}{\mu(X,Z)}

\newcommand{\muXk}{\mu(\Xtil,Z)}

\newcommand{\calU}{\mathcal{U}}

\newcommand{\calB}{\mathcal{B}}

\newcommand{\calF}{\mathcal{F}}

\newcommand{\calN}{\mathcal{N}}
\newcommand{\gauss}[2]{\calN\left(#1,#2\right)}
\newcommand{\calW}{\mathcal{W}}

\newcommand{\cH}{\mathcal{H}}
\newcommand{\cI}{\mathcal{I}}

\newcommand{\cK}{\mathcal{K}}

\newcommand{\cX}{\mathcal{X}}
\newcommand{\cY}{\mathcal{Y}}
\newcommand{\cZ}{\mathcal{Z}}

\newcommand{\thetaQmu}{f^{Q}(\mu)}

\newcommand{\Ijc}{\calI_{\ell_1}}
\newcommand{\kappamu}{f_{\ell_1}(\mu)}

\newcommand{\Yk}{\tilde{Y}}
\newcommand{\indicat}[1]{\mathbbm{1}_{\left\{#1\right\}}}

\newcommand{\Ljc}{L^{\alpha}_{n}}

\newcommand{\calT}{\mathcal{T}}
\newcommand{\Tcondlaw}{\bm{T}|\bs{Z}}
\newcommand{\thetamuT}{f_n^{\calT}
(\mu)}

\newcommand{\thetamusT}{f_n^{\calT}(\mustar)}

\newcommand{\LjT}{L^{\alpha,\calT}_{n}}

\newcommand{\condedT}{\bs{Z},\bm{T}}

\newcommand{\bX}{\bs{X}}

\newcommand{\bUam}{\bs{\Gamma}}
\newcommand{\bOme}{\bs{\Omega}}
\newcommand{\bSig}{\bs{\Sigma}}
\newcommand{\bH}{\bs{H}}

\newcommand{\bU}{\bs{U}}

\newcommand{\bv}{\bs{v}}
\newcommand{\bA}{\bs{A}}
\newcommand{\bZ}{\bs{Z}}
\newcommand{\bT}{\bs{T}}
\newcommand{\projmat}{\bm{P}_{n-1}}
\newcommand{\Ib}{\mathbf{I}} 

\newcommand{\inner}[2]{\left\langle #1,#2 \right\rangle}
\newcommand{\norm}[1]{||#1||}
\newcommand{\rbr}[1]{\left(#1\right)} %% round brackets
\newcommand{\bmu}{\bm{\mu}}

\newcommand{\beq}{\begin{equation}}
\newcommand{\eeq}{\end{equation}}
\newcommand{\beqa}{\begin{eqnarray}}
\newcommand{\eeqa}{\end{eqnarray}}
\newcommand{\balg}{\begin{align}}
\newcommand{\ealg}{\end{align}}

\newcommand{\barr}{\begin{array}}
\newcommand{\earr}{\end{array}}

\newcommand{\blem}{\begin{lemma}}
\newcommand{\elem}{\end{lemma}}

\newcommand{\bthm}{\begin{theorem}}
\newcommand{\ethm}{\end{theorem}}

\newcommand{\bft}{\begin{fact}}
\newcommand{\eft}{\end{fact}}

\newcommand{\beg}{\begin{example}}
\newcommand{\eeg}{\end{example}}

\newcommand{\bpf}{\begin{proof}}
\newcommand{\epf}{\end{proof}}

\newcommand{\bprop}{\begin{proposition}}
\newcommand{\eprop}{\end{proposition}}
\newcommand{\bdf}{\begin{definition}}
\newcommand{\edf}{\end{definition}}
\newcommand{\nline}{\\ \nonumber}
\newcommand{\fmu}{f(\mu)}
\newcommand{\LB}{L^{\alpha}_{n}}
\newcommand{\Q}{Q}
\newcommand{\myq}{\mathfrak{q}}
\newcommand{\fucb}{f^{\text{UCB}}}

\graphicspath{{./fig/}}
\numberwithin{equation}{section}

\definecolor{darkred}{rgb}{0.6, 0.0, 0.0}
\usepackage[normalem]{ulem}
\definecolor{ejc}{RGB}{0,0,200}

\usepackage{authblk}

\newcommand{\acc}[1]{{\color{black}#1}}
\title{Floodgate: inference for model-free variable importance
}
\author{Lu Zhang}
\author{Lucas Janson}
\date{}
\affil{Department of Statistics, Harvard University}
\begin{document}
\maketitle

%possible method names: Atlas, Price is Right, safety net, hats-off, trivet, release valve, surge protector, floodgate, quarantine, blackjack, reservoir, splash guard, mudflap

\begin{abstract}
    Many modern applications seek to understand the relationship between an outcome variable $Y$ and a covariate $X$ in the presence of a (possibly high-dimensional) confounding variable $Z$. Although much attention has been paid to testing \emph{whether} $Y$ depends on $X$ given $Z$, in this paper we seek to go beyond testing by inferring the \emph{strength} of that dependence. We first define our estimand, the minimum mean squared error (mMSE) gap, which quantifies the conditional relationship between $Y$ and $X$ in a way that is deterministic, model-free, interpretable, and sensitive to nonlinearities and interactions. We then propose a new inferential approach called \emph{floodgate} that can leverage any working regression function chosen by the user (allowing, e.g., it to be fitted by a state-of-the-art machine learning algorithm or be derived from qualitative domain knowledge) to construct asymptotic confidence bounds, and we apply it to the mMSE gap. \acc{We additionally show that floodgate's accuracy (distance from confidence bound to estimand) is adaptive to the error of the working regression function.} We then show we can apply the same floodgate principle to a different measure of variable importance when $Y$ is binary. Finally, we demonstrate floodgate's performance in a series of simulations and apply it to data from the UK Biobank to infer the strengths of dependence of platelet count on various groups of genetic mutations.
    \medskip

\end{abstract}
\section{Introduction}\label{sec:introduction}
    \subsection{Problem Statement}\label{sec:setup}
    Scientists looking to better-understand the relationship between a response variable $Y$ of interest and a covariate $X$ in the presence of confounding variables $Z=(Z_1,\dots,Z_{p-1})$ often start by asking \emph{how important} $X$ is in this relationship. Although this question is sometimes simplified by statisticians to the binary question of `is $X$ important or not?', a more informative and useful inferential goal is to provide inference (i.e., confidence bounds) for an interpretable real-valued measure of variable importance (MOVI). The canonical approach of assuming a parametric model for $Y\mid X,Z$ will usually provide obvious MOVI candidates in terms of the model parameters, but the simple models for which it is known how to construct confidence intervals (e.g., low-dimensional or ultra-sparse generalized linear models) often provide at best very coarse approximations to the true $Y\mid X,Z$ (as evidenced by the marked predictive outperformance of nonparametric machine learning methods in many domains), resulting in undercoverage due to violated assumptions \emph{and} lost power due to insufficient capacity to capture complex relationships. This raises the motivating question for this paper: {\bf what is an interpretable, sensitive, and model-free measure of variable importance and how can we provide valid and narrow confidence bounds for it?}

\subsection{Our contribution}
The main contribution of this paper is to introduce \emph{floodgate}, a method for inference of the minimum mean squared error (mMSE) gap, which satisfies the following high-level objectives which we believe are fairly universal for the task at hand.
\setlength{\leftmargini}{3.5cm}
\begin{itemize}
    \item[{\bf (Sensitivity)}] The mMSE gap is strictly positive unless $\Ec{Y}{X,Z} \stackrel{a.s.}{=} \Ec{Y}{\Xnoj}$, and is large whenever $X$ explains a lot of the variance in $Y$ not already explained by $\Xnoj$ alone, making it sensitive to arbitrary nonlinearities and interactions in $Y$'s relationship with $X$. 
    \item[{\bf (Interpretability)}] The mMSE gap has simple predictive, explanatory, and causal interpretations for $Y$'s relationship with $X$, is a {functional} of \emph{only} the joint distribution of $(Y,X,Z)$, and is exactly zero when $\jnull$. 
    \item[{\bf (Validity)}] \acc{We first prove floodgate's asymptotic validity assuming the user knows the distribution of $X\mid Z$, but with essentially no other assumptions
     (in particular we require no smoothness, sparsity, or other constraints on $\Ec{Y}{X,Z}$ that would ensure its learnability at \emph{any} geometric rate). 
     However, to emphasize that the floodgate idea is not tied to such assumptions, we also provide a version of floodgate valid under double-robustness-type assumptions.}
%    \rev{Floodgate is asymptotically valid under extremely mild moment conditions, and in particular requires no smoothness, sparsity, or other constraints on $\Ec{Y}{X,Z}$ that would ensure its learnability at \emph{any} geometric rate. Floodgate requires the user to know the distribution of $X\mid Z$, although we show floodgate is not tied to that model-X assumption by providing another possibility: floodgate is also asymptotically valid under the double-robustness-type conditions.}
%    prove this requirement can sometimes be relaxed to only knowing a model for $X\mid Z$, and we theoretically and numerically characterize floodgate's robustness to misspecification of this distribution.}{should we adjust this now that we have the double-robustness results?} {}    
%    1. Although fg is not tied to model assumption
%    2. we provide two possibilities: validity under model-X assumption and DR assumption 

    \item[{\bf (Accuracy)}] Floodgate derives accuracy from flexibility by allowing the user to estimate $\Ec{Y}{X,Z}$ in whatever way they like, and we prove that the accuracy of inference is adaptive to the mean squared error (MSE) of that estimate. 
\end{itemize}
\setlength{\leftmargini}{1cm}

In a bit more detail, we (in Section~\ref{sec:method}) define the mMSE gap as an interpretable and model-free MOVI (Section~\ref{sec:MOVI}) and present a method, \emph{floodgate}, to construct asymptotic lower confidence bounds for it that provides the user absolute latitude to leverage any domain knowledge or advanced machine learning algorithms to make those bounds as tight as possible (Section~\ref{sec:mock}). We consider upper confidence bounds (Section~\ref{sec:ucb_hardness}), address computational considerations (Section~\ref{sec:computation}), theoretically characterize the width of floodgate's confidence bounds (Section~\ref{sec:accuracy}),
% \rev{\sout{and its robustness to model misspecification (Appendix~\ref{sec:robustness})}}, 
 and briefly address some immediate generalizations (Section~\ref{sec:generalizations}). 

We then proceed to extensions of floodgate (Section~\ref{sec:extensions}), first presenting an alternative MOVI that we can similarly construct asymptotic confidence bounds for when $Y$ is binary (Section~\ref{sec:class}). Second, we present a modification of floodgate that, for certain models, allows asymptotic inference even when $X$'s distribution is only known up to a parametric model (Section~\ref{sec:relax}) {and apply it to multivariate Gaussian (Section~\ref{sec:lowdim_Gaussian}) and discrete Markov chain (Section~\ref{sec:DMC}) covariate models.} %\lzmargin{}{not mention the gap result?}
%\lz{We also theoretically quantify the accuracy gap between the regular floodgate and this modified version for Gaussian and discrete Markov Chain covariate models (Sections~ \ref{sec:lowdim_Gaussian} and \ref{sec:DMC}).}

%\ljmargin{Third, we show that floodgate can easily be generalized to provide inference for distributions with different $X$ distribution than the data (Appendix~\ref{sec:transport}). And fourth, we explain how procedures from selective inference may be applied to floodgate in pursuit of various types of \emph{selective} coverage (Section~\ref{sec:fcr_adj}).}{perhaps get rid of this}

Finally we demonstrate floodgate's performance and support our theory with simulations (Section~\ref{sec:simul}) and an application to data from the UK Biobank (Section~\ref{sec:realdata}). We end with a discussion of the future research directions opened by this work (Section~\ref{sec:discussion}). All proofs are deferred to the appendix.
    
    \subsection{Related work}
    \label{sec:literature}
    \acc{Many existing works consider \emph{marginal} variable importance, i.e., not accounting for the presence of $Z$ in the relationship between $Y$ and $X$ \citep{hirschfeld1935connection,gretton2005measuring,gretton2007kernel,szekely2007measuring,szekely2013energy,heller2013consistent,shao2014martingale,wang2017generalized,chatterjee2021new,deb2021multivariate}, including some that measure that importance via differences in conditional means in a way resembling our mMSE gap \citep{shao2014martingale}. Such approaches address a very different statistical question, and so we focus our literature review on works that, like us, consider conditional variable importance.}
    
    The standard approach to conditional statistical inference in regression is to assume a parametric model for $Y\mid X,Z$, often a generalized linear model (GLM) or cousin thereof. With $Y\mid X,Z$ so parameterized, it is usually straightforward to define a parametric MOVI and a large body of literature is available to provide asymptotic inference for such parametric MOVIs (see, for example, \citet{buhlmann2013statistical,nickl2013confidence,zhang2014confidence,van2014asymptotically,javanmard2014confidence,buhlmann2015high,dezeure2017high,zhang2017simultaneous}). However, when the parametric $Y\mid X,Z$ model is misspecified even slightly, the associated parametric MOVI becomes ill-defined, reducing its interpretability. Furthermore, many $Y\mid X,Z$ models are too simple to capture or detect nonlinearities that may be present in real-world data sets.
    
    One approach to addressing the shortcomings of parametric inference is to generalize the parameters of common parametric models to be well-defined in a much larger nonparametric model class. For example, under mild moment conditions one can generalize the parameters in a linear model for $Y\mid X,Z$ as parameters in the least-squares \emph{projection} to a linear model of any $Y\mid X,Z$ distribution \citep{berk2013valid,taylor2014exact,buja2014discussion,buja2015models,rinaldo2016bootstrapping,lee2016exact,buja2019models1,buja2019models2}. Such a linear projection MOVI can be hard to interpret because it will in general have a non-zero value even when $\jnull$; see Appendix~\ref{sec:comparison} for a simple example. Another example of a generalized parameter is the expected conditional covariance functional $\EE{\Covc{Y}{\Xj}{\Xnoj}}$ (see, for example, \citet{robins2008higher, robins2009semiparametric, li2011higher, robins2017minimax, newey2018cross,shah2018hardness,VC-ea:2018,LL-RM-JR:2019,katsevich2020theoretical}), which represents a generalization of the linear coefficient in a \emph{partially} linear model. $\EE{\Covc{Y}{\Xj}{\Xnoj}}$ always equals zero when $\jnull$, but it shares the shortcoming of linear projection MOVIs that it lacks sensitivity to capture nonlinearities or interactions in $Y$'s relationship with $X$. That is, both MOVIs mentioned in this paragraph will assign any non-null variable that influences $Y$ nonlinearly or through interactions with other covariates a value that can severely underrate that variable's true importance, and can even assign a variable the MOVI value zero when $Y$ is a deterministic non-constant function of it.
    
    A second approach has been to infer model-free MOVIs defined through machine learning algorithms fitted to part of the data itself \citep{lei2018distribution,fisher2018model,watson2019testing}. 
    %defines a MOVI called LOCO as the difference in the prediction error of a machine learning algorithm fitting $Y$ to $X$ and that of a machine learning algorithm fitting $Y$ to $\Xnoj$. 
    By leveraging the expressiveness of machine learning, such a MOVI can be made sensitive to nonlinearities and interactions but is itself \emph{random} and depends both on the data and the choice of machine learning algorithm. This poses a challenge for interpretability and in particular for replicability, since even \emph{identical} analyses run on two independent data sets that are \emph{identically-distributed} will provide inferences for \emph{different} MOVI values.% Furthermore, like linear projection parameters, important variables may 
    
    % \rev{review kernel-based papers: }
%     Kun Zhang, Jonas Peters, Dominik Janzing, and
% Bernhard Scholkopf, ”Kernelbased conditional independence test and application in causal discovery,” 
    
    Another line of work \citep{castro2009polynomial,vstrumbelj2014explaining,owen2017shapley,lundberg2020local,covert2020understanding,williamson2020efficient} considers MOVIs based on the classical form of the Shapley value \citep{shapley1953value,charnes1988extremal}, which in general assigns a non-zero MOVI value to covariates $X$ with $\jnull$, making it hard to interpret its value mechanistically or causally (though it has some appealing properties for a \emph{predictive} interpretation).

    An interesting new proposal for a model-free MOVI was made in \citet{azadkia2019simple}. Their MOVI has the distinction that it equals zero if and only if $\jnull$ \emph{and} it attains the maximum value $1$ if $Y$ is almost surely a measurable function of $\Xj$ given $\Xnoj$. \acc{More recently, \citet{huang2020kernel} proposed a larger class of MOVIs satisfying the same properties. However, both papers focus on consistent estimators and do not provide confidence bounds for their MOVIs.}% and do not provide methods for inference (confidence lower- or upper-bounds).}

    As we will detail in Section~\ref{sec:MOVI}, the MOVI we provide inference for, the mMSE gap, does not suffer from the drawbacks of the MOVIs described in the previous paragraphs, and indeed the same MOVI has been considered before. In the sensitivity analysis literature it is called the ``total-effect index" \citep{saltelli2008global} but to our knowledge its inference (confidence lower- or upper-bounds) is not considered there. In one of the Shapley value papers \citep{covert2020understanding} a generalization of the mMSE gap is used as the input to the Shapley value calculation, but again inferential results (for the mMSE gap or its Shapley version) are not considered in that paper. \acc{Otherwise, \citet{williamson2017nonparametric} appears to be the first to consider inference for the mMSE gap (this inference is then used with neural networks in \citet{feng2018nonparametric}), but the asymptotic normality theory their coverage guarantee relies on fails at the boundary of the parameter space, i.e., the important case of when the mMSE gap is zero, or the variable is unimportant. A recent follow-up work \citep{williamson2020unified} addresses this limitation by combining estimators on two disjoint subsets of the data (though their inference still requires the \emph{group} mMSE gap of the entire covariate vector to be positive). Our different approach avoids altogether this issue when the mMSE gap is zero so that our inference is valid for any value of the mMSE gap (group or otherwise), and although we also use data splitting, we do so in a way that seems to lead to significantly reduced variance (and hence more accurate inference) relative to \citet{williamson2020unified}, as we show in Section~\ref{sec:comp_vimp}.}
    
    \subsection{Notation}
    For two random variables $A$ and $B$ defined on the same probability space, let $P_{A\,|\,B}$ denote the conditional distribution of $A\mid B$. Denote the $(1-\alpha)$th quantile of the standard normal distribution by $z_\alpha$. \acc{Let $\chi^2\left(P\|Q\right)$ denote the $\chi^2$ divergence $\int_{\Omega}(\frac{d P}{d Q}-1)^2 dQ$ between two distributions $P,Q$ on the probability space $\Omega$.} Let $[n]$ denote the set $\{1,\dots,n\}$. 
    
    \section{Methodology}
    \label{sec:method}
    %MOVI names: variance gap, EVE, predictive gap, prediction drop, explanatory gap, MSE gap
    \subsection{Measuring variable importance with the mMSE gap}\label{sec:MOVI}
    We begin by defining the MOVI that we will provide inference for in this paper.
    \begin{definition}[Minimum mean squared error gap]
    \label{def:anova_movi}
    The \emph{minimum mean squared error (mMSE) gap} for variable $X$ is defined as 
    \begin{equation}
        \label{eq:msegap}
        \Ij^2 = \EE{\left(Y - \condmeanj\right)^2} - \EE{\left(Y - \condmean\right)^2}
    \end{equation}
    whenever all the above expectations exist.
    \end{definition}
    We will at times refer to either $\Ij^2$ or $\Ij$ as the mMSE gap when it causes no confusion. Although the same MOVI has been used before (see Section~\ref{sec:literature}), we provide here a number of equivalent definitions/interpretations which we have not seen presented together before.
    \begin{itemize}
        \item Equation~\eqref{eq:msegap} has a direct \emph{predictive} interpretation as the increase in the achievable or minimum MSE for predicting $Y$ when $X$ is removed.
        \item The mMSE gap can also be interpreted as the decrease in the \emph{explainable variance} of $Y$ without $X$:
        \begin{equation} \label{eq:mMSEgap_form2}
        \Ij^2 = \Var{\condmean} - \Var{\condmeanj}.
        \end{equation}
        \item When $X$ is viewed as a treatment level for $Y$ and $Z$ is a set of measured confounders, $\Ij$ can be seen as an \emph{expected squared treatment effect}:
        \begin{equation}
	   %   \mathcal{I}^2 = \frac{1}{2}\EE{\left(\Ec{Y}{X^{(1)},Z}-\Ec{Y}{X^{(2)},Z}\right)^2},
	      \Ij^2 = \frac{1}{2}\Ep{x_1,x_2,Z} {\left(\Ec{Y}{X=x_1,\Xnoj}-\Ec{Y}{X=x_2,\Xnoj}\right)^2}.
	        \end{equation}
	   % where $X^{(1)}$ and $X^{(2)}$ are independently drawn from $P_{X\mid Z}$. 
	    where $x_1$ and $x_2$ are independently drawn from $P_{X\mid Z}$ in the outer expectation. 
	    %{\color{red}I like the WeSCATE connection, but I'm worried it might distract the reader to spend too much time on causal stuff... let's think about this.}
	    \item \acc{We can also rewrite the mMSE gap as:
	    \begin{equation}\label{eq:mMSEgap_mmd}
	        \Ij^2  = \EE{(\condmeanj - \condmean)^2}
	    \end{equation}
	   % }{Is there a reason you put $\condmean$ before $\condmeanj$? We use the opposite order in defining the MACM gap and in the equivalent def of $\Ij$ when we included it in section 3 as well, and hence the current version looks slightly different from what the reviewer suggested, which is not ideal.}
	    and interpret $\Ij$ as the $\ell_2$ distance between the two regression functions $\condmeanj$ and $\condmean$.
	    }
        \item Lastly, we remark that $\Ij^2$ also admits a very compact (if less immediately interpretable) expression:
        \begin{equation}\label{eq:mMSEgap_compact}
            \Ij^2 = \ANOVA.
        \end{equation}
    \end{itemize}
    
    In light of these multiple alternative expressions, we find the mMSE gap remarkably interpretable. Note that it only requires the existence of some low-order conditional and unconditional moments of $Y$ to be well-defined, and its value is invariant to any fixed translation of $Y$ and to the replacement of $X$ or $Z$ by any fixed bijective function of itself. Furthermore, the mMSE gap is zero if and only if $\condmean \stackrel{a.s.}{=}\condmeanj$, and in particular it is exactly zero when $\jnull$ and strictly positive if $\condmean$ depends at all on $X$, allowing it to fully capture arbitrary nonlinearities and interactions in $\condmean$. 
%    \ljmargin{\rev{The idea of using the differences in the conditional means to measure variable importance has been  
    
    Note that $\Ij$ has the same units as $Y$, which can help interpretation when $Y$'s units are meaningful (much like it does for the average treatment effect in causal inference). However, if a unitless quantity is preferred, such as for comparison between MOVIs across $Y$s with different units, we can also measure variable importance by and extend our methodology to a standardized version of $\Ij^2$, namely,  $\Ij^2/\Var{Y}$. In fact, with \acc{some} more work, we can even extend our inferential results to a version of the mMSE gap which is invariant to transformations of $Y$, \acc{or versions that are zero if \emph{and only if} $\jnull$; see Section~\ref{sec:generalizations} and Appendix \ref{sec:iff_movi} for details, with Appendix~\ref{sec:kernelfg} extending our results to the kernel partial correlation of \citet{huang2020kernel}.}
	  
    \subsection{Floodgate: asymptotic lower confidence bounds for the mMSE gap}
    \label{sec:mock}
    As can be seen by Equation~\eqref{eq:mMSEgap_compact}, the mMSE gap is a nonlinear functional of the true regression function $\mustar(x,z):=\Ec{Y}{X=x,Z=z}$. Hence if we had a sufficiently-well-behaved estimator $\muhat$ for $\mustar$ (e.g., asymptotically normal or consistent at a sufficiently-fast geometric rate), there would be a number of existing tools in the literature (e.g., the delta method, influence functions) that we could use to provide inference for the mMSE gap. But such estimation-accuracy assumptions are only known to hold for a very limited class of regression estimators, and in particular preclude most modern machine learning algorithms and methods that integrate hard-to-quantify domain knowledge, which are exactly the types of powerful regression estimators we would most like to leverage for accurate inference.
    
    However, given the centrality of $\mustar$ in the definition of the mMSE gap, it seems we need to at least implicitly estimate it with some \revision{working regression function} $\mu$. And even if we avoid assumptions on $\mu$'s accuracy, if we want to provide rigorous inference then we ultimately still need \emph{some} way to relate $\mu$ to $\Ij$, which is a function of $\mustar$. We address this issue in the context of constructing a lower confidence bound (LCB) for the mMSE gap. The key idea proposed in this paper is to use a functional, which we call a \emph{floodgate}, to relate \emph{any} $\mu$ to $\Ij$. In particular, we will shortly introduce a $\thetamu$ such that for \emph{any} $\mu$, 
    \begin{itemize}
        \item[(a)] $\thetamu\le\Ij$
        \item[(b)] we can construct a lower confidence bound $L$ for $\thetamu$.
    \end{itemize}
    Then by construction $L$ will also constitute a valid LCB for $\Ij$. The term \emph{floodgate} comes from metaphorically thinking of constructing a LCB as preventing flooding \acc{($L>\Ij$, i.e., miscoverage)} by keeping the water level ($L$) below a critical threshold ($\Ij$) under arbitrary weather conditions ($\mu$\acc{, or more specifically, $\mu$'s error, which we may not expect to be able to control well}). Then by controlling $L$ below $\Ij$ for any $\mu$, $f$ acts as a floodgate, and we also use the same name for the inference procedure we derive from $f$.

    In particular, for any (nonrandom) function $\mu:\mathbb{R}^p\rightarrow\mathbb{R}$, define
    \begin{equation}\label{eq:thetaj_mu}
        \thetamu := 
        \frac{\EE{\covc{\mustar(X,Z)}{\muX}{Z}}}{\sqrt{\EE{\varc{\muX}{Z}}}}, 
    \end{equation}
    where by convention we define $0/0=0$ so that $\thetamu$ remains well-defined when the denominator of \eqref{eq:thetaj_mu} is zero. It is not hard to see that $f$ tightly satisfies the lower-bounding property (a) and we formalize this in the following lemma which is proved in Appendix~\ref{pf:lem:max}. 
    \begin{lemma}
    \label{lem:max}
    For any $\mu$ such that $\thetamu$ exists, $\thetamu\le \Ij$, with equality when $\mu=\mustar$.
    \end{lemma}
    \acc{In order to establish property (b) of $f$, we first take a \emph{model-X} approach \citep{LJ:2017,candes2016panning}: we assume we know $P_{X|Z}$ but avoid assumptions on $Y\mid X,Z$.} \acc{We start with such a model-X assumption because its simplicity helps elucidate the key ideas underlying the floodgate method, but floodgate is not tied to such assumptions, and indeed we present alternative versions of floodgate that operate under different assumptions later in the paper (Section~\ref{sec:relax}'s version somewhat relaxes the assumed knowledge of $P_{X|Z}$ without requiring any new assumptions and Remark~\ref{rk:modelx}'s version relies on a \emph{double-robust} set of assumptions). That said, the model-X assumption is sometimes reasonable and has been used before in a number of applications (see Appendix~\ref{sec:modelX} for elaboration and examples), including in genomics like in the application presented in Section~\ref{sec:realdata}, and we theoretically (Appendix \ref{sec:robustness}) and numerically (Section \ref{sec:simul_robust}) characterize model-X floodgate's robustness to misspecification of $P_{X|Z}$.}
    % in Section~\ref{sec:relax} that this assumed knowledge of $P_{X|Z}$ can be relaxed somewhat without sacrificing floodgate's statistical properties, and in Remark~\ref{rk:modelx} we present another version of floodgate with similar properties under a \emph{double robust} set of assumptions.
 		%\rev{At the same time, we shall highlight that floodgate is not tied to the model-X assumption. Later in Remark \ref{rk:modelx}, we present another possibility of floodgate which works under the doubly-robustness-type conditions. Moreover, we show the model-X assumption is reasonable in many real-world applications (Appendix \ref{sec:modelX}) and prove this assumption can sometimes be relaxed to only knowing a model for $X\mid Z$ (Section \ref{sec:relax}), and we theoretically (Appendix \ref{sec:robustness}) and numerically (Section \ref{sec:simul_robust}) characterize floodgate's robustness to misspecification of $P_{X|Z}$.}     
%    Floodgate is not tied to model-X assumption. mention DR; mention the three points on model-X
%    We will show that such an assumption is reasonable in many real-world applications (see some discussion in Appendix) and discuss how the floodgate method is robust to misspecification (details can be found in Appendix XXX) 
%    
%    and allows a doubly-robust approximation; see more detailed comments about this assumption in Remark \ref{rk:modelx}.    
%    model-X assumption is plausible -> doubly-robustness result -> theoretical and empirical study about robustness
    \acc{Knowing $P_{X|Z}$ and $\mu$ means that, given data $\{(X_i,Z_i,Y_i)\}_{i=1}^n$, we also know $\{V_i:=\varc{\mu(X_i,Z_i)}{Z_i}\}_{i=1}^n$ which are i.i.d. and unbiased for the squared denominator in \eqref{eq:thetaj_mu}. And if we rewrite the numerator as
%    In order to estimate the numerator, we first rewrite it as 
    \begin{equation}\label{eq:fnum} \EE{\covc{\mustar(X,Z)}{\muX}{Z}}=\EE{Y\big(\muX-\Ec{\muX}{Z}\big)}, 
    \end{equation}
    %{\color{red}\[  \EE{\covc{\muX}{\mustar(X,Z)}{Z}}=\EE{\big(Y-\Ec{\muX}{Z}\big)\big(\muX-\Ec{\muX}{Z}\big)}, \]}
    then we see we also know
    %and then recognize that we also know
    $\{R_i:=Y_i\big(\mu(X_i,Z_i)-\Ec{\mu(X,Z_i)}{Z_i}\big)\}_{i=1}^n$ which are i.i.d. and unbiased for the numerator. Thus for any given $\mu$, we can use sample means of $R_i$ and $V_i$ to asymptotically-normally estimate both expectations in Equation~\eqref{eq:thetaj_mu}, and then combine said estimators through the delta method to get an estimator of $\thetamu$ whose asymptotic normality facilitates an immediate asymptotic LCB. This strategy is spelled out in Algorithm~\ref{alg:MOCK} and Theorem~\ref{thm:main} establishes its asymptotic coverage.} 
%    \ljmargin{}{It seems to me this should be yet another remark and go after the theorem. For instance, I don't think it's a more important point than the double-robustness or high-dimensional extension}} move to the RK of choice of $\mu$
\acc{We pause to mention a simple but important point: when $\mu(X, Z)$ does not depend on $X$ at all, then $f(\mu)=0$ and all the $V_i$ and $R_i$ are zero with probability 1, making floodgate's LCB computed in Algorithm~\ref{alg:MOCK} deterministically zero as well.
%equal zero as mentioned previously, and so does the output LCB from Algorithm~\ref{alg:MOCK}. 
This implies that when the regression algorithm for obtaining $\mu$ is sparse, in the sense that it only depends on a fraction of its inputs, then floodgate will produce LCBs of zero for many of the covariates. For those covariates, coverage will hold \emph{deterministically}, and hence floodgate will have average coverage even higher than the nominal $1-\alpha$, as observed in some simulations in Section \ref{sec:simul}.}
%\lzmargin{When $\mu(X, Z)$ does not depend on $X$, $f(\mu)$ equals zero as mentioned previously, and so does the output LCB from Algorithm~\ref{alg:MOCK}. This implies that when the regression algorithm for obtaining $\mu$ conducts variable selection, our floodgate procedure will output zero LCBs for those unselected variables. In such situations, we will expect an over-coverage, as observed in some simulations in Section \ref{sec:simul}.}{carefully check}
\begin{algorithm*}[htbp]
\caption{Floodgate}\label{alg:MOCK}
\begin{algorithmic}%[1]
\REQUIRE Data $\{(Y_i,X_i,Z_i)\}\nsubp$, $P_{X\mid \Xnoj}$, a \revision{working regression function} $\mu:\mathbb{R}^p\rightarrow\mathbb{R}$, and a confidence level $\alpha \in (0,1)$.
\vspace{0.05cm}
\STATE Compute $R_{i} = Y_i\big(\mu(X_i,Z_i)-\Ec{\mu(X_i,Z_i)}{\Xinoj}\big)$ and $V_{i} = \Varc{\mu(X_i,Z_i)}{\Xinoj}$ for each $i\in [n]$, and their sample mean $(\bar{R},\bar{V})$ and sample covariance matrix $\hat{\Sigma}$, and compute $s^2 = \frac{ 1 }{ \bar{V} }\left[ \left(\frac{\bar{R} }{2 \bar{V} }\right)^2 \hat{\Sigma}_{22} +  \hat{\Sigma}_{11} - \frac{\bar{R} }{ \bar{V} } \hat{\Sigma}_{12} \right].$
\ENSURE Lower confidence bound $L_{n}^{\alpha} (\mu)=\max\left\{\frac{\bar{R} }{ \sqrt{\bar{V}}}
- \frac{z_{\alpha}s}{\sqrt{n}},\,0\right\}$, with the convention that $0/0=0$.
\end{algorithmic}
\end{algorithm*}

%\ljmargin{
\acc{    
\begin{theorem}[Floodgate validity]\label{thm:main}
    For any given working regression function $\mu:\mathbb{R}^{p} \rightarrow \mathbb{R}$ and i.i.d. data $\{(Y_i,X_i,Z_i)\}_{i=1}^n$,   if $\ee{Y^{4}},~\ee{\mu^{4}(X,Z)} <\infty$, then $L_{n}^{\alpha} (\mu)$ from Algorithm~\ref{alg:MOCK} satisfies 
    % \lzmargin{
    %  $\mathbb{P}\left(L_{n}^{\alpha} (\mu) \le \thetamu \right) \ge 1 - \alpha - O(n^{-1/2})$, which combined with Lemma~\ref{lem:max} immediately establishes}{}
	\begin{equation} 
	\nonumber
        \liminf_{n\rightarrow \infty} \mathbb{P}\left(L_{n}^{\alpha} (\mu) \le \Ij \right) \ge 1 - \alpha.
	\end{equation}
\end{theorem}
}
%    }{I think this theorem statement changed, so (at least parts of) it should be colored, right? Can you please do a pass of the paper to make sure we didn't miss anywhere else that we changed that should be colored?}
\acc{
The proof of Theorem~\ref{thm:main} can be found in Appendix~\ref{pf:thm:main}.
%\ljmargin{When $\mu(X, Z)$ does not depend on $X$, $f(\mu)$ equals zero as mentioned previously, and so does the output LCB from Algorithm~\ref{alg:MOCK}. This implies that when the regression algorithm for obtaining $\mu$ conducts variable selection, our floodgate procedure will output zero LCBs for those unselected variables. In such situations, we will expect an over-coverage, as observed in some simulations in Section \ref{sec:simul}.}{this feels out of place---maybe we should only mention this in the simulations section when it happens}
Fourth moments (as opposed to the usual second moments for the CLT) are required because the estimand itself involves the expectations of $Y\mu(X, Z)$ and $\mu^2(X, Z)$.
%involving product of $Y$ and $\mu(X, Z)$ in the central limit theorem argument. 
With higher moment conditions, we can apply relatively recent Berry--Esseen-type results for the delta method \citep{pinelis2016optimal} to strengthen the pointwise asymptotic coverage of Theorem~\ref{thm:main} to have a rate of $n^{-1/2}$; see Appendix \ref{sec:rate} for details. We note that in both Algorithm~\ref{alg:MOCK} and Theorem~\ref{thm:main}, $Y$ can be everywhere replaced by $Y-g_0(Z)$ for any non-random function $g_0$ (e.g., $\Ec{\mu(X,Z)}{\Xnoj=z}$ would be a natural choice), which can reduce the variance of the $R_i$ terms and hence improve the LCB.} %{I would put this second} 
%\ljmargin{Such 4th moment conditions are required since we work 

\acc{
\begin{remark}[Doubly robust floodgate]\label{rk:modelx}
Although for ease of exposition we have presented Algorithm~\ref{alg:MOCK} and Theorem~\ref{thm:main} under the model-X assumption that $P_{X|Z}$ is known exactly, we emphasize here that the underlying idea of floodgate is \emph{not} tied to this assumption. To reiterate, the key conceptual contribution of this paper is to introduce a lower-bounding functional $f(\mu)$ for $\Ij$ such that $f(\mu)$ provides a tractable statistical target to obtain a LCB for. To underscore this point, we present here a version of floodgate following the same principle but that is valid under standard double-robust assumptions instead of the aforementioned model-X assumption.
\acc{
Consider the following functional that depends not only on a working regression function $\mu(x,z)$,
%(an estimate of $\Ec{Y}{X=x,Z=z}$), 
but also some $Q_y$ estimating the true 
%$Y\mid Z$ distribution (
$P_{Y\mid Z}$ and some $Q_x$ estimating the true $P_{X\mid Z}$:
%Note that the floodgate idea is not tied to the model-X assumption: the deterministic lower-bounding property of the floodgate functional $f(\mu)$ (property (a)) does not rely on any modeling assumptions; constructing the LCBs for $f(\mu)$ (property (b)) does not necessarily require the model-X assumption.
    \begin{equation}\label{eq:approxfg}
        f_{ Q_y, Q_{x}}(\mu) := \frac{\EE{(Y - \Epc{Q_y}{Y}{Z} )(\mu(X, Z) - \Epc{Q_x}{\mu(X, Z)}{Z} )}}{\sqrt{\EE{(\mu(X, Z) - \Epc{Q_x}{\mu(X, Z)}{Z})^2}}},
    \end{equation}   
%    \lzmargin{}{E for true distribution}
    where 
    %we abbreviate $Q_{Y\mid Z}, Q_{X\mid Z}$ as $Q_y, Q_x$ respectively and 
    $\mathbb{E}_{Q_x}$ (resp. $\mathbb{E}_{Q_y}$) denotes expectation with respect to $Q_x$ (resp. $Q_y$) as opposed to the true data-generating distribution, and by convention we again define $0/0=0$. 
    %Note the above expectations without subscripts are taken with respect to the true distribution. 
    Given $Q_y, Q_x$, and $\mu$, i.i.d. unbiased estimates analogous to $R_i$ and $V_i$ in Algorithm~\ref{alg:MOCK} of
    the numerator and squared denominator, respectively, of $f_{ Q_y, Q_x}(\mu)$ 
    can be computed from each data point 
    %can be unbiasedly and asymptotically normally estimated purely with the observed data, 
    under no assumptions whatsoever, 
    thus allowing the exact same kind of LCB as in Algorithm~\ref{alg:MOCK} to be computed for $f_{ Q_y, Q_x}(\mu)$.
    %so $f_{ g_y, g_{x}}(\mu)$ does indeed provide a tractable statistical target to obtain a LCB for. 
    It now just remains to check that $f_{ Q_y, Q_x}(\mu)$ lower-bounds $\Ij$. 
    %We quantify the approximation error as below.  
    \begin{custlem}{2.3}
    \label{lem:doublerobust}
    For any $\mu, Q_y, Q_x$ such that $Q_x$ is absolutely continuous with respect to $P_{X\mid Z}$ and $f_{ Q_y, Q_x}(\mu)$ exists, we have that $ f_{Q_y, Q_x }(\mu) \le \Ij + \Delta$, where
%    \begin{equation} \nonumber
%       f_{g_y, g_x}(\mu) \le f(\mu)  +  \Delta \le \Ij + \Delta,
%    \end{equation}        
% $$f_{Q_y, Q_x }(\mu) \le f(\mu)  +  \Delta \le \Ij + \Delta,$$
    % where 
%    \begin{equation}
%    \Delta = \label{eq:DRDelta_v1}
%        \sqrt{\EE{( \Ec{Y}{Z} - \Epc{Q_y}{Y}{Z} )^2} } \sqrt{\EE{\left(\frac{h(W)}{\sqrt{\EE{h^2(W)}}} \right)^2\chi^2\left(Q_{X\mid Z}\|P_{X\mid Z}\right) }}. 
%    \end{equation}
     \begin{align} \label{eq:DRDelta_v1}
    \Delta =
        \sqrt{\EE{( \Ec{Y}{Z} - \Epc{Q_y}{Y}{Z} )^2} \EE{
        %\left(\frac{h(X,Z)}{\sqrt{\EE{h^2(X,Z)}}}\right)^2
        w_{\mu}(X,Z)
        \chi^2\left(Q_x\|P_{X\mid Z}\right) }}
    \end{align}
    and \acc{$w_\mu(X,Z) = \frac{ (\mu(X,Z) - \Ec{\mu(X,Z)}{Z})^2 }{\EE{ (\mu(X,Z) - \Ec{\mu(X,Z)}{Z})^2 }} $ is non-negative, has mean 1, and does not depend on $Q_y$ or $Q_x$, and we again define $0/0 =0$.} Furthermore, $f_{Q_y, P_{X\mid Z} }(\mu) = f(\mu)$ and thus $f_{Q_y, P_{X\mid Z} }(\mustar) = \Ij$ (for any $Q_y$).
    \end{custlem}
    The proof can be found in Appendix \ref{pf:lem:doublerobust}. Lemma \ref{lem:doublerobust} 
    says that $f_{ Q_y, Q_x}(\mu)$ only fails to lower-bound $\Ij$ to an extent bounded by the square root of the product of two terms: the MSE of $\Epc{Q_y}{Y}{Z}$ and the weighted $\chi^2$ error of $Q_x$. The same result also holds if we move $w_\mu(X,Z)$ in Equation~\eqref{eq:DRDelta_v1} from the second term to the first term; see Equation~\eqref{eq:DRDelta_v2}.
    As the first term measures the error in modeling $Y\mid Z$ and the second term measures the error in modeling $X\mid Z$,
    the square root of their product $\Delta$ is exactly what we would expect to be bounded as $o(n^{-1/2})$ under standard double-robustness assumptions (see, e.g., \citet{chernozhukov2018double}). And indeed, since the LCB for $f_{Q_y, Q_x }(\mu)  $ will be $\Omega(n^{-1/2})$ below $f_{Q_y, Q_x }(\mu) $, $\Delta=o(n^{-1/2})$ implies asymptotic coverage exactly as in Theorem~\ref{thm:main}. }
%    is essentially a double-robustness type result, showing that the floodgate method with an approximation of $\thetamu$ is valid as long as the product of the root MSE terms is $o(n^{-1/2})$.    
%    the model-X assumption can be effectively weakened when one can obtain a reasonably good estimate of $\Ec{Y}{Z=z}$ i.e., the MSE term $\EE{(g_0(Z) - \Ec{\mustar(X, Z)}{Z})^2}$ is small.    
%    We close this remark by mentioning that the deterministic lower-bounding property of the floodgate functional (property (b)) does not rely on any modeling assumption on $Y\mid X$ or $X\mid Z$; but there exist different ways to construct lower confidence bounds for $f(\mu)$.

\end{remark}
\begin{remark}[Floodgate's validity in high dimensions]\label{rk:highdim}
%Now we discuss the applicability of floodgate in the presence of high-dimensional covariates. 
Again for ease of exposition, Theorem~\ref{thm:main} establishes floodgate's pointwise asymptotic coverage for a fixed $\mu$ and a fixed (and hence fixed-dimensional) distribution for $(Y,X,Z)$. 
%In Appendix~\ref{sec:rate}, we 
It is certainly of interest to also consider the high-dimensional regime
%In the following, we consider the regime 
where the data-generating distribution (including the covariate dimension $p$) and the working regression function $\mu$ both depend on $n$, but it turns out that this setting is actually not very different from the simpler setting of Theorem~\ref{thm:main}.
%grows with $n$ and the working regression function $\mu_{n,p}$ is fitted from high-dimensional data. 
%First, the property (a) holds regardless of dimension: $f(\mu_{n,p}) \le \Ij$. 
To see this,
%when floodgate remains valid in this regime, 
first note that Theorem~\ref{thm:main} relies only on Lemma~\ref{lem:max} ($f(\mu)\le \Ij$) and a central limit theorem (CLT) applied to the 2-dimensional mean of the i.i.d. pairs $(R_i,V_i)$. But Lemma~\ref{lem:max} is non-asymptotic, and hence $f(\mu)\le \Ij$ still holds even if $\mu$ varies with $n$. And the pairs $(R_i,V_i)$ remain i.i.d. and 2-dimensional even as $\mu$ and the distribution of $(Y,X,Z)$ vary with $n$, so all that is needed for floodgate's validity is a 2-dimensional i.i.d. triangular array CLT, which only requires that the 2-dimensional random variables $(R_i,V_i)$ remain ``well-behaved". In Appendix~\ref{sec:rate} we show 
%just such a result
in fact an even stronger \acc{(non-asymptotic)} result, 
which, similarly to Theorem~\ref{thm:main}, 
only requires certain moments of $Y$ and $\mu(X, Z)$ to remain bounded
%}{there is no asymptotics in the Appendix~\ref{sec:rate} results, should we change the wording?}
%only requires higher-order moments of $|Y|$ and $|\mu|$ to remain bounded (
(although the result in Appendix~\ref{sec:rate} requires a bound on higher moments than Theorem~\ref{thm:main} so that recent Berry--Eseen-type results for the delta method can be applied to bound floodgate's undercoverage at a rate of $n^{-1/2}$). 
%\lzmargin{}{can we use absolute moments of $Y$ and $\mu$ instead of moments of $|Y|$ and $|\mu|$?}
%\lzmargin{}{Y is bounded means Y is bounded from both below and above. Can we just say $Y$ is a bounded random variable instead of ``$|Y|$ is a bounded random variable"}
In fact, it is even sufficient to replace the bound on $\mu(X,Z)$'s absolute moment with a bound on that of its conditional residual $h(X,Z):= \mu(X,Z)-\Ec{\mu(X,Z)}{Z}$. Note that $h$ only really measures the contribution from the single covariate $X$ to the whole working regression function $\mu$, even when $Z$ is high-dimensional. Hence, we believe that assuming that $Y$'s and $h(X,Z)$'s moments do not explode, even in high dimensions (recall $Y$ and $h(X,Z)$ remain 1-dimensional regardless of the dimension of the data), seems quite mild in practice.
%Then it suffices to check the validity of the LCB (property (b)). 
%Second, note that Theorem~\ref{thm:main} relies only on $f(\mu)\le \Ij$ and a two-dimensional central limit theorem (CLT) applied to the means of the i.i.d. pairs of random variables 
%Note the LCB in Algorithm~\ref{alg:MOCK} is based on 
%$Y_i (\mu_{n,p}(X_i,Z_i)-\Ec{\mu_{n,p}(X_i,Z_i)}{\Xinoj})$ and $\Varc{\mu_{n,p}(X_i,Z_i)}{\Xinoj}$ which are all one-dimensional i.i.d. random variables, regardless of $n$ and $p$. 
%Appendix \ref{sec:rate} establishes non-asymptotic results, i.e., proves coverage validity with a $n^{-1/2}$ rate based on recent Berry-Esseen-type bounds for non-linear statistics. Applying such results to the above regime only requires finiteness of some higher moments and the constant in that $n^{-1/2}$ rate only depends on the moments of $Y$ and $h_{n,p}(X,Z):=\mu_{n,p}(X_i,Z_i)-\Ec{\mu_{n,p}(X_i,Z_i)}{\Xinoj}$. Such moment assumptions are reasonable in high dimension. 
% for given $\mu_{n,p}$, where the constant only depends on the moments of $Y$ and $h_{n,p}$.
%For example, the higher moment of $h_{n,p}(X,Z)$ can still be reasonably bounded even $\mu_{n,p}(X,Z)$ is fitted from a dense high dimensional model. 
For instance, if $|Y|$ is a bounded random variable (as it often will be in practice), then as long as $\mu$ is winsorized at some level (which, as long as the level is at least as large as $|Y|$'s bound, can only improve $\mu$'s performance) \citep{rinaldo2016bootstrapping}, then floodgate's asymptotic validity is automatically ensured in the most general high-dimensional regime. Even when $Y$ is unbounded, we would usually not expect the moments of $Y$ or $h(X,Z)$ to diverge.
% }{I softened this a bit---it used to say we wouldn't expect this ``except in adversarial settings''}.
Indeed in Section~\ref{sec:simul_vary_p} we conduct high-dimensional simulations with unbounded $Y$ and $\mu$ fitted via various parametric and nonparametric machine learning algorithms, yet floodgate's coverage remains empirically valid regardless of the dimension.
		\end{remark}

\begin{remark}[Choosing $\mu$]\label{rk:mu}
\acc{
The final missing piece in our LCB procedure is the choice of $\mu$. In terms of how to obtain a working regression function $\mu$, the flexibility of our procedure thus far finally pays off: $\mu$ can be chosen in \emph{any} way that does not depend on the data used for inference. Normally we expect this to be achieved through data-splitting, i.e., a set of data samples is divided into two independent parts, and one part is used to produce an estimate $\mu$ of $\mustar$ while floodgate is applied to the other part with input $\mu$; we will explore this strategy in simulations in Section~\ref{sec:simul}. But in general, $\mu$ can be derived from any independent source, including mechanistic models or data of a completely different type than that used in floodgate (see, for example, \citet{bates2020causal} for an example of using a regression model fitted to a separate data set in the context of variable selection). The goal is to allow the user as much latitude as possible in choosing $\mu$ so that they can leverage every tool at their disposal, including modern machine learning algorithms and qualitative domain knowledge, to get as close to $\mustar$ as possible. We show in Section~\ref{sec:accuracy} that there is a direct relationship between the accuracy of $\mu$ and the accuracy of the resulting floodgate LCB.

In fact, an interesting and surprising feature of floodgate (both $f$ and Algorithm~\ref{alg:MOCK}) is that it is invariant to certain transformations of $\mu$, making floodgate work well even sometimes when $\mu$ is quite far from $\mustar$. In particular,
%invariant to two aspects of $\mu$:
%\begin{itemize}
%    \item[(i)] floodgate is invariant to any additive term in $\mu$ that depends only on $Z$,
%    \item[(ii)] floodgate is invariant to any positive global constant multiplying $\mu$.
%\end{itemize}
%This means that 
everything about floodgate remains identical if $\mu$ is replaced by any member of the set $S_\mu=\{c\mu(\cdot,\cdot)+g(\cdot,\cdot):c>0,\; g(x,\cdot)=g(x',\cdot)\, \forall x,x'\}$.
An immediate consequence is that if $\mu$ is a partially linear function in $x$, i.e., $\mu(x,z)=cx+g(z)$ for some $c$ and $g$, then floodgate only depends on $\mu$ through the sign of $c$, making floodgate particularly forgiving for partially linear working models. To be precise, floodgate using $\mu(x,z)=cx+g(z)$ will perform \emph{identically} to floodgate using the \emph{best} partially linear approximation to $\mustar$ as long as $c$ has the same sign as the coefficient in that best approximation (regardless of $c$'s magnitude or anything about $g$). %\lzmargin{In Section \ref{sec:accuracy}, we will see a clear advantage of such invariance property in terms of statistical accuracy.}{since this invariance stuff and the notation $S_\mu$ will appear again in statistical accuracy section, should we keep this cyan-colored sentence?}
}	
\end{remark}
}
    
    \subsection{Upper confidence bounds for the mMSE gap}\label{sec:ucb_hardness}
    Before continuing our study of floodgate LCBs, we first pause to address a natural question: what about an \emph{upper} confidence bound (UCB)? {One way to get a UCB is to follow a workflow similar to the previous subsection},
    as follows. For any working regression function $\nu$ for $\Ec{Y}{Z}$, consider the functional 
    \begin{equation*}
    \fucb(\nu) = \EE{(Y-\nu(Z))^2}.
    \end{equation*}
    Then $\fucb$ plays an analogous role to $f$ in the opposite direction, in that for \emph{any} $\nu$, (a) $\fucb(\nu)\ge \Ij^2$ and (b) we can construct a level $\alpha$ UCB 
    $U_n^\alpha(\nu)$ 
    for $\fucb(\nu)$. Property (a) is immediate from the minimality of the first term and non-negativity of the second term in definition \eqref{eq:msegap}, {while property (b) can be established without even making model-X assumptions}: simply take the CLT-based UCB from the estimator $\frac{1}{n}\sum_{i=1}^n(Y_i-\nu(Z_i))^2$, which is unbiased for $\fucb(\nu)$. 
    
    Unfortunately, there is no value of $\nu$ such that $\fucb(\nu)=\Ij^2$ except in the noiseless setting where $Y$ is a \emph{deterministic} function of $(X,Z)$. In particular, no matter how well $\nu$ is chosen and how large $n$ is, $U_n^\alpha(\nu)-\Ij^2\ge \EE{\Varc{Y}{X,Z}}$ with probability at least $1-\alpha$. This shortcoming is perhaps foreseeable given that $U_n^\alpha(\nu)$ never even uses the $X_i$, but it turns out to be unimprovable (even using model-X information), as we now prove in Theorem~\ref{thm:UCB_hardness}.
    
    \begin{theorem}\label{thm:UCB_hardness}
    % Consider any joint distribution of $(X,Y)$ with general conditional model of $Y$ given $X$ satisfying $(1)$ $\mathbb{E}(Y^2)\le c$ for some constant $c>0$; $(2)$ $\mathbb{E}\mathrm{Var}(Y|X)>0$ and denote the set of these distributions by $\calF$,  given the confidence level $1-\alpha$,
    Fix a continuous joint distribution $P_{X,Z}$ for $(X,Z)$, and let $\calF$ denote the class of joint distributions $F$ for $(Y,X,Z)$ such that $F$ is compatible with $P_{X,Z}$ and $\Var{Y}<\infty$.
    %Let $\calF$ denote the class of joint distributions $F$ for $(Y,X,Z)$ such that $\Var{Y}<\infty$ and 
    Let $U(D_n)$ denote a scalar-valued function of the $n$ i.i.d. samples $D_n=\{Y_i,X_i,Z_i\}\nsubp$; if $U(D_n)$ outputs a UCB for the mMSE gap that is pointwise asymptotically valid for any $F\in\calF$, i.e.,
    $$
    \inf_{F\in \calF}\liminf_{n\rightarrow\infty}\mathbb{P}_{F}(U(D_n)\ge \Ij^2_F)\ge 1-\alpha,
    $$
    then
    \begin{equation}\label{eq:UCB_trivial}
    \sup_{F\in \calF}\limsup_{n\rightarrow\infty}\Pp{F}{U(D_n)-\Ij^2_F < \Ep{F}{\Varpc{F}{Y}{X,\Xnoj}}} \le \alpha,
     \end{equation}
     where the subscript $F$ denotes quantities computed with $F$ as the data-generating distribution.
     %where $\Ij^2(F)$, $\mathbb{E}_F$, $\textnormal{Var}_F$, and $\mathbb{P}_F$ denote the mMSE gap, expectation, variance, and probability under $F$, respectively.
     %where $\Ep{F}{\Varpc{F}{Y}{\Xnoj}}$ is a trivial upper confidence bound.
    \end{theorem}
    The proof of Theorem~\ref{thm:UCB_hardness} can be found in Appendix~\ref{app:ucb_hardness}. Note that since we fix $P_{X,Z}$ at the beginning of the theorem statement, $U$ is allowed to use model-X information.
    %It is easy to see a method to achieve this bound asymptotically, which simply fits a regression function $\mu_{-x}(Z)$ to the first split of the data, and then computes an unbiased estimate of its MSE on the second split of the data and takes the UCB from the CLT approximation of the MSE estimator. Then clearly if $\mu_{-x}$ is consistent for $\Ec{Y}{Z}$, the bound in the above theorem will be attained asymptotically.
    As just mentioned above, this theorem provides no cause for concern in the noiseless setting when $\EE{\Varc{Y}{X,\Xnoj}}=0$. 
    %
    %While this method and the asymptotic UCB it attains are in some sense trivial as they do not use $X$ at all (and hence in fact could be computed without even observing $X$), if $\Varc{Y}{X,Z} = 0$, then $\Ij^2 = \EE{\Varc{Y}{Z}}$, which may at least approximately hold in some applications, so that nothing is lost. 
    However, in many applications we may expect $\EE{\Varc{Y}{X,Z}}$ to be substantial, and the above theorem guarantees \emph{any} pointwise asymptotically valid UCB must be conservative by this amount. 
    The only way to overcome this problem \acc{would be} to assume some sort of structure on $Y \mid X, Z$, such
as smoothness or sparsity,  
%In particular, the assumed structure of $Y \mid X, Z$ must be incorporated into the UCB procedure itself to attain nontrivial results, 
in contrast to \acc{model-X} floodgate which requires no
information about $Y \mid X, Z$ and can certainly produce nontrivial LCBs and even achieve the parametric rate
with sufficiently-accurate $\mu$; see Section~\ref{sec:accuracy}.
    %\lz{By contrast, floodgate applied with any $\mu$ produces a LCB that is pointwise asymptotically valid, but can produce bounds that approach the mMSE gap quite quickly in $n$, even achieving the parametric rate with sufficiently accurate $\mu$; see Section~\ref{sec:accuracy}.}
    %unless some structural assumption is made about $\Ec{Y}{X,Z}$, such as smoothness or sparsity, that allows it to be estimated at some rate. 
    %
    %Another takeaway is that, in order to overcome this UCB impossibility result, the assumed structure of $Y\mid X,Z$ must be incorporated into the UCB procedure itself, in stark contrast to floodgate which requires no information about $Y\mid X,Z$ and can certainly produce nontrivial LCBs and even achieve the parametric rate with sufficiently-accurate $\mu$; see Section~\ref{sec:accuracy}. 
    %We note that for the purposes of discovery, an LCB $L$ alone is already very useful as it probabilistically-rigorously establishes that a variable is \emph{at least as} important as $L$. 
    Although it is disappointing that a better UCB is not achievable,
    %Although it is disappointing that further assumptions would be needed for a better UCB in many applications, 
    we envision MOVI inference often being used
    %employed in scenarios where the goal is 
    to quantify \emph{new} important relationships, in which case we expect it to be more useful to know a variable is \emph{at least as} important as some LCB than to upper-bound its importance with a UCB. Given this perspective and the negative UCB result of Theorem~\ref{thm:UCB_hardness}, we return for the remainder of the paper to the study of using floodgate to obtain LCBs.
    
    \subsection{Computation}
    \label{sec:computation}
    
    Astute readers may have noticed that the quantities $R_i$ and $V_i$ in Algorithm~\ref{alg:MOCK} involve conditional expectations/variances which, though in principle known due to \acc{the assumed model-X} knowledge of $P_{X|Z}$, may be quite hard to compute in practice. In certain cases these conditional expectations can have simple or even closed-form expressions, such as when $\mu$ is a generalized linear model and $X\mid Z$ is Gaussian, but otherwise a more general approach is needed. Monte Carlo provides a natural solution: assume that we can sample $K$ copies $\Xtil_i^{(k)}$ of $X_i$ from $P_{X_i|Z_i}$ conditionally independently of $X_i$ and $Y_i$ and thus
    replace $R_i$ and $V_i$, respectively, by the sample estimators
    \[R_i^K = Y_i\left(\mu(X_i,Z_i) - \frac{1}{K}\sum_{k=1}^{K} \mu(\Xtil_i^{(k)} ,Z_i)\right), \]%\qquad 
    \[ V^{K}_{i} = \frac{1}{K-1}\sum_{k=1}^{K}\left( \mu(\Xtil_i^{(k)} ,Z_i)-  \frac{1}{K}\sum_{k=1}^{K} \mu(\Xtil_i^{(k)} ,Z_i) \right)^2. \]
    %thereby obtain Monte Carlo sample estimators $R^K_i$ and $V^K_i$ of $R_i$ and $V_i$, respectively. 
    Luckily the same guarantees hold for the Monte Carlo analogue of floodgate, even for fixed $K$.
    \acc{
	\begin{theorem}\label{thm:main_general}
	 Under the conditions of Theorem~\ref{thm:main}, for any given $K>1$, $L_{n,K}^{\alpha} (\mu)$ computed by replacing $R_i$ and $V_i$ with $R_i^K$ and $V_i^K$, respectively, in Algorithm \ref{alg:MOCK} satisfies
	\begin{equation} 
	\label{eq:unif_valid_general}\nonumber
	\liminf_{n\rightarrow \infty} 
	\PP{L_{n,K}^{\alpha} (\mu) \le \Ij }  \ge 1 - \alpha.
	\end{equation}
    \end{theorem}
    The proof can be found in Appendix~\ref{pf:thm:main_general}. In general we expect larger values of $K$ to produce more accurate LCBs, but we found the difference between $K=2$ and $K=\infty$ to be surprisingly small in our simulations and, of course, it will always be computationally faster to use smaller $K$. Although Theorem~\ref{thm:main_general} is a pointwise result holding for any fixed $K>1$, it can be generalized to a uniform result over all $K>1$ with miscoverage bounded by a $n^{-1/2}$ rate using higher moment conditions and a variance lower bound assumption; see Appendix \ref{sec:rate} for details.}

   % Theorem \ref{thm:max_2} tells us 
   
    \subsection{Accuracy adaptivity to $\mu$'s mean squared error}
    \label{sec:accuracy}
    Having established floodgate's validity and computational tractability, the natural next question is: how accurate is it, i.e., how close is the LCB to the mMSE gap? The answer depends on the accuracy of $\mu$---the better that $\mu$ approximates $\mustar$, the more accurate the floodgate LCB is, as formalized in the following theorem.      
    \begin{theorem}[Floodgate accuracy and adaptivity]\label{thm:accuracy}
        For i.i.d. data $\{(Y_i,X_i,Z_i)\}_{i=1}^n$ such that $\ee{Y^{12}}<\infty$, $\Varc{Y}{X,Z} \ge \tau$ a.s. for some $\tau>0$, and a sequence of \revision{working regression functions} $\mu_n:\mathbb{R}^{p} \rightarrow \mathbb{R}$ such that for some $C$ and all $n$ either $\EE{\Varc{\mu_n(X,Z)}{Z}}=0$ or $\frac{\EE{ {\mu}_n^{12}(X,Z)}}{\EE{\Varc{\mu_n(X,Z)}{Z}}^{6}} \le C$, the output of Algorithm~\ref{alg:MOCK} satisfies
        \begin{equation}\label{eq:accuracy_gap}
            \Ij - L_n^{\alpha}(\mu_n)  = O_p\left(\inf_{\mu\in S_{\mu_n}}\EE{( \mu(X,Z)-\mustar(X,Z))^2} +n^{-1/2}\right),
        \end{equation}
        \acc{where $S_{\mu_n}  = \{c\mu_n(\cdot,\cdot)+g(\cdot,\cdot):c>0,\; g(x,\cdot)=g(x',\cdot)\, \forall x,x'\}$ as defined in Remark \ref{rk:mu}.}
    \end{theorem}
    The proof can be found in Appendix~\ref{pf:thm:accuracy}. The above condition that ``$\EE{\Varc{\mu_n(X,Z)}{Z}}=0$ or $\frac{\EE{ {\mu}_n^{12}(X,Z)}}{\EE{\Varc{\mu_n(X,Z)}{Z}}^{6}} \le C$'' is a scale-free moment condition on $\mu_n$ which says that $\mu_n(X,Z)$ can have no dependence on $Z$ at all or have a non-vanishing conditional variance (given $Z$) relative to its higher moments. \acc{The high-order moments in our assumptions are likely a technical artifact of our proof (which actually proves a somewhat stronger result than stated in the theorem), and could perhaps be relaxed with a different approach.}
    \acc{
    %Existence of such high moments are assumed \ljmargin{to avoid lengthy discussions}{not really clear what you mean by this---what length discussion about triangular arrays would we have to include if we didn't make these higher moment conditions? Can you rephrase or elaborate? I think readers may get confused here since none of our other results in the main text require such higher moments}. 
    As it stands, these assumptions allow us to utilize the Berry--Esseen-type results in Appendix \ref{pf:thm:main_rate} to handle the fact that $\mu_n$ varies with $n$.} %\lzmargin{}{rewrite it by explaining in detail}
     
     We call the left-hand side of Equation~\eqref{eq:accuracy_gap} the \emph{half-width} (by analogy with the \emph{width} that would measure the accuracy of a two-sided confidence interval) and Theorem~\ref{thm:accuracy} shows it is \emph{adaptive} to the accuracy of $\mu_n$ through the MSE of the best element of its equivalence class $S_{\mu_n}$, up to a limit of the parametric or central limit theorem rate of $n^{-1/2}$. 
    So in principle floodgate can achieve $n^{-1/2}$ accuracy if a member of $S_{\mu_n}$ converges very quickly to $\mustar$, but in general floodgate's accuracy decays gracefully with $\mu_n$'s accuracy. \acc{Note} that the infimum in Equation~\eqref{eq:accuracy_gap} means that floodgate is \emph{self-correcting} with respect to $\mu_n$'s conditional mean given $Z$, \acc{as explained in the second paragraph of Remark~\ref{rk:mu}}.%(through invariance (i)) and global scale (through invariance (ii)).  \lzmargin{}{rewrite here to connect to Remark}

    \subsection{Straightforward generalizations}\label{sec:generalizations}
    Before moving onto extensions, we briefly address a few relatively straightforward generalizations of floodgate.
    
    \paragraph*{Extending the mMSE gap}
    The mMSE gap can be very naturally made invariant to the scale of $Y$ and bounded between 0 and 1 by dividing it by $\Var{Y}$. And since $\Var{Y}$ can be easily and asymptotically-normally estimated under weaker conditions than already assumed for floodgate's validity in Theorem~\ref{thm:main_general}, it is straightforward to extend the floodgate procedure and its validity to perform inference on the scale-free version $\Ij^2_{\text{sf}} = \Ij^2/\Var{Y}$. \acc{We also consider two ways of extending the mMSE gap such that the key property of the MOVI in \citet{azadkia2019simple} is satisfied, i.e., the MOVI equals zero if \emph{and only if} $\jnull$. Details about defining the MOVIs and providing inference can be found in Appendices \ref{sec:max_mmse} and \ref{sec:kernelfg}.}

    \paragraph*{Inference for group variable importance}
    In applications where a group of variables share a common interpretation or are too correlated to powerfully distinguish, it is often necessary to infer a measure of \emph{group} importance instead of a MOVI. Luckily, when $X$ is multivariate, the mMSE gap remains perfectly well-defined and interpretable and floodgate (both $f$ and Algorithm~\ref{alg:MOCK}) retain all the same inferential properties. Indeed, we apply floodgate to groups of variables in our genomics application in Section~\ref{sec:realdata}.
    
    \paragraph*{Transporting inference to other covariate distributions} In some applications, the samples we collect may not be uniformly drawn from the population we are interested in studying. For instance, our data may come from a lab experiment with covariates randomized according to one distribution, while our interest lies in inference about a population outside the lab whose covariates follow a different distribution. As long as the samples at hand share a common conditional distribution $Y\mid X,Z$ with the target population, it is relatively straightforward to perform an importance-weighted version of floodgate that provides inference for the target population's mMSE gap. We provide the details in Appendix~\ref{sec:transport}.
    %The mMSE gap (and hence also floodgate's inference) depends on the entire joint distribution of $(X,Y,Z)$, not just the conditional distribution of $Y\mid X,Z$. Hence we may want to conduct inference

    \paragraph*{Adjusting for selection} When inference is required for many variables simultaneously, it is often preferable to focus attention on a subset of variables whose inferences appear particularly interesting. But if we only report the set of LCBs that are, say, farthest from zero, then our coverage guarantees will fail to hold for this set due to selection bias (this is not a defect of floodgate, but a property of nearly every non-selective inferential procedure). One way to address this may be to apply false coverage-statement rate adjustments \citep{benjamini2005false} to floodgate LCBs. The application is straightforward, and floodgate LCBs satisfy the monotone property required by \citet{benjamini2005false}, although they do not in general satisfy the independence or positive regression dependence on a subset (PRDS) condition and hence would require a correction \citep{benjamini2001control} for strict guarantees to hold. We leave a more formal treatment of selection adjustment to future work, but note also some simple ways to perform benign selection. 
    
    First, if selection is performed using $\mu$ and/or independent data, then no adjustment is needed for validity. For instance, if floodgate is run by data-splitting, we could arbitrarily use the first half of the data (which is also used for choosing $\mu$, but not for running floodgate) for selection, including selecting precisely the subset of variables that $\mu$ depends on. In fact, we can even perform a certain type of benign post-hoc data processing based on the floodgate data itself: if the floodgate data are used to construct a \emph{transformation} of the floodgate LCBs such that every transformed LCB either shrinks or remains the same, then the transformed LCBs retain their marginal asymptotic validity. This is because any such transformation, even one depending on the data or LCBs themselves, can only \emph{increase} coverage of each LCB by reducing it or leaving it unchanged; this is related to the screening procedure in \citet{ML-LJ:2020}. This means, for instance, that if a selection procedure is applied to the floodgate data and used to zero out any unselected LCBs, then as long as the zeroed-out LCBs are reported alongside the rest, the marginal validity of all reported LCBs remains intact even though the same data was used to construct the LCBs and to perform the selection that transformed them.
    %{\color{red} summarize changing distribution of X (refer to appendix), selection adjustment stuff (no need if selection based on first half, can do CDR analogue if selection uses second half, otherwise can do FCR), maybe group variable selection?}
    
    \section{Extensions}\label{sec:extensions}
    % Our framework can be extended in several aspects. First, we consider an alternative MOVI to the mMSE gap, which takes a absolute value form instead. In classification setting, we provide a valid inferential procedure for it using similar strategy as in Section~\ref{sec:mock}.
    % % we consider a different MOVI and provide inference in classification setting where the response variable is discrete therefore the conditional model of $Y$ given $X$ is simpler and the loss function is different. 
    % % By using the similar strategy as before, we also define variable importance measure based on predictive performance and provide inferential guarantees. 
    % Second, we extend our main methodology to the setting where only a model for the covariate distribution is known and illustrate with two specific examples.
    % % Third, we show how to transport the inferential results to the new environments or populations where the covariate distribution is different. 
    % Finally, in light of the multiplicity concern often in applied research, we describe how to adjust the multiple lower confidence bounds to guarantee false coverage-statement rate (FCR) control.
    % a multiple adjustment procedure when the multiple confidence lower bounds are constructed or reported for selected with false coverage-statement rate (FCR) control guarantee.
    % the case where multiple confidence bounds are reported after selection and show that the resulting lower confidence bounds are ready to be adjusted for the selected covariates by a false coverage-statement rate (FCR) control procedure.
    \subsection{Beyond the mMSE gap}
    \label{sec:class}
    To demonstrate that the floodgate idea can be used beyond the mMSE gap, we consider the following MOVI.
    % Readers may ask whether there are other appropriate MOVIs alternative to the mMSE gap. For example, we can quickly come up with the following quantity.
    \begin{definition}[Mean absolute conditional mean gap]
    \label{def:class_movi}
    The \emph{mean absolute conditional mean (MACM) gap} for variable $X$ is defined as 
    \begin{equation}
        \label{eq:class_movi}
        \Ijc =  \EE{\left|\condmeanj - \condmean\right|}
    \end{equation}
    whenever all the above expectations exist.
    \end{definition}
    The subscript in $\Ijc$ reflects its similarity to $\Ij^2 =  \EE{(\condmeanj - \condmean)^2}$ 
    % the following expression for $\Ij$,% Recall that \eqref{eq:mMSEgap_compact} directly implies the following expression of the mMSE gap
    % \begin{equation} \nonumber
    %     \label{eq:mMSEgap_form3}
    %     \Ij^2 =  \EE{(\condmeanj - \condmean)^2},
    % \end{equation}
    except with the square replaced by the absolute value (also known as the $\ell_1$ norm). %then $\Ijc$ basically replaces the squared form by the absolute value.
    %(up to a factor of $0.5$).
    % ---------------------
    % Outline:
    % 1. It is nature to ask whether there are other MOVIs with different loss
    % 2. We do not have a general characterization 
    % 3. For example, the absolute one we know how to do it in classification setting
    % 4. And it is nice that we have if and only if
    % ---------------------
    Although we have not found a floodgate function to enable inference for arbitrary $Y$, the remainder of this subsection shows how to perform floodgate inference when $Y$ is binary (coded as $Y\in\{-1,1\}$). %In terms of conducting inference on $\Ijc$, currently we do not have a general answer but are able to provide a valid inferential procedure in classification setting. More specifically, when $Y$ is binary, say $Y\in \calY=\{-1,1\}$, we propose how to produce lower confidence bounds, as described in this section. 
    We note that when $Y$ is binary, $\Ijc$ is zero if and \emph{only if} $\jnull$ holds (the ``if" part holds for non-binary $Y$ as well), since the expected value uniquely determines the distribution of a binary random variable. 
    In particular, for any (nonrandom) function $\mu:\mathbb{R}^p\rightarrow\mathbb{R}$, define
%    \begin{definition}\label{def:kappaj_u}
    \begin{equation}\label{eq: kappaj_u}
    \kappamu = 
    % 2\EE{\mathbbm{1}_{\{Y\cdot (\muXk - \Ec{\muX}{\Xnoj}) < 0\}} - \mathbbm{1}_{\left\{Y\cdot (\muX - \Ec{\muX}{\Xnoj}) < 0\right\}}},
    2\mathbb{P}\big(Y (\muXk - \Ec{\muX}{\Xnoj}) < 0\big)-2\mathbb{P}\big(Y (\muX - \Ec{\muX}{\Xnoj}) < 0\big)
    \end{equation}
    where $\tilde{X}\sim P_{X|Z}$ and is conditionally independent of $X$ and $Y$.
    \begin{lemma}\label{lem:class_max}
    If $|Y|\stackrel{a.s.}{=}1$, then for any $\mu$ such that $\kappamu$ exists, $\kappamu\le \Ijc$, with equality when $\mu=\mustar$.
    \end{lemma}

Obtaining an LCB for $\kappamu$ is even easier than it was for $\thetamu$ because $\kappamu$ is essentially just one expectation instead of a ratio of expectations, so a straightforward central limit theorem argument suffices; \acc{Algorithm~\ref{alg:class_MOCK} (presented in Appendix \ref{app:class_computation}) formalizes the procedure and Theorem~\ref{thm:class_main} establishes its asymptotic coverage.}
%\begin{algorithm*}[h]
%	\caption{Floodgate for the MACM gap}\label{alg:class_MOCK}
%	\begin{algorithmic}
%	\REQUIRE  Data $\{(Y_i,X_i,Z_i)\}\nsubp$, $P_{X\mid \Xnoj}$, a \revision{working regression function} $\mu:\mathbb{R}^p\rightarrow\mathbb{R}$, and a confidence level $\alpha \in (0,1)$.
%	\STATE Let $U_i = \mu(X_i,Z_i)-\ec{\mu(X_i,Z_i)}{Z_i}$ and % and $\tilde{U}_i = 
%	compute
%	\vspace{-.2cm}\[ R_i = \left\{\begin{array}{rl} \Pc{U_i<0}{\Xinoj}-\indicat{U_i<0} & \text{ if }\;Y_i=1\\
%	\Pc{U_i>0}{\Xinoj}-\indicat{U_i>0} & \text{ if }\;Y_i=-1
%	\end{array}\right. \vspace{-.2cm}\]
%	for $i\in [n]$,
%	and compute its sample mean $\bar{R}$ and sample variance $s^2$.%by $(\bar{R}, s^2)$.
%	\RETURN Lower confidence bound $L_{n}^{\alpha} (\mu)=2 \max\left\{\bar{R} - \frac{z_{\alpha}s}{\sqrt{n}},0\right\}$.
%	\end{algorithmic}
%\end{algorithm*}

	\begin{theorem}[MACM gap floodgate validity]\label{thm:class_main}
	For any given \revision{working regression function} $\mu:\mathbb{R}^{p} \rightarrow \mathbb{R}$ and i.i.d. data $\{(Y_i,X_i,Z_i)\}_{i=1}^n$,
	$L_{n}^{\alpha} (\mu)$ from Algorithm~\ref{alg:class_MOCK} satisfies 
	\begin{equation} 
	\label{eq:class_ptwise_valid}\nonumber
 	\PP{L_{n}^{\alpha} (\mu) \le \Ijc } 
	    \ge 1 - \alpha - O(n^{-1/2}).
	\end{equation}	
    \end{theorem}
Theorem~\ref{thm:class_main} is proved in Appendix~\ref{pf:thm:class_main}, and perhaps its most striking feature is its lack of assumptions, which follows from the boundedness of $\kappamu$ and the $R_i$. 
%\acc{This boundedness also allows us to immediately bound the rate of miscoverage by $O(n^{-1/2})$, since unlike unlike the analogous results for the mMSE gap in Appendix~\ref{sec:rate} which require}
%\rev{\ljmargin{Thanks to such boundedness, we are also able to prove coverage validity with a rate of $n^{-1/2}$, which is directly presented in the above statement; by contrast, the rate versions of Theorems \ref{thm:main} and \ref{thm:main_general} require extra higher moment conditions and are deferred to Appendices.}{I actually don't think we need to explain this, so I would just delete this text}}
Like $f$, $f_{\ell_1}$ is invariant to any transformation of $\mu$ that leaves $\mathrm{sign}(\mu(X,Z)-\ec{\mu(X,Z)}{Z})$ unchanged on a set of probability 1, making its validity immediately uniform over large classes of $\mu$.

Although the boundedness of the $R_i$ streamlines the coverage guarantees, their conditional probabilities make it somewhat more complicated to carry out efficient computation of Algorithm~\ref{alg:class_MOCK}. In particular, the sharp boundary at zero inside the probabilities requires a certain degree of smoothness in $\mu$ and $P$ to be able to estimate the $R_i$ by Monte Carlo samples analogously to Section~\ref{sec:computation}. We give precise sufficient conditions and a proof of their validity in Appendix~\ref{app:class_computation}, and defer study of Algorithm~\ref{alg:class_MOCK}'s accuracy and robustness to future work.
 
    \subsection{Relaxing the assumptions by conditioning}
    \label{sec:relax}
    In this section we show that we can relax the \acc{model-X} assumption that $P_{X|Z}$ be known exactly and apply floodgate when only a \emph{parametric model} is known for $P_{X|Z}$. This is inspired by \citet{DH-LJ:2019} which similarly relaxes the assumptions of model-X knockoffs. We follow the same general principle of conditioning on a sufficient statistic of the parametric model for $P_{X|Z}$, but doing so in floodgate requires a somewhat different approach than \citet{DH-LJ:2019}. \acc{Note that this section's method and assumptions are also distinct from the double robust assumptions in Remark~\ref{rk:modelx}, further emphasizing that the key ideas underlying floodgate are not tied to any particular set of assumptions.}

    The approach we take in this section will involve computations on the entire matrix of observations, i.e., $(\bX,\bZ)\in \mathbb{R}^{n\times p}$ whose rows are the covariate samples $(X_i,Z_i)$ and $\bm{y}\in \mathbb{R}^{n}$ whose entries are the response samples $Y_i$.
    Now suppose that we know a model $F_{X|Z}$ for $P_{X|Z}$ with a sufficient statistic functional for $n$ independent (but not necessarily identically distributed) samples $\bX\mid\bZ$ given by $\calT(\bX,\bZ)$, whose random value we will denote simply by $\bT$. We will assume that $\calT$ is invariant to permutation of the rows of $(\bX,\bZ)$ (as we would expect for any reasonable $\calT$, since these rows are i.i.d.).
    
    The key idea that allows us to perform floodgate inference without knowing the distribution of $\bX\mid \bZ$ is that, by definition of sufficiency, we \emph{do} know the distribution of $\bX\mid\bZ,\bT$. Leveraging this idea requires some adjustment to the floodgate procedure, and we start by defining a conditional analogue of $f$.
    \begin{equation}\label{eq:thetamuT}
        \thetamuT:=
        \frac{\EE{\covc{\mustar(X_i,Z_i)}{\mu(X_i,Z_i)}{\bs{Z},\bs{T}}}}{\sqrt{\EE{\varc{\mu(X_i,Z_i)}{\bs{Z},\bs{T}}}}}, 
        \end{equation}
         again with the convention $0/0=0$. Note that $\thetamuT$ does not depend on the choice of $i$ thanks to $\calT$'s permutation invariance, but it \emph{does} depend on the sample size $n$. Nevertheless, it follows immediately from the proof of Lemma~\ref{lem:max} that $\thetamuT \le \thetamusT$ for any nonrandom $\mu$. On the other hand, $\thetamusT\neq\Ij$, but instead a different relationship that is nearly as useful holds:
    $$ 
     \thetamusT \le \thetamus = \Ij,
    $$ 
    due to the monotonicity of conditional variance.
    
    With floodgate property (a) ($\thetamuT\le\Ij$) established, we now turn to property (b): the ability to construct a LCB for $\thetamuT$. In an analogous way as for $\thetamu$, we can compute $n$ unbiased estimators of the numerator and the squared denominator, but these estimators are no longer i.i.d. because they are linked through $\bT$, so we cannot immediately apply the central limit theorem or delta method as we did in Section~\ref{sec:mock}. Our workaround is to split the data into \acc{$n_2$} \emph{batches} \acc{of size $n_1$} and only condition on the sufficient statistic within each batch. This way, there is still independence between batches and we can apply the central limit theorem and delta method across batches. \acc{This strategy is spelled out in Algorithm \ref{alg:batch_cond} (see Appendix \ref{sec:relax_details} for details) and Theorem~\ref{thm:batch_cond} establishes its asymptotic coverage. We call this procedure \emph{co-sufficient} floodgate because the term ``co-sufficiency" describes sampling conditioned on a sufficient statistic \citep{stephens2012goodness}.}
   % \ljmargin{\acc{$n_2$} \emph{batches} \acc{of size $n_1$}}{$n_1,n_2$ are referred to later but used to only be introduced in the Alg environment}  
%      \begin{algorithm*}[h!]
%			\caption{Co-sufficient floodgate} \label{alg:batch_cond}
%		\begin{algorithmic}[1]
%			\REQUIRE The inputs of Algorithm \ref{alg:MOCK}, a sufficient statistic functional $\calT$, and a batch size $n_{2}$.
%			%a batching rule (i.e. $n=n_{1} n_{2}$, with $n_{1}$ batches, each of size $n_{2}$).
%			\STATE Let $n_1=n/n_2$ and for $m\in[n_1]$, denote $(\bX_m,\bZ_m) = \{X_i,Z_i\}_{i=(m-1)n_2+1}^{mn_2}$, and let $\bT_m=\calT(\bX_m,\bZ_m)$.
%      	            \STATE For $m\in[n_1]$, compute
%      	          \begin{align*}
%      	          R_m &= \frac{1}{n_2}\sum_{i=(m-1)n_2+1}^{mn_2} Y_i\, (\mu(X_i,Z_i) - \Ec{\mu(X_i,Z_i)}{\bs{Z}_m, \bT_m} ),\\
%      	          V_m &= \frac{1}{n_2}{\sum_{i=(m-1)n_2+1}^{mn_2} \Varc{\mu(X_i,Z_i)}{\bs{Z}_m,\bT_m}},
%      	          \end{align*}
%		  their sample mean $(\bar{R},\bar{V})$, their sample covariance matrix $\hat{\Sigma}$, and $s^2 = \frac{ 1 }{ \bar{V} }\left[ \left(\frac{\bar{R} }{2 \bar{V} }\right)^2 \hat{\Sigma}_{22} +  \hat{\Sigma}_{11} - \frac{\bar{R} }{ \bar{V} } \hat{\Sigma}_{12} \right].$
%		    \RETURN Lower confidence bound $\LjT (\mu)=\max\left\{\frac{\bar{R} }{ \sqrt{\bar{V}}}
%        - \frac{z_{\alpha}s}{\sqrt{n_1}},\,0\right\}$, with the convention that $0/0=0$. 
%		\end{algorithmic}
%	\end{algorithm*}
    \acc{
    \begin{theorem}[Co-sufficient floodgate validity]\label{thm:batch_cond}
    For any given working regression function $\mu:\mathbb{R}^{p} \rightarrow \mathbb{R}$, i.i.d. data $\{(X_i,Z_i,Y_i)\}_{i=1}^n$, and permutation-invariant sufficient statistic functional $\calT$, if $\ee{Y^4}<\infty$ and $\ee{\mu^4(X,Z)} <\infty$, then $\LjT(\mu)$ from Algorithm~\ref{alg:batch_cond} satisfies 
	\begin{equation} 
	\nonumber
        \liminf_{n\rightarrow \infty}\mathbb{P}\left(\LjT(\mu) \le \Ij \right) \ge 1 - \alpha.
	\end{equation}
    \end{theorem}
    The proof can be found in Appendix~\ref{pf:thm:batch_cond}. }
%    the weaker moment conditions than Theorem~\ref{thm:main} correspond to the weaker $o(1)$ term, and we defer to future work strengthening it to $O(n^{-1/2})$ following similar techniques as earlier results in the paper. 
    Regarding computation, as in Section~\ref{sec:computation}, we can replace the conditional expectations in \acc{Algorithm~\ref{alg:batch_cond} with Monte Carlo estimates};
    %the expressions for $R_m$ and $V_m$ with Monte Carlo estimates based on resampling $\bX_m\mid \bZ_m,\bT_m$ conditionally independently of $\bX$ and $\bm{y}$;
%    }{this used to refer to $R_m$ and $V_m$, which no longer appear in the main text}
    see Appendix~\ref{sec:mccf} for details. For a given $\mu$, we may worry that co-sufficient floodgate loses some accuracy relative to regular floodgate due to the gap between $\thetamu$ and $\thetamuT$, but in fact this gap is typically $O(n_2^{-1})$ for fixed-dimensional parametric models. {We quantify this gap for multivariate} Gaussian and discrete Markov chain covariate models in the following two subsections{, showing that, at least in these two cases, co-sufficient floodgate relaxes the assumptions of \acc{model-X} floodgate with only a minimal loss in accuracy}.

	\subsubsection{Low-dimensional multivariate Gaussian model}
	\label{sec:lowdim_Gaussian}
	
	In this section we let $\calB_{m}=\{(m-1)n_2+1,\dots,mn_2\}$.
	  \begin{proposition}\label{prop:gap_rate_gaussian}
	  Suppose samples $\{X,Z\}\nsubp$ are i.i.d. multivariate Gaussian parameterized as $X_{i}\mid\Xinoj \sim \gauss{(1,Z_{i})\gamma}{\sigma^{2}}$ for some $\gamma\in \mathbb{R}^{p}$ and $\sigma^2>0$, and $\Xinoj \sim \gauss{\bv_0}{\bSig_0}$. Assume  $\sigma^{2}$ is known and the batch size $n_{2}$ satisfies $n_{2}>p+2$. {Let} $\calT$ be the following sufficient statistic functional
	  $$
	  \bT_{m}:=\calT(\bs{X}_m, \bs{Z}_m)=\left(\sum_{i\in\calB_{m}}X_{i}, \sum_{i\in\calB_{m}}X_{i}Z_{i}\right).
	  $$
	  Then if $\EE{\mu^4(X,Z)},\EE{(\mustar)^4(X,Z)} <\infty$, we have
	  \begin{equation}\label{eq:gap_rate_gaussian}
     \thetamu - \thetamuT = O\left(\frac{p}{n_{2} - p -2}\right).
	  \end{equation}
      \end{proposition}
        The proof can be found in Appendix~\ref{pf:prop:gap_rate_gaussian}. Note the condition $n_{2}>p+2$ is not surprising as when the sample size is smaller than $p$, the sufficient statistic functional is degenerate, resulting in a zero value of $\thetamuT$. The bound in \eqref{eq:gap_rate_gaussian} allows $p$ to grow with $n$ in general, but when $p$ is fixed, it gives the rate of $O(n_2^{-1})$, as mentioned earlier in Section~\ref{sec:relax}.

	 \subsubsection{Discrete Markov chains}
	 \label{sec:DMC}
	 
	 To present our second example model, we define some new notation.
% 	 The second example is one of the discrete Graphical models, the Markov chains. Same with most discrete models, it assumes sort of local independence, formally characterized as local Markov property. The construction of knockoff by conditioning for these models is essentially taking advantage of this structure. We investigate this example under our framework, showing that the negligible accuracy loss from conditioning is benefiting from this structure, too. 
	 Consider a random variable $W$ following a discrete Markov chain with $K$ states with $X=W_j$, $Z=W_{\noj}$, then the model parameters include the initial probability vector $\pi^{(1)}\in \mathbb{R}^{K}$ with $\pi^{(1)}_{k} = \PP{W_{1}=k}$ and the transition probability matrix $\Pi^{(j)} \in \mathbb{R}^{K\times K}$ (between $W_{j-1}$ and $X = W_{j}$) with $\Pi^{(j)}_{k,k'} = \Pc{W_{j} = k'}{W_{j-1}=k} $. Further denoting $q({k,k_{1},k_{2}}) = \Pc{W_j=k}{W_{j-1}=k_1, W_{j+1} =k_2}$, we have
    \[
    q({k,k_{1},k_{2}}) = \frac{  \Pi^{(j)}_{k_1,k} \Pi^{(j+1)}_{k,k_2}  }{\sum_{k=1}^{K} \Pi^{(j)}_{k_1,k} \Pi^{(j+1)}_{k,k_2}},
    \]
    so that the conditional distribution of $\bm{X}_m\mid \bm{Z}_m$ can be compactly written down as
    \begin{equation}\label{eq:cond_dist_DMC}
        \Pc{ \bm{X}_m}{\bm{Z}_m} =  %\prod_{k=1}^{K}
        \prod_{k,k_{1},k_{2} \in [K]} 
        (q({k,k_{1},k_{2}}))^{N(k,k_{1},k_{2})},
    \end{equation}
    where $N(k,k_{1},k_{2}) = \sum_{i\in \calB_m} \indicat{X_{i}=k,W_{i,j-1} =k_{1},W_{i,j+1}=k_{1} }$.
    Thus we finally conclude that $\{N(k,k_{1},k_{2})\}_{(k,k_{1},k_{2}\in[K] )}$ is sufficient, and we proceed with this sufficient statistic. 
    %Throughout this section, we will work with this sufficient statistics. 
    % The proposition below quantifies the gap between $\thetamu$ and $\thetamuT$.
% 	 Without loss of generality, assume the transition probabilities $\{ \Pi^{(j)}_{k_1,k}, \Pi^{(j+1)}_{k,k_2} \}_{k,k_{1},k_{2}\in [K]}$ are nonzero (since we can always exclude the states with zero probability)
% 	 Following the derivations in Proposition \ref{prop:gap_rate_gaussian} and \ref{prop:cov_rate_gaussian}, it suffices to consider the $\chi^2$ divergence between $F$ and $G$.
	 \begin{proposition}\label{prop:gap_rate_DMC}
        Consider the above discrete Markov chain model and define the sufficient statistic functional $\calT$ as
        $$
        \bT_{m}=\calT(\bs{X}_m, \bs{Z}_m)=\{N(k,k_{1},k_{2})\}_{(k,k_{1},k_{2}\in[K] )}.
        $$
        Then if for variable $X=W_j$, $K^2\min\{ 
        \PP{W_{j-1}=k_{1},W_{j+1}=k_{2}}\}_{k_{1},k_{2}\in [K]}
        \} \ge q_{0}>0$ holds and
        % where $N(k,k_{1},k_{2}) = \sum_{i\in\calB_{m}} \indicat{X_{i}=k,X_{i,j-1} =k_{1},X_{i,j+1}=k_{1} }$, 
        % the moment condition
        $\EE{(\mustar)^2(X,Z)}$, $\EE{\mu^2(X,Z)}<\infty$, we have
        \[
       	 \thetamu - \thetamuT  = O\left(\frac{K^3}{n_{2}}\right).
        \]
	 \end{proposition}
	  The proof can be found in Appendix~\ref{pf:prop:gap_rate_DMC}. Note that $\calT$ here is not minimal sufficient and the above rate is cubic in $K$. The non-minimal sufficient statistic is adopted for the discrete Markov chain model in this paper since it is easier to work with and gives the desired rate in $n_2$, but we expect the rate in $K$ could be improved by using the minimal sufficient statistic. Again, $K$ is allowed to grow with $n$ in general, but when it is fixed we get a rate of $O(n_2^{-1})$, as mentioned earlier in Section~\ref{sec:relax}.     

    \section{Simulations}
    \label{sec:simul}
    Source code for conducting 
    %floodgate in 
    our simulation studies can be found at \url{https://github.com/LuZhangH/floodgate}.
    \subsection{Setup}\label{sec:simul_setup}
   {In the following subsections,
    %of this section, 
    we conduct simulation studies to complement the main theoretical claims of the paper. We study the effects of the sample-splitting proportion (Section~\ref{sec:simul_split}), covariate dimension (Section~\ref{sec:simul_vary_p}), and model misspecification (Section~\ref{sec:simul_robust}) on floodgate. Additional simulation studies on the effect of covariate dependence and sample size can be found in Appendix \ref{app:sim:cover}.}
    % covariate dependence (Section~\ref{sec:simul_corr}), sample size (Section~\ref{sec:simul_vary_n}),
    In Section~\ref{sec:comp_vimp}, we numerically compare floodgate with the method proposed in \citet{williamson2020unified}. We also study the extensions to floodgate for the MACM gap (Section~\ref{sec:simul_binary}) and co-sufficient floodgate (Section~\ref{sec:simul_relax}).
    Each simulation study generates a set of covariates and performs floodgate inference on each in turn (i.e., treating each covariate as $X$ and the rest as $Z$) before averaging its results (either coverage or half-width) over the covariates. 
    
    This paragraph describes the simulation setup for all but the simulation of Section~\ref{sec:comp_vimp}. The covariates are sampled from a Gaussian autoregressive model of order 1 (AR(1)) with autocorrelation 0.3, except in Section~\ref{sec:simul_corr} where this value is varied over. The conditional distribution of $Y\mid X,Z$ is given by $\mustar(X,Z)$ plus standard Gaussian noise, and in each subsection we perform experiments with both a linear and a highly nonlinear model. The linear model is sparse with non-zero coefficients' locations independently uniformly drawn from among the covariates, and the non-zero coefficients' values having uniform random signs and identical magnitudes (5, unless stated otherwise) divided by $\sqrt{n}$. The nonlinear model combines zero'th-, first-, and second-order interactions between nonlinear (mostly trigonometric and polynomial) transformations of elementwise functions of a subset of covariates, and then multiplies this entire function by an amplitude (50, unless stated otherwise) divided by $\sqrt{n}$; see Appendix~\ref{app:sim:nonlinear} for details. Both models use $n=1100$, $p=1000$, and a sparsity of 30 unless stated otherwise.
    
    In our implementations of floodgate, we split the sample into two equal parts (justified by the results of Section~\ref{sec:simul_split}) and use the first half to fit $\mu$. In most of the simulations, we consider four fitting algorithms (two linear, two nonlinear): the LASSO \citep{tibshirani1996regression}, Ridge regression, Sparse Additive Models (SAM; \citep{ravikumar2009sparse}), and Random Forests \citep{breiman2001random}; when the response is binary there are two additional fitting algorithms: logistic regression with an L1 penalty and an L2 penalty; see Appendix~\ref{app:sim:algorithms} for implementation details of these algorithms. The Monte Carlo version of floodgate from Section~\ref{sec:computation} is not needed for the linear methods, and for the nonlinear methods, $K=500$ is used.
    
    % \ljc{first half of this paragraph should be reworded}
    {Given the novelty of considering inference for the mMSE gap, it is challenging to compare floodgate to alternatives except in special cases. For instance, in {low-dimensional} Gaussian linear models the mMSE gap is a simple function of the coefficient and thus ordinary least squares (OLS) inference can be compared to floodgate; see Appendix~\ref{app:sim:ols} for details of how it is made comparable. Thus, in the {low-dimensional} linear-$\mustar$ simulations of Sections~\ref{sec:simul_vary_p} and \ref{sec:simul_corr}, we compare floodgate's inference to that of OLS, which acts as a sort of oracle since its inference relies on very strong knowledge of $Y\mid X,Z$ which floodgate does not rely on, and OLS is not valid without that knowledge (and does not apply in high dimensions). Another example is when we can assume the group mMSE gap of all of $(X,Z)$ is bounded away from zero, in which case the method of \citet{williamson2020unified} applies, so in Section~\ref{sec:comp_vimp} we compare their method with floodgate in such a setting.} 
    
    \acc{\begin{premark}[Floodgate's connection to conditional independence testing]
    %Lastly, we comment on floodgate's connection to conditional independence testing. 
    Recall that $\jnull$ implies $\Ij=0$, and hence rejecting $\jnull$ when $L_{n}^{\alpha} (\mu)>0$ constitutes an asymptotically valid level-$\alpha$ conditional independence test (which could then be combined with a multiple testing procedure to perform variable selection). However, floodgate was explicitly designed to solve the harder problem of quantifying strength of dependence, as opposed to the conditional independence problem of whether any dependence exists at all. Due to the methodological constraints imposed by the more challenging nature of our problem, especially the need for data splitting, we do not expect this test derived from floodgate to be competitive with (and hence do not compare with) the many excellent conditional independence tests available in the literature (see, e.g., \citet{candes2016panning, DH-LJ:2019, berrett2019permutation, ML-LJ:2020, barber2020testing, tansey2022holdout,fukumizu2008kernel,zhang2011kernel,wang2015conditional,shah2018hardness, park2020measure,huang2020kernel}).
    \end{premark}
    }

    We always take the significance level $\alpha=0.05$, and all results are averaged over $64$ independent replicates unless stated otherwise (although in most cases each plotted point is averaged over multiple covariates per replicate as well, since we apply floodgate to each covariate in turn in each replicate).

    \subsection{Effect of sample splitting proportion}
    \label{sec:simul_split}
    % As mentioned in Section~\ref{sec:computation}, practically we can do sample splitting to obtain a regression function estimator $\muhat$ from the first part of the data then plug $\muhat$ into Algorithm \ref{alg:MOCK} to produce lower confidence bounds using the second part of the data. The splitting proportion (i.e. $n_1/n$, where $n_{1}$ is the number of samples used for fitting $\muhat$ and $n_2=n-n_1$ samples are directly used in Algorithm \ref{alg:MOCK}) will play a impact on the performance of our method. In Figure $\ref{fig:split}$, splitting proportions ranging from $5\%$ to $95\%$ with $5\%$ as the step size are considered. As quantified by the statistical accuracy result in Section~\ref{sec:accuracy}, the precision is determined by two parts: $(1)$. the gap between $\thetamuh$ and $\cal{I}_j$, which relates to the mean squared prediction error of $\muhat$; $(2)$. the usual length of confidence intervals based on central limit theorem, which is of order $O_p(1/n_2)$ in our sample splitting context. Intuitively, with a fixed budge of total available samples, the more data are used for model fitting, we are more likely to do good in terms of the first part. At the same time, $n_2$ will be smaller, thus resulting a loss of precision in the second part.
    %%Lucas' version of previous paragraph:
    As mentioned in Section~\ref{sec:mock}, we can split a fixed sample size $n$ into a first part of size $n_e$ for estimating $\mustar$ and use the remaining $n-n_e$ samples for floodgate inference via Algorithm~\ref{alg:MOCK}. The choice of $n_e$ represents a tradeoff between higher accuracy in estimating $\mustar$ (larger $n_e$) and having more samples available for inference (smaller $n_e$).  
    
    {In Figure \ref{fig:split}, we vary the sample splitting proportion and plot the average half-widths of floodgate LCBs of non-null covariates under distributions with the linear and the nonlinear $\mustar$ described in Section~\ref{sec:simul_setup}. Corresponding coverage plots and additional plots with different simulation parameters can be found in Appendix~\ref{app:sim:cover}. 
    Our main takeaway from these plots is that, while the optimal choice of splitting proportion varies between distributions and algorithms, the choice of 0.5 seems to frequently achieve a half-width close to the optimum. 
    %As one would expect, however, as the signal or sample size grows, there are diminishing returns to $n_e$, and the optimal sample split for some algorithms moves to the left. 
    Acknowledging that in some circumstances a more informed choice than 0.5 can be made, we nevertheless choose 0.5 as the default splitting proportion throughout the rest of our simulations.}
    
    \begin{figure}[tb]
        \centering
        \includegraphics[width = 1\linewidth]{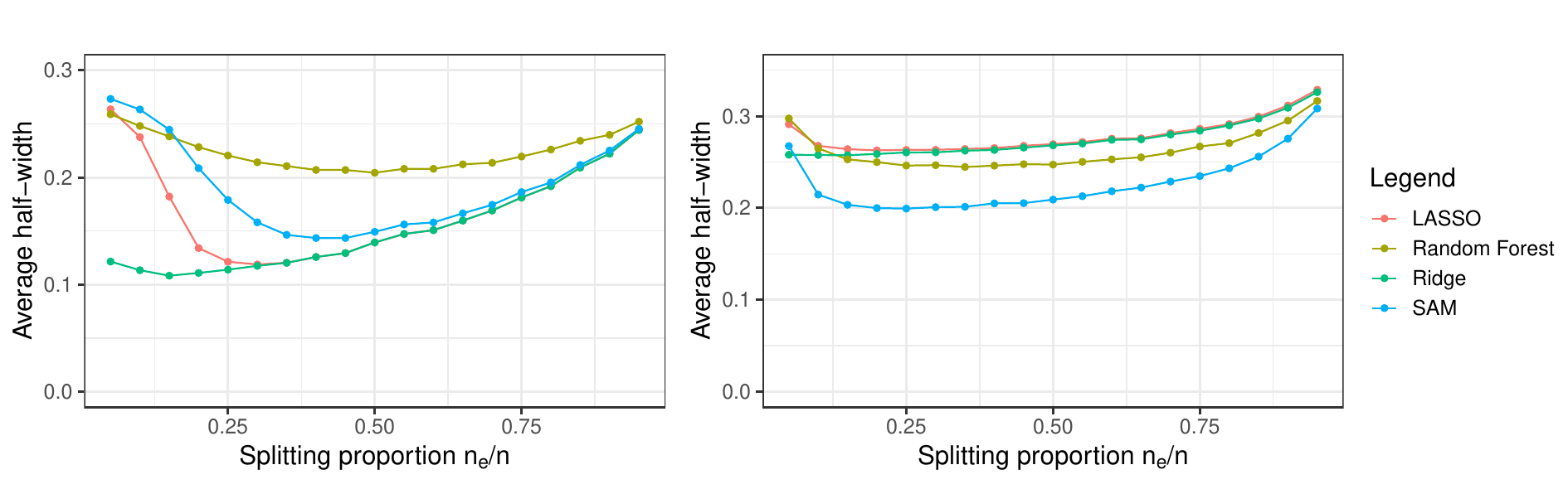}
        \caption{{Average half-widths for the linear-$\mustar$ (left) and nonlinear-$\mustar$ (right) simulations of Section~\ref{sec:simul_split}. The coefficient amplitude is chosen to be 10 for the left panel and the sample size $n$ equals 3000 in the right panel; see Section~\ref{sec:simul_setup} for remaining details. Standard errors are below 0.005 (left) and 0.006 (right).}}
%        Half width plots: design matrix with $n$ i.i.d. rows from an AR(1) model with autocorrelation $0.3$; $Y|X \sim \calN(X\beta,1)$, where $\beta$ has non-zero entries with random signs and equal absolute values (which will be the amplitude value divided by $\sqrt{n}$); there are $30$ non-null variables; $p = 1000$, $n = 1100$; number of null copies $K = 500$ for nonlinear fitting algorithms.

        \label{fig:split}
    \end{figure}

    In addition to displaying the dynamics of sample splitting proportion, these plots also demonstrate two other phenomena. First, the linear algorithms (LASSO and Ridge) dominate when $\mustar$ is linear, and the nonlinear algorithms (SAM and Random Forest) dominate when $\mustar$ is nonlinear. Second, Ridge has smaller half-width than LASSO for all sample splitting proportions, which can be explained by floodgate's invariance to (partially-)linear $\mu$: all that matters is getting the sign of the coefficient right, and setting a coefficient to zero guarantees a zero LCB. So the LASSO suffers from being a sparse estimator, although in practice we may still prefer it because of the corresponding computational savings of only having to run floodgate on a subset of covariates.

     \subsection{Effect of covariate dimension}
      \label{sec:simul_vary_p}
    {To understand the dependence of dimension on floodgate, we perform simulations varying the dimension. In particular, in the first panel of Figure \ref{fig:vary_p_comb}, we vary the covariate dimension and plot the average half-widths of floodgate LCBs of non-null covariates when $\mustar$ is linear. This setting enables comparison with OLS because it is linear and low-dimensional, so we also include a curve for OLS. 
    % The second panel of Figure~\ref{fig:vary_p_comb} is similar except with a smaller $n_e$ that is favorable for the linear algorithms in floodgate. 
    
    The main takeaway is that floodgate's accuracy is relatively unaffected by dimension, and although for very low dimensions (where OLS is known to be essentially optimal) it is less accurate than OLS, for a good choice of $n_e$ floodgate's half-widths are at most about 50\% larger than OLS's and actually narrower than OLS's when $p\approx n/2$. A similar message is found with nonlinear $\mustar$ in the second panel of Figure~\ref{fig:vary_p_comb}, except OLS no longer applies and in this case the nonlinear algorithms outperform the linear ones in floodgate. Coverage plots corresponding to Figure \ref{fig:vary_p_comb} and additional plots with different simulation parameters can be found in Appendix~\ref{app:sim:cover}.}
    % The coverage never falls below nominal, and these plots can be found in Appendix~\ref{app:sim:cover}. 
    \begin{figure}[tb]
        \centering
        \includegraphics[width = 1\linewidth]{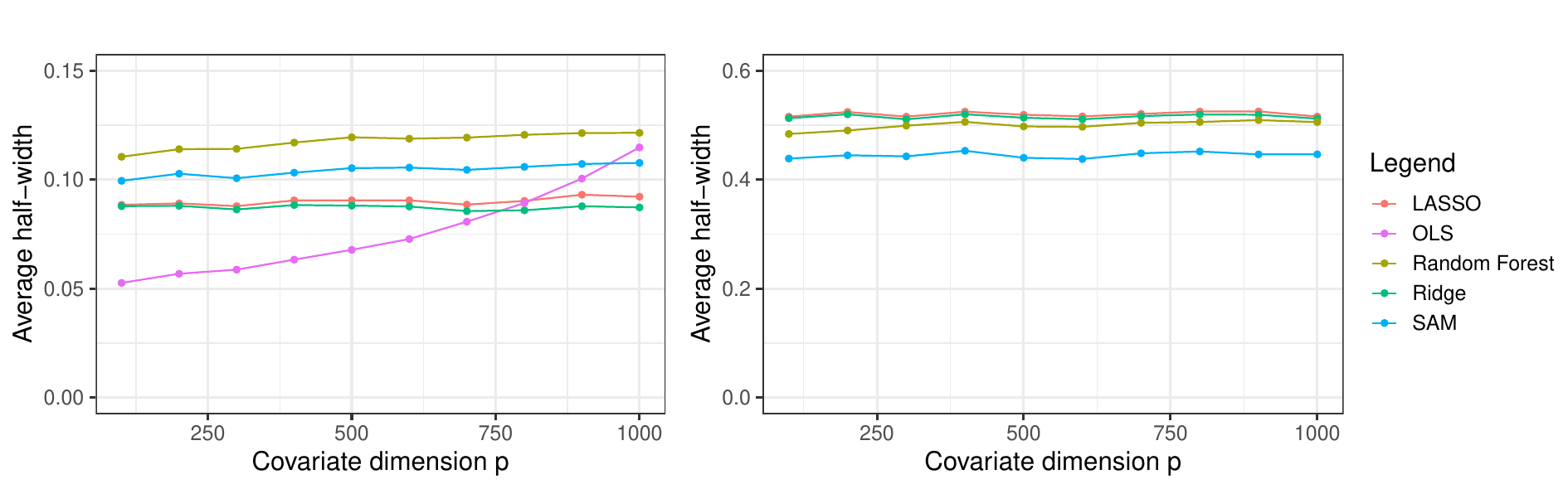}
        \caption{{Average half-widths for the linear-$\mustar$ (left) and nonlinear-$\mustar$ (right) simulations of Section~\ref{sec:simul_vary_p}. OLS is run on the full sample. $p$ is varied on the x-axis; see Section~\ref{sec:simul_setup} for remaining details. Standard errors are below 0.002 (left) and 0.008 (right).}
        % Half width plots: design matrix with $n$ i.i.d. rows from an AR(1) model with autocorrelation $0.3$; $Y|X \sim \calN(X\beta,1)$, where $\beta$ has non-zero entries with random signs and equal absolute values (which will be the amplitude value divided by $\sqrt{n}$); amplitude equals $5$; there are $30$ non-null variables; $n = 1100$; number of null copies $K = 500$ for nonlinear fitting algorithms.
        }
        \label{fig:vary_p_comb}
    \end{figure}       
    
\subsection{Comparison with \citet{williamson2020unified}}
    \label{sec:comp_vimp}
    {Although \citet{williamson2020unified}'s method {(which we refer to as W20b)} is only valid when the group mMSE gap of all the covariates is bounded away from zero, we can compare it with floodgate in that setting. 
    %In the following, method is called W20b. 
    We use W20b according to that paper's instructions for ensuring validity for any value of $\Ij$ {(as long as the group mMSE gap for all the variables put together is bounded away from zero), which seems most comparable to floodgate. That is,} we implement the sample-split and cross-fitted version using the default function \texttt{vimp\_rsquared} in the W20b authors' R package \texttt{vimp} (version 2.1.0). Since W20b gives confidence intervals for $\Ij^2/\Var{Y}$, we transform its inference into a $1-\alpha$ coverage LCB for $\Ij$ by taking the lower bound from its $1-2\alpha$ confidence interval, multiplying it by $\Var{Y}$, and then taking the square root.}
    % Since W20b gives confidence intervals for $\Ij^2/\Var{Y}$, we transform its inference into a $1-\alpha$ coverage LCB for $\Ij$ by taking the lower bound from its $1-2\alpha$ confidence interval, multiplying it by $\Var{Y}$, and then taking the square root. %In Figures \ref{fig:comp_freq} and \ref{fig:comp_p}, we show the coverage and average LCB values of both methods. By plotting the average LCB values it is easy to compare the informativeness of the two LCBs relative to the actual value of $\Ij$. 
    % Note that the theory underpinning W20b implies it should have coverage \emph{exactly} $1-\alpha$ asymptotically, as opposed to floodgate's guarantee of just having asymptotic coverage \emph{at least} $1-\alpha$. Nevertheless, in the settings we consider in this subsection, we find W20b's coverage is sometimes below $1-\alpha$ and often far above it (while floodgate's remains at or above $1-\alpha$ throughout, as expected). And in all settings, even when W20b's coverage is below that of floodgate, floodgate's LCB is considerably farther from zero than W20b's, thus floodgate provides more-informative inference than W20b in these simulations.
Our simulation example uses a sine function of varying frequency for $\mustar$. In particular, $p=2$, the covariates $(X,Z)\in\mathbb{R}^2$ are i.i.d. uniformly distributed on $(-1,1)$, and $Y$ equals $A(\lambda)\sin(\lambda X)$ plus standard Gaussian noise, where
$\lambda>0$ controls the frequency and $A(\lambda)$ is chosen so that $\Ij=0.5$ regardless of $\lambda$ {(thus ensuring the group mMSE gap of $(X,Z)$ is always bounded away from zero, as required by W20b)}. 
%where $A>0$ is the amplitude and $\lambda$ is the frequency of the sine function. Hence $W_1$ is non-null and $W_2$ is null. 
Both floodgate and W20b must internally fit an estimate of $\mustar$, and for both methods we use locally-constant loess smoothing with tuning parameters selected by 5-fold cross-validation, following a different two-dimensional simulation example from \citet{williamson2017nonparametric}.

\begin{figure}
\centering
        \includegraphics[width = 0.54\linewidth]{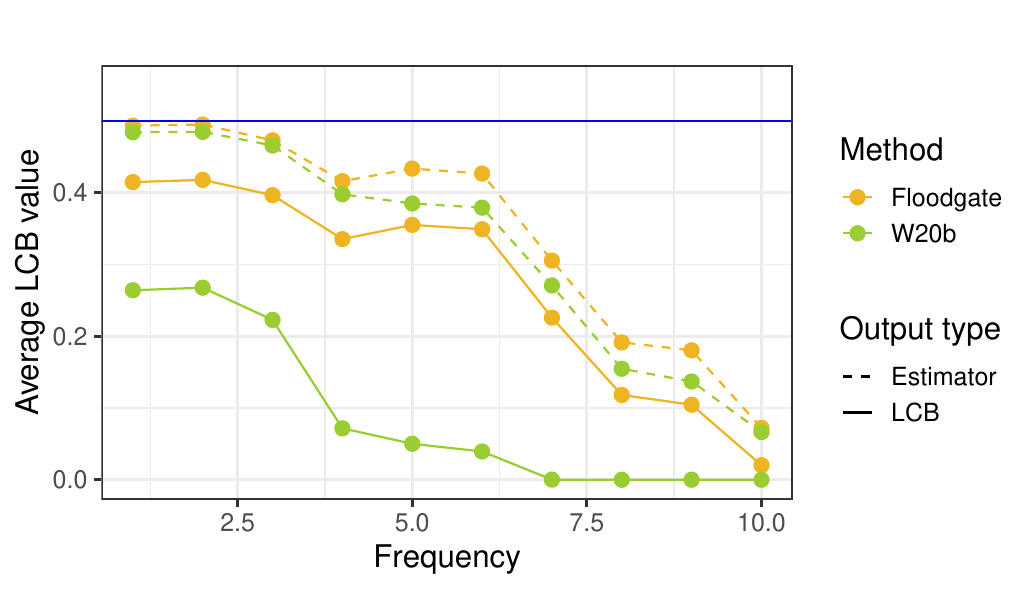}
        \vspace{0.25cm}
        % \vspace{-0.45cm}
        \caption{
        Average LCB values (solid lines) for floodgate and W20b in the sine function simulation of Section~\ref{sec:comp_vimp}. The frequency $\lambda$ is varied on the x-axis, and the solid blue line in the plot shows the true value of $\Ij$. The dashed lines correspond to the average estimator values of $\Ij$.
        %, while the frequency value $\lambda$ is varied on the x-axis. 
        The results are averaged over 640 independent replicates, and the standard errors are below 0.01.
             %(left) and 0.01 (right).
        }
        \label{fig:comp_freq}
\end{figure}

{The solid curves in Figure \ref{fig:comp_freq} show the average LCBs of the two methods applied to the non-null variable $X$} as $\lambda$ varies. %we vary the frequency value $\lambda$ from 1 to 10 and fix the signal-to-noise ratio to be $0.2$ (thus $\Ij=0.5$) by setting the amplitude value $A$ to be %$0.5/\sqrt{0.5 - \sin(2\lambda)/(4\lambda)}$. 
%The corresponding coverage plots for both the null and the non-null variables are deferred to Appendix \ref{app:sim:cover}. 
Larger $\lambda$ corresponds to less-smooth $\Ec{Y}{X,Z}$ and hence a more challenging estimation problem (for both methods), and both methods become generally more conservative and less accurate as $\lambda$ grows {(both methods achieve at or above nominal coverage throughout this simulation; see Appendix~\ref{app:sim:cover} for the coverage plot)}. Yet floodgate's LCB provides consistently and considerably more accurate inference over the entire range of $\lambda$. To better understand this performance difference, we additionally plot as dashed curves the average of the asymptotically normal estimators of $\Ij$ each method uses for inference. We see from the plot that the two estimators have similar bias, but the gap between the LCB and the estimator is much smaller for floodgate, reflecting a smaller variance. This is likely due to the form of W20b's estimator, which is the difference of two asymptotically normal test statistics, one computed on each half of the split data. Heuristically, one would expect this to lead to higher variance than an estimator computed on (and hence whose variance comes only from) one half of the data, like floodgate's.
% We find this sample splitting
%In Figure \ref{fig:comp_freq}, by checking the dashed lines corresponding to the average estimators, we find both methods have similar performances in producing the estimators of the
This general picture is reinforced by a {higher-dimensional simulation} given in Appendix \ref{app:sim:cover}.

    \subsection{Robustness}
    \label{sec:simul_robust}  
    In order to study the robustness of floodgate to misspecification of $P_{X|Z}$, we consider a scenario we expect to arise in practice: a data analyst does not know $P_{X|Z}$ exactly, so instead they estimate it using the data they have, and then treat the estimate as the ``known" $P_{X|Z}$ and proceed with floodgate. Note that if the analyst splits the data and uses the same subset for estimating $\mu$ and for estimating $P_{X|Z}$, then Theorem~\ref{thm:robust_no_adj} applies, but if they use \emph{all} of their data to estimate $P_{X|Z}$, then our theory does not apply. Also note we are not studying the performance of co-sufficient floodgate in this subsection.

    Note that if the analyst splits the data and uses the same subset for estimating $\mu$ and for estimating $P_{X|Z}$, then Theorem~\ref{thm:robust_no_adj} applies, but if they use \emph{all} of their data to estimate $P_{X|Z}$, then our theory does not apply. Also note we are not studying the performance of co-sufficient floodgate in this subsection.
    
    {Figure \ref{fig:robust_linear} varies how much in-sample data is used in $P_{X|Z}$-estimation and shows the coverage of floodgate for null and non-null variables in a linear setting. {The estimation procedure is to fit the graphical LASSO (GLASSO) with $3$-fold cross-validation to a subset of the in-sample data and treat $P_{X|Z}$ as conditionally Gaussian with covariance matrix given by the GLASSO estimate.} Since $n=1100$ in all these simulations and the sample splitting proportion is 0.5, when the x-axis value passes 550 is when {the $P_{X|Z}$-estimation and inference sets start to overlap, and at the value 1100, all of the data is being used to estimate $P_{X|Z}$, including the half used for inference (violating Theorem~\ref{thm:robust_no_adj}'s assumptions)}. Nevertheless, we see the coverage is consistently quite high, \acc{only dropping slightly from that with true $P_{X\mid Z}$ for very low estimation sample sizes (i.e., very bad estimates of the covariance matrix). Note that some $\mu$-fitting algorithms in Figure \ref{fig:robust_linear} have higher-than-nominal coverage; this is largely because the floodgate procedure will deterministically output a zero LCB (and hence have 100\% coverage) when $\mu(x, z)$ does not depend on $x$. This happens for many covariates when $\mu$ is fitted via a sparse regression such as the LASSO and SAM (short for Sparse Additive Models), but also for our version of Random Forests which we effectively sparsify for computational reasons (see Appendix~\ref{app:sim:algorithms} for details). Figures \ref{fig:binary} and \ref{fig:sine_relax2} show similar overcoverage for the same reason.}
%    As a result, we will observe over-coverages in the corresponding simulation figures, as indicated in Section \ref{sec:mock}.
%    }    

    Average half-width plots corresponding to Figure~\ref{fig:robust_linear} can be found in Appendix~\ref{app:sim:cover}. In additional to the linear setting in Figure \ref{fig:robust_linear}, we also observe robust empirical coverage of floodgate when the conditional model of $Y$ is nonlinear; see Appendix \ref{app:sim:cover} for details.} 
    
    %  As we have quantitatively conducted robustness analysis in Section~\ref{sec:robustness}, we also empirically demonstrate the robustness of our method. In practice, the distribution of $X$ has to be estimated from available data. In this section, we consider two different data sources: either we can estimate from unlabelled data or reuse the labelled data for modeling $X$. As the generating model of $X$ is chosen to be multivariate Gaussian, we can either plug in the empirical covariance matrix as use it as the distribution of $X$ when $n>p$, or use the graphical lasso to estimate the covariance structure. 
     
    \begin{figure}[tb]
        \centering
        \includegraphics[width =1\linewidth]{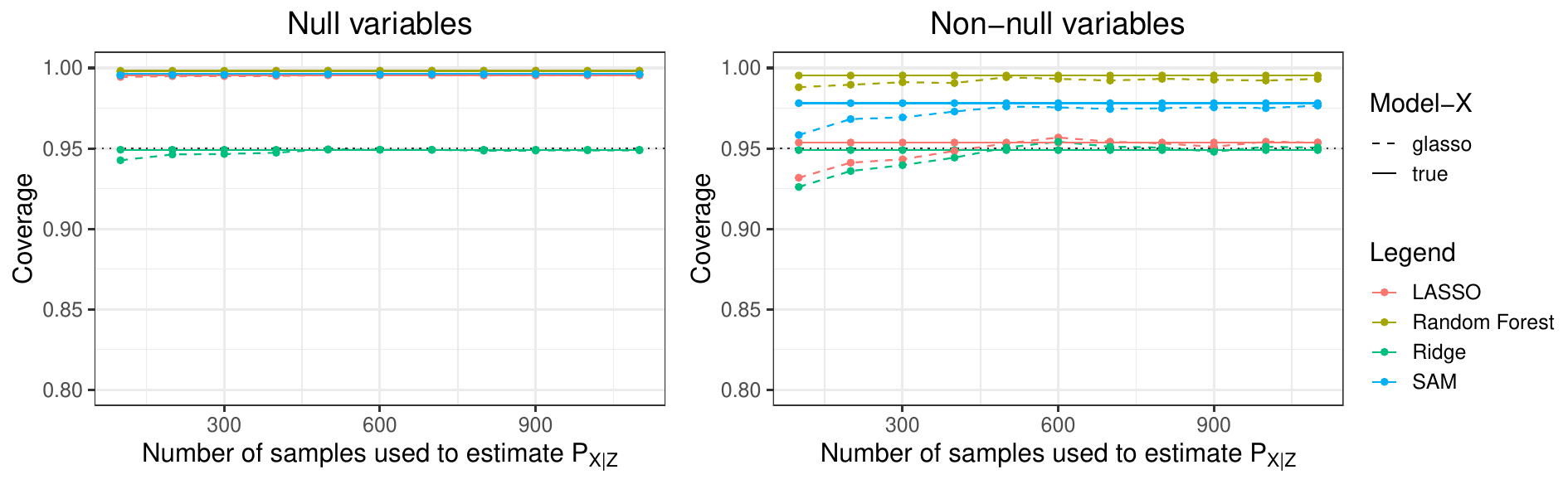}
        \caption{Coverage of null (left) and non-null (right) covariates when the covariate distribution is estimated in-sample for the linear-$\mustar$ simulations of Section~\ref{sec:simul_robust}. See Section~\ref{sec:simul_setup} for remaining details. Standard errors are below 0.001 (left) and 0.008 (right).
        % Coverage plots for null and non-null variables, with the sample size for model-X part varying; design matrix with $n$ i.i.d. rows from an AR(1) model with the auto-correlation coefficient being $0.3$; $Y|X \sim \calN(X\beta,1)$, where $\beta$ has non-zero entries with random signs and equal absolute values (which will be the amplitude value divided by $\sqrt{n}$); there are $30$ non-null variables; amplitude equals $5$ ; $n = 1100$; number of null copies $K = 500$ for nonlinear fitting algorithms. 
        }
        \label{fig:robust_linear}
    \end{figure}

    \subsection{Floodgate for the MACM gap}
    \label{sec:simul_binary}
%    \subsubsection{Simulations for MACM gap}
%\ljmargin{Figure~\ref{fig:binary} shows that the coverage of floodgate applied to the MACM gap as described in Algorithm~\ref{alg:class_MOCK} in Section~\ref{sec:class} has consistent coverage over a range of algorithms for fitting $\mu$, and we see the dynamics of the average half-width as the explained variance proportion in $P_{Y|X,Z}$ increases.}{Lu please fill in details of the simulation and caption of the figure. I think this should go in the main paper at the end of Section 4.}

Here we study the empirical performance of floodgate applied to the MACM gap as described in Section~\ref{sec:class}. Conditional on the covariates, the binary response is generated from a logistic regression with $\frac{\log(\Pc{Y=1}{X,Z})}{\log(\Pc{Y=-1}{X,Z})}$
%natural parameter 
given by
%$2~\text{Bernoulli}({e^{\eta}}/{(1+e^{\eta})}) -1$ with $\eta$ being similarly defined as 
the linear $\mustar(X,Z)$ in Section~\ref{sec:simul_setup}. We set the sample size $n=1000$, and the remaining simulation parameters to be the values described in Section~\ref{sec:simul_setup}. Figure~\ref{fig:binary} shows that floodgate has consistent coverage over a range of algorithms for fitting $\mu$, and we see the dynamics of the average half-width as the explained variance proportion in $P_{Y|X,Z}$ increases. Note that $R_i$ in Algorithm \ref{alg:class_MOCK} needs to in general be estimated by Monte Carlo samples (see Appendix \ref{app:class_computation} for details) and in Figure~\ref{fig:binary}, we set $K = 100$ and $M = 400$ whenever the Monte Carlo version is used. 

      \begin{figure}[tb]
        \centering
        \includegraphics[width =1\linewidth]{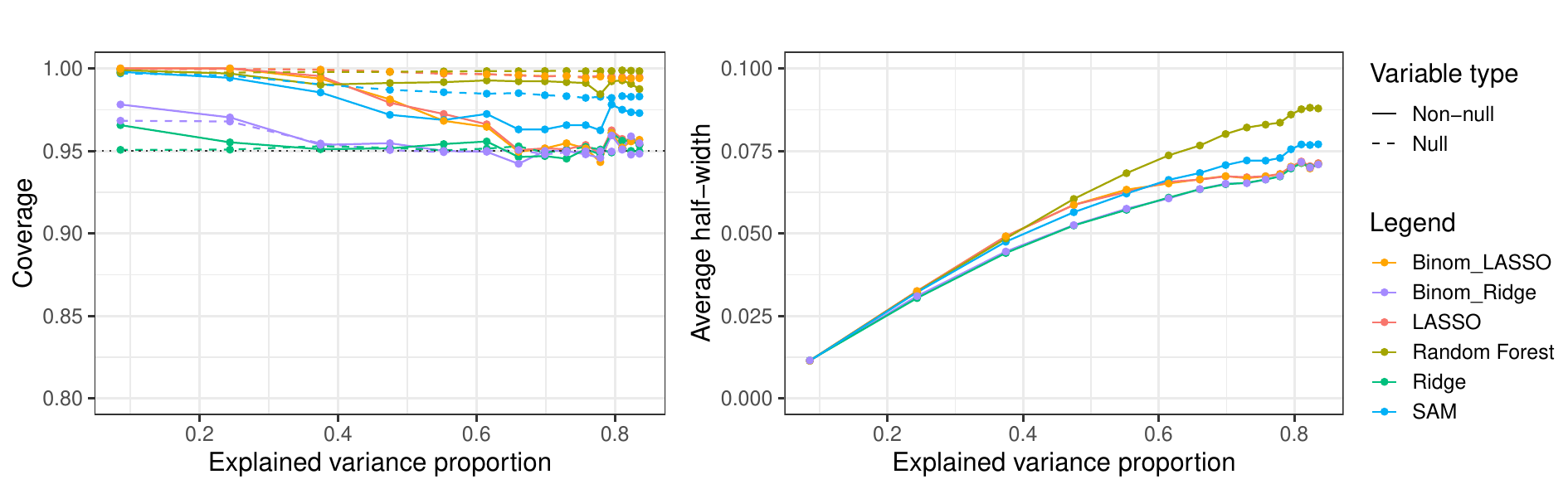}
        \caption{Coverage (left) and average half-widths (right) for the binary response simulations of Section~\ref{sec:simul_binary}. The explained variance proportion is varied over the x-axis. See Section~\ref{sec:simul_setup} and \ref{sec:simul_binary} for remaining details. Standard errors are below 0.006 (left) and 0.001 (right).
        % Coverage plots for null and non-null variables, with the sample size for model-X part varying; design matrix with $n$ i.i.d. rows from an AR(1) model with the auto-correlation coefficient being $0.3$; $Y|X \sim \calN(X\beta,1)$, where $\beta$ has non-zero entries with random signs and equal absolute values (which will be the amplitude value divided by $\sqrt{n}$); there are $30$ non-null variables; amplitude equals $5$ ; $n = 1100$; number of null copies $K = 500$ for nonlinear fitting algorithms. Coverage of null (left) and non-null (right) covariates when the covariate distribution is estimated in-sample for the linear-$\mustar$ simulations of Section~\ref{sec:simul_robust}. See Section~\ref{sec:simul_setup} for remaining details.
        }
        \label{fig:binary}
    \end{figure}
    
     \subsection{Co-sufficient floodgate}
     \label{sec:simul_relax}
     Finally, we study the empirical performance of co-sufficient floodgate as described in Section~\ref{sec:relax} as compared to the original floodgate method which is given full knowledge of $P_{X|Z}$.
    %  , co-sufficient floodgate relaxes the assumption to only knowing $X$'s distribution up to a parametric model. Here we demonstrate its empirical performance in contrast to the original floodgate procedure. 
     We set the covariate dimension $p=50$, the number of Monte Carlo samples $K=100$, and the amplitude value for nonlinear-$\mustar$ to $30$. The remaining simulation parameters are set to the values described in Section~\ref{sec:simul_setup}. Co-sufficient floodgate and the original floodgate procedure use the same \revision{working regression function}, fitted from $n_e = 500$ samples, and use the same number of samples $n-n_e$ for inference. 
     %As for the number of inference samples, both methods have the same budget. 
     %But co-sufficient floodgate spends them in the form of batches. 
     The batch size $n_{2}$ for co-sufficient floodgate is $300$ and we vary the number of batches $n_1 = (n-n_e)/n_2$ on the $x$-axes.
     %of Figures \ref{fig:relax} and \ref{fig:sine_relax}. 
     Co-sufficient floodgate is given the conditional variance of the Gaussian distribution of $X\mid Z$, but not its conditional mean, parameterized by a $(p-1)$-dimensional coefficient vector multiplying $Z$.
    %  Each time when focusing on one covariate denoted by $X$, with the others by $Z$, the conditional distribution of $X\mid Z$ is Gaussian, we only assume knowing the variance of that Gaussian distribution when running co-sufficient floodgate. 
     Figure \ref{fig:sine_relax2} shows that co-sufficient floodgate has satisfying coverage even when the number of batches is small, and has average half-width quite close to the original floodgate procedure which is given the conditional mean of $X\mid Z$ exactly. {In additional to the nonlinear setting in Figure \ref{fig:sine_relax2}, simulations for a linear $\mustar$ lead to similar conclusions; see Appendix \ref{app:sim:cover}.}
     
    %     \begin{figure}[tb]
    %     \centering
    %     \includegraphics[width =1\linewidth]{Pics/Relax/sine_relax.pdf}
    %     \caption{Coverage (left) and average half-widths (right) for co-sufficient floodgate against the original floodgate in the nonlinear-$\mustar$ simulations. The number of batches $n_{1}$ is varied over the x-axis. See Section~\ref{sec:simul_setup} and \ref{sec:simul_relax} for remaining details. Standard errors are all below 0.008.
    %     % Coverage plots for null and non-null variables, with the sample size for model-X part varying. Design matrix with $n$ i.i.d. rows follows the Gaussian copula model in  \ref{sec:nonlinear_setup} with the auto-correlation coefficient being $0.3$; $Y|X$ follows the nonlinear model described in \ref{sec:nonlinear_setup} with amplitude being $50$ and $30$ non-null variables; $n = 1100, p = 1000$; number of null copies $K = 500$ for nonlinear fitting algorithms.
    %     }
    %     \label{fig:sine_relax}

    % \end{figure}
    
            \begin{figure}[tb]
        \centering
        \includegraphics[width =1\linewidth]{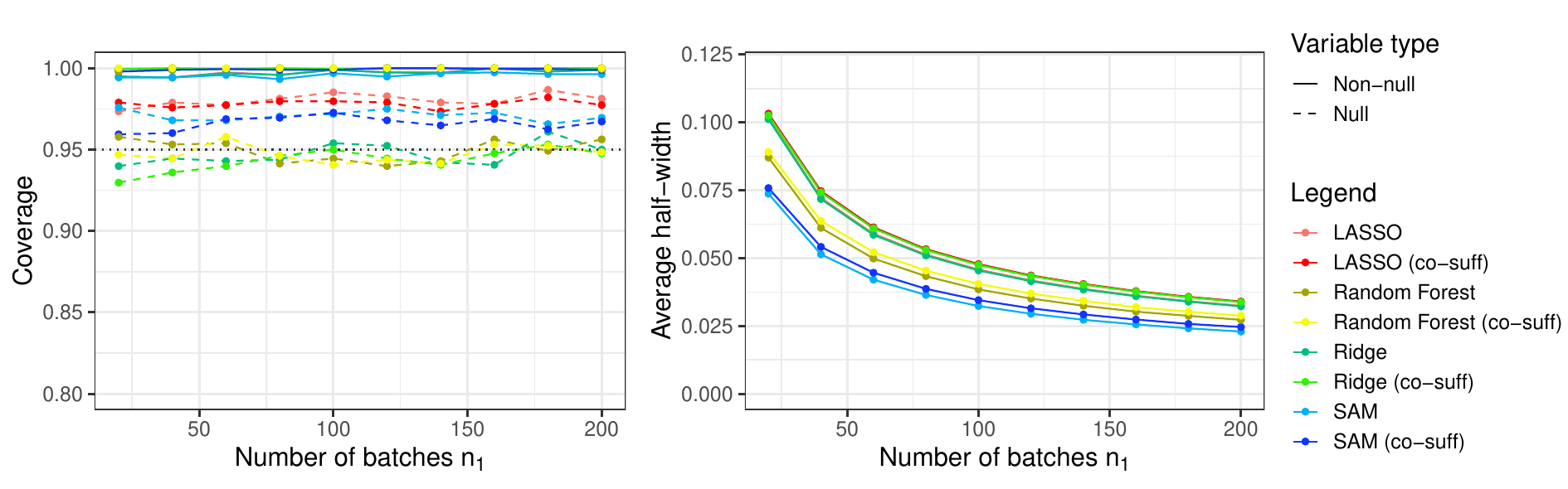}
        \caption{Coverage (left) and average half-widths (right) for co-sufficient floodgate and original floodgate in the nonlinear-$\mustar$ simulations. The number of batches $n_{1}$ is varied over the x-axis. See Section~\ref{sec:simul_setup} and \ref{sec:simul_relax} for remaining details. Standard errors are below 0.009 (left) and 0.002 (right).
        % Coverage plots for null and non-null variables, with the sample size for model-X part varying. Design matrix with $n$ i.i.d. rows follows the Gaussian copula model in  \ref{sec:nonlinear_setup} with the auto-correlation coefficient being $0.3$; $Y|X$ follows the nonlinear model described in \ref{sec:nonlinear_setup} with amplitude being $50$ and $30$ non-null variables; $n = 1100, p = 1000$; number of null copies $K = 500$ for nonlinear fitting algorithms.
        }
        \label{fig:sine_relax2}

    \end{figure}

    \section{Application to genomic study of platelet count}\label{sec:realdata}
     %To link this work with genomics, we first define some notation. Let $G = (G_1,\dots,G_p)$ denote a genotype, where $G_j$ represents the $j$th single nucleotide polymorphism (SNP), let $Y$ be a trait of interest such as platelet count, and let $E=(E_1,\dots,E_q)$ denote a set of non-genetic factors such as age.
    The study of genetic \emph{heritability} is the study of how much variance in a trait can be explained by genetics. Precise definitions vary based on modeling assumptions \citep{zuk2012mystery}, but the fundamental concept is intuitive and central to genomics; indeed the goal of genome-wide association studies (GWAS) is often precisely to identify single nucleotide polymorphisms (SNPs) or loci that explain the most variance in a trait. To connect heritability with the present paper, suppose $Y$ denotes a trait, $X$ denotes a SNP or group of SNPs, and $Z$ denotes all the remaining SNPs not included in $X$. Then {as can be seen in Equation~\eqref{eq:mMSEgap_form2},} the mMSE gap $\Ij^2$ \emph{exactly} measures the variance in $Y$ that is attributable to $X$. Thinking of $\Ij^2$ as a sort of \emph{conditional} heritability also makes it easy to include non-genetic factors such as age in $Z$, since such factors may influence $Y$ but not be of {direct} interest to geneticists. Thus $\Ij^2$ can capture both gene-gene and gene-environment interactions. %\ljmargin{Formally, }{needed? how to phrase?}
    
    Having established $\Ij^2$ as a quantity of interest, we proceed to infer it for blocks of SNPs at various resolutions of the human genome by applying floodgate to a platelet GWAS from the UK Biobank. Our analysis builds on the work of \citet{sesia2020multi}, which carefully applied model-X knockoffs to the same data to perform multi-resolution \emph{selection} of important SNPs, and in doing so require, like floodgate, a model for the SNPs $X,Z$ and a working regression function, both of which we reuse in our own analysis. In particular, we follow the literature on genotype/haplotype modeling \citep{stephens2001new, zhang2002dynamic, li2003modeling, scheet2006fast, sesia2019gene, sesia2020controlling, sesia2020multi} and model the SNPs as following a  hidden Markov model, and use the cross-validated Lasso as the algorithm to fit our working regression function $\mu$. Although we use a linear $\mu$ to match the existing analysis in \citet{sesia2020multi}, we remind the reader that one is in general free to use any $\mu$ with floodgate, and we hope that domain experts applying floodgate in the future to GWAS data can tailor $\mu$ to be even more powerful. The output of the analysis in \citet{sesia2020multi} is a so-called ``Chicago plot", which plots stacked blocks of selected SNPs at a range of block resolutions. The height of the Chicago plot at a given location on the genome reflects the resolution at which the SNP at that location was rejected, with a greater height corresponding to a smaller block of SNPs being rejected. However, since the Chicago plot is derived from a pure selection method, it contains no information about the \emph{strength} of the relationship between the trait and any of the blocks of SNPs. Floodgate enables us to construct a \emph{colored} Chicago plot by computing an LCB for each selected block of SNPs and reporting an LCB of zero (without computation) for all unselected blocks of SNPs; see Appendix~\ref{app:realdata} for implementation details.
    
    \begin{figure}[tb]
        \centering
         \includegraphics[width=1\linewidth]{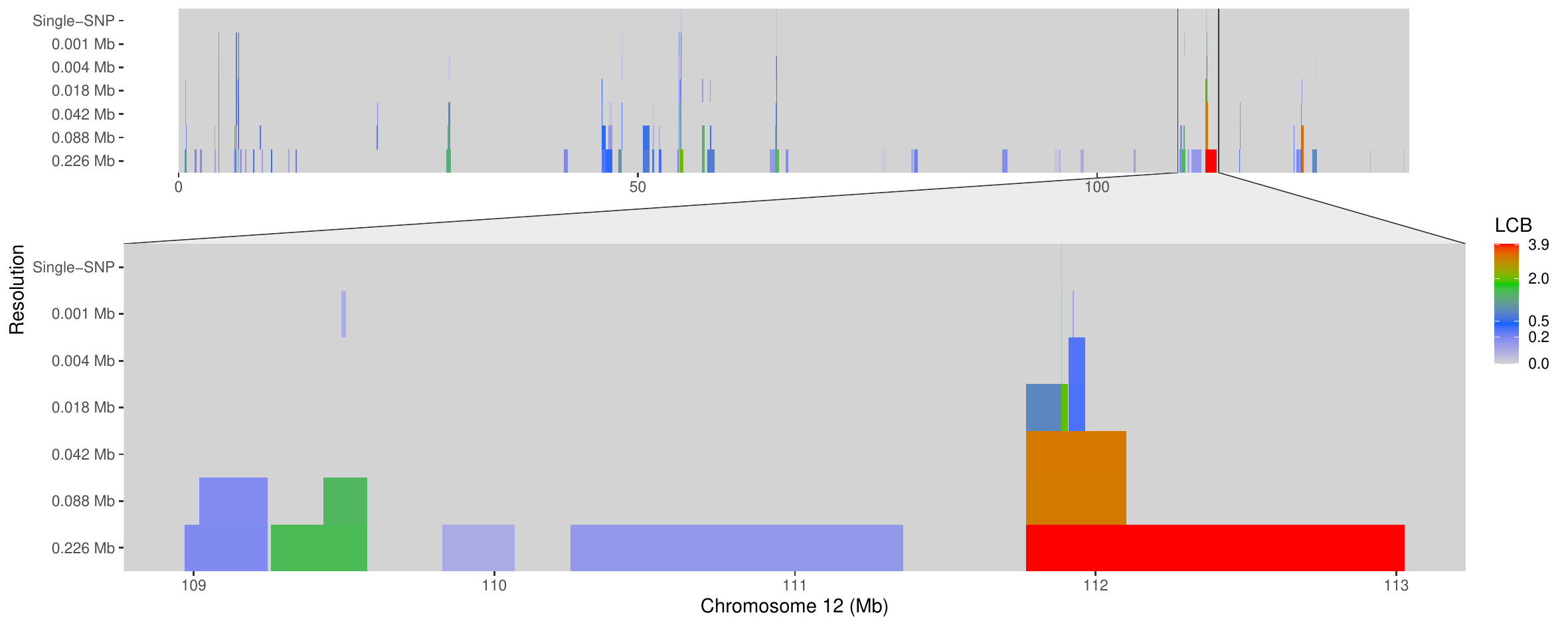}
        \caption{Colored Chicago plot analogous to Figure 1a of \citet{sesia2020multi}. The color of each point represents the floodgate LCB for the block that contains the SNP at the location indicated on the x-axis at the resolution (measured by average block width) indicated on the y-axis (note some blocks appearing in the original Chicago plot have an LCB of zero and hence are colored grey). The second panel zooms into the region of the first panel containing the largest floodgate LCB.
        % {\color{red}get rid of plot titles, make colors non-transparent, make all text bigger, remove dead space to the left of 0 and beyond the right end point as well}
        }
        \label{fig:shanghai}
    \end{figure}
    
    In particular, Figure~\ref{fig:shanghai} is a colored version of Figure 1a of \citet{sesia2020multi}, which displayed the genomic regions on chromosome 12 that those authors found to be related to platelet count in the UK Biobank data. Our colored figure shows how informative floodgate LCBs can be over and beyond a pure selection method, as it shows the signal is far from being spread evenly over the SNPs selected by \citet{sesia2020multi}.
    %a small fraction of selected blocks show evidence of considerably more signal than the rest. 
    This information is crucial for the \emph{prioritization} of selected regions, as without color the Chicago plot does not give any indication which of the selected SNPs the data indicates are most important (we note that the height of the tallest selected block at a SNP need \emph{not} correspond to its importance, and indeed there are many pairs of locations in the figure such that one has a taller block in the original Chicago plot but the other has a brighter color in Figure~\ref{fig:shanghai}).
    % color floodgate coloring shows there is at best a weak correspondence between strong colors and tall ``towers").

    \section{Discussion}\label{sec:discussion}
    Floodgate is a powerful and flexible framework for rigorously inferring the strength of the conditional relationship between $Y$ and $X$. We prove results about floodgate's validity, accuracy, and robustness and address a number of extensions/generalizations, but a number of questions remain for future work and we highlight two here:
    \begin{itemize}
    \item Floodgate relies on a \revision{working regression function} that is not estimated from the same data used for inference, which usually will require data splitting. It would be desirable, both from an accuracy standpoint and a derandomization standpoint, to remove the need for data splitting or at least find a way for samples in one or both splits to be recycled between regression estimation and inference.
    \item The floodgate framework is applied here to the mMSE gap and the MACM gap, but more generally it constitutes a new tool for flexible inference of nonparametric functionals, and we expect it can find use for inferring other MOVIs. The main challenge for its application is the identification of an appropriate floodgate functional, and it would be of interest to better understand principles or even heuristics for finding such functionals for a given MOVI. Indeed we make no claim that the functionals proposed in this paper are unique for their respective MOVIs, and there may be others that lead to better floodgate procedures.
%    \item \cyancom{As mentioned in Section~\ref{sec:method}, there is a straightforward modification to floodgate, which can improve the inference in terms of the variance. It is of great interest how this and possibly other modifications to the method can promote robustness.}
    % \item Lastly, we expect floodgate to be a powerful tool for applications, as it enables the user to incorporate every bit of domain expertise and regression technology into improved accuracy. Given its model-X assumptions, floodgate may be most appealing at least at first for applications with available unlabeled data or other information about Extending co-sufficient floodgate from Section~\ref{sec:relax} beyond the two models considered there would further enhance floodgate's applicability.
    \end{itemize}
    % Future directions:
    
    % Get rid of data splitting
    
    % Other floodgate functionals for $\Ij$ (and perhaps better understanding of how to generate floodgate functionals and/or how to choose among them)
    
    % Floodgate for other estimands, e.g., quantities in terms of quantile regression.
    
    % Formal treatment of selective inference issues
    
    % Extending conditional floodgate to other models for $X\mid Z$
    
    % Other domain applications of floodgate
    
    % Other ways to condition on more to induce robustness or address confounding

    \section*{Acknowledgements}
    L.Z. is partially supported by the NSF-Simons Center for Mathematical and Statistical Analysis of Biology at Harvard, award number \#1764269 and the Harvard Quantitative Biology Initiative. L.J. is partially supported by the William F. Milton Fund. We would like to thank the Neale Lab at the Broad Institute of MIT and Harvard for including us in their application to the UK Biobank Resource (application 31063), Sam Bryant for the access to on-premise data files, Matteo Sesia, Eugene Katsevich, Asher Spector, Benjamin Spector, Masahiro Kanai and Nikolas Baya for the help with the genomics application, and Dongming Huang for helpful discussions.
    
\bibliographystyle{apalike}
\bibliography{mybib}

    \appendix
\newpage
    \section{Proofs for main text}
     Throughout the proofs, we will abbreviate $(X,Z)=W, ~(\Xtil,Z)=\Wtil$ for simplicity and write $w=(x,z)$. And  $\gstar,g:\mathbb{R}^{p-1}\rightarrow \mathbb{R};~\hstar,h:\mathbb{R}^{p}\rightarrow \mathbb{R}$ are defined as below:
    \begin{eqnarray}\label{eq:def_ggstar}
    \gstar(\xnoj) = \ec{\mustar(W)}{\Xnoj =\xnoj},~
    g(\xnoj) = \ec{\muW}{\Xnoj=\xnoj},\\
    \hstar(w) = \mustar(w) -\gstar(\xnoj),~
    h(w) = \mu(w) - g(\xnoj)\label{eq:def_hhstar}.
    \end{eqnarray}
    And we can further decompose $Y$:
    \begin{equation} \label{eq:4_lem:max}
        Y = \condmean + \eps(Y,X,Z) = \mustar(W) + \eps(Y,W) = \gstar(\Xnoj) + \hstar(W) + \eps(Y,W).
    \end{equation}
    \acc{Let $L_2(\Omega,\calF,P)$ denote the vector space of real-valued random variables with finite second moments, which is a Hilbert space, and define its subspace $L_2(W):=L_2(\Omega,\sigalg(W),P)$, where $\sigalg(W)$ is the sub $\sigma$-algebra generated by $W=(X,Z)$. ($L_2(\Xnoj):=L_2(\Omega,\sigalg(\Xnoj),P)$ is defined analogously).} \acc{Then $\mustar(W)$ and $\gstar(\Xnoj)$ can be interpreted as the projections of $Y$ onto the subspaces $L_2(W)$ and $L_2(\Xnoj)$, respectively. $Y$ and $\mustar(W)$ admit the orthogonal decompositions $Y = \mustar(W) + \eps(Y,W)$ and $\mustar(W) = \gstar(\Xnoj) + \hstar(W)$, respectively.  Similarly note the projection of $\mu(W)$ onto $L_2(\Xnoj)$ and the decomposition $\mu(W) = g(Z) + h(W)$.
    % we remark that $\gstar$ is the further projection of $\condmean$ onto the subspace $L_2(\Xnoj)$ and $\hstar$ is the orthogonal complement of $\gstar$. {\color{red} last sentence should be integrated with definitions of g's and h's.} 
    We remark these imply the following facts:
    \begin{align}\label{eq:5_lem:max}
    \begin{split}
           \ec{\eps(Y,W)}{W}=0,~\EE{\eps(Y,W) \lambda(W)}=0, \\
         \ec{\hstar(W)}{\Xnoj}=0,~\EE{\hstar(W) \gamma(Z)}=0, \ec{h(W)}{\Xnoj}=0, ~\EE{h(W) \gamma(Z)}=0.
    \end{split}
    \end{align}
    for any function $\lambda(w)$ and any function $\gamma(z)$.} 
    % \begin{equation}
    %  \ec{\eps(Y,W)}{W}=0,~\EE{\eps(Y,W) \lambda(W)}=0,~\ec{\hstar(W)}{\Xnoj}=0,~\EE{\hstar(W) \gamma(Z)}=0,
    % \end{equation}
    % Similarly considering $\mu(W) = g(Z) + h(W)$, we have 
    % \begin{equation}\label{eq:h_proj_fact}
    %     \ec{h(W)}{\Xnoj}=0, ~\EE{h(W) \gamma(Z)}=0
    % \end{equation}
    % for any $\gamma(z)$.
\acc{Thus we can rewrite the denominator of $\thetamu$ by noticing the equivalence below:
    \begin{equation}\label{eq:h2W}
    	\EE{\varc{\muX}{Z}} = \EE{\varc{h(W)}{Z}}  = \EE{\Ec{h^2(W)}{Z}} = \EE{h^2(W)}.
    \end{equation}
    As for the numerator of $\thetamu$, \eqref{eq:fnum} mentions the rewritten expression. Here we formally derive the following equivalent expressions of $\thetamu$,
    \begin{eqnarray}
        \thetamu \nonumber
        &:=& 
        \frac{\EE{\covc{\mustar(X,Z)}{\muX}{Z}}}{\sqrt{\EE{\varc{\muX}{Z}}}} \\ \nonumber
        &=& 
        \frac{\EE{\covc{\hstar(W)}{h(W)}{Z}}}{\sqrt{\EE{h^2(W)}}} \\ \label{eq:fmu_simplified}
        &=&\frac{\EE{\hstar(W)h(W)} }{\sqrt{\EE{h^2(W)}}} \\ \nonumber
        &=&\frac{\EE{Yh(W)} }{\sqrt{\EE{h^2(W)}}} - \frac{\EE{\eps(Y,W)h(W)} }{\sqrt{\EE{h^2(W)}}} - \frac{\EE{\gstar(Z)h(W)} }{\sqrt{\EE{h^2(W)}}} \\ 
        &=&\frac{\EE{Yh(W)} }{\sqrt{\EE{h^2(W)}}}  \label{eq:equiv_fmu}
    %  \EE{\covc{\mustar(X,Z)}{\muX}{Z}}=\EE{Y\big(\muX-\Ec{\muX}{Z}\big)}, 
    \end{eqnarray}
    where the second equality is by \eqref{eq:h2W} and the definitions of $\hstar(W),h(W)$}, the third equality holds by the total law of conditional expectation and \eqref{eq:5_lem:max}, the fourth equality comes from \eqref{eq:4_lem:max}, and the last equality holds due to \eqref{eq:5_lem:max} and the total law of conditional expectation. As \eqref{eq:equiv_fmu} is very concise, we will work with this expression of $\thetamu$ throughout the following proof. Also note we have an equivalent expression of $\Ij$.
     \begin{equation}\label{eq:6_lem:max}
    \sqrt{\EE{(\hstar)^2(W)}} = \sqrt{\EE{\Ecmid{\left(\mustar(W) -  \Ec{\mustar(W)}{\Xnoj}\right)^2}{\Xnoj}}} = \sqrt{\ANOVA} = \Ij.
    \end{equation}
    
          \acc{
          Note that the proofs of Theorems \ref{thm:main} and \ref{thm:main_general} only require moment conditions on $h(W)$, which will hold under the corresponding moment conditions on $\mu(X, Z)$. This can be seen from the following example where the finiteness of $\EE{\mu^r(W)}$ implies that of $\EE{h^r(W)}$ for some positive integer $r$:
%      We also note that 
%      Recall $h(W) = \mu(W) - \Ec{\mu(W)}{\Xnoj}$ as defined in \eqref{eq:def_hhstar}. We have
\begin{align}
\label{eq:element_mmt_bound}
\begin{split}
	\EE{h^r(W)} 
	&= \EE{ (\mu(W) - \Ec{\mu(W)}{\Xnoj} )^r}\\ 
	&\le  2^{r-1}(\EE{\mu^r(W)}+ \EE{(\Ec{\mu(W)}{\Xnoj})^r}) \\ 
	&\le  2^{r-1}(\EE{\mu^r(W)}+ \EE{\Ec{\mu^r(W)}{\Xnoj}}) =  2^r \EE{\mu^r(X, Z)},
\end{split}
\end{align}
where the first inequality holds due to the $C_r$ inequality (which states that $\EE{|X+Y|^r}\le C_r (\EE{|X|^r} + \EE{|Y|^r})$ with $C_r = 1$ for $0 < r \le 1$ and $C_r =2^{r-1}$ for $r\ge 1$), the second inequality holds by Jensen's inequality, and the last equality holds due to the tower property of conditional expectation. 
%Thus $\EE{\mu^4(X, Z)} < \infty$ implies $\EE{h^4(W)} < \infty$. 

   In the proofs of Theorems \ref{thm:main} and \ref{thm:main_general}, we will use a key fact to simplify exposition: when $\EE{h^2(W)}> 0$, $\EE{h^2(W)}=1$ can be assumed without loss of generality. This is because \eqref{eq:equiv_fmu} says $f(\mu)=\frac{\EE{Yh(W)} }{\sqrt{\EE{h^2(W)}}}$ and $R_i$, $V_i$ in Algorithm \ref{alg:MOCK} can be rewritten as
    	\begin{eqnarray*}
    		R_i &=& Y_i (\mu(X_i, Z_i) - \Ec{\mu(X_i, Z_i}{Z_i}) = Y_i h(W_i), \\
    		V_i &=& \Varc{\mu(X_i,Z_i)}{\Xinoj} = \Varc{h(X_i,Z_i)}{\Xinoj} 
    	\end{eqnarray*}
    	by definition of $h$. Regarding Theorem \ref{thm:main_general}, $R_i^K$, $V_i^K$ can be rewritten as 
    	\begin{align}\nonumber
	    \begin{split}
	    R_i^K &= Y_i\left(h(W_i) - \frac{1}{K}\sum_{k=1}^{K} h(X_i^{(k)} ,Z_i)\right),\\
	    V^{K}_{i} &= \frac{1}{K-1}\sum_{k=1}^{K}\left( h(X_i^{(k)} ,Z_i)-  \frac{1}{K}\sum_{k=1}^{K} h(X_i^{(k)} ,Z_i) \right)^2
	   	\end{split}
	   	\end{align}
    	due to \eqref{eq:RiK_in_h}, \eqref{eq:ViK_in_h}. It is immediate that the floodgate procedure is invariant to positive scaling thus we assume $\EE{h^2(W)}=1$ without loss of generality.}

    \subsection{Proofs in Section~\ref{sec:mock}}
    \subsubsection{Lemma \ref{lem:max}}\label{pf:lem:max}
      \begin{proof}[Proof of Lemma \ref{lem:max}]
      When $\EE{\varc{\muX}{Z}}=0$, the numerator must also be zero, and hence the ratio is 0 by convention and $\thetamu\le \Ij$. Now assuming $\EE{\varc{\muX}{Z}}>0$,
      \begin{align*}
          \thetamu &= \frac{\EE{\covc{\muX}{\mustar(X,Z)}{Z}}}{\sqrt{\EE{\varc{\muX}{Z}}}} \\
          &= \frac{\EE{\sqrt{\varc{\muX}{Z}}\sqrt{\varc{\mustar(X,Z)}{Z}}\mathrm{Cor}\left(\muX, \mustar(X,Z)\,|\,Z\right)}}{\sqrt{\EE{\varc{\muX}{Z}}}} \\
          &\le \frac{\EE{\sqrt{\varc{\muX}{Z}}\sqrt{\varc{\mustar(X,Z)}{Z}}}}{\sqrt{\EE{\varc{\muX}{Z}}}} \\
          &\le \frac{\sqrt{\EE{\varc{\muX}{Z}}}\sqrt{\EE{\varc{\mustar(X,Z)}{Z}}}}{\sqrt{\EE{\varc{\muX}{Z}}}} = \Ij,
      \end{align*}
      where the first inequality uses the fact that correlation is bounded by 1, and the second inequality uses Cauchy--Schwarz. Finally, it is immediate that $f(\mustar) = \Ij$.
\end{proof}

%\begin{proof}{TBD for triangular array}
%%  For given $\mu_n$, showing the above is quite straightforward: in the proof of Theorem \ref{thm:main}, we establish the asymptotic normality of $T$; we also show $s$ converges in probability to $\tilde{\sigma}_0$ (which is the variance of the asymptotic normal distribution, as defined in \eqref{eq:sigma_tilde_def}). For a sequence of working regression functions $\mu_n$, we need to deal with the following triangular array  
%
%  Theorem \ref{thm:main} is proved based on the CLT and the delta method. In this proof, we have a sequence of working regression functions $\mu_n$ and need to deal with the following triangular array  
%    \begin{equation}
%    	R_{ni} = Y_ih_n(X_i,Z_i), V_{ni} = \Varc{h_n(X_i, Z_i)}{Z_i}, \quad i \in [n]
%    \end{equation}	
%    where $h_n(X_i, Z_i) = (\mu_n(X_i,Z_i) - \Ec{\mu_n(X_i, Z_i)}{Z_i}$. Therefore, we will utilize the triangular array CLT and a extended version of the delta method where the function can depend on $n$. Recall the proof of Theorem \ref{thm:main} considers $4$ different cases then deals with them separately. To avoid length proof, we will focus on the most complicated case where $\Var{Yh(W)}>0$ and $ \Var{\Varc{h(X)}{\Xnoj}}>0$ and omit the similar derivations for Cases \myRom{1}, \myRom{2} and \myRom{3}.
%     
%	Consider the i.i.d. triangular array of two-dimensional 
\acc{
\subsubsection{Theorem \ref{thm:main}}
\begin{proof}[Proof of Theorem \ref{thm:main}]
\label{pf:thm:main}
%Recall $h(W) = \mu(W) - \Ec{\mu(W)}{\Xnoj}$ as defined in \eqref{eq:def_hhstar}. We have
%\begin{align}\label{eq:element_mmt_bound}
%\begin{split}
%	\EE{h^4(W)} 
%	&= \EE{ (\mu(W) - \Ec{\mu(W)}{\Xnoj} )^4}\\ 
%	&\le  2^{4-1}(\EE{\mu^4(W)}+ \EE{(\Ec{\mu(W)}{\Xnoj})^4}) \\ 
%	&\le  2^{4-1}(\EE{\mu^4(W)}+ \EE{\Ec{\mu^4(W)}{\Xnoj}}) =  16~ \EE{\mu^4(X, Z)},
%\end{split}
%\end{align}
%where the first inequality holds due to the $C_r$ inequality, the second inequality holds by Jensen's inequality, and the last equality holds due to the tower property of conditional expectation. 
Due to \eqref{eq:element_mmt_bound}, $\EE{\mu^4(X, Z)} < \infty$ implies $\EE{h^4(W)} < \infty$. 
In the following proof, we will only assume the weaker moment conditions $\EE{Y^4}, \EE{h^4(W)} < \infty$. Under such moment conditions, we also have $\EE{Yh(W)} \le \sqrt{\EE{Y^2}}\sqrt{\EE{h^2(W)}}$ and $\EE{h^2(W)} < \infty$  since the finiteness of higher moments implies that of lower moments. 

	When $\mu(X,Z) \in \sigalg(Z)$, i.e., $\EE{\varc{\muX}{Z}}=0$, we immediately have coverage since $L_{n}^{\alpha} (\mu) =0$ by construction and $\Ij \ge 0$ by its definition. Regarding the case where $\EE{\varc{\muX}{Z}}  \neq 0$, we have $\EE{h^2(W)} = \EE{\varc{\muX}{Z}} > 0$ due to \eqref{eq:h2W}. Based on the discussions in the part after \eqref{eq:element_mmt_bound}, we can assume $\EE{h^2(W)}=1$ without loss of generality.
%	\begin{eqnarray}\nonumber
%		0 < \EE{\Varc{\mu (X,Z)}{\Xnoj}} 
%		&=& \EE{\Varc{h (W)}{\Xnoj}} \\ \nonumber
%		&=& \EE{\Ec{h^2(W)}{\Xnoj} -  (\Ec{h(W)}{\Xnoj})^2} \\ 
%		&=& \EE{h^2(W)}
%		\label{eq:Eh2W}
%	\end{eqnarray}
%	 where the first equality holds by the definition of $h(W)=h(X,Z)$, the second equality comes from the expansion of $\Varc{h (W)}{\Xnoj}$, and the last equality holds as a result of $\Ec{h(W)}{\Xnoj}=0$ in \eqref{eq:5_lem:max} and the tower property of conditional expectation. 	 
%	\rev{ 
%	 When $\EE{h^2(W)}> 0$, we can assume $\EE{h^2(W)}=1$ without loss of generality. This is because $f(\mu)=\frac{\EE{Yh(W)} }{\sqrt{\EE{h^2(W)}}}$ and $R_i = Y_i (\mu )$
%	 Therefore, we assume $\EE{h^2(W)}=1$ for the following proof without loss of generality (since floodgate is invariant to positive scaling).
%	 }
%	

    Recall in Algorithm \ref{alg:MOCK}, we denote $R_{i} = Y_i\big(\mu(X_i,Z_i)-\Ec{\mu(X_i,Z_i)}{\Xinoj}\big)$ and $V_{i} = \Varc{\mu(X_i,Z_i)}{\Xinoj}$ for each $i\in [n]$, and compute their sample mean $(\bar{R},\bar{V})$ and sample covariance matrix $\hat{\Sigma}$. The LCB is constructed as
    $$
             L_{n}^{\alpha} (\mu)=\max\left\{ \frac{\bar{R}}{ \sqrt{\bar{V}  }} -  \frac{z_{\alpha}s}{\sqrt{n}}   ,0\right\},~~\text{where}~~  s^2 = \frac{ 1 }{ \bar{V} }\left[ \left(\frac{\bar{R} }{2 \bar{V} }\right)^2 \hat{\Sigma}_{22} +  \hat{\Sigma}_{11} - \frac{\bar{R} }{ \bar{V} } \hat{\Sigma}_{12} \right].     
     $$ 
     And we have 
     $$
       \{{L_{n}^{\alpha} (\mu) \le \Ij}\} =   \left\{\frac{\bar{R}}{ \sqrt{\bar{V}  }} -  \frac{z_{\alpha}s}{\sqrt{n}}    \le \Ij\right\} \supset \left\{\frac{\bar{R}}{ \sqrt{\bar{V}  }} -  \frac{z_{\alpha}s}{\sqrt{n}}    \le \fmu \right\},
     $$
     where the first equality holds since $\Ij \ge 0$ and the subset relation holds due to Lemma \ref{lem:max}. Hence it suffices to show that 
        \begin{equation}
            \label{eq:valid_atmu}
        	   % \inf_{P\in \calP,~ \mu \in \calU}
        	    \PP{\frac{\bar{R}}{ \sqrt{\bar{V}  }} -  \frac{z_{\alpha}s}{\sqrt{n}}    \le \thetamu} \ge 1 - \alpha - o(1).
        	\end{equation}
       We will utilize the central limit theorem (CLT) and the delta method to prove the above result. Now we consider four different cases. 
        	\begin{enumerate}[(I)]
        	    \item $\Var{Yh(W)}=0$ and $ \Var{\Varc{h(W)}{\Xnoj}}=0$.
        	    \item $\Var{Yh(W)}>0$ and $ \Var{\Varc{h(W)}{\Xnoj}}=0$.
        	    \item $\Var{Yh(W)}=0$ and $ \Var{\Varc{h(W)}{\Xnoj}}>0$.
        	    \item $\Var{Yh(W)}>0$ and $ \Var{\Varc{h(W)}{\Xnoj}}>0$.
        	\end{enumerate}
   Note that assuming $\EE{Y^4}$ and $\EE{h^4(W)} < \infty$ ensures all the above variances exist; the bounding strategy is the same as \eqref{eq:element_mmt_bound}, thus we omit the proof. When $\Var{Yh(W)}=0$ and $\Var{\Varc{h(W)}{\Xnoj}}=0$, respectively, we have the following facts.
  \begin{align}\label{eq:fact1}
  	\Var{Yh(W)}=0 ~\Rightarrow~ R_{i} = {\EE{Yh(W)}}, \forall~ i \in [n], \bar{R} = {\EE{Yh(W)}},~ \hat{\Sigma}_{11} =\hat{\Sigma}_{12} =0,\\ \label{eq:fact2}
  	\Var{\Varc{h(W)}{\Xnoj}}=0 ~\Rightarrow~ V_{i} = \EE{h^2(W)}, \forall~ i \in [n],\bar{V} = {\EE{h^2(W)}},~ \hat{\Sigma}_{22} =\hat{\Sigma}_{12} =0.
  \end{align} 
   
    Case (\myRom{1}): due to \eqref{eq:fact1} and \eqref{eq:fact2}, we simply have $\frac{\bar{R}}{ \sqrt{\bar{V}  }} = \EE{Yh(W)}/\sqrt{\EE{h^2(W)}} = \thetamu$ and $s = 0$, thus \eqref{eq:valid_atmu} holds.
    
    Case (\myRom{2}): due to \eqref{eq:fact2}, $s^2=\hat{\Sigma}_{11}/\bar{V} = \hat{\Sigma}_{11}/\EE{h^2(W)}$, hence we have the following equivalence
    $$
     \left\{\frac{\bar{R}}{ \sqrt{\bar{V}  }} -  \frac{z_{\alpha}s}{\sqrt{n}}    \le \fmu \right\} = \left\{
             \bar{R} - \frac{z_{\alpha}(\hat{\Sigma}_{11})^{1/2}}{\sqrt{n}} \le  \EE{Yh(W)}  
             \right\}.
    $$
    Thus the problem is reduced to showing that
 \begin{equation} 
    \label{eq:case2_coverage}
         \PP{   \bar{R} - \frac{z_{\alpha}(\hat{\Sigma}_{11})^{1/2}}{\sqrt{n}} \le  \EE{Yh(W)} } \ge 1 - \alpha - o(1).
    \end{equation}
        Notice $\bar{R}$ is simply the sample mean estimator of the quantity $\EE{Yh(W)}$ and $\hat{\Sigma}_{11}$ is the corresponding sample variance. \eqref{eq:case2_coverage} is an immediate result of the central limit theorem and Slutsky's theorem. 
              
        Case (\myRom{3}): due to \eqref{eq:fact1}, we have 
            $$
            \frac{\bar{R}}{ \sqrt{\bar{V}  }} -  \frac{z_{\alpha}s}{\sqrt{n}} =\frac{\EE{Yh(W)}}{ \sqrt{\bar{V}  }} -  \frac{z_{\alpha}s}{\sqrt{n}},~~\text{where}~~ s^2 = \frac{1}{\bar{V}}\left(\frac{\EE{Yh(W)}  }{2\bar{V}}\right)^2 \hat{\Sigma}_{22}.
             $$
        	$\frac{\EE{Yh(W)}}{ \sqrt{\bar{V} }}$ is a nonlinear function of the moment estimators. We will use the delta method to establish the asymptotic normality result. In case (\myRom{4}), $\EE{Yh(W)}$ is further replaced by its moment estimator, and we are dealing with a bit more complicated nonlinear statistic than $1/{ \sqrt{\bar{V} }}$. Hence we focus on case (\myRom{4}) and omit the very similar proof for case (\myRom{3}). 
        
        Case (\myRom{4}): since $\Var{Yh(W)}>0$ and $ \Var{\Varc{h(W)}{\Xnoj}}>0$, 
        we have as $n\rightarrow \infty$,
        \begin{equation}\label{eq:2dclt}
      \sqrt{n} \left( 
        \begin{array}{c}
\bar{R} -\EE{Yh(W)} \\
\bar{V}-\EE{h^2(W)}
\end{array} \right) \stackrel{d}{\rightarrow} \gauss{\bm{0}}{ \Sigma}
           \end{equation}
        by the multivariate central limit theorem, where the covariance matrix of the random vector $(R_i, V_i) \in \mathbb{R}^2$ is denoted by
         $ \Sigma $ with
            \begin{equation}\nonumber
	   \left(
	    \begin{array}{cc}
		\Sigma_{11}   &  \Sigma_{12} \\ 
	   \Sigma_{21}    &   \Sigma_{22}
	    \end{array} 
	    \right) = \left(
	    \begin{array}{cc}
	   \Var{Yh(W)}      &  \Cov{Yh(W)}{ \Varc{h(W)}{\Xnoj} } \\ 
	   \Cov{Yh(W)}{ \Varc{h(W)}{\Xnoj} }      &  \Var{\Varc{h(W)}{\Xnoj}  }
	    \end{array} 
	    \right). 
	  \end{equation}
	$\EE{Y^4},\EE{h^4(W)} < \infty$ ensures the finiteness of $\Sigma_{11},\Sigma_{12}, \Sigma_{22}$. Denote 
	\begin{equation}\label{eq:tsigma0_def}
		\tilde{\sigma}_0^2=  \frac{ 1 }{ \EE{h^2(W)}  }\left[ \left(\frac{\EE{Yh(W)}}{2 \EE{h^2(W)}}\right)^2 {\Sigma}_{22} +  {\Sigma}_{11} - \frac{\EE{Yh(W)}}{ \EE{h^2(W)} } {\Sigma}_{12} \right],
	 \end{equation}
	and we will show $\tilde{\sigma}_0 > 0$ over the course of derivations from \eqref{eq:tsigma0_AB} to the end of the proof. Now consider
	\begin{equation}\label{eq:H_def}
	\left(\frac{\bar{R}}{ \sqrt{\bar{V}}} - f(\mu)\right)/s = \left(\frac{\bar{R}}{ \sqrt{\bar{V}}} - f(\mu)\right)/\tilde{\sigma}_0\cdot \frac{\tilde{\sigma}_0}{s} :=\frac{H(\bar{R}, \bar{V}) - \fmu}{\tilde{\sigma_0}}\cdot \left(\frac{s}{\tilde{\sigma_0}}\right)^{-1},
	\end{equation}
      where $H(x_1, x_2):\mathbb{R}^2 \rightarrow \mathbb{R} $ is defined as $H(x_1, x_2) = x_1/\sqrt{x_2}$ for $x_2 > 0$ 	and its gradient equals $\nabla H(x_1, x_2) = (\frac{\partial H}{\partial x_1}, \frac{\partial H}{\partial x_2}) = \frac{1}{\sqrt{x_2}}(1, -\frac{x_1}{2x_2})$. Let $\theta = (\EE{Yh(W)}, \EE{h^2(W)})$, then 
    \begin{equation}\label{eq:grad_def}
      \nabla H(\theta) = \frac{1}{\sqrt{\EE{h^2(W)}}}\left(1,-\frac{\EE{Yh(W)}}{2 \EE{h^2(W)}}
	\right),
	\end{equation} 
	and we obtain 
	\begin{equation}\label{eq:tsigma0_rewrite}
		\Var{\nabla H(\theta)^\top \left( 
        \begin{array}{c}
\sqrt{n} \bar{R}  \\
\sqrt{n} \bar{V}
\end{array} \right) } =
	\Var{\nabla H(\theta)^\top \left( 
        \begin{array}{c}
R_i  \\
V_i
\end{array} \right) }= 	\nabla H(\theta)^\top~ \Sigma ~\nabla H(\theta) =  \tilde{\sigma}_0^2
	\end{equation}
	 where the second equality holds by the definition of $\Sigma$ and the last equality holds by elementary calculation. Therefore, by applying the multivariate delta method to \eqref{eq:2dclt}, we have $\sqrt{n}(H(\bar{R}, \bar{V}) - H(\theta)) \stackrel{d}{\rightarrow}  \gauss{0}{ 		\nabla H(\theta)^\top~ \Sigma ~\nabla H(\theta) }$, i.e.,
	  \begin{equation}\label{eq:deltaclt}
	 \sqrt{n} (H(\bar{R}, \bar{V}) - \fmu)/\tilde{\sigma}_0 \stackrel{d}{\rightarrow}  \gauss{0}{ 1}.
	  \end{equation}
	Replacing the means, variances and covariances in $	\tilde{\sigma}_0^2$ by their moment estimators, we obtain
	  $$
	  \frac{ 1 }{ \bar{V} }\left[ \left(\frac{\bar{R} }{2 \bar{V} }\right)^2 \hat{\Sigma}_{22} +  \hat{\Sigma}_{11} - \frac{\bar{R} }{ \bar{V} } \hat{\Sigma}_{12} \right],    
	  $$
	 which equals $s^2$ by its definition. Due to the finiteness of $\EE{Yh(W)}, \EE{h^2(W)},\Sigma_{11},\Sigma_{12},\Sigma_{22}$, we have
	 $$
	 (\bar{R}, \bar{V}, \hat{\Sigma}_{11}, \hat{\Sigma}_{12}, \hat{\Sigma}_{22}) \stackrel{p}{\rightarrow} (\EE{Yh(W)}, \EE{h^2(W)}, {\Sigma}_{11}, {\Sigma}_{12}, {\Sigma}_{22})
	 $$
	 by the law of large numbers. Then by the continuous mapping theorem, we have
	 $ s \stackrel{p}{\rightarrow} \tilde{\sigma}_0$ as $n\rightarrow \infty$. Combining this with \eqref{eq:H_def} and \eqref{eq:deltaclt}, we have 
	 $$
	 \sqrt{n}\left(\frac{\bar{R}}{ \sqrt{\bar{V}}} - f(\mu)\right)/s \stackrel{d}{\rightarrow}  \gauss{0}{ 1},
	 $$
	as $n\rightarrow \infty$, which establishes \eqref{eq:valid_atmu}. 
	 
	 Now we will verify the positiveness of $\tilde{\sigma}_0$. Recall $\EE{h^2(W)} = 1$ as assumed without loss of generality; we rewrite $\tilde{\sigma}_0^2 $ 
	 	 \begin{eqnarray} \label{eq:tsigma0_AB}
	 	 	 \tilde{\sigma}_0^2 
	 &=& \Var{\nabla H(\theta)^\top \left( 
        \begin{array}{c}
R_i -\EE{Yh(W)} \\
V_i-\EE{h^2(W)}
\end{array} \right) }\\ \nonumber
 &=& \Var{\left(1, -\frac{\EE{Yh(W)}}{2}\right)^\top \left( 
        \begin{array}{c}
R_i-\EE{Yh(W)} \\
V_i-\EE{h^2(W)}
\end{array} \right)} \\ \label{eq:easy_tsigma0}
&=& \EE{(R_i-\EE{Yh(W)}  - 0.5~{\EE{Yh(W)}}(V_i - 1) )^2} \\ \nonumber
&:=& \EE{(A+B)^2}	
	 \end{eqnarray}
	 where the first equality holds due to \eqref{eq:tsigma0_rewrite} and the basic property of variance, the second equality holds due to \eqref{eq:grad_def}, and the last equality is by rearranging and the terms $A,B$ are defined as below:
	   \begin{eqnarray} \label{eq:thm1_A_Def}
	  A &:=& R_i -\Ec{Y_i h(W_i)}{\Xinoj} = Y_i h(W_i) - \Ec{Y_i h(W_i)}{\Xinoj},\\  \label{eq:thm1_B_Def}
	  B &:=& \Ec{Y_i h(W_i)}{\Xinoj} -  \EE{Yh(W)}  - 0.5~ \EE{Yh(W)}( \Varc{h(W_i)}{\Xinoj} -1).
	   \end{eqnarray}
	  Now we can expand \eqref{eq:tsigma0_AB} as
	  \begin{eqnarray} \nonumber
	   \tilde{\sigma}_0^2 = \EE{(A+B)^2 }  
	   &=&  \EE{\Ec{ (A + B)^2}{\Xinoj} }\\ \nonumber
	   &=&  \EE{ \Ec{ A^2}{\Xinoj} - 2 B~\Ec{A}{\Xinoj} + B^2 }\\ \nonumber
	   &=&  \EE{ \Ec{ A^2}{\Xinoj} + B^2 } \\ \label{eq:L2V_bound_v1}
	   &\ge&  \EE{\Varc{Yh(W)}{\Xnoj}},
	  \end{eqnarray}
where the first equality comes from the tower property of conditional expectation, the second equality holds since $B\in \sigalg(\Xinoj)$ and the third equality holds due to $\Ec{A}{\Xinoj}=0$. 

 Since \eqref{eq:L2V_bound_v1} gives a lower bound for $\tilde{\sigma}_0^2$, we are done when $\EE{\Varc{Yh(W)}{\Xnoj}} >0$. Otherwise, we assume $\EE{\Varc{Yh(W)}{\Xnoj}} =0$, then $\tilde{\sigma}_0^2   = 	\nabla H(\theta)^\top~ \Sigma ~\nabla H(\theta)  =0$ implies the degeneracy of $\Sigma$ since the vector $\nabla H(\theta) = (1, -0.5~\EE{Yh(W)})$ is nonzero. It suffices to show it is impossible to have $\Sigma$ degenerate when $\EE{\Varc{Yh(W)}{\Xnoj}} =0$. According to the definition of $\Sigma$, we have that $Yh(W)$ is a linear function of $\Varc{h(W)}{\Xnoj}$ in the degenerate case. This means $ Yh(W) = c \Varc{h(W)}{\Xnoj} +d $ for some constants $c,d$. Then we obtain 
	  $$
	  \Varc{Yh(W)}{\Xnoj} = \Varc{c \Varc{h(W)}{\Xnoj} +d   }{\Xnoj} = c^2 \Var{\Varc{h(W)}{\Xnoj}  } > 0,
	  $$
	  since we are dealing with case (\myRom{4}) where $ \Var{\Varc{h(W)}{\Xnoj}}>0$ and $\Var{Yh(W)}>0$ (thus $c^2 >0$). The above result contradicts the assumption $\EE{\Varc{Yh(W)}{\Xnoj}} =0$. This finishes showing the positiveness of $\tilde{\sigma}_0$, 	  
\end{proof}
}

        \acc{
         \subsubsection{Lemma \ref{lem:doublerobust}}\label{pf:lem:doublerobust}  
    \begin{proof}[Proof of Lemma \ref{lem:doublerobust}]
      \acc{Recall the notations $g(z) = \Ec{\mu(X, Z)}{Z=z}$ and $h(w)=h(x,z)=\mu(x,z)- g(z)$ introduced in \eqref{eq:def_ggstar} and \eqref{eq:def_hhstar}. When $Q_x = P_{X\mid Z}$, we immediately have $\mu(X, Z) - \Epc{Q_x}{\mu(X, Z)}{Z} = \mu(X,Z) -  \Ec{\mu(X, Z)}{Z} = h(W)$, thus
      $$
      f_{ Q_y, Q_{x}}(\mu) = \frac{\EE{(Y - \Epc{Q_y}{Y}{Z} )h(W) }}{\sqrt{\EE{ h^2(W)}}} = \frac{\EE{Yh(W) }}{\sqrt{\EE{ h^2(W)}}} = f(\mu)
      $$ 
      where the second equality holds since $\EE{\Epc{Q_y}{Y}{Z} h(W)} = 0$ by \eqref{eq:5_lem:max} and the last equality holds by \eqref{eq:equiv_fmu}. Hence $ f_{ Q_y, P_{X\mid Z} }(\mu) =  f(\mu)$ is proved. For convenience, we also use the following notations throughout this proof: $ P_x := P_{X\mid Z}$, $g_y(Z):= \Epc{Q_y}{Y}{Z}$, $g_x(Z):= \Epc{Q_x}{\mu(X,Z)}{Z}$. Thus we rewrite $f_{ Q_y, Q_{x}}(\mu)$ in \eqref{eq:approxfg} as
      \begin{equation}\label{eq:db_f_equi}
       f_{ Q_y, Q_{x}}(\mu) = \frac{\EE{(Y - g_y(Z) )(\mu(X, Z) -g_x(Z) )}}{\sqrt{\EE{(\mu(X, Z) - g_x(Z) )^2}}} \le \frac{\sqrt{\EE{(Y - g_y(Z) )^2}}\sqrt{\EE{(\mu(X, Z) -g_x(Z) )^2}}}{\sqrt{\EE{(\mu(X, Z) - g_x(Z) )^2}}},
       \end{equation} 
       where the inequality holds by the Cauchy--Schwarz inequality.
%          \begin{equation}\label{eq:approxfg}
%        f_{ Q_y, Q_{x}}(\mu) := \frac{\EE{(Y - \Epc{Q_y}{Y}{Z} )(\mu(X, Z) - \Epc{Q_x}{\mu(X, Z)}{Z} )}}{\sqrt{\EE{(\mu(X, Z) - \Epc{Q_x}{\mu(X, Z)}{Z})^2}}},
%    \end{equation}     
%     \ljmargin{When $\mu(X, Z) - g_{x}(Z) = 0$}{throughout this proof you still use the old notation $g_0$ and $g_\mu$, please update it to the current notation of $g_y$ and $g_x$},
     If $\EE{(\mu(X, Z) - g_x(Z) )^2} = 0$, $ f_{ Q_y, Q_{x}}(\mu) $ is $0/0=0$ by convention and thus $f_{g_y, g_x}(\mu) \le \Ij + \Delta$ automatically holds due to the non-negativeness of $\Delta$ and $\Ij$. 
    %  When $\mu(X, Z) - g_{x}(Z) != 0$. 
    %  Otherwise, we assume $\EE{(\mu(X, Z) - g_{x}(Z))^2}=1$ since rescaling $\mu(X, Z) - g_{x}(Z)$ by $\sqrt{\EE{(\mu(X, Z) - g_{x}(Z))^2}}$ does not change $f_{g_y, g_x}(\mu)$ neither the corresponding inferential results. 
     Otherwise, we notice that 
     \begin{align}\nonumber
      \EE{(\mu(X, Z) - g_{x}(Z))^2} &= \EE{(\mu(X, Z) - \Ec{\mu(X, Z)}{Z} +  \Ec{\mu(X, Z)}{Z}  - g_{x}(Z))^2} \\ \nonumber
      & = \EE{(\mu(X, Z) - g(Z))^2} + \EE{(g(Z)  - g_{x}(Z))^2} \\ \label{eq:denom_bd}
      & \ge \EE{h^2(W)},
     \end{align}
     where the first equality holds due to rearranging, the second equality holds since 
     \begin{align} \nonumber
      ~ & \EE{(\mu(X, Z) - \Ec{\mu(X, Z)}{Z})(\Ec{\mu(X, Z)}{Z}  - g_{x}(Z))} = \EE{h(W)(g(Z)  - g_{x}(Z)) } = 0 
     \end{align}
     by \eqref{eq:5_lem:max}, and the last inequality holds due to the definition of $h(w)$ and the non-negativeness of $\EE{(g(Z)  - g_{x}(Z))^2}$. We note that $ f_{ Q_y, Q_{x}}(\mu) \le 0 \le \Ij + \Delta$ when the numerator $\EE{(Y - g_y(Z))(\mu(X, Z) - g_{x}(Z))} \le 0$. Thus it remains to deal with the case where $\EE{(Y - g_y(Z))(\mu(X, Z) - g_{x}(Z))} > 0$. Now we expand $ f_{ Q_y, Q_{x}}(\mu)$ and bound it as below:
    \begin{align}\nonumber
               f_{ Q_y, Q_{x}}(\mu) 
              &= ~\frac{\EE{(Y - g_y(Z))(\mu(X, Z) - g_{x}(Z))}}{\sqrt{\EE{(\mu(X, Z) - g_{x}(Z))^2}}}\\ \nonumber
              &= ~\frac{\EE{(Y - g_y(Z))(\mu(X, Z) - g(Z))}}{\sqrt{\EE{(\mu(X, Z) - g_{x}(Z))^2}}} 
              +
              \frac{\EE{(Y - g_y(Z))(g(Z) - g_{x}(Z))}}{\sqrt{\EE{(\mu(X, Z) - g_{x}(Z))^2}}} \\ \nonumber
              & \le ~ \frac{\EE{(Y - g_y(Z)) (\mu(X, Z) -  g(Z) )}}{\sqrt{ \EE{h^2(W)}}} +  
              \frac{\EE{( Y - g_y(Z))(g(Z) - g_{x}(Z))}}{\sqrt{   \EE{(\mu(X, Z) - g_{x}(Z))^2} }} \\ \nonumber
              &= ~ \frac{\EE{Yh(W) )}}{\sqrt{ \EE{h^2(W)}}} + \frac{\EE{( \epsilon(Y, W) + \hstar(W) + \gstar(Z) - g_y(Z))(g(Z) - g_{x}(Z))}}{\sqrt{ \EE{(\mu(X, Z) - g_{x}(Z))^2}  }} \\ \nonumber
              & = ~ f(\mu) + \frac{\EE{( \gstar(Z) - g_y(Z))(g(Z) - g_{x}(Z))}}{{\sqrt{ \EE{(\mu(X, Z) - g_{x}(Z))^2} }}}, \\ 
              &\le ~  \Ij+ \frac{\EE{ | \gstar(Z) - g_y(Z)|\cdot |g(Z) - g_{x}(Z)|}}{{ \sqrt{ \EE{h^2(W)} } }}
              \label{eq:db_2f}
            %   \\
            %  &:= \mathrm{II}_1 + \mathrm{II}_2, 
            %  \label{eq:db_II_decom}
            % & = f(\mu) +  
            %   \frac{\EE{( \epsilon(Y, W) + \hstar(W) + \gstar(Z) - g_y(Z))(g(Z) - g_{x}(Z))}}{\sqrt{ \EE{\varc{\mu(X,Z)}{Z}}  }}  \\ \nonumber
            %  & = f(\mu) +  \frac{\EE{( \gstar(Z) - g_y(Z))(g(Z) - g_{x}(Z))}}{\sqrt{\EE{h^2(W)} }}
    \end{align}
     where the first equality comes from \eqref{eq:db_f_equi}, the second equality is by rearranging, the first inequality holds due to $\EE{(Y - g_y(Z))(\mu(X, Z) - g_{x}(Z))} > 0$ and \eqref{eq:denom_bd}, the third equality holds since $\EE{g_y(Z)(\mu(X, Z) -  g(Z)) } =\EE{g_y(Z) h(W) } =0$ by \eqref{eq:5_lem:max} and we expand $Y$ as in \eqref{eq:4_lem:max}, the last equality holds by \eqref{eq:5_lem:max}, \eqref{eq:h2W} and \eqref{eq:equiv_fmu}, and the last inequality holds due to Lemma \ref{lem:max}, $\EE{ | \gstar(Z) - g_y(Z)|\cdot |g(Z) - g_{x}(Z)|} > 0$ and \eqref{eq:denom_bd}.
%      we can rewrite $\mathrm{II}_1$ as
%      \begin{equation}
%      \label{eq:db_II1_bound}
%       \mathrm{II}_1 =  \frac{\EE{Y(\mu(X, Z) -  g(Z) )}}{\sqrt{\EE{\varc{\mu(X,Z)}{Z}}}} = 
%       =f(\mu)
%         \end{equation}
%   where the second equality holds due to  \eqref{eq:equiv_fmu}
%      \eqref{eq:5_lem:max}. Regarding $\mathrm{II}_2$, we expand $Y$ as in \eqref{eq:4_lem:max} and apply \eqref{eq:5_lem:max} and \eqref{eq:h2W}, thus obtain
%      \begin{align}\nonumber
%           \mathrm{II}_2 
%           &~= \frac{\EE{( \epsilon(Y, W) + \hstar(W) + \gstar(Z) - g_y(Z))(g(Z) - g_{x}(Z))}}{\sqrt{ \EE{\varc{\mu(X,Z)}{Z}}  }} \\
%           &~= \frac{\EE{( \gstar(Z) - g_y(Z))(g(Z) - g_{x}(Z))}}{\sqrt{\EE{h^2(W)} }}.
%           \label{eq:db_II2_bound}
%      \end{align}
    %  Combining \eqref{eq:db_II_decom}, \eqref{eq:db_II1_bound} with \eqref{eq:db_II2_bound} results 
    %  \begin{equation}\label{eq:db_2f}
    % f_{ Q_y, Q_{x}}(\mu)  \le f(\mu) + \frac{\EE{( \gstar(Z) - g_y(Z))(g(Z) - g_{x}(Z))}}{\sqrt{\EE{h^2(W)} }}.
    %  \end{equation}
     In the following, we bound $\EE{ | \gstar(Z) - g_y(Z)|\cdot |g(Z) - g_{x}(Z)|}$. Since we denote $g_x(z) =\Epc{Q_x}{\mu(X,Z)}{Z=z} $ with $Q_x$ being the estimate of the true conditional distribution of $X$ given Z (i.e., $P_{X\mid Z}$, abbreviated as $P_x$), we can rewrite $|g(Z) - g_x(Z)|$ then bound it as:
     \begin{eqnarray} \nonumber
     |g(Z) - g_x(Z)| 
     &=& \left|\Epc{P_x}{\mu(X,Z)}{Z} - \Epc{Q_x}{\mu(X,Z)}{Z} \right|\\ \nonumber
     &=& \left| \Epc{P_x}{h(W) + g(Z)}{Z}   - \Epc{Q_x}{h(W) + g(Z)}{Z}\right| \\ \nonumber
     &=&  \left|\Epc{P_x}{h(W)}{Z}   - \Epc{Q_x}{h(W)}{Z}\right| \\ \nonumber
     &=&  \left|\int h(x,Z)(1- \delta(x,Z))d P_{X\mid Z}(x\mid Z)\right| \\
     &=& |\Epc{P_x}{h(W)(1- \delta(W))}{Z}| \le \sqrt{\Epc{P_x}{h^2(W)}{Z}}\sqrt{\chi^2\left(Q_x\|P_{X\mid Z}\right)},
     \label{eq:g_diff_bound}
     \end{eqnarray}
     where the second equality holds due to \eqref{eq:def_hhstar}, the third equality holds since $g(Z) \in \sigalg(Z)$, the fourth equality holds since $Q_x$ is absolutely continuous with respect to $P_{X\mid Z}$ and we denote $\delta(x,Z):=\frac{dQ_x(x\mid Z)}{dP_{X\mid Z}(x\mid Z)}$ and rewrite the third line in the form of integral, and the last inequality holds by the Cauchy--Schwarz inequality and the definition of the $\chi^2$ divergence. Hence replacing the term $| g(Z) - g_x(Z) |$ in \eqref{eq:db_2f} by its upper bound in \eqref{eq:g_diff_bound}
     produces the following
     \begin{equation}\label{eq:DRf_bound}
      f_{ Q_y, Q_{x}}(\mu)   \le \Ij + \frac{\EE{|\gstar(Z) - g_y(Z)| \sqrt{\Epc{P_x}{h^2(W)}{Z}}\sqrt{\chi^2\left(Q_x\|P_{X\mid Z}\right)}}}{\sqrt{\EE{h^2(W)} }}.
     \end{equation}
     Now we will bound $\mathrm{III}:=\EE{ |\gstar(Z) - g_y(Z)| \sqrt{\Epc{P_x}{h^2(W)}{Z}}\sqrt{\chi^2\left(Q_x\|P_{X\mid Z}\right)}}$ in three different versions. 
     
     Firstly, we apply the Cauchy--Schwarz inequality to $\sqrt{\Epc{P_x}{h^2(W)}{Z}}\sqrt{\chi^2\left(Q_x\|P_{X\mid Z}\right)}$ and $|\gstar(Z) - g_y(Z)|$, producing
     \begin{eqnarray}\nonumber
     \mathrm{III}
     &=& \EE{| \gstar(Z) - g_y(Z)| \sqrt{\Epc{P_x}{h^2(W)}{Z}}\sqrt{\chi^2\left(Q_x\|P_{X\mid Z}\right)}}\\ \nonumber
     &\le & \sqrt{\EE{( \gstar(Z) - g_y(Z) )^2} } \sqrt{\EE{ \Epc{P_x}{h^2(W)}{Z} \chi^2\left(Q_x\|P_{X\mid Z}\right)}} \\ \nonumber
     & = & \sqrt{\EE{( \gstar(Z) - g_y(Z) )^2} } \sqrt{\EE{ \Epc{P_x}{h^2(W) \chi^2\left(Q_x\|P_{X\mid Z}\right) }{Z}}}\\
     &=&  \sqrt{\EE{( \gstar(Z) - g_y(Z) )^2} } \sqrt{\EE{h^2(W) \chi^2\left(Q_x\|P_{X\mid Z}\right) }},
     \label{eq:DRbound_v1}
     \end{eqnarray}
     where the second equality holds since $\chi^2\left(Q_x\|P_{X\mid Z}\right) \in \sigalg(Z)$, and the last equality holds due to the notation $P_x = P_{X\mid Z}$ and the law of total expectation. Noting the definition of $\mathrm{III}$ and combining \eqref{eq:DRf_bound} with \eqref{eq:DRbound_v1} yields 
     $$
      f_{ Q_y, Q_{x}}(\mu) \le \Ij + \sqrt{\EE{( \gstar(Z) - g_y(Z) )^2} } \sqrt{\EE{\left(\frac{h(W)}{\sqrt{\EE{h^2(W)}}} \right)^2\chi^2\left(Q_x\|P_{X\mid Z}\right) }}.
     $$
    Recalling the notations: 
    \begin{equation}\label{eq:recall_notation}
        \gstar(Z) = \Ec{Y}{Z},\quad g_y(Z) = \Epc{Q_y}{Y}{Z},\quad h(W) = \mu(X, Z) -  \Ec{\mu(X, Z)}{Z},
    \end{equation}
    and $w_\mu(X,Z) = \frac{ (\mu(X,Z) - \Ec{\mu(X,Z)}{Z})^2 }{\EE{ (\mu(X,Z) - \Ec{\mu(X,Z)}{Z})^2 }} $, \eqref{eq:DRDelta_v1} is thus established. 
    
    Secondly, we apply the Cauchy--Schwarz inequality to $\sqrt{\chi^2\left(Q_x\|P_{X\mid Z}\right)}$ and $|\gstar(Z) - g_y(Z)|\sqrt{\Epc{P_x}{h^2(W)}{Z}}$ in $\mathrm{III}$, producing
          \begin{eqnarray}\nonumber
     \mathrm{III}
     &=& \EE{|\gstar(Z) - g_y(Z)| \sqrt{\Epc{P_x}{h^2(W)}{Z}}\sqrt{\chi^2\left(Q_x\|P_{X\mid Z}\right)}}\\ \nonumber
     &\le & \sqrt{\EE{ \chi^2\left(Q_x\|P_{X\mid Z}\right) } } \sqrt{\EE{ \Epc{P_x}{h^2(W)}{Z} ( \gstar(Z) - g_y(Z))^2 }} \\ \nonumber
     & = & \sqrt{\EE{ \chi^2\left(Q_x\|P_{X\mid Z}\right) } } \sqrt{\EE{ \Epc{P_x}{h^2(W) ( \gstar(Z) - g_y(Z))^2 }{Z}}}\\
     &=&  \sqrt{\EE{ \chi^2\left(Q_x\|P_{X\mid Z}\right) } }\sqrt{\EE{h^2(W) ( \gstar(Z) - g_y(Z))^2}}, 
     \label{eq:DRbound_v2}
     \end{eqnarray}
      where the second equality holds since $( \gstar(Z) - g_y(Z))^2 \in \sigalg(Z)$, and the last equality holds due to the notation $P_x = P_{X\mid Z}$ and the law of total expectation. Combining \eqref{eq:DRf_bound} and \eqref{eq:recall_notation} with \eqref{eq:DRbound_v2} and recalling  $w_\mu(X,Z) = \frac{ (\mu(X,Z) - \Ec{\mu(X,Z)}{Z})^2 }{\EE{ (\mu(X,Z) - \Ec{\mu(X,Z)}{Z})^2 }} =\frac{h^2(W)}{\EE{h^2(W)}} $ yields a different bound on $ f_{ Q_y, Q_{x}}(\mu)$, namely,
     \begin{align}\label{eq:DRDelta_v2}
     \begin{split}
  &    f_{ Q_y, Q_{x}}(\mu) \le f(\mu) + \Delta', \text{   where }\\
    \Delta'& ~= \sqrt{\EE{ \chi^2\left(Q_x \|P_{X\mid Z}\right) } } \sqrt{\EE{w_\mu(X,Z) ( \Ec{Y}{Z}- \Epc{Q_y}{Y}{Z} )^2 }}.
         \end{split}
     \end{align}
     
	Lastly, we apply the Cauchy--Schwarz inequality to $\sqrt{\Epc{P_x}{h^2(W)}{Z}}$ and $|\gstar(Z) - g_y(Z)|\sqrt{\chi^2\left(Q_x\|P_{X\mid Z}\right)}$ in $\mathrm{III}$, producing
	    \begin{eqnarray}\nonumber
     \mathrm{III}
     &=& \EE{ |\gstar(Z) - g_y(Z)| \sqrt{\Epc{P_x}{h^2(W)}{Z}}\sqrt{\chi^2\left(Q_x\|P_{X\mid Z}\right)}}\\ \nonumber
     &\le & \sqrt{\EE{ \Epc{P_x}{h^2(W)}{Z} } } \sqrt{\EE{ \chi^2\left(Q_x\|P_{X\mid Z}\right)( \gstar(Z) - g_y(Z))^2 }} \\ \nonumber
     & = & \sqrt{\EE{ h^2(W)} }  \sqrt{\EE{ \chi^2\left(Q_x\|P_{X\mid Z}\right)( \gstar(Z) - g_y(Z))^2 }} \\
     &=&  \sqrt{\EE{ h^2(W)} } \left(\EE{( \gstar(Z) - g_y(Z))^4}\right)^{1/4}  \left(\EE{\left( \chi^2\left(Q_x\|P_{X\mid Z}\right)\right)^2}\right)^{1/4},
     \label{eq:DRbound_v3}
     \end{eqnarray}
      where the second equality holds due to the notation $P_x = P_{X\mid Z}$ and the law of total expectation, and the last inequality holds by applying the Cauchy--Schwarz inequality again. Combining \eqref{eq:DRf_bound} and \eqref{eq:recall_notation}  with \eqref{eq:DRbound_v3} yields a final different bound on $  f_{ Q_y, Q_{x}}(\mu)  $, namely,
    \begin{align}\label{eq:DRDelta_v3}
     \begin{split}
  &    f_{ Q_y, Q_{x}}(\mu) \le f(\mu) + \Delta'', \text{   where }\\
    \Delta''& ~= \left(\EE{( \Ec{Y}{Z}- \Epc{Q_y}{Y}{Z})^4}\right)^{1/4}  \left(\EE{\left( \chi^2\left(Q_x \|P_{X\mid Z}\right)\right)^2}\right)^{1/4}.
         \end{split}
     \end{align}
     }

%     |\Epc{P}{\mu(W)}{\Xnoj} -\Epc{Q}{\mu(W)}{\Xnoj}|
%     &=&\left|\int \mu(x,Z)(1- \omega(x,Z))d Q_{X\mid Z}(x\mid Z)\right| \\ \nonumber
%     &\le& \sqrt{\Epc{Q_{X\mid Z}}{\mu^2(X,Z)}{\Xnoj}} \sqrt{\int (1-w(x,Z))^2 d Q_{X\mid Z}(x\mid Z)}\\
%     &=&  \sqrt{\Epc{Q_{X\mid Z}}{\mu^2(W)}{\Xnoj}} \sqrt{\chi^2\left(P_{X\mid Z}\|Q_{X\mid Z}\right)},
%     \label{eq:2_lem:robust_gap}
%     where $\delta(x,Z) = \frac{d P_{X\mid Z}(x\mid Z)}{d Q_{X\mid Z}(x\mid Z)}$ and the above inequality is from the Cauchy--Schwarz inequality.
%      the last inequality holds by the Cauchy--Schwarz inequality, and $\Delta$ in the last line denotes the term $
%        \sqrt{\EE{(g_y(Z) - \Ec{\mustar(X, Z)}{Z})^2} \EE{(g_{x}(Z) - \Ec{\mu(X, Z)}{Z})^2}/\EE{\varc{\mu(X,Z)}{Z}} }$.  Thus Lemma \ref{lem:doublerobust} is established.
%        }{I don't think this is valid---by ``assuming'' $\EE{(\mu(X, Z) - g_{x}(Z))^2}=1$, you are effectively redefining $\mu$ and $g_\mu$, but the big-O bound as presented in the main text is in terms of the originally-defined $\mu$ and $g_\mu$. I don't think this is too big a problem, it just means our big-O bound needs to be a bit more complicated with a normalizing term (I guess if we keep the normalizing term, we can drop the big-O notation). I think it is still interpretable as a double robust type of bound, but it needs a bit more explanation in the main text where presented. let's discuss}. Thus Lemma \ref{lem:doublerobust} is established.
    \end{proof}
    }
    
    \subsection{Proofs in Section~\ref{sec:ucb_hardness}}
    \label{app:ucb_hardness}

    \begin{proof}[Proof of Theorem \ref{thm:UCB_hardness}]
        We prove by contradiction. Suppose there exists an upper confidence bound procedure ensuring asymptotic coverage such that \eqref{eq:UCB_trivial} holds, that is, there exists a joint law over $(Y,X,Z)$, denoted by $F_{\infty}\in \calF$ such that 
        \begin{equation}\label{eq:UCB_nontrivial}
            \limsup_{n\rightarrow\infty} \Pp{\infty}{U(D_n)-\Ij^2_{F_\infty} < \Ep{\infty}{\Varpc{\infty}{Y}{X,\Xnoj}}} > \alpha.
        \end{equation}
        where $\mathbb{P}_{\infty},~ \mathbb{E}_{\infty},~\mathrm{Var}_{\infty}$ denote that the data generating distribution for i.i.d. sample $D_n$ is $F_{\infty}$. Note that $\Pp{\infty}{U(D_n)-\Ij^2_{F_\infty} < \Ep{\infty}{\Varpc{\infty}{Y}{X,\Xnoj}}} = \Pp{\infty}{U(D_n) < \Ep{\infty}{\Varpc{\infty}{Y}{\Xnoj}}}$ by the definition of $\Ij^2_{F_\infty}$.
        % denote the $X$'s distribution under $F_{\infty}$ by $F_{\infty}(x)$ and 
        Let $\lambda_{1} = \Ep{\infty}{\Varpc{\infty}{Y}{\Xnoj}} $. When $\lambda_1 =0$, we have $\Ep{\infty}{\Varpc{\infty}{Y}{\Xnoj}}= \Ep{\infty}{\Varpc{\infty}{Y}{X,\Xnoj}} = \Ij^2_{F_{\infty}}=0 $ and immediately show 
        \begin{equation}\nonumber
            \limsup_{n\rightarrow\infty} \Pp{\infty}{U(D_n)-\Ij^2_{F_\infty} < \Ep{\infty}{\Varpc{\infty}{Y}{X,\Xnoj}}}  = \limsup_{n\rightarrow\infty} \Pp{\infty}{U(D_n) < \Ij^2_{F_\infty} } \le \alpha,
        \end{equation} 
        which contradicts \eqref{eq:UCB_nontrivial}. In the following we consider the case where $\lambda_{1} >0$.
        % then we have $\lambda_1>0$, since otherwise $\Ep{\infty}{\Varpc{\infty}{Y}{\Xnoj}}=\Ij^2_{F_{\infty}}=0 $ and $\eqref{eq:UCB_nontrivial}$ does not hold.
        % and $ \lambda_{2} =\Ij(F_{\infty})$ and $\lambda_{2}=(c- \Ij(F_{\infty}) )/2$. 
        Now we construct a sequence of joint laws over $(Y,X,Z)$, denoted by $\{F_{k}\}_{k=1}^{\infty}$, $F_k \in \calF$, such that the conditional distribution of $\eps \mid X,Z$ 
        is the same as that under $F_{\infty}$, where $\eps = Y - \Ec{Y}{X,Z}$, that is, 
        % $(X,Z)$ follows the same distribution as that under $F_{\infty}$ and so does the conditional distribution of $\eps \mid X,Z$, where $\eps = Y - \Ec{Y}{X,Z}$, that is,
        \begin{eqnarray} \label{eq:eps_same}
       \Ppc{k}{\eps}{X,Z} =  \Ppc{\infty}{\eps}{X,Z},  &~~&\forall~ k \ge 1
       %~~ X\sim F_{\infty}(x) ~~\text{under}~~ F_{k},~\forall~ k \ge 1,
       \end{eqnarray}
        and there exist Borel sets $A_{k}\in \mathbb{R}^{p-1}$ satisfying the following:
        % \lz{P xz is continuius and delete A.29}
        \begin{enumerate}[(a)]
            \item $\Pp{k}{\Xnoj \in A_k} = 1/k$;
            \item $\Ppc{k}{Y}{X,Z} =  \Ppc{\infty}{Y}{X,Z}$ when $\Xnoj \notin A_k$;
            % \item $\Ppc{k}{\eps}{X} =  \Ppc{\infty}{\eps}{X}$ when $\Xnoj \in A_k$;
            \item ${\Epc{k}{\mustar_k(X,Z)}{\Xnoj}} = {\Epc{\infty}{\mustar_\infty(X,Z)}{\Xnoj}}$ when $\Xnoj \in A_k$;
            \item ${\Varpc{k}{\mustar_k(X,Z)}{\Xnoj}} = {\Varpc{\infty}{\mustar_\infty(X,Z)}{\Xnoj}} + {k}\left(2\lambda_1 - \Ij^2_{F_{\infty}} \right)$ when $\Xnoj \in A_k$;
        \end{enumerate}%\frac{c+\lambda_1}{2}
        % \begin{eqnarray}\nonumber
        % \Pp{k}{\Xnoj \in A_k} = 1/k;~~
        % \Ppc{k}{Y}{X} =  \Ppc{\infty}{Y}{X} ~~\text{when}~~ \Xnoj \notin A_k  \\ \nonumber 
        %  \Ppc{k}{\eps}{X} =  \Ppc{\infty}{\eps}{X},~~ {\Varpc{k}{\mustar_k(X)}{\Xnoj}} = {\Varpc{\infty}{\mustar_\infty(X)}{\Xnoj}} + {k}\left(\frac{c+\lambda_1}{2} - \Ij^2(F_{\infty}) \right)~~\text{when}~~ \Xnoj \in A_k  
        % \end{eqnarray}
        where $\mathbb{P}_{k}~,\mathbb{E}_k,~ \mathrm{Var}_k$ denote that the data generating distribution for i.i.d. sample $D_n$ is $F_{k}$, and $\mustar_k(X,Z):=\Epc{k}{Y}{X,Z},\mustar_\infty(X,Z):=\Epc{\infty}{Y}{X,Z}$. According to the statement of Theorem \ref{thm:UCB_hardness}, the covariate distribution $P_{X, Z}$ is continuous and fixed. Therefore we have (a) is possible and immediately know
        \begin{eqnarray}\label{eq:X_same}
        \Pp{k}{X,Z} =  \Pp{\infty}{X,Z}, &~~&\forall~ k \ge 1.
       \end{eqnarray}  
        Note here $\Epc{k}{\cdot}{\Xnoj}, \Varpc{k}{\cdot}{\Xnoj}$ are the same as $\Epc{\infty}{\cdot}{\Xnoj}, \Varpc{\infty}{\cdot}{\Xnoj}$ due to \eqref{eq:X_same}. Hence we can calculate $\Ij_{F_k}$ through the following
        \begin{eqnarray} \nonumber
             \Ij^2_{F_k} - \Ij^2_{F_{\infty}} 
             &=& \Ep{\infty}{\indicat{A_k} \left( {\Varpc{\infty}{\mustar_k(X,Z)}{\Xnoj}}-   {\Varpc{\infty}{\mustar_\infty(X,Z)}{\Xnoj}}\right) }\\ \nonumber
             &=& \Ep{\infty}{ \indicat{A_k} {k}\left({2\lambda_1}{} - \Ij^2_{F_{\infty}} \right)} \\  \label{eq:Ij_gap}
             &=& {2\lambda_1}{} - \Ij^2_{F_{\infty}}=:\lambda_2  ,
        \end{eqnarray}
        where the first equality comes from the definition of $\Ij^2_F$, \eqref{eq:X_same} and (b), the second equality holds due to (d) and the third equality holds due to (a). Therefore $ \Ij^2_{F_k} = 2\lambda_1$. We should also check whether $F_k$ belongs to $\calF$. Indeed, we consider the following
        \begin{eqnarray} \nonumber
            \Varp{k}{Y} &=& \Ep{k}{\Varpc{k}{Y}{X,Z}} + \Varp{k}{ \Epc{k}{Y}{X,Z}} \\ \nonumber
                        &=& \Ep{k}{\Varpc{k}{\eps}{X,Z}} + \Varp{k}{\Epc{k}{Y}{\Xnoj}} + \Ij^2_{F_k} \\  \nonumber
                        &=& \Ep{\infty}{\Varpc{k}{\eps}{X,Z}} + \Varp{\infty}{\Epc{k}{Y}{\Xnoj}} + \Ij^2_{F_k} \\ \nonumber
                        &=& \Ep{\infty}{\Varpc{\infty}{\eps}{X,Z}} + \Varp{\infty}{\Epc{\infty}{Y}{\Xnoj}} + \Ij^2_{F_k} \\ \nonumber
                        &=& \Ep{\infty}{\Varpc{\infty}{\eps}{X,Z}} + \Varp{\infty}{\Epc{\infty}{Y}{\Xnoj}} + \Ij^2_{F_\infty} + \lambda_2 \\ \nonumber
                        &=& \Varp{\infty}{Y} + \lambda_2 < \infty,
        \end{eqnarray}
        where the first equality comes from the law of total variance, the second equality holds as a result of the decomposition $Y=\mustar(X,Z) + \eps$ and the equivalent expression of the mMSE gap \eqref{eq:mMSEgap_form2}, the third equality holds due to \eqref{eq:X_same}, the fourth equality holds due to \eqref{eq:eps_same}, (b) and (c), the fifth equality comes from \eqref{eq:Ij_gap}. Thus we verify $F_k \in \calF,~\forall~ k \ge 1$.
        % $$
        % \Ij(F_k)= \Ij(F_{\infty}) + (c-\lambda_{1})/2 < c.
        % $$
        As the upper confidence bound procedure $U$ ensures asymptotic coverage validity and $ \Ij^2_{F_k} = 2\lambda_1$, we have
        \begin{equation}\label{eq:UCB_uniform_cover}
              \Pp{k}{ U(D_n)\ge  {2\lambda_1}{} } \ge 1-\alpha + o_{k}(1)
        \end{equation}
        where the subscript in $o_{k}(1)$ emphasizes that the convergence is with respect to data generating function $F_k$. Remark we only require for fixed $k$, $o_{k}(1) \rightarrow 0$ as $n\rightarrow \infty$. Also notice the following
        \begin{equation}\label{eq:close_in_TV}
            \left|  \Pp{\infty}{ U(D_n)\ge  {2\lambda_1}{} }  -  \Pp{k}{ U(D_n)\ge {2\lambda_1}{} } \right| \le d_{TV}(F_k, F_\infty) \le \frac{1}{k}, ~~~~\forall~ k\ge 1,
        \end{equation}
        where the first inequality comes from the property of total variation distance and the second equality holds as a result of (a), according to the construction of $F_k$. Combining \eqref{eq:UCB_uniform_cover} and \eqref{eq:close_in_TV} yields the following
        \begin{equation}\nonumber
            \Pp{\infty}{U(D_n) \ge  {2\lambda_1}{} } \ge 1 - \alpha -1/k + o_k(1), ~~~~\forall~ k\ge 1. 
        \end{equation}
        First let $n\rightarrow \infty$ then send $k$ to infinity, we obtain
        \begin{equation}\nonumber
            \liminf_{n \rightarrow \infty} \Pp{\infty}{U(D_n) \ge 2\lambda_1} \ge 1- \alpha,
        \end{equation}
        which contradicts 
        \begin{equation}\nonumber
             \limsup_{n \rightarrow \infty} \Pp{\infty}{U(D_n) < \Ep{\infty}{\Varpc{\infty}{Y}{\Xnoj}} =\lambda_1 } > \alpha .
        \end{equation}
        % \begin{eqnarray}\label{eq:dense3}
        %     \mathbb{P}_{k}(X_{\noj}\in A_{k})=1/k\\ \mathrm{Var}_{k}(\mu^{\star}(X)|X_{\noj}) = \mathrm{Var}_{\infty}(\mu^{\star}(X)|X_{\noj})+k(\lambda_{1} - \lambda_{2}) + k(c-\lambda_{1})/2 ~~\text{when}~X_{\noj}\in A_{k}\\
        %     \mathbb{P}_{k}(\epsilon|X)=\mathbb{P}_{\infty}(\epsilon|X)~~\text{when}~X_{\noj}\in A_{k},~~   \mathbb{P}_{k}(Y|X)=\mathbb{P}_{\infty}(Y|X)~~\text{when}~X_{\noj}\notin A_{k}\\
        % \end{eqnarray}
        % hence we have $\pi_{j}(\mu^{\star})=\lambda_{1}+(c-\lambda_{1})/2 < c 
        % $, thus $F_{k}\in \calF$. By the exact coverage validity,
        % \begin{equation}
        %     \mathbb{P}_{k}(B(X,Y)\ge (c+\lambda_{1})/2)\ge 1-\alpha
        % \end{equation}
        % together with 
        % \begin{equation}
        % |\mathbb{P}_{\infty}(B(X,Y)\ge (c+\lambda_{1})/2) -  \mathbb{P}_{k}(B(X,Y)\ge (c+\lambda_{1})/2)|\le d_{TV}(P_{k},P_{\infty})\le 1/k, ~~ \forall~ k \ge 1
        % \end{equation}
        % yields $\mathbb{P}_{\infty}(B(X,Y)\ge (c+\lambda_{1})/2) \ge 1-\alpha -1/k,~\forall ~k\ge 1$. Since $c>\lambda_{1}$, we further obtain $\mathbb{P}_{\infty}(B(X,Y)> \lambda_{1}) \ge 1-\alpha$, which is a contradiction.
    \end{proof}

    \acc{
     \subsection{Proofs in Section~\ref{sec:computation}}
     \begin{proof}[Proof of Theorem \ref{thm:main_general}]
     \label{pf:thm:main_general}
	As in the proof of Theorem~\ref{thm:main}, we immediately have coverage validity when $\mu(X,Z) \in \sigalg(\Xnoj)$. Otherwise, it suffices to show 
	  \begin{equation}\label{eq:general_valid}
	  \PP{\frac{\bar{R}}{ \sqrt{\bar{V}  }} -  \frac{z_{\alpha}s}{\sqrt{n}}  \le \thetamu} \ge 1 - \alpha -o(1).
	  \end{equation}
	  for any given $K > 1$, where the sample mean $(\bar{R},\bar{V})$ and sample covariance matrix $\hat{\Sigma}$ are defined the same way as in Algorithm \ref{alg:MOCK} except that $R_i, V_i$ are replaced by their Monte Carlo estimators $ R_i^K, V_i^K$ as defined below.
	    \begin{align}
	    \label{eq:Rik_ViK_def}
	    \begin{split}
	    R_i^K &:= Y_i\left(\mu(X_i,Z_i) - \frac{1}{K}\sum_{k=1}^{K} \mu(X_i^{(k)} ,Z_i)\right),\\
	    V^{K}_{i} &:= \frac{1}{K-1}\sum_{k=1}^{K}\left( \mu(X_i^{(k)} ,Z_i)-  \frac{1}{K}\sum_{k=1}^{K} \mu(X_i^{(k)} ,Z_i) \right)^2,
	   	\end{split}
	   	\end{align}
	   for any fixed $K>1$. 
	   
	   First we verify 
	   \begin{equation}\label{eq:mc_esti_yes}
	   	\EE{R_i^K} = \EE{Yh(W)}, \quad \EE{V_i^K} = \EE{h^2(W)}.
	   \end{equation} 
	   By the construction of the null samples, $X_i^{(k)}$ satisfy the following properties:
	 \begin{eqnarray}\label{eq:nulls_cond_indep}
	 \{X_i^{(k)}\}_{k=1}^{K} \independent (X_{i},Y_i) \mid \Xinoj,\\
	  \{X_i^{(k)}\}_{k=1}^{K}\mid \Xinoj \iid X_{i} \mid \Xinoj,
	 \label{eq:nulls_cond_ident}
	 \end{eqnarray}
	 thus we have 
	 \begin{align}  \label{eq:mean_estimator}
	 \Ecmid{ \frac{1}{K}\sum_{k=1}^{K}\mu(\Xtil_i^{(k)},Z_i) }{\Xinoj} &= \Ec{\mu(X_i,Z_i)}{\Xinoj},\\ 
	 \Ecmid{ \frac{1}{K-1}\sum_{k=1}^K \left( \mu(\Xtil_i^{(k)},Z_i) -  \frac{1}{K} \sum_{k=1}^K \mu(\Xtil_i^{(k)},Z_i)  \right)^2 }{\Xinoj} &= \Varc{\mu (X_i,Z_i)}{\Xinoj}, \label{eq:var_estimator}
	 \end{align}
	 and further obtain
	 \begin{eqnarray*} \nonumber
	 \EE{R_i^K} 
	 &=& \EE{ Y_i\left( \mu(X_i,Z_i) - \frac{1}{K}\sum_{k=1}^{K}\mu(\Xtil_i^{(k)},Z_i)\right) } \\ \nonumber
	&=& \EE{Y_i \mu(W_i)} - \EE{\Ec{Y_i}{\Xinoj} \Ecmid{\frac{1}{K}\sum_{k=1}^{K}\mu(\Xtil_i^{(k)},Z_i)}{\Xinoj}  }  \\ \nonumber 
	&=& \EE{Y_i \mu(W_i)} - \EE{\Ec{Y_i}{\Xinoj} \Ec{\mu(X_i,Z_i)}{\Xinoj}  } \\ \nonumber
	&=& \EE{Y_i \mu(W_i)} - \EE{Y_i \Ec{\mu(X_i,Z_i)}{\Xinoj}  } = \EE{Yh(W)},
	 \end{eqnarray*}
	where the first equality holds due to \eqref{eq:Rik_ViK_def}, the second equality holds due to \eqref{eq:nulls_cond_indep}, the third equality holds due to \eqref{eq:mean_estimator}, the fourth equality comes from the tower property of total expectation and the last one is by the definition of $h(W)$. Regarding the term $\EE{V_i^K }$, \eqref{eq:var_estimator} and \eqref{eq:h2W} immediately imply $\EE{V_i^K } = \EE{h^2(W)}$. 

	To prove \eqref{eq:general_valid}, we can follow a similar strategy as in the proof of Theorem \ref{thm:main}. Note Appendix \ref{pf:thm:main} considers $4$ different cases then deals with them separately. Essentially we can conduct similar analysis, but to avoid lengthy proof, we focus on the most complicated case where $\Var{Yh(W)}>0$ and $ \Var{\Varc{h(X)}{\Xnoj}}>0$ and omit the derivations for the other three cases.
%	assume the moment condition $\EE{\Varc{Yh(W)}{Z}}=  \EE{\Varc{Y (\mu(X,Z)-\Ec{\mu(X,Z)}{Z})}{\Xnoj}}>0$ and thus focus on this specific case. 
	Under the moment conditions $\EE{Y^4}, \EE{h^4(W)} < \infty$, we have 
	   $\EE{R_i^K } = \EE{Yh(W)} < \infty $ and $\EE{V_i^K } = \EE{h^2(W)} < \infty$.
	   
	   By applying the multivariate central limit theorem and the delta method, we obtain the following asymptotic normality result as in the proof of Theorem \ref{thm:main}: as $n\rightarrow \infty$,
	 \begin{equation}\label{eq:normal_RK}
	     \sqrt{n}\left(\frac{ \frac{1}{n}\sum\nsubp R_i^K }{\sqrt{\frac{1}{n} \sum\nsubp V_i^K  }} -   \thetamu \right) \stackrel{d}{\rightarrow } \gauss{0}{\tilde{\sigma}_0^2},
	 \end{equation}
	 where $\tilde{\sigma}_0^2$ is similarly defined as in \eqref{eq:tsigma0_def} and its positiveness will be proved over the course of derivations from \eqref{eq:recallAB} toward the end of this proof. Due to the law of large numbers and the continuous mapping theorem, we can prove $s \stackrel{p}{\rightarrow} \tilde{\sigma}_0$ as in Appendix \ref{pf:thm:main}. The asymptotic normality and the consistency result only require us to verify the finiteness of $\Sigma_{11} = \Var{R_i^K}, \Sigma_{12} = \Cov{R_i^K}{V_i^K}, \Sigma_{22} = \Var{V_i^K}$. Since $\EE{R_i^K }, \EE{V_i^K } < \infty$ under the stated moment conditions and $\Cov{R_i^K}{V_i^K} \le \sqrt{ \Var{R_i^K}{ \Var{V_i^K}}} $ by the Cauchy--Schwarz inequality, it suffices to prove 
	  \begin{equation}\label{eq:clt_mmc}
	 \EE{|R_i^K|^2} < \infty,  \EE{|V_i^K|^2} < \infty. 
	  \end{equation}
	Denote $\bar{h}_i^K = \frac{1}{K}\sum_{k=1}^{K}h(\Xtil_i^{(k)},Z_i)$ and we rewrite $R_i^K$ and $V_i^K$.
	 \begin{align} \nonumber
	R_i^K &=Y_i\left(\mu(X_i,Z_i) - \frac{1}{K}\sum_{k=1}^{K}\mu(\Xtil_i^{(k)},Z_i)\right) \\ \nonumber
	 &= Y_i\left( \mu(X_i,Z_i) - \Ec{ \mu(X_i,Z_i)}{Z_i} -  \frac{1}{K}\sum_{k=1}^{K}(\mu(\Xtil_i^{(k)},Z_i) - \Ec{\mu(\Xtil_i^{(k)},Z_i)}{Z_i}) \right) \\  \label{eq:RiK_in_h}
	 &=Y_i\left( h(X_i,Z_i) - \frac{1}{K}\sum_{k=1}^{K} h(\Xtil_i^{(k)},Z_i)\right)\\
	 &= Y_i(h(X_i,Z_i) - \bar{h}_i^K)
	  \label{eq:RiK_rewrite}
	  \end{align}
 	where the first equality holds by \eqref{eq:Rik_ViK_def}, the second equality holds by \eqref{eq:mean_estimator} and the third equality holds by the definition of $h(w)$. 
 	\begin{eqnarray} \nonumber
 	V_i^K 
 	&=& \frac{1}{K-1}\sum_{k=1}^{K}\left( \mu(X_i^{(k)} ,Z_i)-  \frac{1}{K}\sum_{k=1}^{K} \mu(X_i^{(k)} ,Z_i) \right)^2 \\ \nonumber
 	&=& \frac{1}{K-1} \sum_{k=1}^{K} \left( h(X_i^{(k)} ,Z_i)-  \frac{1}{K}\sum_{k=1}^{K} h(X_i^{(k)} ,Z_i) \right)^2 \\  \label{eq:ViK_in_h}
 	&=& \frac{1}{K-1} \sum_{k=1}^{K} h^2(X_i^{(k)} ,Z_i) - \frac{K}{K-1} \left( \frac{1}{K}\sum_{k=1}^{K} h(X_i^{(k)} ,Z_i) \right)^2  \\ 
 	&=& \frac{K}{K-1}\left( \frac{1}{K} \sum_{k=1}^{K} h^2(X_i^{(k)} ,Z_i) - (\bar{h}_i^K)^2 \right)
 	 \label{eq:ViK_rewrite}
 	  \end{eqnarray}
 	  where the first equality holds by \eqref{eq:Rik_ViK_def}, the second equality holds due to similar derivations as \eqref{eq:RiK_in_h} and the last two equalities are simply by expanding and rearranging.
 	Now we bound
 	\begin{eqnarray} \nonumber
 	\left(\EE{|R_i^K|^2}\right)^2
 	&=& \left(\EE{Y_i^2 (h(X_i,Z_i) - \bar{h}_i^K)^2} \right)^2\\ \nonumber
 	&\le& \EE{Y^4} \EE{(h(X_i,Z_i) - \bar{h}_i^K
)^4} \\ 
 	&\le & 	\EE{Y^4}\cdot 2^{4-1}\left(
 	\EE{h^4(X_i, Z_i) + \EE{\left( \bar{h}_i^K\right)^4}}
 	\right) \label{eq:RiK2_bound}
 	\end{eqnarray}
	where the first equality holds due to \eqref{eq:RiK_rewrite}, the first inequality holds by the Cauchy--Schwarz inequality, the second inequality comes from the $C_r$ inequality. 
 	Regarding $\EE{|V_i^K|^2}$, we have
	 	 \begin{align}\nonumber
	 	\EE{|V_i^K|^2} 
	 	=~& \EE{\left| \frac{K}{K-1}\left( \frac{1}{K} \sum_{k=1}^{K} h^2(X_i^{(k)} ,Z_i) - (\bar{h}_i^K)^2 \right)  \right|^2} \\ \nonumber
	 	\le ~& \frac{2^{2-1}K^2}{(K-1)^2} \EE{\left(\frac{1}{K} \sum_{k=1}^{K} h^2(X_i^{(k)} ,Z_i)\right)^2} +  \frac{2^{2-1}K^2}{(K-1)^2}  \EE{(\bar{h}_i^K)^4} \\ 
	 	\le ~ & 2^3\left( \EE{\left(\frac{1}{K} \sum_{k=1}^{K} h^2(X_i^{(k)} ,Z_i)\right)^2}  + \EE{(\bar{h}_i^K)^4} \right):= 2^3(\mathrm{II} + \EE{(\bar{h}_i^K)^4}),
	 	\label{eq:ViK2_bound}
	 \end{align}
	 where the first equality holds by \eqref{eq:ViK_rewrite}, the first inequality holds due to the $C_r$ inequality, and the second inequality comes from rearranging and the fact that $K \le 2(K-1)$ (since $K >1$). The term $\mathrm{II}$ and $\EE{(\bar{h}_i^K)^4}$ can be bounded using the same strategy. Below we give the bounding details of $\EE{(\bar{h}_i^K)^4}$ and omit that of $\mathrm{II}$. By the tower property of conditional expectation, we have   
	   \begin{eqnarray}\nonumber
	  \EE{(\bar{h}_i^K)^4} 
	  &=& \EE{\left( \frac{1}{K}\sum_{k=1}^{K}h(\Xtil_i^{(k)},Z_i) \right)^4}\\
	  &= &  \EE{\Ecmid{\left( \frac{1}{K}\sum_{k=1}^{K}h(\Xtil_i^{(k)},Z_i) \right)^4}{Z_i}}. \label{eq:h4_rewrite}
	   \end{eqnarray}
	   To bound $\Ecmid{\left( \frac{1}{K}\sum_{k=1}^{K}h(\Xtil_i^{(k)},Z_i) \right)^4}{Z_i}$, we notice that, conditional on $\Xinoj$, 
	   $\{h(\Xtil_i^{(k)},Z_i)\}_{k=1}^{K}$ are i.i.d. mean zero random variables, hence we can apply the extension of the Bahr--Esseen inequality in \citet{dharmadhikari1969bounds} to obtain 
	\begin{eqnarray}\label{eq:Bahr_h}
	\Ecmid{ \left(
	 \sum_{k=1}^K h(\Xtil_i^{(k)},Z_i)\right)^4 }{\Xinoj} \le c_{4,K}~ \sum_{k=1}^{K} \Ecmid{ h^4(\Xtil_i^{(k)},Z_i)  }{\Xinoj},
	 \end{eqnarray}
	 Note for generic $d \ge 2$ and $n$, the term $c_{d,n}$ is defined as 
	 $$
	 c_{d,n} = n^{d/2 - 1} \frac{d(d-1)}{2}\max\{1, 2^{d-3}\}\left[1 + 2d^{-1} D_{2m}^{(d-2)/2m} \right]
	 $$
	 where the integer $m$ satisfies $2m\le d <2m+2$, and
%	 }{Since we already use $p$ for the covariate dimension, please switch this notation to something else we haven't used or at least only used more narrowly. I think a similar argument appears a few times throughout the appendix---please revise the notation analogously everywhere}
	 $$
	 D_{2m} = \sum_{t=1}^m \frac{t^{2m-1}}{(t-1)!}.
	 $$
	 We then can simply bound $c_{4,K}$ by $C_4 K$ for some universal constant $C_4$ which do not depend on $K$. Therefore, combining \eqref{eq:h4_rewrite} and \eqref{eq:Bahr_h} gives us 
	 \begin{eqnarray} \nonumber
	 \EE{(\bar{h}_i^K)^4}  &\le &\EE{\frac{C_4 K}{K^4}  \sum_{k=1}^{K} \Ecmid{ h^4(\Xtil_i^{(k)},Z_i)  }{\Xinoj} } \\  \nonumber
	 &=& \frac{C_4}{K^2}\EE{\Ecmid{ h^4(X_i,Z_i)  }{\Xinoj}  } =\frac{C_4}{K^2}\EE{h^4(W)}
	 \end{eqnarray} 
	 where the equality holds by \eqref{eq:nulls_cond_ident} and the second equality holds by the tower property of conditional expectation. Since $\EE{h^4(W)}< \infty$, we have $\EE{(\bar{h}_i^K)^4} < \infty$. The finiteness of $\mathrm{II}$ is similarly proved. Due to \eqref{eq:RiK2_bound} and \eqref{eq:ViK2_bound}, we thus establish \eqref{eq:clt_mmc} under the stated moment conditions $\EE{Y^4}, \EE{h^4(W)}< \infty$. Applying Slutsky's theorem to \eqref{eq:normal_RK} and the consistency result that $s \stackrel{p}{\rightarrow} \tilde{\sigma}_0$, we have
	 	 \begin{equation}\nonumber
	     \frac{\sqrt{n}}{s}\left(\frac{ \frac{1}{n}\sum\nsubp R_i^K }{\sqrt{\frac{1}{n} \sum\nsubp V_i^K  }} -   \thetamu \right) \stackrel{d}{\rightarrow } \gauss{0}{1},
	 \end{equation}  
	   which establishes \eqref{eq:general_valid}.
	
	Now we will verify the positiveness of $\tilde{\sigma}_0$ as promised.
	 Recall in the proof of Theorem \ref{thm:main}, the variance term in the asymptotic normality result is also denoted as $\tilde{\sigma}_0^2$ and admits the following expression 
	   \begin{equation}\label{eq:recallAB} 
	    \EE{(R_i-\EE{Yh(W)}  - 0.5~{\EE{Yh(W)}}(V_i - 1) )^2} = \EE{(A+B)^2} > 0
	   \end{equation}
	   according to \eqref{eq:easy_tsigma0}, where $A$ and $B$ are defined in \eqref{eq:thm1_A_Def} and \eqref{eq:thm1_B_Def} and $ \EE{(A+B)^2} > 0 $ 
	  as proved over the course of derivations from \eqref{eq:easy_tsigma0} to the end of the proof of Theorem \ref{thm:main}. In this proof, it is not hard to see $\tilde{\sigma}_0^2$ has a similar form except that $R_i, V_i$ in the above expression are replaced by their Monte Carlo estimators $ R_i^K, V_i^K$, thus giving
	   \begin{eqnarray} \nonumber
	   	\tilde{\sigma}_0^2 
	   	&=& \EE{(R_i^K-\EE{Yh(W)}  - 0.5~{\EE{Yh(W)}}(V_i^K - 1) )^2} \\ \nonumber
	   	&=&  \EE{(Y_i(h(X_i,Z_i) - \bar{h}_i^K)-\EE{Yh(W)}  - 0.5~{\EE{Yh(W)}}(V_i^K - 1) )^2}  \\
	   	&=&  \EE{\left( \mathrm{III}_{1}
	    - \mathrm{III}_{2}  \right)^2 },\label{eq:sigma_tilde_simple}
	   \end{eqnarray}
	   where the second equality holds by \eqref{eq:RiK_rewrite} and rearranging, the terms $\mathrm{III}_{1}, \mathrm{III}_{2} $ in the last equality are defined as:
	  \begin{align}\nonumber
	   \mathrm{III}_{1} ~:=~&     Y_i h(W_i) - \EE{Yh(W)} - 0.5~ \EE{Yh(W)}   ( \Varc{h(W_i)}{\Xinoj} -1)\\ \nonumber
	    \mathrm{III}_{2} ~:=~& Y_i  \bar{h}_i^K +  0.5~ \EE{Yh(W)}(V_i^K -  \Varc{h(W_i)}{\Xinoj}). 
	  \end{align}
	 To bound $\EE{\left( \mathrm{III}_{1}
	    - \mathrm{III}_{2}  \right)^2 }$, we will show $\Ec{\mathrm{III}_{2} }{Y_i, W_i} = 0$.
	 Recall the definition that $\bar{h}_i^K = \frac{1}{K}\sum_{k=1}^{K}h(\Xtil_i^{(k)},Z_i)$, we obtain
	 \begin{eqnarray}\nonumber
	 	\Ec{\bar{h}_i^K }{Y_i, W_i}  = \Ec{\bar{h}_i^K }{Z_i}  = \Ec{h(W_i)}{Z_i} = 0,
	 \end{eqnarray}
	 where the first equality holds due to $W_i =(X_1, Z_i)$ and \eqref{eq:nulls_cond_indep}, the second equality holds by \eqref{eq:nulls_cond_ident}, and the last equality holds due to \eqref{eq:5_lem:max}. Similarly we have 
	 \begin{eqnarray}\nonumber
	 	\Ec{V_i^K }{Y_i, W_i}  = \Ec{V_i^K }{Z_i} = \Varc{h(W_i)}{\Xinoj},
	 \end{eqnarray}
	 due to \eqref{eq:Rik_ViK_def}, \eqref{eq:nulls_cond_indep}, and \eqref{eq:mean_estimator}. Thus we have shown
	     \begin{eqnarray}\label{eq:III2_cancel}
	  \Ec{\mathrm{III}_{2} }{Y_i, W_i} = 0.
	  \end{eqnarray}
	  Applying the tower property of conditional expectation to \eqref{eq:sigma_tilde_simple} then expanding yields the following expression:
	  \begin{eqnarray} \nonumber
	       \tilde{\sigma}_0^2  
	       &=& \EE{ \Ec{\left( \mathrm{III}_{1}^2 + \mathrm{III}_{2}^2 -2 \mathrm{III}_{1}\mathrm{III}_{2} \right)}{Y_i, W_i}}\\ \nonumber 
	       &=&  \EE{ \mathrm{III}_{1}^2 + \Ec{ \mathrm{III}_{2}^2}{Y_i, W_i}  -2 \mathrm{III}_{1} \Ec{\mathrm{III}_{2} }{Y_i, W_i}}\\ \nonumber
	       &=&   \EE{ \mathrm{III}_{1}^2 + \Ec{ \mathrm{III}_{2}^2}{Y_i, W_i} } \\ \label{eq:tsigma_byAB}
	     &\ge&   \EE{ \mathrm{III}_{1}^2  } = \EE{(A+B)^2} , 
	  \end{eqnarray}
	  where the second equality holds since $ \mathrm{III}_{1} \in \sigalg(Y_i, W_i)$, and the third equality comes from \eqref{eq:III2_cancel}. Note in the last line we have $ \mathrm{III}_{1} = A+B$ due to the definitions of $A,B$ in \eqref{eq:thm1_A_Def} and \eqref{eq:thm1_B_Def} and $\EE{(A+B)^2} > 0 $ due to \eqref{eq:recallAB}. Note $\EE{(A+B)^2}$ does not depend on $K$, therefore we establish the positiveness of $\tilde{\sigma}_0$ for any $K>1$.	   	 
%	  Then according to those derivations, we have $ \EE{ \mathrm{III}_{1}^2  } > 0$ under the assumed condition $ \EE{\Varc{Yh(W)}{\Xnoj}} = \EE{\Varc{Y (\mu(X,Z)-\Ec{\mu(X,Z)}{Z})}{\Xnoj}} >0$. The above derivations holds for any $K>1$ and the lower bound $\EE{ \mathrm{III}_{1}^2  }$ for $ \tilde{\sigma}_0^2  $ does not depend on K. These finish showing the positiveness of $\tilde{\sigma}_0$ for any $K>1$.	   
     \end{proof}
     }
	 
    \subsection{Proofs in Section~\ref{sec:accuracy}}
    \begin{proof}[Proof of Theorem \ref{thm:accuracy}]
    \label{pf:thm:accuracy}
    First we write 
    $$
     \Ij - L^{n}_{\alpha}(\mu_n) =   \Ij - f(\mu_n)  + f(\mu_n) -  L^{n}_{\alpha}(\mu_n),
    $$
    where $f(\mu_n)$ is defined as
    % $$
    % f(\mu_n) =: \frac{\EE{(Y - \mu_n \Xk)^2} - \EE{(Y - \mu_n (X,Z))^2}}{\sqrt{2\EE{(\muX - \mu_n\Xk)^2}}}.
    % $$
     $$
    f(\mu_n) =:  \frac{\EE{\covc{\mustar(X,Z)}{\mu_n(X,Z)}{Z}}}{\sqrt{\EE{\varc{\mu_n(X,Z)}{Z}}}}.
    $$
    Then it suffices to separately show 
    \begin{eqnarray} \label{eq:1_thm:accuracy}
    &~&  \Ij - f(\mu_n) = O_p\left(\inf_{\mu'\in S_{\mu_n}}\EE{( \mu'_n(X,Z)-\mustar(X,Z))^2} \right),
    \\ \label{eq:2_thm:accuracy}
    &~& f(\mu_n) - L^{n}_{\alpha}(\mu_n)   = O_{p}\left(n^{-1/2}\right).
    \end{eqnarray}
    % Note that the randomness in \eqref{eq:1_thm:accuracy} only comes from the training data (which is used to obtain $\mu_n$), and the randomness in \eqref{eq:2_thm:accuracy} comes from $\{(X_i,Z_i,Y_i)\}_{i=1}^{n}$ and the samples for fitting $\mu_n$. 
    In the following, we first show \eqref{eq:2_thm:accuracy}. Recall the definitions in Algorithm \ref{alg:MOCK}, when $\mu(X,Z)\in \sigalg(Z)$, we have $  f(\mu_n)= L^{n}_{\alpha}(\mu_n)=0$, hence in the following we focus on the case where $\mu(X,Z)\notin \sigalg(Z)$. Note we have
    %and can be thought as conditioning on $\muhat$. Showing \eqref{eq:2_thm:accuracy} is quite straightforward. 
    
    $$
    % \hat{\theta}_{j}=n^{-1}\sum_{i\in [n]}R_{ij},~
    % s_{j}^2 = (n-1)^{-1}\sum_{i\in [n]}(R_{ij} - \hat{\theta}_{j})^2,~
    L^{n}_{\alpha}(\mu_n) \ge \frac{\bar{R}}{\sqrt{V}} - \frac{z_{\alpha}s}{\sqrt{n}},
%     = \left(  \right):=T - \frac{z_{\alpha}s}{\sqrt{n}},
    $$
    %where $R_{ij} = Y_i(\muhat(X_i) - \Ec{\muhat(X_i)}{\Xinoj,\muhat})/\sqrt{\EE{\Varc{\mu(X)}{\Xnoj,\muhat}}},~ i\in [n]$. 
    then since $f(\mu_n) - 
     L^{n}_{\alpha}(\mu_n) \le s \left(\left|\left( \frac{\bar{R}}{\sqrt{V}} - f(\mu_n) \right) /s\right| + \frac{z_{\alpha}}{\sqrt{n}} \right)  $,
    it suffices to show
    $$
   T:= \frac{ {\bar{R}}/{\sqrt{V}} - f(\mu_n)}{s}  = O_{p}\left(n^{-1/2}\right),~~~ s = O_{p}(1).
    $$
    For given $\mu_n$, showing the above is quite straightforward: in the proof of Theorem \ref{thm:main}, we establish the asymptotic normality of $T$; we also show $s$ converges in probability to $\tilde{\sigma}_0$ (which is the variance of the asymptotic normal distribution, as defined in \eqref{eq:tsigma0_def}). For a sequence of working regression functions $\mu_n$, we need more work and the stated uniform moment conditions.   
%    we need to deal with the following triangular array  
%    \begin{equation}
%    	R_{ni} = Y_i\left(h_n(X_i,Z_i) \right), V_{ni} = \Varc{h_n(X_i, Z_i)}{Z_i}, \quad i \in [n]
%    \end{equation}	
%    where $h_n(X_i, Z_i) = (\mu_n(X_i,Z_i) - \Ec{\mu_n(X_i, Z_i)}{Z_i}$. Recall the proof \ref{pf:thm:main} considers $4$ different cases then deals with them separately. As in the proof of \ref{pf:thm:main_general}, we focus on the most complicated case where $\Var{Yh(W)}>0$ and $ \Var{\Varc{h(X)}{\Xnoj}}>0$ to avoid lengthy proof, thus omit the similar derivations for Cases \myRom{1}, \myRom{2} and \myRom{3}.
    The proof proceeds through verifying the following: note that 
    by definition of bounded in probability, $T = O_{p}\left(n^{-1/2}\right)$ says for any $\eps>0$, there exists $M$ for which 
    $$
    \sup_{n} P( \sqrt{n}|T|> M) \le \eps.
    $$
    {The case that $\mu(X,Z)\in \sigalg(Z)$, i.e., $\EE{\Varc{\mu_n(X,Z)}{Z}}=0$, was dealt with in the first sentence after \eqref{eq:2_thm:accuracy}. Now it suffices to show} for any $\mu_n$ in the function class $\calU :=\{\mu: \EE{ {\mu}^{12}(X,Z)}/(\EE{\Varc{\mu(X,Z)}{Z}})^{6} \le C \} $,
    \begin{equation}\label{eq:bigO_accuracy_pf}
    \sup_{n} \PP{ \sqrt{n}|T|> M} \le \eps,
    \end{equation} 
    and the choice of $M$ (when fixing $\eps$) is uniform over $\mu_n\in \calU$.
    Define the standard Gaussian random variable by $G$.
%     $G_{\mu_n}$ by $G_{\mu_n} \dsim \gauss{0}{\tilde{\sigma}^2(\mu_n)}$, where $\tilde{\sigma}^2(\mu_n)$ denotes the variance $\tilde{\sigma}^2$ with the input of Algorithm 1 being $\mu_n$
     Then we have
    \begin{equation}\label{eq:1_accuracy_pf}
    \PP{ \sqrt{n}|T  | > M}
    \le  
    \PP{ |G| > M} + \Delta,
    \end{equation}
    where $ \Delta$ is defined as
    \begin{equation}\label{eq:2_accuracy_pf}
     \Delta:= \sup_{\mu_n \in \calU} \sup_{M>0}
     \left|
     \PP{ \sqrt{n}|T| > M} -    \PP{ |G| > M} 
     \right|.
     \end{equation}
    %Note that $\calU$ is defined to be the class of $\mu$ which satisfies $\EE{\mu^8(X,Z)}\le c$.
    %    }{do you mean this for an exponent of 12 instead of 4? It seems like you're applying \eqref{eq:element_mmt_bound} with $r=12$ in the rest of this paragraph}
    \acc{Due to \eqref{eq:element_mmt_bound}, $\EE{\mu^{12}(X, Z)} < \infty$ implies $\EE{h^{12}(W)} < \infty$, where $h$ is defined in \eqref{eq:def_hhstar}. In the following proof, we will only assume weaker moment conditions, i.e., $\EE{\Varc{\mu_n(X,Z)}{Z}}=0$ or $\frac{\EE{ {\mu}_n^{12}(X,Z)}}{\EE{\Varc{\mu_n(X,Z)}{Z}}^{6}} \le C$ stated in Theorem \ref{thm:accuracy} is replaced by $\EE{\Varc{\mu_n(X,Z)}{Z}}=0$ or $\frac{\EE{ {h}_n^{12}(X,Z)}}{\EE{\Varc{\mu_n(X,Z)}{Z}}^{6}} \le C$, where $h_n$ is defined accordingly.
     
    In the proof of Theorem \ref{thm:main_rate}, we assume $\EE{h^2(W)}=1$ without loss of generality. This is because we can always scale $h$ by dividing by $\sqrt{\EE{h^2(W)}}$ when the given working regression function satisfies $\mu(X, Z) \notin \sigalg(Z)$. The floodgate inference procedure and results are the same with the corresponding scaled version $\tilde{h}(W)$. And the scaled version still satisfies the finite moment condition $\EE{\tilde{h}^{12}(W)} < \infty$. Now we are dealing with a sequence of working regression functions $\mu_n$. If we scale $h_n$ analogously by dividing it by $\sqrt{\EE{h_n^2(W)}}$, the corresponding function sequence $\{\tilde{h}_n\}$ does not necessarily satisfy the uniform moment condition, i.e., for all $n$, $\EE{\tilde{h}_n^{12}(W)} < C$ for some constant $C$.
    % which does not admit the same scaling. 
    % (since we can always scale $h$ by $\sqrt{\EE{h^2(W)}}$), here we have a sequence of \revision{working regression functions} $\mu_n$ which does not admit the same scaling. 
    But the moment conditions $\ee{Y^{12}}<\infty$ and $\EE{ {h}_n^{12}(X,Z)}/(\EE{\Varc{\mu_n(X,Z)}{Z}})^{6} = \EE{ {h}_n^{12}(W)}/( \sqrt{\EE{h_n^2(W)}})^{12} \le C$ for all $n$ ensure the uniform moment bound after scaling, hence for the following we can assume $\EE{h_n^2(W)}=1$. }
 
    According to the proof of Theorem \ref{thm:main_rate}, we have the following Berry--Esseen bound
    $$
    \sup_{M>0}
     \left|
     \PP{ \sqrt{n}|T| > M} -    \PP{ |G| > M} 
     \right| =  O\left( \frac{1}{\sqrt{n}} \right),
    $$
    which relies on verifying the following:
    \begin{enumerate}[(i)]
	 %\label{eq:moments_finite}
	     \item  $\EE{|U_{01}|^3}$, $\EE{|U_{02}|^3}$, $\EE{|U_{03}|^3}$, $\EE{|U_{04}|^3}$, $\EE{|U_{05}|^3} <\infty$,
	    \item $\tilde{\sigma}_0^2(\mu_n) = H_2(\bs{0}) >0$,
	    \item $\tilde{\sigma}^2(\mu_n) =  \|L(U_{0})\|_2 >0 $.
	 \end{enumerate}
	Note the above terms are defined similarly as in the proof of Theorem \ref{thm:main_rate} except the dependence on $\mu_n$ (but we abbreviate the notation dependence on $\mu_n$ for the random variables). We have $\tilde{\sigma}^2(\mu_n) =1$ due to the derivations after \eqref{eq:L2V_bound_v2} in the proof of Theorem \ref{thm:main_rate}. To show the constant in the above rate of $ \frac{1}{\sqrt{n}}$ is uniformly bounded, we need to prove $\inf_{\mu_n \in \calU}\tilde{\sigma}^2(\mu_n) >0$ and uniformly control the the 3rd moments in the condition \myrom{1}. First notice that
    \begin{eqnarray}\nonumber
    \inf_{\mu_n \in \calU}\tilde{\sigma}^2(\mu_n)
    &\ge&   \inf_{\mu_n \in \calU} \EE{\Varc{Yh_n(W)}{\Xnoj}} \\ \nonumber
    &\ge&   \inf_{\mu_n \in \calU} \EE{\Varc{Yh_n(W)}{X,Z}} \\ \nonumber
    &=&   \inf_{\mu_n \in \calU} \EE{h^2_n(W)\Varc{Y}{X,Z}} \\ \nonumber
    &\ge & \tau >0
    \end{eqnarray}
    where the first inequality holds due to \eqref{eq:L2V_bound_v1}, the second inequality holds as a result of the law of total conditional variance, the last equality holds by the assumption that $\EE{h_n^2(W)}=1$ and the moment lower bound condition $\Var{Y|X,Z} \ge \tau>0$. Assuming $\ee{Y^{12}}<\infty$ and $\EE{ {\mu}_n^{12}(X,Z)}/(\EE{\Varc{\mu_n(X,Z)}{Z}})^{6} \le C$, we can uniformly control the moments $\EE{|U_{01}|^3}$, $\EE{|U_{02}|^3}$, $\EE{|U_{03}|^3}$, $\EE{|U_{04}|^3}$, $\EE{|U_{05}|^3}$, therefore establish the rate of $ \frac{1}{\sqrt{n}}$ in \eqref{eq:2_accuracy_pf}:
    \begin{equation}\nonumber
        \Delta =  O\left( \frac{1}{\sqrt{n}} \right).
    \end{equation}
    %the constant in the above rate of $ \frac{1}{\sqrt{n}}$ is uniformly bounded. 
    Combining this with \eqref{eq:1_accuracy_pf}, we have
    $$
    \sup_{\mu_n \in \calU}\PP{ \sqrt{n}|T| > M}
    \le  
     \PP{ |G| > M}  + \frac{C'}{\sqrt{n}}
    $$
    for some constant $C'$ depending on $C,\tau$ and $\EE{Y^{12}}$. Therefore we obtain \eqref{eq:bigO_accuracy_pf} and the choice of $M$ can be universally chosen over $\mu_n \in \calU$, which finally establishes $T = O_{p}\left(n^{-1/2}\right)$.
    Using similar strategies, we can prove $ s =  O_{p}\left(1 \right) $. Hence we have shown \eqref{eq:2_thm:accuracy}.
%     $$
%     \EE{\hat{\theta}_{j}|\muhat} = \thetamuh,~ s_j|\muhat 
%     \stackrel{a.s.}{\rightarrow} \Varcmid{\hat{\theta}}{\muhat}
%     $$ 
%     where the derivation is exactly the same as in \eqref{eq:3_lem:max}, except that we are conditioning on $\muhat$. Now we have 
%     $$
%     \left|\thetamuh - L_j\right| \le \left|\hat{\theta}_{j} - \thetamuh \right| + \left|\frac{z_{\alpha}s_{j}}{\sqrt{n}} \right|
%     $$
%     As $\hat{\theta}_{j}$ is the sample mean estimator for $\thetamuh$, we can show $\left|\thetamuh - L_j\right|$ is of rate $O_{P}\left(n^{-1/2}\right)$ if $\Varcmid{R_{ij}}{\muhat}$ is finite. The moment conditions $\ee{Y^4},\ec{\muhat^4(X)}{\muhat} < \infty$ imply
%     $$
%     \Varcmid{
%     Y(\muhat(X) - \ec{\muhat(X)}{\Xnoj})}{\muhat}\le C
%     $$ 
%     for some constant C, by using similar derivations as in \eqref{eq:2_thm:main}. Then together with $\ee{\varc{\muhat(X)}{\Xnoj,\muhat}} \ge \tau_0 > 0$, we have
%     $$
%     \Varcmid{R_{ij}}{\muhat} = \frac{\Varcmid{
%     Y(\muhat(X) - \ec{\muhat(X)}{\Xnoj})}{\muhat}}{\ee{\varc{\muhat(X)}{\Xnoj,\muhat}}}\le \frac{C}{\tau_0} < \infty.
%     $$
% 	hence \eqref{eq:2_thm:accuracy} is done. 	
% 	hence $\var{\hat{\theta}_{j}}$ can be uniformly bounded under the assumption $\ee{\varc{\muhat(X)}{\Xnoj}} > c_0$, where $\hat{\theta}_{j}$ is the sample mean estimator of the term $\thetamuh$, as defined in Algorithm \ref{alg:MOCK}. Thus we have $|\thetamuh - L_j|\le |\hat{\theta}_{j} - \thetamuh| + O\left( n^{-1/2} \right)  = O_{p}\left(n^{-1/2} \right)$. 

	Now we proceed to prove \eqref{eq:1_thm:accuracy}, first it can be simplified into the following form due to \eqref{eq:fmu_simplified} and \eqref{eq:6_lem:max},
	\begin{equation} \label{eq:3_thm:accuracy}
	    \Ij - f(\mu_n)
	    = \sqrt{\ee{(\hstar)^2(W)}} - 
        \frac{\EE{h_n(W)\hstar(W)}}{\sqrt{\EE{h_n^2(W)}}} 
	\end{equation}
	%where $\ghat(X) := \ec{\muhat(X)}{\Xnoj,\muhat}$, $h_n(X):=\muhat(X) - \ghat(X)$ and $\gstar(X),\hstar(X)$ 
	where $h_n(W)=\mu_n(W) - \Ec{\mu_n(W)}{Z}$ and $\hstar$ 
	are defined the same way. Remark we have $0/0=0$ by convention for \eqref{eq:3_thm:accuracy}.
	%Without loss of generality ${\Ecmid{\hhat^2(X)}{\muhat}} = \ee{(\hstar)^2(X)}$, since ${\Ecmid{\hhat^2(X)}{\muhat}}>0$ and scaling $\hhat$ by a positive constant will not change the value of \eqref{eq:3_thm:accuracy}. 
% 	Then it suffices to quantify the gap between $\thetamus$ and $\thetamuh$. As in the proof of Theorem \ref{thm:main}, we denote $\gstar:= \gstar (X) = \ec{\mustar(X)}{\Xnoj}$, $\hstar=\mustar - \gstar$ and similarly denote $\ghat,\hhat$ for $\muhat$, 
%     % $\hat{g}=\mathbb{E}(\hat{\mu}(X)|X_{\noj})$, 
%     % With the notations $h,\hat{h}$, 
%     then the gap $|\thetamuh - \thetamus|$  can be simplified as
    We also find it is more convenient to work with $f(\bar{\mu}_n)$ (note $f(\mu_n)=f(\bar{\mu}_n))$, recall that 
    %Note that $\theta(\mu_n)$ remains the same for any $\mu' \in S_{\mu_n}$, we specifically choose the representative 
    the definition of $\bar{\mu}_n$:
    \begin{equation} \nonumber
        \bar{\mu}_n(x,z):= \sqrt{\frac{\Ij}{\EE{h^2_n(W)}}}
        \left(\mu_n(x,z) - \Ec{\mu_n(X,Z)}{Z=z} \right) +  \Ec{\mustar(X,Z)}{Z=z},
    \end{equation}
    and similarly denote $\bar{h}_n(w) =  \bar{\mu}_n(x,z) - \Ec{\bar{\mu}_n(X,Z)}{Z=z}$. When $\mu(X,Z)\in \sigalg(Z)$, we have $ \bar{\mu}_n(x,z) =  \Ec{\mustar(X,Z)}{Z=z}, \bar{h}_n(w) =0 $, thus
     \begin{equation}\label{eq:movi_gap_case1}
      \Ij - f(\mu_n) = \Ij = \frac{ \EE{(\bar{h}_n(W)-\hstar(W))^2}  }{ \sqrt{\ee{(\hstar)^2(W)}} }
    \end{equation}
    Otherwise when $\EE{h^2_n(W)}>0$, we have $ \sqrt{\EE{\bar{\mu}_n^2(W)}} = \Ij$. In this case, we rewrite the right hand side of \eqref{eq:3_thm:accuracy} in terms of $\bar{\mu}_n$ and further simplify it as below,
    $$
    \frac{ 
    \EE{(\bar{h}_n(W)-\hstar(W))^2}
    - \left(
    \sqrt{\EE{\bar{h}_n^2(W)}}
     - \sqrt{\ee{(\hstar)^2(W)}}
    \right)^{2}
    }{2\sqrt{\EE{\bar{h}_n^2(W)}}} 
    = \frac{
    \EE{(\bar{h}_n(W)-\hstar(W))^2}
    }{2\sqrt{\ee{(\hstar)^2(W)}}}
    $$  
    which says that 
    \begin{equation}\label{eq:movi_gap_case2}
         \Ij - f(\mu_n) = \frac{
    \EE{(\bar{h}_n(W)-\hstar(W))^2}
    }{2\sqrt{\ee{(\hstar)^2(W)}}}
    \end{equation}
     Note that $\sqrt{\ee{(\hstar)^2(W)}} = \Ij$ which does not depend on $\mu$, hence it suffices to show 
       \begin{equation} \label{eq:4_thm:accuracy}
      \EE{(\bar{h}_n(W)-\hstar(W))^2} = O_p\left(\inf_{\mu'\in S_{\mu_n}}\EE{( \mu'(X,Z)-\mustar(X,Z))^2} \right).
        \end{equation}
     We prove it by considering two cases:
     \begin{enumerate}[(a)]
         \item $\EE{h_n(W)\hstar(W)}\le 0$,
         \item $\EE{h_n(W)\hstar(W)}> 0$.
     \end{enumerate}
    %For the following derivations, we omit the notation of conditioning on $\mu_n$ without causing confusion. 
    Regarding case (a), we have
     \begin{eqnarray}\nonumber
     \inf_{\mu'\in S_{\mu_n}}\EE{( \mu'(X,Z)-\mustar(X,Z))^2}
     &=& \inf_{c>0, \forall g(z)}\left( 
     \EE{(c h_n(W) - \hstar(W) )^2} + \EE{(g(Z) - \Ec{\mustar(W)}{Z} )^2}\right)  \\ \nonumber
     &=& \inf_{c>0 }
     \EE{(c h_n(W) - \hstar(W) )^2}  \\ \nonumber
     &=& \EE{(\hstar)^2(W)}+ \inf_{c>0 } 
     c^2\EE{ h^2_n(W)} - 2c \EE{h_n(W)\hstar(W)} \\ \nonumber
     &=&  \EE{(\hstar)^2(W)}
     \end{eqnarray}
     where the first equality holds by the definition of $S_{\mu_n}$ and the fact that, for any $g(Z)$,
     $$
     \EE{\hstar(W) g(Z)} =  \EE{g(Z)\Ec{\hstar(W)}{Z}} = 0
     $$
     and similarly $\EE{h_n(W) g(Z)} =0$.
     The second equality holds by choosing $g(z)$ to be $\Ec{\hstar(W)}{Z=z}$. The third equality is simply from expanding and the last equality holds in case (a). Noticing
     $$
     \EE{(\bar{h}_n(W)-\hstar(W))^2} \le 2\left(\EE{\bar{h}_n^2(W)}+ \EE{(\hstar)^2(W)}  \right) = 4\EE{(\hstar)^2(W)}
     $$
     we thus establish \eqref{eq:4_thm:accuracy}. Regarding case (b), we have
    \begin{eqnarray}\nonumber
    \inf_{\mu'\in S_{\mu_n}}\EE{( \mu'(X,Z)-\mustar(X,Z))^2}
     &=& \inf_{c>0 }
     \EE{(c h_n(W) - \hstar(W) )^2}  \\ \nonumber
     &=& \inf_{c>0 }
     \EE{( c h_n(W) - h_{0}(W) + h_0(W)  -  \hstar(W) )^2}  \\ \nonumber
     &=&  \EE{(h_0(W)  -  \hstar(W) )^2} + \inf_{c>0 }
     \EE{( c h_n(W) - h_{0}(W) )^2}   \\ \nonumber
       &=&  \EE{(h_0(W)  -  \hstar(W) )^2} \\  \label{eq:5_thm:accuracy}
       &=& \EE{(\hstar)^2(W)} -  \EE{(h_0(W)  )^2}
      \end{eqnarray}
     where in the second equality, $h_0$ is defined to be 
     $$
     h_0(w) := \frac{ \EE{h_n(W)\hstar(W)}}{{\EE{h_n^2(W)}}} h_n(w).
     $$
    It satisfies the property $\EE{h_n(W)\left(\hstar(W) - h_0(W)  \right)}=0 $ thus the third equality holds. The fourth equality comes from choosing $c$ to be $ \frac{ \EE{h_n(W)\hstar(W)}}{{\EE{h_n^2(W)}}} $, which is positive in case (b). The last equality holds again due to $\EE{h_n(W)\left(\hstar(W) - h_0(W)  \right)}=0$. And we have
     \begin{eqnarray}\nonumber
     \EE{(\bar{h}_n(W)-\hstar(W))^2} 
     &=& 2  \EE{(\hstar)^2(W)} - 2 \EE{\bar{h}_n(W)\hstar(W)} \\ 
     &=& 2  \EE{(\hstar)^2(W)} - 2 \EE{(h_0(W)  )^2}\rho    \label{eq:6_thm:accuracy}
     \end{eqnarray}    
     where $\rho$ denotes the following term and can be further simplified based on the definition of $\bar{h}_n(W)$ and $ h_0(W)$.
      \begin{eqnarray*}
        \rho
        &:=&\frac{\EE{\bar{h}_n(W)\hstar(W)}}{\EE{(h_0(W)  )^2}}  \\
        &=&  \frac{\Ij\sqrt{\EE{h_n^2(W)}}}{\EE{h_n(W)\hstar(W)} }
     \end{eqnarray*}     
     thus we have $\rho > 0$ in case (b) and $\rho \ge 1$ by the Cauchy--Schwarz inequality. Combining this with \eqref{eq:5_thm:accuracy} and \eqref{eq:6_thm:accuracy} yields \eqref{eq:4_thm:accuracy}. Finally we establish the bound in \eqref{eq:accuracy_gap}.

    \end{proof}

    \subsection{Proofs in Section~\ref{sec:class}}
     \begin{proof}[Proof of Lemma \ref{lem:class_max}]\label{pf:lem:class_max}
        We prove this lemma by a small trick, taking advantage of the idea of symmetry. Remember as in \eqref{eq:nulls_cond_indep}, $\Xj$'s null copy $\Xtil$ is constructed such that
         \begin{equation}\label{eq:nulls_prop}
          \Xkj \independent (\Xj, Y) \mid \Xnoj,\quad\text{and}\quad
      \Xkj \mid \Xnoj \deq \Xj \mid \Xnoj. 
         \end{equation}
        We can define the null copy of $\Yk$ by drawing from the conditional distribution of of $Y$ given $Z$, without looking at $(\Xj,Y)$. Remark that introducing $\Yk$ is just for the convenience of proof and does not necessarily mean we need to be able to sample it. Formally it satisfy
        \begin{equation}\label{eq:Ynull_prop}
             \Yk \independent (\Xj, Y) \mid \Xnoj,~~
            \Yk \mid \Xnoj \deq  Y \mid \Xnoj 
        \end{equation}
        More specifically, we ``generate" $\Yk$ conditioning on $\Xk$, following the same conditional distribution as $Y|X,Z$ (It can be verified this will satisfy \eqref{eq:Ynull_prop}). Now by the symmetry argument, we have 
        \begin{equation}\label{eq:1_lem:class_max}
        \EE{\mathbbm{1}_{\{
        Y\cdot [\muXk - \Ec{\muX}{\Xnoj}] < 0
        \}}} = 
        \EE{\mathbbm{1}_{\{
        \Yk\cdot [\muX - \Ec{\muX}{\Xnoj}] < 0
        \}}}.        
        \end{equation}
        Let $W = (X,Z)$ and define $g(\Xnoj):= \Ec{\mu(W)}{\Xnoj},~h(W):=\mu(W) - g(\Xnoj)$ with the associated functions denoted by $g(\xnoj),~h(w)$, we can rewrite $\kappamu/2$ as
        \begin{eqnarray} \nonumber
        f_{\ell_1}(\mu) /2
        &=&  \mathbb{P}\big(Y (\muXk - \Ec{\muX}{\Xnoj}) < 0\big)-\mathbb{P}\big(Y (\muX - \Ec{\muX}{\Xnoj}) < 0\big) \\ \nonumber
         &=& 
        \EE{\mathbbm{1}_{\{
        \Yk\cdot [\mu(W) - \Ec{\mu(W)}{\Xnoj}] < 0
        \}}} - 
        \EE{\mathbbm{1}_{\left\{
        Y\cdot [\mu(W) - \Ec{\mu(W)}{\Xnoj} ] < 0
        \right\}}} \\ \nonumber
        &=&
        \EE{\Ecmid{\left(
        \mathbbm{1}_{\{
        \Yk\cdot [\mu(W) - \Ec{\mu(W)}{\Xnoj}] < 0
        \}} - 
        \mathbbm{1}_{\left\{
        Y\cdot [\mu(W) - \Ec{\mu(W)}{\Xnoj} ] < 0
        \right\}} 
        \right)
        }{W}} \\ \nonumber
        &=&
        \EE{\Ecmid{\left(
        \mathbbm{1}_{\{
        \Yk\cdot h(W) < 0
        \}} - 
        \mathbbm{1}_{\left\{
        Y\cdot h(W) < 0
        \right\}} 
        \right)
        }{W}}          
        \end{eqnarray}
        where the second equality is by \eqref{eq:1_lem:class_max}, the third one comes from the law of total expectation and the fourth one is by the definition of $h(W)$. Now it suffices to consider maximizing the following quantity
        \begin{equation}\label{eq:2_lem:class_max}
            \Ecmid{\left(
        \mathbbm{1}_{\{
        \Yk\cdot h(W) < 0
        \}} - 
        \mathbbm{1}_{\left\{
        Y\cdot h(W) < 0
        \right\}} 
        \right)
        }{W=w}
        \end{equation}
        for each $w=(x,z)$. Due to the property \eqref{eq:Ynull_prop}, we have 
        $$
        \Pc{\Yk=y}{W} = \Pc{\Yk=y}{\Xnoj}= \Pc{Y=y}{\Xnoj}~~~ y \in \{-1,1\}
        $$
        hence we can simplify the conditional expectation of the first indicator function in \eqref{eq:2_lem:class_max} into the following 
        \begin{eqnarray} \nonumber
             \Ec{
        \mathbbm{1}_{\{
        \Yk\cdot h(W) < 0
        \}} 
        }{W=w} 
        &=& 
        \Pc{\Yk = 1, ~h(W) < 0}{W=w} +  \Pc{\Yk = -1, ~h(W) > 0}{W=w}\\
        &=& 
        \Pc{Y=1}{\Xnoj = \xnoj}\indicat{h(w) < 0} + \Pc{Y=-1}{\Xnoj =\xnoj}\indicat{h(w) > 0}\label{eq:3_lem:class_max}
        \end{eqnarray}
        Similarly we have
        \begin{equation}\label{eq:4_lem:class_max}
             \Ec{
        \mathbbm{1}_{\{
        Y\cdot h(W) < 0
        \}} 
        }{W=w} =
        \Pc{Y=1}{W=w}\indicat{h(w) < 0} + \Pc{Y=-1}{W=w}\indicat{h(w) > 0}
        \end{equation}
        when $\Ec{Y}{W=w}> \Ec{Y}{\Xnoj=\xnoj}$, we have 
        $$
        \Pc{Y=1}{W=w} >  \Pc{Y=1}{\Xnoj = \xnoj}, ~\Pc{Y=-1}{W=w} <  \Pc{Y=-1}{\Xnoj = \xnoj},
        $$
        hence in this case, by comparing \eqref{eq:3_lem:class_max} and \eqref{eq:4_lem:class_max} we know $h(w)> 0$ will maximize \eqref{eq:2_lem:class_max} with maximum value 
        \begin{eqnarray}\nonumber
         \Pc{Y=-1}{\Xnoj = \xnoj} - \Pc{Y=-1}{W=w} 
         &=& (1- \Ec{Y}{\Xnoj=\xnoj})/2 - (1- \Ec{Y}{W=w})/2 \\
         &=& (\Ec{Y}{W=w} - \Ec{Y}{\Xnoj=\xnoj})/2
        \end{eqnarray}
        Similarly we can figure out the maximizer of $h(w)$, when $\Ec{Y}{W=w}< \Ec{Y}{\Xnoj=\xnoj}$. Finally we have 
        \begin{equation}\label{eq:class_maximizer1}
        h(w)  \left\{
                \begin{array}{ll}
                  > 0,  ~~\text{when}~\Ec{Y}{W=w}> \Ec{Y}{\Xnoj=\xnoj}\\
                  < 0,  ~~\text{when}~\Ec{Y}{W=w}< \Ec{Y}{\Xnoj=\xnoj} \\
                 \text{can be any choice},  ~~\text{when}~\Ec{Y}{W=w}= \Ec{Y}{\Xnoj=\xnoj}
                \end{array}
              \right.
         \end{equation}
        will maximize \eqref{eq:2_lem:class_max} with the maximum value $|\Ec{Y}{W=w} - \Ec{Y}{\Xnoj=\xnoj}|/2$. Remark the definition of $h(w)  = \mu(w) - g(z) $, we can restate \eqref{eq:class_maximizer1} as 
         \begin{equation}\label{eq:class_maximizer2}
         \left\{
                \begin{array}{ll}
                 \mu(x,z) = \mu(w) > g(\xnoj),  ~~\text{when}~\Ec{Y}{W=w}> \Ec{Y}{\Xnoj=\xnoj}\\
                 \mu(x,z) =  \mu(w) < g(\xnoj),  ~~\text{when}~\Ec{Y}{W=w}< \Ec{Y}{\Xnoj=\xnoj} \\
                 \text{can be any choice},  ~~\text{when}~\Ec{Y}{W=w}= \Ec{Y}{\Xnoj=\xnoj}
                \end{array}
              \right.
         \end{equation} 
         where again $g(\xnoj)=\Ec{\muX}{\Xnoj = \xnoj}$. Apparently, choosing $\mu(x,z)$ to be the true regression function $\mustar(x,z)$ will satisfy \eqref{eq:class_maximizer2}. Hence we show $\kappamu$ is maximized at $\mustar$ with maximum value 
         $$
         \mathbb{E}\left|\condmeanj - \condmean\right|
         $$
         which equals $\Ijc$. Clearly from \eqref{eq:class_maximizer2}, 
         $\mustar(x,z)$ is not the unique maximizer and any function in the set described in the following set can attain the maximum.
        \begin{equation}\label{eq:class_maximizer_all}
    	\{\mu: \mathbb{R}^{p} \rightarrow \mathbb{R} \mid \text{sign}\left(\mu(x,z)-\Ec{\muX}{\Xnoj = \xnoj}\right) =\text{sign}\left(\Ec{Y}{X = x} - \Ec{Y}{\Xnoj =\xnoj}\right)\}.
        \end{equation}
    \end{proof}
       \begin{proof}[Proof of Theorem \ref{thm:class_main}]\label{pf:thm:class_main}
% 		Noticing $\Ljc$ is the asymptotic confidence lower bound for $\kappamu$ based on the normal approximations, where we apply the central limit theorem to i.i.d. random variables $R_{ij}$. 
		According to Algorithm \ref{alg:class_MOCK}, we first denote 
		\begin{eqnarray}
		 &~& U:=\mu(X,Z),~~g(\xnoj):=\ec{\mu(X,Z)}{\Xnoj = \xnoj},\\ \nonumber
        &~& G_{z}(u):=\Pc{U < u}{\Xnoj =\xnoj},~~F_{z}(u):=\Pc{U \le u}{\Xnoj =\xnoj}. \\ \nonumber
		\end{eqnarray}
		thus have the following expression of $R_i$:
		$$
		R_{i} = G_{\Xinoj}(g(Z_i))\indicat{Y_i=1} + (1-F_{\Xinoj}(g(Z_i)))\indicat{Y_i=-1} - \indicat{Y_i(\mu(W_i) -g(Z_i) ) < 0} 
		%~~\text{where}~ g(Z_i) = \Ec{\mu(W_i)}{\Xinoj}
		$$
		%and $ F_{\Xinoj}(\cdot), G_{\Xinoj}(\cdot)$ are defined in \eqref{eq:cdf_notation}.
% 		$$ 
% 		R_{ij}=\frac{1}{K}\sum_{k=1}^{K}\left(\indicat{Y_i(\mu(\Xk_{ij}^{(k)}) -G_{ij})\le 0} \right)- \indicat{Y_i(\mu(W_i)) -G_{ij})\le 0},~~\text{where}
% 		~G_{ij} = \Ec{\mu(W_i)}{\Xinoj})
% 		 $$
        First we prove that $\EE{R_{i}} = \kappamu/2$. Recall the definition of $f_{\ell_1}(\mu)$ in \eqref{eq: kappaj_u},
        $$
        f_{\ell_1}(\mu)/2 = 
        \EE{\mathbbm{1}_{\{
        Y\cdot [\muXk - \Ec{\muX}{\Xnoj}] < 0
        \}}} - 
        \EE{\mathbbm{1}_{\left\{
        Y\cdot [\muX - \Ec{\muX}{\Xnoj} ] < 0
        \right\}}},
        $$
        let $W=(X,Z)$, then it suffices to show the following
        \begin{equation}\label{eq:1_thm:class_main}
        \EE{ G_{\Xnoj}(g(Z))\indicat{Y=1} + (1-F_{\Xnoj}(g(Z)))\indicat{Y=-1}} 
        = \EE{\mathbbm{1}_{\{
        Y\cdot [\muXk - \Ec{\mu(X,Z)}{\Xnoj}] < 0
        \}}}.
        \end{equation}
        By the law of total expectation we can rewrite the right hand side as
        \begin{equation}\nonumber
            \EE{\Ec{\mathbbm{1}_{\{
        Y\cdot [\muXk - \Ec{\muX}{\Xnoj}] < 0
        \}}}{\Xnoj,Y}}.
        \end{equation}
        Due to the property \eqref{eq:nulls_prop}, we have $\Xtil \independent  (Y,Z) \mid\Xnoj$ and $\Xtil \mid \Xnoj \sim X \mid \Xnoj$, which yields
        $$
        \Ec{\mathbbm{1}_{\{
        Y\cdot [\muXk - \Ec{\muX}{\Xnoj}] < 0
        \}}}{\Xnoj=\xnoj,Y=1} = G_{\Xnoj}(g(Z))\indicat{Y=1}.
        $$
        And we can do similar derivations when $Y=-1$. Thus we can prove $\EE{R_{i}} = \kappamu/2$ by showing \eqref{eq:1_thm:class_main}. In light of the deterministic relationship in Lemma \ref{lem:class_max}, we have $\{ \Ljc(\mu) \le \kappamu  \}\subset \{ \Ljc(\mu) \le \Ijc  \}$, hence it suffices to prove
    \begin{equation} 
	\label{eq:MC_cover_restate}
	%\nonumber
	    %\liminf_{n \rightarrow \infty}\mathbb{P}_{{P}}\left(\Ljc(\mu) \le \kappamu \right) 
	    \PP{ \Ljc(\mu) \le \kappamu}
	    \ge 1 - \alpha - O(n^{-1/2}).
	\end{equation}
      Note that $\Var{R_{i}}$ always exist due to the boundedness. When $\Var{R_{i}}=0$, we have $R_{i}=\kappamu/2 = \bar{R}$ and $s=0$, thus $\Ljc(\mu)=  \kappamu$, hence \eqref{eq:MC_cover_restate} trivially holds. Remark this includes the case when $\muX \in \sigalg(\Xnoj)$. Otherwise, applying Lemma \ref{lem:berry_t} to i.i.d. bounded random variables $R_{i}$ will yield \eqref{eq:MC_cover_restate},
    %     \begin{equation} 
	   %  %\label{eq:1_thm:class_main}
    % 	\nonumber
	   % \mathbb{P}_{{P}}\left(\Ljc(\mu) \le    \kappamu\right) \ge
 	  %   1 - \alpha - O(1/\sqrt{n}).
	   % \end{equation}
	    where the constant will depend on $\Var{R_{i}}$. 
     \end{proof}

    \subsection{Proofs in Section~\ref{sec:relax}}
    \begin{proof}[Proof of Theorem \ref{thm:batch_cond}]
    \label{pf:thm:batch_cond}  
    %\begin{enumerate}[(1)]
     When $\calT$ is degenerate or $\mu(X) \in \sigalg(\Xnoj)$, we immediately have $\LjT(\mu)= 0$ according to Algorithm \ref{alg:batch_cond}, which implies the coverage validity. Below we focus on the non-trivial case. 
        %where $\Delta_{mj}$ is well-defined and assumed to satisfy $0<c_{0}< n_{2}\mathrm{Var}(\Delta_{mj}) < c_{1}$. 
        Due to the deterministic relationship 
    $$
    \thetamuT \le \thetamusT \le \thetamus = \Ij,
    $$ 
    it suffices to prove 
   \begin{eqnarray}\label{eq:relax_coverage}
    \Pp{P}{ \LjT(\mu) \le \thetamuT}     
	           \ge 1 - \alpha - o(1).
    \end{eqnarray}
    which can be reduced to establishing certain asymptotic normality based on i.i.d. random variables $R_{m},V_m, m\in [n_{1}]$ whenever the variance of the asymptotic distribution is nonzero. First, we verify that under the stated conditions, all the involving moments are finite, which can be reduced to show
    $$
    \Var{R_m},\Var{V_m}< \infty.
    $$
    For a given $n_2$, it can be further reduced to the following
    \begin{eqnarray}\nonumber
    \Var{Y_i\, (\mu(X_i,Z_i) - \Ec{\mu(X_i,Z_i)}{\bs{Z}_m, \bT_m}}\\ \nonumber
    \Var{ \Varc{\mu(X_i,Z_i)}{\bs{Z}_m,\bT_m}}<\infty. 
    \end{eqnarray}
    Using similar strategies in the proof of Theorem \ref{thm:main}, we can show the above holds under the moment conditions $\EE{Y^4},\EE{\mu^4(X)}<\infty$ by the Cauchy--Schwarz inequality and the tower property of conditional expectation.
    
    Note that in the proof of the main result, i.e. Theorem \ref{thm:main}, we consider four different cases based on whether some variances are zero or not. Here we only pursue the asymptotic coverage validity, then the discussion on those four different cases becomes very straighforward. When both the variances of $R_m, V_m$ are zero, we have $\bar{R}/\bar{V} = \thetamuT,~s^2 =0$, then \eqref{eq:relax_coverage} holds immediately. When $\Var{V_m}=0$, we can simply establish the asymptotic normality by the central limit theorem. Otherwise, delta method can be applied. Here we give the derivation for the most non-trivial case where $ \Var{R_m},\Var{V_m}>0$. Denote random vectors $\{U_m\}_{m=1}^{n_1}  = \{(U_{m1},U_{m2})\}_{m=1}^{n_1} \iid U = (U_1, U_2)$ to be 
    \begin{eqnarray}
        \label{eq:thm_relax_V1V2_def}
	    U_{m1} &=& R_m - \EE{Y_i\, (\mu(X_i,Z_i) - \Ec{\mu(X_i,Z_i)}{\bs{Z}_m, \bT_m}},\\
        U_{m2} &=& V_m - \EE{\Varc{\mu(X_i,Z_i)}{\bs{Z}_m,\bT_m}}
    \end{eqnarray}
    hence we have $\EE{U}=0$. Denote $h^{\calT}(W_i) = \mu(X_i,Z_i) - \Ec{\mu(X_i,Z_i)}{\bs{Z}_m, \bT_m}$, we have the following holds
      \begin{eqnarray} \nonumber
     \thetamuT
        &=& \frac{\EE{\covc{\mustar(X_i,Z_i)}{\mu(X_i,Z_i)}{\bs{Z},\bs{T}}}}{\sqrt{\EE{\varc{\mu(X_i,Z_i)}{\bs{Z},\bs{T}}}}} \\ \nonumber
        &=& \frac{\EE{\covc{\mustar(X_i,Z_i)}{h^{\calT}(W_i)}{\bs{Z},\bs{T}}}}{\sqrt{\EE{ \EE{(h^{\calT}(W_i)^2)}}}} \\ \nonumber
        &=& \frac{\EE{{\mustar(X_i,Z_i)}{h^{\calT}(W_i)}}}{\sqrt{\EE{ \EE{(h^{\calT}(W_i)^2)}}}} \\ \nonumber
        &=& \frac{\EE{Y_ih^{\calT}(W_i)}  }{ \sqrt{\EE{(h^{\calT}(W_i)^2)} }} , 
    \end{eqnarray}
    where the first equality holds by the definition of $\thetamuT$, the second inequality holds by the definition of $h^{\calT}(W_i) $. Regarding the third equality, we make use of the fact $\Ec{h^{\calT}(W_i)}{\bs{Z}_m, \bT_m} = 0$ and the tower property of conditional expectation. The last inequality holds by the tower property of conditional expectation and the fact that $h^{\calT}(W_i)  \in \sigalg(\bs{X}_m,\bs{Z}_m) $.
    Let $T=\bar{R}/\bar{V}$, then $T - \thetamuT$ can be rewritten as 
      $$
	  T - \thetamuT = \frac{\bar{U}_1 +  \EE{Y_ih^{\calT}(W_i)} }{ \sqrt{\bar{U}_2 +  \EE{(h^{\calT}(W_i)^2)} }} - \frac{\EE{Y_ih^{\calT}(W_i)}  }{ \sqrt{\EE{(h^{\calT}(W_i)^2)} }} : = H(\bar{U})
	  $$
	  where $\bar{U} = (\bar{U}_1, \bar{U}_2) = \frac{1}{n_1}\sum\nsubp U_m$ and $H:\mathbb{R}^2 \rightarrow \mathbb{R}$ is defined through the following:
      $$
      H(x) = H(x_1,x_{2}) :=  \frac{x_1 +  \EE{Y_ih^{\calT}(W_i)} }{ \sqrt{x_2 +  \EE{(h^{\calT}(W_i)^2)} }} - \frac{\EE{Y_ih^{\calT}(W_i)}  }{ \sqrt{\EE{(h^{\calT}(W_i)^2)} }} : = H(\bar{U})
      $$
      when $x_{2} > -  \EE{(h^{\calT}(W_i)^2)}$ and is set to be $\frac{\EE{Y_ih^{\calT}(W_i)}}{ \sqrt{\EE{(h^{\calT}(W_i)^2)} }} $ otherwise. Note that the first order derivatives of $H(x)$ exists, by applying the multivariate Delta method to mean zero random vectors $\{(U_{m1},U_{m2})\}_{m=1}^{n_1} $ with the nonlinear function chosen as $H$, we have
       \begin{equation}\nonumber
    	     \sqrt{n_1}(T -   \thetamuT ) \stackrel{d}{\rightarrow} \gauss{0}{\tilde{\sigma}^2}
    	 \end{equation}
      whenever the variance term $\tilde{\sigma}^2$ is nonzero. Exactly following the strategy in the proof of Theorem \ref{thm:main}, we have $\tilde{\sigma}^2>0$ under the case where $ \Var{R_m},\Var{V_m}>0$. Also notice $s^2$ is a consistent estimator of $\tilde{\sigma}^2$, then by the argument of Slutsky's Theorem, \eqref{eq:relax_coverage} is established.   
 \end{proof}

    \section{An example for projection methods}
    \label{sec:comparison}
    % As the name suggests, (linear) projection methods are still seeking sort of "linear" solutions to the variable importance measure problem even though they are not assuming the true model to be linear. However, (linear) projection methods can be shown to give undesirable results in some situations, compared with the floodgate approach (even with linear regression algorithms chosen). Below we 
    Consider covariates $W=(W_1,W_2)$ distributed as $W_1\sim\calN(0,1)$ and $W_2= W_1^2+\calN(0,1)$. Let $Y = W_1^2+\calN(0,1)$, with all the Gaussian random variables independent.
    % \begin{equation*}
    % Y = W_{1}^2+\epsilon,~~~
    % W_{2}= W_{1}^2+G,~~~W_1 \sim \gauss{0}{1},~\epsilon \sim \calN(0,1),~G  \sim \calN(0,1)
    % \end{equation*}
    % where the Gaussian random variables are all independent, and denote $W=(W_1,W_2)$. 
    %Clearly, $W_{1}$ is a important variable for $Y$ and $W_{2}$ is not important. In the context of conditional independence, w
    Then $W_1$ is the only important variable; formally: $W_{1}\centernot{\independent}Y\mid W_{2}$ and $W_{2} \independent Y\mid W_{1}$. 
    % However, projection approaches fail to respect these and give a completely opposite answer. Note that they are targeting the 
    But the projection parameters are $(\EE{W^{\top}W} )^{-1}\EE{WY}=(0,\frac{3}{4})^{\top}$, i.e., zero for the non-null covariate and non-zero for the null covariate. 
    % Floodgate approach (even with $\mu$ fitted from linear models) will still not contradict the conditional independence framework in the sense $\thetamu$ equals $ \EE{\Covc{Y}{W_1}{W_2}} =0$, $ \EE{\Covc{Y}{W_2}{W_1}} =0$
    % %$\mathbb{E}\mathrm{Cov}(Y,X_1|X_{2})=0 \EE{\Covc{Y}{W_1}{W_2}}, \mathbb{E}\mathrm{Cov}(Y,X_2|X_{1})=0$ 
    % (up to some constant). Using linear model as fitting algorithm for floodgate is not reasonable and that is why the non-null random variable $X_{1}$ can not successfully discovered. But we claim even with undesirable regression algorithms which can be completely wrong choices, floodgate will not given wrong answers in terms of the validity of confidence lower bounds.

    \section{Rate results}
    \label{sec:rate}
  	\acc{
    \begin{theorem}[Floodgate validity]\label{thm:main_rate}
    For any given working regression function $\mu:\mathbb{R}^{p} \rightarrow \mathbb{R}$ and i.i.d. data $\{(Y_i,X_i,Z_i)\}_{i=1}^n$,   if $\ee{Y^{12}},~\ee{\mu^{12}(X,Z)} <\infty$, then $L_{n}^{\alpha} (\mu)$ from Algorithm~\ref{alg:MOCK} satisfies 
	\begin{equation} 
	\nonumber
        \mathbb{P}\left(L_{n}^{\alpha} (\mu) \le \Ij \right) \ge 1 - \alpha - Cn^{-1/2}
	\end{equation}
	for some constant $C$ depending only on the moments of $Y$ and $\mu(X, Z)$.
    \end{theorem}
    The proof can be found in Appendix \ref{pf:thm:main_rate}. Establishing the $n^{-1/2}$ rate requires relatively recent Berry--Esseen-type results for the delta method \citep{pinelis2016optimal} and also necessitates the existence of 12th moments.

    \begin{theorem}\label{thm:main_general_rate}
	 Assume the conditions of Theorem~\ref{thm:main_rate} and $\EE{\Varc{Y (\mu(X,Z)-\Ec{\mu(X,Z)}{Z})}{\Xnoj}}>0$. $L_{n,K}^{\alpha} (\mu)$ computed by replacing $R_i$ and $V_i$ with $R_i^K$ and $V_i^K$, respectively, in Algorithm \ref{alg:MOCK} satisfies
	\begin{equation} 
	\nonumber
	\inf_{K>1} 
	\PP{L_{n,K}^{\alpha} (\mu) \le \Ij }  \ge 1 - \alpha - Cn^{-1/2}
	\end{equation}
	for some constant $C$ depending only on the moments of $Y$ and $\mu(X, Z)$.
    \end{theorem}
    The proof can be found in Appendix \ref{pf:thm:main_general_rate}. Note that the additional assumption beyond Theorem~\ref{thm:main_rate} of $\EE{\Varc{Y (\mu(X,Z)-\Ec{\mu(X,Z)}{Z})}{\Xnoj}}>0$ is only needed for $n^{-1/2}$-rate coverage validity \emph{uniformly} over $K>1$, and could be removed for the same result for any fixed $K>1$.
    }
     \subsection{Proofs in Appendix \ref{sec:rate}}
    \subsubsection{Theorem \ref{thm:main_rate}}
    \label{pf:thm:main_rate}
    \begin{proof}[Proof of Theorem \ref{thm:main_rate}]
%    Due to \eqref{eq:element_mmt_bound}, $\EE{\mu^{12}(X, Z)} < \infty$ implies $\EE{h^{12}(W)} < \infty$. In the following proof, we will only assume the weaker moment conditions that $\EE{Y^{12}}, \EE{h^{12}(W)} < \infty$.  Under such moment conditions, we also have $\EE{Yh(W)} \le \sqrt{\EE{Y^2}}\sqrt{\EE{h^2(W)}}$ and $\EE{h^2(W)} < \infty$  since the finiteness of higher moments implies that of lower moments. 
    \acc{Recall in Algorithm \ref{alg:MOCK}, we denote $R_{i} = Y_i\big(\mu(X_i,Z_i)-\Ec{\mu(X_i,Z_i)}{\Xinoj}\big)$ and $V_{i} = \Varc{\mu(X_i,Z_i)}{\Xinoj}$ for each $i\in [n]$, and compute their sample mean $(\bar{R},\bar{V})$ and sample covariance matrix $\hat{\Sigma}$. The LCB is constructed as
    $$
             L_{n}^{\alpha} (\mu)=\max\left\{ \frac{\bar{R}}{ \sqrt{\bar{V}  }} -  \frac{z_{\alpha}s}{\sqrt{n}}   ,0\right\},~~\text{where}~~  s^2 = \frac{ 1 }{ \bar{V} }\left[ \left(\frac{\bar{R} }{2 \bar{V} }\right)^2 \hat{\Sigma}_{22} +  \hat{\Sigma}_{11} - \frac{\bar{R} }{ \bar{V} } \hat{\Sigma}_{12} \right].     
     $$ 
    Following exactly the same discussions as those from the beginning to \eqref{eq:valid_atmu} in the proof of Theorem \ref{thm:main}, we have 
    \begin{itemize}
    	\item Theorem \ref{thm:main_rate} can be proved under the weaker moment conditions that $\EE{Y^{12}}, \EE{h^{12}(W)} < \infty$, which is assumed for the following proof;
    	\item it suffices to prove
    	       \begin{equation}
            \label{eq:unif_valid_atmu}
        	   % \inf_{P\in \calP,~ \mu \in \calU}
        	    \PP{\frac{\bar{R}}{ \sqrt{\bar{V}  }} -  \frac{z_{\alpha}s}{\sqrt{n}} \le \thetamu} \ge 1 - \alpha - C/\sqrt{n}
        	\end{equation}
        	for some constant $C$ when $\EE{\varc{\muX}{Z}}  \neq 0$;
        \item we can assume $\EE{h^2(W)}=1$ without loss of generality.
    \end{itemize}
  	We will utilize Berry--Esseen-type bounds to prove \eqref{eq:unif_valid_atmu}. Now we still consider the following four cases. 
        	\begin{enumerate}[(I)]
        	    \item $\Var{Yh(W)}=0$ and $ \Var{\Varc{h(W)}{\Xnoj}}=0$.
        	    \item $\Var{Yh(W)}>0$ and $ \Var{\Varc{h(W)}{\Xnoj}}=0$.
        	    \item $\Var{Yh(W)}=0$ and $ \Var{\Varc{h(W)}{\Xnoj}}>0$.
        	    \item $\Var{Yh(W)}>0$ and $ \Var{\Varc{h(X)}{\Xnoj}}>0$.
        	\end{enumerate}
        	Note that assuming $\EE{Y^{12}}$ and $\EE{h^{12}(W)} < \infty$ ensures all the above variances exist due to the same bounding strategy as \eqref{eq:element_mmt_bound}.
%        	Notice that, when $\Var{Yh(W)}=0$, we have
%             $
%             R_{i} = {\EE{Yh(W)}}
%             $ for $i \in [n]$, and thus $\bar{R} = {\EE{Yh(W)}},~ \hat{\Sigma}_{11} =\hat{\Sigma}_{12} =0$; when $\Var{\Varc{h(W)}{\Xnoj}}=0$, we have  $
%             V_{i} = \EE{h^2(W)}$ for $i\in [n]$, and thus $\bar{V} = {\EE{h^2(W)}},~ \hat{\Sigma}_{22} =\hat{\Sigma}_{12} =0$.   

             Case (\myRom{1}): \eqref{eq:unif_valid_atmu} holds by the discussion for Case (\myRom{1}) in the proof of Theorem \ref{thm:main}. 
            
             Case (\myRom{2}): due to the derivations for Case (\myRom{2}) in the proof of Theorem \ref{thm:main}, the problem is reduced to showing
             \begin{equation} 
        	    \label{eq:case2_coverage_rate}
                     \PP{   \bar{R} - \frac{z_{\alpha}(\hat{\Sigma}_{11})^{1/2}}{\sqrt{n}} \le  \EE{Yh(W)} } \ge 1 - \alpha - C/\sqrt{n}.
        	    \end{equation}
             As mentioned in the proof of Theorem \ref{thm:main}, $\bar{R}$ is simply the sample mean estimator of the quantity $\EE{Yh(W)}$ and $\hat{\Sigma}_{11}$ is the corresponding sample variance. Therefore, the CLT and Slutsky's theorem immediately establish the asymptotic coverage validity. To prove the $1/\sqrt{n}$ rate in \eqref{eq:case2_coverage_rate}, stronger results are needed.} The classical Berry--Esseen bound serves as the main ingredient, which states that
            \begin{lemma}[Berry--Esseen bound]
        	\label{lem:berry}
        	    There exists a positive constant C, such that for $i.i.d.$ mean zero random variables $X_1,\dots,X_n$ satisfying
        	   \begin{enumerate}[(1)]
        	       \item $\ee{X_{1}^{2}}=\sigma^2 > 0$
        	       \item $\ee{|X_1|^3} = \rho < \infty$
        	   \end{enumerate}
                if we define $F_n(x)$ to be the cumulative distribution function (CDF) of the scaled average ${\sqrt{n}\bar{X}}/{\sigma}$ and denote the CDF of the standard normal distribution by $\Phi(x)$, then we have
        	    \begin{equation} 
        	    \label{eq:1_lem:berry}
        	        \sup_{x\in \mathbb{R}}\left| F_n(x) - \Phi(x) \right| \le \frac{C\rho}{\sigma^3 \sqrt{n}}.
        	    \end{equation}
            \end{lemma}
            
            Since $\sigma$ in the above result is generally unknown and usually replaced by the sample variance $s_{\sigma}^2 = \frac{1}{n}\sum_{i=1}^n(X_i - \bar{X})^2$, we need the following lemma, which is proved in \citet{bentkus1996berry}.
           \begin{lemma}[Berry--Esseen bound for Student's statistic]
            \label{lem:berry_t}
        	Under the same conditions as in Lemma \ref{lem:berry}, if we redefine $F_n(x)$ to be the cumulative distribution function (CDF) of the Student t-statistic ${\sqrt{n}\bar{X}}/{s_{\sigma}}$, then we have the following Berry--Esseen bound
        	\begin{equation} 
        	    \label{eq:1_lem:berry_t}
        	        \sup_{x\in \mathbb{R}}\left| F_n(x) - \Phi(x) \right| \le \frac{C'\rho}{\sigma^3 \sqrt{n}}.
             \end{equation}        	
        \end{lemma}
        
         To apply Lemma \ref{lem:berry_t}, since we are in Case (\myRom{2}) where $\Var{\Varc{h(W)}{\Xnoj}}=0$ and $\Var{Yh(W)}>0$, it suffices to verify the finiteness of the term ``$\rho$" in our context:
          \begin{eqnarray} \nonumber
            \rho 
            &=& {\EE{\left|Yh(W) - \EE{Yh(W)}  \right|^3 }} \\
            &\le&  \nonumber
             { 2^{3-1}\left( \EE{Y^3h^3(W)}
             +  |\EE{Yh(W)} |^3
             \right)
             } 
             < \infty
            \end{eqnarray}
            \acc{where the equality holds since we assume $\EE{h^2(W)}=1$ and the inequality comes from the $C_r$ inequality.} For the last inequality, using the Cauchy--Schwarz inequality and the fact that higher moments dominate lower moments, we obtain the finiteness when assuming $\EE{Y^6}, \EE{h^6(W)}<\infty$, which holds under the assumed moment conditions. Now by applying the Berry--Esseen bound in Lemma \ref{lem:berry_t} with $\bar{X}=\bar{R} - \EE{Y h(W)}$ and $s_{\sigma}^2= \hat{\Sigma}_{11}$, we obtain \eqref{eq:case2_coverage_rate}.

        \acc{Case (\myRom{3}): due to \eqref{eq:fact1}, we have 
            $$
            \frac{\bar{R}}{ \sqrt{\bar{V}  }} -  \frac{z_{\alpha}s}{\sqrt{n}} =\frac{\EE{Yh(W)}}{ \sqrt{\bar{V}  }} -  \frac{z_{\alpha}s}{\sqrt{n}},~~\text{where}~~ s^2 = \frac{1}{\bar{V}}\left(\frac{\EE{Yh(W)}  }{2\bar{V}}\right)^2 \hat{\Sigma}_{22}.
             $$}
        	Note $\frac{\EE{Yh(W)}}{ \sqrt{\bar{V} }} $ is a nonlinear function of the moment estimators,
        % 	By applying CLT to i.i.d. random variables $Y_{i}h(X_i)/\sqrt{\ee{h^2(W_i)}}$ with $h(X_i) = \mu(W_i) - \ec{\mu(W_i)}{\Xinoj}$ together with a Slutsky's theorem to deal with the replacement of $\sqrt{\ee{h^2(W_i)}}$ by $\sqrt{\bar{\sigma}_j }$ , 
        	so the following asymptotic normality result is a direct consequence of the multivariate delta method,
    	 \begin{equation}\nonumber
    	     \sqrt{n}\left(\frac{\EE{Yh(W)}}{ \sqrt{\bar{V} }}-   \thetamu \right) \stackrel{d}{\rightarrow} \gauss{0}{\tilde{\sigma}_0^2},
    	 \end{equation}
          {where $\tilde{\sigma}_0^2 = H_2(\bs{0})$ will be specified later (see the definition of $H_2(x)$ in \eqref{eq:H2_def})} and $s^2$ in ${L}_n^{\alpha}(\mu)$ is a consistent estimator of it. To establish the rate $1/\sqrt{n}$, the classical Berry--Esseen result needs to be extended for nonlinear statistics. Note that Case (\myRom{4}) involves a nonlinear statistic too, and is a bit more complicated. Hence we focus on Case (\myRom{4}) and omit the very similar proof for Case (\myRom{3}). 
          
          Case (\myRom{4}):
             \acc{Denote $T:=\left(\frac{\bar{R}}{ \sqrt{\bar{V}}} - f(\mu)\right)/s$.} Under specific moment conditions, we will establish the Berry--Esseen-type bound below:
	         \begin{equation}\label{eq:T_BE_bound}
	          \sup_{t \in \mathbb{R}} \left|
    	     \PP{\sqrt{n} T\le t }  -  \Phi(t)
    	     \right|  = O\left(\frac{1}{\sqrt{n}}\right)
    	    \end{equation}
	     where $\Phi(t)$ denotes the CDF of the standard normal distribution.
        %   To establish the rate $1/\sqrt{n}$ in \eqref{eq:case3_coverage}, Berry--Esseen type bounds are needed. In the case where $\Var{\Varc{h(X)}{\Xnoj}}=0$, we have $ \sigma_{ij} =\EE{h^2(X)} $ for $i \in[n]$ thus $\sqrt{\bar{\sigma}_j}=\sqrt{\EE{h^2(X)}} $ and $ \tilde{\sigma} = \sqrt{\Var{Yh(W)}/\EE{h^2(X)}}$, then exactly following the same idea as proving \eqref{eq:case2_?coverage} and a classical Berry--Esseen result in Lemma \ref{lem:berry} would be sufficient. The details are thus omitted. Below we focus on the case where both 
        %   $\Var{Yh(W)}>0$ and $\Var{\Varc{h(X)}{\Xnoj}}>0$, and the classical Berry--Esseen result needs to be extended for nonlinear statistics. 
%         	\begin{enumerate}[(a)]
% 	        \item Under specific moment conditions, we will establish a Berry--Esseen-type bound for the nonlinear statistic $T$ with the usual rate:
% 	         \begin{equation}\label{eq:T_BE_bound}
% 	          \sup_{t \in \mathbb{R}} \left|
%     	     \PP{\sqrt{n}(T - \thetamu)\le t\tilde{\sigma} }  -  \Phi(t)
%     	     \right|  = O\left(\frac{1}{\sqrt{n}}\right)
%     	    \end{equation}
% 	     where $\Phi(t)$ denotes the CDF of the standard normal distribution.
% 	     \item By verifying that $s$ satisfies Lemma~\ref{lem:generic_cover_rate}'s consistency rate assumption and combining it with the above Berry--Esseen bound, we apply Lemma \ref{lem:generic_cover_rate} to establish \eqref{eq:unif_valid_atmu}.
% 	  \end{enumerate}

 The proof relies on a careful analysis of nonlinear statistics. We take advantage of the results in a recent paper \citep{pinelis2016optimal} that establishes Berry--Esseen bounds with rate $1/\sqrt{n}$ for the multivariate delta method when the function applied to the sample mean estimator satisfies certain smoothness conditions. And the constants in the rate depend on the distribution only through several moments. Specifically, consider $U,U_1,\dots,U_n$ to be i.i.d. random vectors on a set $\mathcal{X}$ and a functional $H: \mathcal{X} \rightarrow \mathbb{R}$ which satisfies the following smoothness condition: 
	  \begin{cond}\label{cond:f_smoothness}
	   There exists $\varepsilon, M_{\varepsilon} >0$ and a continuous linear functional $L:\mathcal{X} \rightarrow \mathbb{R}$ such that   
	 \begin{equation}\label{eq:f_smoothness}
	 |H(x) - L(x)| \le   M_{\varepsilon} \|x\|^2 ~~~\text{for all }x \in  \mathcal{X}  ~\text{with } \|x\|\le \varepsilon 
	 \end{equation}
	  \end{cond}
	  We can think of $L$ as the first-order Taylor expansion of $H$. This smoothness condition basically requires $H$ to be nearly linear around the origin and can be satisfied if its second derivatives are bounded in the small neighbourhood $\{x: \|x\|\le \varepsilon\}$ . Before stating \citet{pinelis2016optimal}'s result (we change their notation to avoid conflicts with the notation in the main text of this paper), define $\bar{U} := \frac{1}{n}\sum\nsubp U_i$ and
	   \begin{eqnarray}\nonumber
	    \tilde{\sigma} :=  \|L(U)\|_2,~~
	     \nu_p := \|U\|_p,~~
	     \varsigma_p := \frac{\|L(U)\|_p }{\tilde{\sigma}}, 
	 \end{eqnarray}
	 where for a given random vector $U=(U_{1},\cdots, U_d) \in \mathbb{R}^d$, $\|U\|_p$ is defined as $\|U\|_p = (\EE{\|U\|^p})^{1/p} $ with $\|u\|^{p}:= \sum_{j=1}^{d}|u_{j}|^p$.
	  \begin{theorem}{\cite[Theorem 2.11]{pinelis2016optimal}}\label{thm:Berry_Esseen_nonlinear}
	  Let $\mathcal{X}$ be a Hilbert space, let $H$ satisfy Condition \ref{cond:f_smoothness} for some $\eps>0$, and assume $\EE{U}=0$, $\tilde{\sigma}>0$ and $\nu_3 < \infty$, then
	  \begin{equation}\label{eq:Berry_Esseen_nonlinear}
	       \sup_{t \in \mathbb{R}} \left|
	     \PP{ \frac{\sqrt{n}H(\bar{U})}{\tilde{\sigma}}\le t }  -  \Phi(t)
	     \right|  \le \frac{C}{\sqrt{n}}
	  \end{equation}
	  where the constant $C$ depends on the distribution of $U$ only through $\tilde{\sigma},\nu_{2},\nu_{3},\varsigma_3$ (it also depends on the smoothness of the functional $H$ through $\eps,M_{\eps}$).
	  \end{theorem}
	  Note that the above result is a generalization of the standard Berry--Esseen bound. $\tilde{\sigma}^2$ is the variance term of the asymptotic normal distribution. $\varsigma_3$ is closely related to the term ${\rho}/{\sigma^2}$ in \eqref{eq:1_lem:berry}. The quantities $\tilde{\sigma},\nu_{2},\nu_{3},\varsigma_3$ involved in the constant $C$ only involve up to third moments, which is in accordance with the standard Berry--Esseen bound in Lemmas \ref{lem:berry} and \ref{lem:berry_t}. Note the existence of $\tilde{\sigma},\nu_{2}, \varsigma_{3}$ is implied by $\nu_{3} < \infty$ due to the fact that lower moments can be controlled by higher moments, together with the linearity of the functional $L$. To apply Theorem~\ref{thm:Berry_Esseen_nonlinear} to our problem, we first let $\mathcal{X} = \mathbb{R}^5$ and random vectors $\{U_i\}\nsubp  = \{(U_{i1},U_{i2}, U_{i3}, U_{i4}, U_{i5})\}\nsubp  \iid U_{0} = (U_{01}, U_{02},U_{03}, U_{04}, U_{05})$ to be 
	  	 \begin{eqnarray}
        \label{eq:thm1_V1V2_def}
	     U_{i1}= R_i - \EE{Yh(W)},~~ U_{i2} = V_i - \EE{h^2(W)},\\ \nonumber
	   %  \label{eq:thm1_V3V4_def}
	     U_{i3}= Y_i^2 h^2(W_i)- \EE{Y^2h^2(W)},~~
	     U_{i4}= (\Varc{h(W_i)}{Z_i})^2 - \EE{(\Varc{h(W)}{Z})^2},\\ \nonumber
	     U_{i5}=  R_i\Varc{h(W_i)}{Z_i} - \EE{Yh(W)\Varc{h(W)}{Z}}.
	     \end{eqnarray}
	     Recall the definition $R_{i} = Y_i\big(\mu(X_i,Z_i)-\Ec{\mu(X,Z_i)}{\Xinoj}\big)$ and $V_{i} = \Varc{\mu(X_i,Z_i)}{\Xinoj}$, hence we have $\EE{U_i}=\EE{U_{0}}=\mathbf{0}$. Let $\bar{U} = (\bar{U}_1, \bar{U}_2, \bar{U}_3, \bar{U}_4, \bar{U}_5) = \frac{1}{n}\sum\nsubp U_i \in \mathbb{R}^5$, recall the definition $T =\left(\frac{\bar{R}}{ \sqrt{\bar{V}}} - f(\mu)\right)/s$ where  $s^2 = \frac{ 1 }{ \bar{V} }\left[ \left(\frac{\bar{R} }{2 \bar{V} }\right)^2 \hat{\Sigma}_{22} +  \hat{\Sigma}_{11} - \frac{\bar{R} }{ \bar{V} } \hat{\Sigma}_{12} \right]$, then $T$ can be rewritten as
	     $$
	     T = H(\bar{U}):=\frac{H_1(\bar{U}_1, \bar{U}_2)}{\sqrt{H_2(\bar{U})}},
         $$
         where $H_1(\bar{U}_1, \bar{U}_2)$ and $ H_2(\bar{U})$ are defined as
	   \begin{flalign} \label{eq:H1_def}
% 	  T = H(\bar{U}):=\frac{H_1(\bar{U}_1, \bar{U}_2)}{\sqrt{H_2(\bar{U})}}, \quad \text{where  } 
	  &H_1(\bar{U}_1, \bar{U}_2):=\frac{\bar{U}_1 + \EE{Yh(W)} }{ \sqrt{\bar{U}_2 +  \EE{h^2(W)} }} - \frac{\EE{Yh(W)} }{ \sqrt{\EE{h^2(W)}}}, \\ \nonumber
	  &H_2(\bar{U}):=\frac{1}{ {\bar{U}_2 +  \EE{h^2(W)}}} \Bigg[ 
	  \left( \frac{\bar{U}_1 + \EE{Yh(W)} }{ 2(\bar{U}_2 +  \EE{h^2(W)}) } \right)^2\left(
	  \bar{U}_4 +  \EE{(\Varc{h(W)}{Z})^2} - (  \bar{U}_2 + \EE{h^2(W)} )^2
	  \right)  \\ \nonumber
	  &\quad \quad \quad \quad + \bar{U}_3 +  \EE{Y^2h^2(W)} - (  \bar{U}_1 + \EE{Yh(W)} )^2  \\ 
	  &\quad \quad \quad \quad - \frac{\bar{U}_1 + \EE{Yh(W)} }{\bar{U}_2 +  \EE{h^2(W)} }\left(
	   \bar{U}_5 + \EE{Yh(W)\Varc{h(W)}{Z}} - (\bar{U}_1 + \EE{Yh(W)})(  \bar{U}_2 + \EE{h^2(W)} )
	   \right)
	   \Bigg].  \label{eq:H2_def}
% 	 \left.  \right].
	  \end{flalign}

% 	  where $\bar{U} = (\bar{U}_1, \bar{U}_2) = \frac{1}{n}\sum\nsubp U_i$ and
Note $H(x)= H(x_1,x_2,x_3,x_4,x_5):\mathbb{R}^5 \rightarrow \mathbb{R}$ is defined by replacing the above $\bar{U} = (\bar{U}_1, \bar{U}_2, \bar{U}_3, \bar{U}_4, \bar{U}_5)$ by $x:=(x_1,x_2,x_3,x_4,x_5)$ respectively.
    %   $$
    %   H(x) = H(x_1,x_{2}) :=  \frac{x_{1} +\EE{Yh(W)} }{\sqrt{x_{2} + \EE{h^2(W)}}  }  - \frac{\EE{Yh(W)} }{\sqrt{ \EE{h^2(W)}} }
    %   $$
      When $x_{2} > -  \EE{h^2(W)}$ or $H_2(x) =0 $, $H(x)$ is set to be $0$. If we can verify the conditions for $T = H(\bar{U})$, Theorem \ref{thm:Berry_Esseen_nonlinear} implies
      \begin{equation}\nonumber
            \sup_{t \in \mathbb{R}} \left|
	     \PP{ \sqrt{n}T\le t \tilde{\sigma}}  -  \Phi(t)
	     \right|  \le \frac{C}{\sqrt{n}},
      \end{equation}
      \acc{for some constant $C$, where $\tilde{\sigma} =  \|L(U_0)\|_2> 0$ (we will define $L(x)$ shortly and subsequently show $\tilde{\sigma} =1$). Theorem \ref{thm:Berry_Esseen_nonlinear} says that the constant $C$ above only depends on some universal constants and $\tilde{\sigma},\nu_{2},\nu_{3},\varsigma_3$, which are the moments of $U$ (i.e., the moments of $U_i$). Since $U_i$ (defined in the three lines around \eqref{eq:thm1_V1V2_def}) is a function of $Y_i $ and $h(W_i)$, we will apply the Cauchy--Schwarz inequality to further bound the moments of $U$ by the moments of $Y$ and $h(W) = \mu(X,Z) - \Ec{\mu(X,Z)}{Z}$.} First we need to verify Condition \ref{cond:f_smoothness}, i.e., there exists $\varepsilon, M_{\varepsilon} >0$ and a continuous linear functional $L:\mathbb{R}^5 \rightarrow \mathbb{R}$ such that   
	 \begin{equation}
        \label{eq:my_f_smooth}
	 |H(x) - L(x)| \le   M_{\varepsilon} \|x\|^2 ~~~\text{for all }x \in  \mathbb{R}^5  ~\text{with } \|x\|\le \varepsilon.
	 \end{equation}
    Second, we will show $\tilde{\sigma}$, $\nu_3$, and $\varsigma_3$ are finite under the stated moment conditions.  %\lzmargin{}{say via moments of Y}
    
    Regarding the smoothness condition, consider the first order Taylor expansion of $H$ at zero,
	 $$
	H(\bs{0}) +  \frac{\partial H}{\partial x_1}(\bs{0}) x_1 + \frac{\partial H}{\partial x_2}(\bs{0}) x_2 +  \frac{\partial H}{\partial x_3}(\bs{0}) x_3 + \frac{\partial H}{\partial x_4}(\bs{0}) x_4 +  \frac{\partial H}{\partial x_5}(\bs{0}) x_5.  
% 	 = \frac{1}{\sqrt{\EE{h^2(W)} }}x_1
% 	 -\frac{ \EE{Yh(W)}  }{ 2 (\sqrt{\EE{h^2(W)} })^3}  x_2.
	 $$
	 Note that for $H(\bs{0}) = H_1(\bs{0})/\sqrt{H_2(\bs{0})}$, we have $H_1(\bs{0}) =0$ and $H_2(\bs{0})>0 $ (denote $\tilde{\sigma}_0^2:=H_2(\bs{0}) $ and we will show it is positive over the course of derivations from \eqref{eq:sigma_tilde_onlyU} to \eqref{eq:L2V_bound_v2}. After simplifying the expression of $H_2(\bs{0})$, we give the explicit form of $\tilde{\sigma}_0^2$ below:
	 \begin{equation}
	 \begin{aligned}
	 \label{eq:sigma_tilde_def}
	  \tilde{\sigma}_0^2~=&~  \frac{1}{ \EE{h^2(W)}} \Bigg[ 
	  \left( \frac{ \EE{Yh(W)} }{ 2( \EE{h^2(W)}) } \right)^2
	  \Var{\Varc{h(W)}{\Xnoj}}
	  +  \Var{Yh(W)} \\ 
	  ~~~& -\frac{\EE{Yh(W)} }{  \EE{h^2(W)} } \Cov{Yh(W)}{ \Varc{h(W)}{\Xnoj} }
	   \Bigg].
	 \end{aligned}
	 \end{equation}
	 Using the chain rule of derivatives, we have for $m\in [5]$,
	 $$
	  \frac{\partial H}{\partial x_m}(\bs{0}) =  \frac{\partial H_1}{\partial x_m}(\bs{0})/\sqrt{H_2(\bs{0})} - \frac{H_1(\bs{0})}{2 H_2(\bs{0})^{3/2}  }\cdot \frac{\partial H_2}{\partial x_m}(\bs{0}) = \frac{\partial H_1}{\partial x_m}(\bs{0})/\tilde{\sigma}_0.
	 $$
	 Since $H_1(x_1,x_2)$ only depends on $x_1,x_2$, we only need to evaluate two partial derivatives to compute the first order Taylor expansion of $H$ at zero, \acc{yielding the following linear function
	 \begin{equation} \label{eq:Lfunc_def}
	  \frac{1}{\tilde{\sigma}_0}\left(\frac{1}{\sqrt{\EE{h^2(W)} }}x_1
	 -\frac{ \EE{Yh(W)}  }{ 2 (\sqrt{\EE{h^2(W)} })^3}  x_2\right):=L(x),
	 \end{equation}
	 which is denoted by $L(x) = L(x_1, x_2)$ and satisfies $L(\bs{0})=0$.} Note that when $\eps = {\EE{h^2(W)}}/{2}$, we have
	 $$
	 \min_{\|x\|\le \eps} (x_2 + \EE{h^2(W)}) = \EE{h^2(W)} - \eps >0.
	 $$
	 {Since $H_2(x)$ is continuous around zero and $H_2(\bs{0}) >0$ (which will be shown in the following proof), we can similarly choose $\epsilon$ sufficiently small such that $\min_{\|x\|\le \eps} H_2(x)>0 $. Recall $H(x) = H_1(x)/\sqrt{H_2(x)}$, where $H_1,H_2$ are defined in \eqref{eq:H1_def} and \eqref{eq:H2_def}, so $H(x)$ is continuous on $\{x:\|x\|\le \varepsilon \}$. Furthermore, its second partial derivatives exist and are continuous over the compact set $\{x:\|x\|\le \varepsilon\}$, thus are also bounded, which implies that there exists $M_{\eps} >0$ such that \eqref{eq:my_f_smooth} holds.}

     As for $\tilde{\sigma}$, $\nu_3$, and $\varsigma_3$, we will now establish the following moment bounds: 
	 \acc{
	 \begin{eqnarray}\label{eq:sigmatilde_def}
	    0<\tilde{\sigma} :=  \|L(U_{0})\|_2 <\infty,\qquad \\ \nonumber
	     \nu_2 := \|U_{0}\|_2,\qquad \nu_3 := \|U_{0}\|_{3} < \infty, \\ \label{eq:varsigma_def}
	     \varsigma_3 := \frac{\|L(U_{0})\|_3 }{\tilde{\sigma}} < \infty.\;\;\qquad
	 \end{eqnarray}}
% 	 where for given random vector (variable) $Z=(Z_{1},\cdots, Z_d) \in \mathbb{R}^d$, $\|Z\|_p$ is defined as $\|Z\|_p = (\EE{\|Z\|^p})^{1/p} $ with $\|z\|^{p}:= \sum_{j=1}^{d}|z_{j}|^p$. 
	 Note that $\nu_3^3 = \|U_{0}\|_3^3= \EE{|U_{01}|^3}+ \EE{|U_{02}|^3} + \EE{|U_{03}|^3}+ \EE{|U_{04}|^3} + \EE{|U_{05}|^3} $ and
	 	\begin{align}\nonumber
	\hspace{-.7cm}(\varsigma_3 \tilde{\sigma})^3 = \EE{|L(U_{0})|^3}
	 =~& \frac{1}{\tilde{\sigma}_0^3 }
	 \EE{ \left| \frac{1}{\sqrt{\EE{h^2(W)} }}U_{01}	 -\frac{ \EE{Yh(W)}  }{ 2 (\sqrt{\EE{h^2(W)} })^3}  U_{02} \right|^3}  \\
	 \le ~&  \frac{2^{3-1} }{\tilde{\sigma}_0^3 }
	 \left(
	  \frac{1}{ (\sqrt{\EE{h^2(W)} } ) ^3} \EE{|U_{01}|^3} +
	 \frac{(\EE{Yh(W)} )^3 }{8 (\sqrt{\EE{h^2(W)} } ) ^9}\EE{|U_{02}|^3}
	 \right) \label{eq:L3V_bound}
	 \end{align}
	 \acc{where the equalities hold due to the definitions of $L$ and $\varsigma_3$ in \eqref{eq:Lfunc_def}, \eqref{eq:varsigma_def}, and the inequality holds as a result of the $C_r$ inequality.} Due to the fact that the finiteness of higher moments implies that of lower moments and \eqref{eq:L3V_bound}, we only need to show
	 \begin{enumerate}[(i)]
	 %\label{eq:moments_finite}
	     \item  $\EE{|U_{01}|^3}$, $\EE{|U_{02}|^3}$, $\EE{|U_{03}|^3}$, $\EE{|U_{04}|^3}$, $\EE{|U_{05}|^3} <\infty$,
	    \item $\tilde{\sigma}_0^2 = H_2(\bs{0}) >0$,
	    \item $\tilde{\sigma}^2 =  \|L(U_{0})\|_2 >0 $,
	 \end{enumerate}
% 	 \begin{eqnarray}\label{eq:moments_finite}
% 	  \EE{|V_1|^3}<\infty,~~\EE{|V_2|^3} < \infty~~ \text{and}~~
% 	 \tilde{\sigma}^2:=  \EE{L^2(V)} >0
% 	 \end{eqnarray}
    under the stated moment conditions. For \myrom{3}, actually we will show $\tilde{\sigma}^2=1$.
    
    \acc{Starting with \myrom{1},} we have
	 	\begin{eqnarray}\nonumber
	   \EE{|U_{02}|^3} = \EE{|U_{i2}|^3} 
	   &=&  \EE{\left|V_i - \EE{h^2(W)} \right|^3} \\ \nonumber
	   &\le & 2^{3-1} \left(
	   \EE{\left|\Varc{\mu(W_i)}{\Xinoj}\right|^3} + (  \EE{h^2(W)})^3
	   \right) \\ \nonumber
	    &\le & 2^{3-1} \left(
	   \EE{\Ec{h^6(W_i)}{\Xinoj}} + (  \EE{h^2(W)})^3
	   \right)  < \infty,
	 \end{eqnarray}
	  where the first inequality comes from the $C_r$ inequality, \acc{the second holds by the definition of $h$ and Jensen's inequality}, and the third inequality holds due to the tower property of conditional expectation and $\EE{h^6(W)} <\infty$ under the assumed moment conditions. For the term $\EE{|U_{01}|^3}$, we have
	  \begin{eqnarray} \nonumber
	 \EE{|U_{01}|^3} =  \EE{|U_{i1}|^3} 
	 &=& \EE{\left|R_{i} - \EE{Yh(W)} \right|^3} \\ \nonumber
	 &\le & 2^{3-1} \left(
	 \EE{|Y_i (\mu(W_i)- \Ec{\mu(W_i)}{\Xinoj} )  |^3} + (\EE{Yh(W)} )^3 
	 \right)   \\ \nonumber
	 &= & 2^{3-1} \left(\ee{|Y^3 h^3(W)|}   + (\EE{Yh(W)} )^3  
	 \right) < \infty,
	 \end{eqnarray}
	  \acc{where the first inequality holds due to the $C_r$ inequality and the second inequality holds due to the Cauchy--Schwarz inequality and the assumed moment conditions.} The same approach and inequalities can be used for the other three terms, i.e., we have $\EE{|U_{03}|^3}, \EE{|U_{04}|^3}, \EE{|U_{05}|^3} <\infty$. Note that $U_{03}$, $U_{04}$, and $U_{05}$ involve higher-order polynomials of $Y_ih(W_i)$ and $\Varc{h(W_i)}{Z_i}$ than $ U_{01},U_{02}$, and thus require assuming bounded 12th moments to ensure the boundedness of their third absolute moments, hence the assumptions in Theorem \ref{thm:main_rate} that $\EE{Y^{12}}<\infty$ and $\EE{h^{12}(W)} < \infty$.        
%        According to the definitions, replacing $h(W)$ by the scaled version $h(W)/\sqrt{\EE{h^2(W)}}$ will not change the value of $f(\mu)$, $T$ nor $\tilde{\sigma}^2$, thus we can assume $\EE{h^2(W)}=1$ without loss of generality. 
        
        \acc{
        Regarding \myrom{2} and \myrom{3}: recalling the definitions of $\tilde{\sigma}^2$ and $L$ in \eqref{eq:Lfunc_def}, \eqref{eq:sigmatilde_def}, we have 
%	  \begin{eqnarray}\label{eq:sigma_tilde_def_moments}
%	    \tilde{\sigma}_0^2  \tilde{\sigma}^2  = \tilde{\sigma}_0^2 \|L(U_{0})\|_2= \frac{1}{\EE{h^2(W)} } \EE{\left(
%	     U_{i1} 
%	     -\frac{ \EE{Yh(W)}  }{ 2 \EE{h^2(W)} }U_{i2} 
%	    \right)^2 }.
%	  \end{eqnarray}
%	  The above can be rewritten as
	  \begin{eqnarray} \nonumber
	   \tilde{\sigma}_0^2  \tilde{\sigma}^2  = \tilde{\sigma}_0^2 \|L(U_{0})\|_2
	   &=& \frac{1}{\EE{h^2(W)} } \EE{\left(
	     U_{i1} 
	     -\frac{ \EE{Yh(W)}  }{ 2 \EE{h^2(W)} }U_{i2} 
	    \right)^2 } \\ \label{eq:sigma_tilde_onlyU}
	    &=&  \EE{\left(
	     U_{i1}
	     -\frac{ \EE{Yh(W)}  }{ 2 }  U_{i2}  \right)^2 } \\ \nonumber
	   &=&    \EE{\left(
	    R_{i} - \EE{Yh(W)} 
	     -\frac{ \EE{Yh(W)}  }{ 2 }  ( \Varc{h(W_i)}{\Xinoj} -1)  \right)^2 } \\ \label{eq:sigma_tilde_expansion}
	   &=&   \EE{\left( A
	    + B  \right)^2 },
	  \end{eqnarray}
     where the third equality holds since $\EE{h^2(W)}=1$ as assumed without loss of generality, the fourth one comes from \eqref{eq:thm1_V1V2_def}, and the last one is by rearranging with $A,B$ defined as:
	   \begin{eqnarray} \label{eq:thm1rate_A_Def}
	  A &:=& Y_i h(W_i) - \Ec{Y_i h(W_i)}{\Xinoj},\\  \label{eq:thm1rate_B_Def}
	  B &:=& \Ec{Y_i h(W_i)}{\Xinoj} -  \EE{Yh(W)}  - \frac{ \EE{Yh(W)}  }{ 2 }  ( \Varc{h(W_i)}{\Xinoj} -1).
	   \end{eqnarray}	  
%	  Now we can expand \eqref{eq:sigma_tilde_expansion} as
%	  \begin{eqnarray} \nonumber
%	   \EE{(A+B)^2 }  
%	   &=&  \EE{\Ec{ (A + B)^2}{\Xinoj} }\\ \nonumber
%	   &=&  \EE{ \Ec{ A^2}{\Xinoj} - 2 B~\Ec{A}{\Xinoj} + B^2 }\\ \nonumber
%	   &=&  \EE{ \Ec{ A^2}{\Xinoj} + B^2 } \\ \label{eq:L2V_bound_v1}
%	   &\ge&  \EE{\Varc{Yh(W)}{\Xnoj}},
%	  \end{eqnarray}
%	  where the first equality comes from the tower property of conditional expectation, the second equality holds since $B\in \sigalg(\Xinoj)$ and the third equality holds due to $\Ec{A}{\Xinoj}=0$. \eqref{eq:L2V_bound_v1} gives one lower bound for $\tilde{\sigma}^2$. 
%	  
The above terms $A, B$ have equivalent expressions as the terms $A, B$ defined in the proof of Theorem \ref{thm:main} (see \eqref{eq:thm1_A_Def}, \eqref{eq:thm1_B_Def}). Note $\EE{\left( A + B  \right)^2 } > 0$, as proved over the course of derivations from \eqref{eq:tsigma0_AB} to the end of the proof of Theorem \ref{thm:main}. Due to \eqref{eq:sigma_tilde_expansion}, we then have $\tilde{\sigma}_0^2  \tilde{\sigma}^2 $ in this proof is nonzero, thus finish showing \myrom{2}.  

	Now we will verify $\tilde{\sigma} = 1$. According to  \eqref{eq:sigma_tilde_onlyU}, we equivalently write down
	 	 \begin{eqnarray} \nonumber
        \tilde{\sigma}_0^2  \tilde{\sigma}^2  
	    &=&  \EE{\left(
	     U_{i1}
	     -\frac{ \EE{Yh(W)}  }{ 2 }  U_{i2}  \right)^2 } \\ \nonumber
	    &=& \EE{\left(
	    \left( 1, -\frac{ \EE{Yh(W)}  }{ 2 } \right)(U_{i1},U_{i2})^{\top} \right)^2
	    } \\
	    &=& \bm{a}^{\top}\Sigma_{U} \bm{a}  \label{eq:L2V_bound_v2},
	  \end{eqnarray} %
	  where $\bm{a}^{\top}:= \left(  1, -\frac{ \EE{Yh(W}  }{ 2 } \right)$ and $\Sigma_{U}$ is the covariance matrix for the random vector $U_i$, which can be explicitly written as
	  \begin{equation}\nonumber
	    \Sigma_{U}= \left(
	    \begin{array}{cc}
	   \Var{Yh(W)}      &  \Cov{Yh(W)}{ \Varc{h(W)}{\Xnoj} } \\ 
	   \Cov{Yh(W)}{ \Varc{h(W)}{\Xnoj} }      &  \Var{\Varc{h(W)}{\Xnoj}  }
	    \end{array} 
	    \right). 
	  \end{equation}
	  We immediately have $ \tilde{\sigma}_0^2  \tilde{\sigma}^2 = \bm{a}^{\top}\Sigma_{U} \bm{a} =  H_2(\bs{0}) =   \tilde{\sigma}_0^2 $ due to the expression of $\tilde{\sigma}_0^2$ in \eqref{eq:sigma_tilde_def}, \eqref{eq:L2V_bound_v2} and $\EE{h^2(W)}=1$ as assumed; hence, $ \tilde{\sigma}= 1$.
 
Having verified \myrom{1}, \myrom{2} and \myrom{3}, we thus prove the Berry--Esseen-type bound in \eqref{eq:T_BE_bound},  which completes the proof for case (\myRom{4}). Therefore, the asymptotic coverage validity with a rate of $1/ \sqrt{n}$ for the lower confidence bounds produced by Algorithm \ref{alg:MOCK} has been established.}
     \end{proof}
    \subsubsection{Theorem \ref{thm:main_general_rate}}
     \label{pf:thm:main_general_rate}
    \begin{proof}[Proof of Theorem \ref{thm:main_general_rate}]
	 \acc{Similarly as in the proofs of Theorem~\ref{thm:main} and Theorem~\ref{thm:main_rate}, we immediately have coverage validity when $\mu(X,Z) \in \sigalg(\Xnoj)$. Otherwise, it suffices to show 
	  \begin{equation}\label{eq:thm2_my_cover_rate}
	  \inf_{K>1}\PP{\frac{\bar{R}}{ \sqrt{\bar{V}  }} -  \frac{z_{\alpha}s}{\sqrt{n}}  \le \thetamu} \ge 1 - \alpha - C/\sqrt{n}
	  \end{equation}
	  for some constant $C$, where the sample mean $(\bar{R},\bar{V})$ and sample covariance matrix $\hat{\Sigma}$ are defined the same way as in Algorithm \ref{alg:MOCK} except that $R_i, V_i$ are replaced by their Monte Carlo estimators $ R_i^K, V_i^K$ as defined below}:
	    \begin{align}\nonumber
%	    \label{eq:Rik_ViK_def}
	    \begin{split}
	    R_i^K &= Y_i\left(\mu(X_i,Z_i) - \frac{1}{K}\sum_{k=1}^{K} \mu(X_i^{(k)} ,Z_i)\right),\\
	    V^{K}_{i} &= \frac{1}{K-1}\sum_{k=1}^{K}\left( \mu(X_i^{(k)} ,Z_i)-  \frac{1}{K}\sum_{k=1}^{K} \mu(X_i^{(k)} ,Z_i) \right)^2,
	   	\end{split}
	   	\end{align}
	   \acc{Recall that the proof in Appendix \ref{pf:thm:main} considers $4$ cases then deals with them separately. Essentially we can conduct similar analysis, but to avoid lengthy derivations, we focus on Case \myRom{4}. Note we also make the extra assumption $ \EE{\Varc{Y (\mu(X,Z)-\Ec{\mu(X,Z)}{Z})}{\Xnoj}}>0$ to simplify the proof.

	 In the proof of Theorem \ref{thm:main_general}, we have the following asymptotic normality result:}
	 \begin{equation}\nonumber
	     \sqrt{n}\left(\frac{ \frac{1}{n}\sum\nsubp R_i^K }{\sqrt{\frac{1}{n} \sum\nsubp V_i^K  }} -   \thetamu \right) \stackrel{d}{\rightarrow } \gauss{0}{\tilde{\sigma}_0^2}.
	 \end{equation}
%	 where the variance $\tilde{\sigma}_0^2$ is proved to be positive. 
	 \acc{To establish \eqref{eq:thm2_my_cover_rate}, we follow the proof strategy of Theorem \ref{thm:main_rate}. Specifically, we apply the Berry--Esseen bound for nonlinear statistics (see Theorem \ref{thm:Berry_Esseen_nonlinear} in Appendix \ref{pf:thm:main_rate}).}

	 Again we first introduce some new notations for the following proof: let random vectors $\{U_i\}\nsubp  = \{(U_{i1},U_{i2},U_{i3}, U_{i4}, U_{i5})\}\nsubp \iid U_{0} = (U_{01}, U_{02},U_{03}, U_{04}, U_{05})$ to be 
	  \begin{equation}
        \label{eq:thm2_U1U2_def}
	     U_{i1}= R_i^K - \EE{Yh(W)},~~ U_{i2} = V_i^K - \EE{h^2(W)},
	    \end{equation}
	    \begin{eqnarray}\nonumber
	    U_{i3}= (R_i^K)^2- \EE{ (R_i^K)^2}, ~~
	     U_{i4}= (V_i^K)^2 - \EE{ (V_i^K)^2},~~
	     U_{i5}=  R_i^K  V_i^K- \EE{R_i^K  V_i^K}.
	     \end{eqnarray}
	 \acc{Note by the construction of the null samples, $X_i^{(k)}$ satisfy the two properties in \eqref{eq:nulls_cond_indep} and \eqref{eq:nulls_cond_ident} and we have \eqref{eq:mean_estimator}, \eqref{eq:var_estimator} hold.}
	 \acc{Recall \eqref{eq:mc_esti_yes} in the proof of Theorem \ref{thm:main_general} states $\EE{R_i^K} = \EE{Yh(W)}, \EE{V_i^K} = \EE{h^2(W)}$, hence $ \EE{U_{i1}} =  \EE{U_{i2}} = 0$.} Straightforwardly, $\EE{U_{i3} }= \EE{U_{i4} } = \EE{U_{i5} } = 0$. \acc{Thus we have $\EE{U_{0}} =\bs{0}$.} Now we denote $\bar{U} = (\bar{U}_1, \bar{U}_2, \bar{U}_3, \bar{U}_4, \bar{U}_5) = \frac{1}{n}\sum\nsubp U_i$ and rewrite the following expression,
	  $$
\frac{1}{s} \left(\frac{ \frac{1}{n}\sum\nsubp R_i^K }{\sqrt{\frac{1}{n} \sum\nsubp V_i^K  }} -   \thetamu \right):= H(\bar{U}):=\frac{H_1(\bar{U}_1, \bar{U}_2)}{\sqrt{H_2(\bar{U})}},
	  $$
	  where $s$ is similarly defined as in Algorithm \ref{alg:MOCK} except that $R_i,V_i$ are replaced by $R_i^K, V_i^K$. Here $H(x)= H(x_1,x_2,x_3,x_4,x_5):\mathbb{R}^5 \rightarrow \mathbb{R}$ is the same as in the proof of Theorem \ref{thm:main_rate}. Therefore the smoothness condition, i.e., Condition \eqref{cond:f_smoothness}, holds by the same argument as in Appendix \ref{pf:thm:main_rate}. 
	 The continuous linear functional $L$ is also defined the same way. To apply Theorem \ref{thm:Berry_Esseen_nonlinear}, it remains to verify the following moment bound conditions on $U_0$ and $L(U_0)$,
	 \begin{eqnarray}\nonumber
	     0<\tilde{\sigma} :=   \|L(U_{0})\|_2 <\infty, \\ \nonumber
	     \nu_2 := \|U_{0}\|_2,~ \nu_3 := \|U_{0}\|_{3} < \infty, \\ \nonumber
	     \varsigma_3 := \frac{\|L(U_{0})\|_3 }{\tilde{\sigma}} < \infty.
	 \end{eqnarray}
% 	 where for given random vector (variable) $Z=(Z_{1},\cdots, Z_d) \in \mathbb{R}^d$, $\|Z\|_p$ is defined as $\|Z\|_p = (\EE{\|Z\|^p})^{1/p} $ with $\|z\|^{p}:= \sum_{j=1}^{d}|z_{j}|^p$. 
	 Note that $\nu_3^3 = \|U_{0}\|_3^3= \EE{|U_{01}|^3}+ \EE{|U_{02}|^3} + \EE{|U_{03}|^3}+ \EE{|U_{04}|^3} + \EE{|U_{05}|^3} $ and \acc{we can bound $(\varsigma_3 \tilde{\sigma})^3$ similarly as in the proof of Theorem \ref{thm:main_rate}:}
	\begin{eqnarray}\nonumber
	\hspace{-1cm}(\varsigma_3 \tilde{\sigma})^3 &=& \EE{|L(U_{0})|^3}
	 =
	 \EE{ \left| \frac{1}{\sqrt{\EE{h^2(W)} }}U_{01}	 -\frac{ \EE{Yh(W)}  }{ 2 (\sqrt{\EE{h^2(W)} })^3}  U_{02}   \right|^3}  \\
	 &\le& 
	 2^{3-1} \left(
	 A\frac{1}{ (\sqrt{\EE{h^2(W)} } ) ^3} \EE{|U_{01}|^3} + \frac{(\EE{Yh(W)} )^3 }{8 (\sqrt{\EE{h^2(W)} } ) ^9}\EE{|U_{02}|^3}
	 \right), \label{eq:LV_bound}
	 \end{eqnarray}
	Due to the fact that the finiteness of higher moments implies that of lower moments and \eqref{eq:LV_bound}, we only need to show
	 \begin{enumerate}[(i)]
	   	 \item $\EE{|U_{01}|^3}$, $\EE{|U_{02}|^3}$, $\EE{|U_{03}|^3}$, $\EE{|U_{04}|^3}$, $\EE{|U_{05}|^3} <\infty$
	    \item $\tilde{\sigma}_0^2 = H_2(\bs{0}) >0$
	    \item $\tilde{\sigma}^2 =  \|L(U_{0})\|_2 >0 $
	 \end{enumerate}
	 under the stated moment conditions. \acc{For \myrom{3}, we have $\tilde{\sigma}^2=1$, due to the derivations in the proof of Theorem \ref{thm:main_rate}. Hence we will focus on the first two conditions in the following. Appendix \ref{pf:thm:main_rate} verifies \myrom{1} and \myrom{2} for any given $K > 1$. In this proof, we will actually show}
	 $$
	 \sup_{K>1}\EE{|U_{0j}|^3} < \infty, ~~\forall j \in [5],\quad \inf_{K>1} \tilde{\sigma}_0^2>0.
	$$
	Note the definitions of $U_{0} = (U_{01}, U_{02},U_{03}, U_{04}, U_{05})$ and $\tilde{\sigma}_0^2$ depend on $K$. To simplify notations, we do not make this dependence explicit. \acc{By the definitions in \eqref{eq:thm2_U1U2_def}, we bound $U_{01}, U_{02}$  as below:
	 \begin{eqnarray} \nonumber
	 \EE{|U_{01}|^3} =  \EE{|U_{i1}|^3} 
	 &=& \EE{|R_i^K  - \EE{Yh(W)} |^3} \\ \nonumber
	 &\le & 2^{3-1} \left(
	 \EE{|R_i^K |^3} + (\EE{Yh(W)} )^3 
	 \right),   \\ \nonumber
	   \EE{|U_{02}|^3} = \EE{|U_{i2}|^3} 
	   &=&  \EE{|V_i^K - \EE{h^2(W)} |^3} \\ \nonumber
	   &\le & 2^{3-1} \left(
	   \EE{|V_i^K|^3} + (  \EE{h^2(W)})^3
	   \right),
	 \end{eqnarray}
	 where the inequalities hold due to the $C_r$ inequality. Recalling in the proof of Theorem \ref{thm:main_general}, we show $\EE{|R_i^K|^2} < \infty,  \EE{|V_i^K|^2} < \infty$ under the condition $\EE{Y^4}, \EE{h^4(W)} < \infty$ over the course of derivations from \eqref{eq:clt_mmc} to the end of that proof. The derivations are mainly based on the $C_r$ inequality and the Bahr--Esseen inequality in \citet{dharmadhikari1969bounds}. Using the same bounding strategy, we can show $ \EE{|R_i^K |^3}, \EE{|V_i^K|^3} < \infty$ when assuming $\EE{Y^6}, \EE{h^6(W)} < \infty$. Hence we obtain $\sup_{K>1}\EE{|U_{01}|^3}, \sup_{K>1}\EE{|U_{02}|^3} < \infty$ under the above moment conditions.} 
	  \acc{And nearly identical derivations as in bounding $\EE{|U_{01}|^3}$ and $\EE{|U_{02}|^3}$} suffice to show $\sup_{K>1}\EE{|U_{03}|^3}$, $\sup_{K>1}\EE{|U_{04}|^3}$, $ \sup_{K>1}\EE{|U_{05}|^3} <\infty$ under the stronger moment boundedness conditions $\EE{Y^{12}}<\infty, \EE{h^{12}(W)} < \infty$ stated in Theorem \ref{thm:main_general}. 
	  
	 \acc{Regarding \myrom{2}, we notice that
	 \begin{eqnarray}
	 	\tilde{\sigma}_0^2 \ge \EE{(A+B)^2} \ge  \EE{\Varc{Yh(W)}{\Xnoj}}, 
	 \end{eqnarray} 
	 where the first inequality holds due to \eqref{eq:tsigma_byAB}, $A,B$ are defined as \eqref{eq:thm1_A_Def} and \eqref{eq:thm1_B_Def} in the proof of Theorem \ref{thm:main}, and the second inequality holds by \eqref{eq:L2V_bound_v1}. The above lower bound for $\tilde{\sigma}_0^2$ does not depend on $K$ and implies the positiveness of $\inf_{K>1}\tilde{\sigma}_0$ under the assumed condition $ \EE{\Varc{Yh(W)}{\Xnoj}} = \EE{\Varc{Y (\mu(X,Z)-\Ec{\mu(X,Z)}{Z})}{\Xnoj}} >0$. } 
	 
	 Therefore, we obtain the Berry--Esseen bound for nonlinear statistics by applying Theorem \ref{thm:Berry_Esseen_nonlinear}. Finally we conclude the asymptotic coverage with a rate of $n^{-1/2}$, i.e.,
	\begin{equation} 
	\nonumber
	\inf_{K>1} 
	\PP{L_{n,K}^{\alpha} (\mu) \le \Ij }  \ge 1 - \alpha - C n^{-1/2},
	\end{equation}
	\acc{where the constant $C$ only depends on the moments of $Y$ and $h(X, Z)= \mu(X,Z) - \Ec{\mu(X, Z)}{Z}$.}
	 \end{proof}  
	 
    \section{Applicability of the Model-X assumption}
    \label{sec:modelX}
    \acc{
     Model-X floodgate assumes knowing the distribution of $P_{X|Z}$. This may not always hold in practice, but in some important instances, $P_{X|Z}$ may be (A) known due to experimental randomization, (B) well-modeled a priori due to domain expertise, or (C) accurately estimated from a large unlabeled data set. For example, (A) holds in the high-dimensional experiments of conjoint analysis \citep{luce1964simultaneous,hainmueller2014public}, (B) holds in the study of the microbiome where accurate covariate simulators exist \citep{ren2016sparsedossa}, and a combination of (B) and (C) hold in genomics, where the model-X framework has been repeatedly and successfully applied for controlled variable selection \citep{sesia2019gene,katsevich2019multilayer,sesia2020multi,bates2020causal,sesia2020controlling}.  
    
    We also quantify the robustness of our inferences to this assumption in Appendix~\ref{sec:robustness} and show it can be relaxed to parametric models (Section \ref{sec:relax}), and indeed model-X approaches have shown promising empirical performance in a number of applications in which it is unclear whether any of (A), (B), or (C) hold, such as bacterial classification from spectroscopic data \citep{chia2020interpretable} and single cell regulatory screening \citep{katsevich2020conditional}. 
    }

 \section{Robustness}
\label{sec:robustness}
\acc{To explain how the floodgate idea is not tied to the model-X assumption, a double-robustness type result (Lemma \ref{lem:doublerobust}) is presented in Remark \ref{rk:modelx}. It involves an approximated floodgate functional \eqref{eq:approxfg} and says that the inferential statements are valid as long as either of the models of $X\mid Z$ or $Y\mid Z$ is correctly specified. For ease of exposition, Algorithm \ref{alg:MOCK} and Theorem \ref{thm:main} focus on a particular floodgate procedure which requires knowing $P_{X|Z}$. However, it is still of interest to study the robustness of floodgate (in Algorithm \ref{alg:MOCK}) to misspecification of $P_{X|Z}$. Specifically, we consider the case when the true distribution $P_{X|Z}$ used in floodgate is replaced by an approximation $Q_{X|Z}$. 
%We will find that the inference is also valid when the approximation of $P_{X|Z}$ is better than that of $\Ec{Y}{X,Z}$.
}
%of either $\Ec{Y}{Z}$ and $\Ec{\mu(X, Z)}{Z}$ (for given $\mu$)

%    Recall in Remark \ref{rk:modelx}, we mentioned the floodgate idea is not tied to the model-X assumption: double-robustness type results (Lemma \ref{lem:doublerobust}}) show that the floodgate method is valid when the estimation of $\Ec{Y}{Z}$ and $\Ec{\mu(X, Z)}{Z}$ (for given $\mu$) are accurate combinedly, however, is currently assumed in Algorithm \ref{thm:main} and Theorem \ref{thm:main} for ease of exposition. Therefore, it is still of interest to study the robustness of floodgate to misspecification of $P_{X|Z}$ 
%We now consider what happens when the distribution used in floodgate is not the true $P_{X|Z}$ but an approximation $Q_{X|Z}$. 

Notationally, let $Q = P_{Y|X,Z}\times Q_{X|Z}\times P_{Z}$ (we need not consider misspecification in the distributions of $Z$ or $Y\mid X,Z$ since these are not inputs to floodgate), and let $f^Q$ be an analogue of $f$ with certain expectations replaced by expectations over $Q$ (we will denote such expectations by $\Ep{Q}{\cdot}$); see Equation~\eqref{eq:thetaj_Q} for a formal definition. It is not hard to see that floodgate with input $Q_{X|Z}$ produces an asymptotically-valid LCB for $f^Q(\mu)$, from which we immediately draw the following conclusions. 

First, if $\mu$ does not actually depend on $X$, i.e., $\Varpc{Q}{\mu(X,Z)}{Z}\stackrel{a.s.}{=}0$, then $f^Q(\mu)=0$ regardless of $Q$ and floodgate is trivially asymptotically-valid. Second, when $\mu$ does depend on $X$, floodgate's inference will still be approximately valid as long as $f^Q(\mu)-f(\mu)\approx 0$, and this difference can be bounded by, for instance, the $\chi^2$ divergence between $P_{X|Z}$ and $Q_{X|Z}$. The third, and perhaps most interesting, conclusion is that the gap between $\Ij$ and $f(\mu)$ grants floodgate an \emph{extra} layer of robustness as long as $\Ij-f(\mu)$ is large compared to $f^Q(\mu)-f(\mu)$. Thus even if $Q_{X|Z}$ is a bad approximation of $P_{X|Z}$, floodgate's inference may be saved if $f(\mu)$ is an \emph{even worse} approximation of $\Ij$, and this latter approximation is related to that of $\mu$ for $\mustar$. To make this last relation precise, we quantify $\mu$'s approximation of $\mustar$ by focusing on a particular representative of $S_\mu$: 
for any $\mu:\mathbb{R}^p\rightarrow\mathbb{R}$, 
\begin{equation}\label{eq:M_v}
\bar\mu(x,z) = \sqrt{\frac{\EE{\varc{\mustar(X,Z)}{Z}} }{\EE{\varc{\mu(X,Z)}{Z}}}}\Big(\mu(x,z)-\Ec{\mu(X,Z)}{Z=z}\Big)+\Ec{\mustar(X,Z)}{Z=z},
\end{equation}
where $0/0=0$. We can think of $\bar\mu$ as a generally accurate representative from $S_\mu$, in that it takes $\mu$ and corrects its conditional mean and expected conditional variance to match $\mustar$. Note that $\bar\mu=\mustar$ whenever $\mustar\in S_\mu$, which includes anytime $\Ij=0$. \acc{Since the LCB from floodgate with input $Q_{X|Z}$ is asymptotically-valid for $f^Q(\mu)$ under certain moment conditions and the proof can be done similarly as Theorem \ref{thm:main}, we will focus on quantifying the difference between $f^Q(\mu)$ and $\Ij$ in the following robustness result.
    \begin{theorem}[Floodgate robustness]\label{thm:robust_no_adj}
    For data $\{(Y_i,X_i,Z_i)\}\nsubp$ i.i.d. draws from $P$ satisfying $\ee{Y^{4}}<\infty$, a sequence of working regression functions $\mu_n:\mathbb{R}^p\rightarrow\mathbb{R}$ such that for some $C$ and all $n$ either $\Varpc{Q^{(n)}}{\mu_n(X,Z)}{Z}$ $\stackrel{a.s.}{=}0$ or $\frac{\max\left\{\EE{\mu_n^{4}(X,Z)},\, \Ep{Q^{(n)}}{\mu_n^{4}(X,Z)}\right\}}{\mathbb{E}\Big[\Varpc{Q^{(n)}}{\mu_n(X,Z)}{Z}\Big]^{2}} \le C$, and a sequence of conditional distributions $Q_{X|Z}^{(n)}$,
    %\newline \vspace{-.4cm}\noindent 
%    the output of Algorithm~\ref{alg:MOCK} when $Q_{X|Z}^{(n)}$ is used as input satisfies
%    \begin{equation}
    %\label{eq:robust_no_adj}
%        \mathbb{P}\left(L^{\alpha}_{n}(\mu_n) \le \Ij +  \Delta_n\right)\ge 1-\alpha - O(n^{-1/2}),
%    \end{equation}
%    where
	the difference between $ f^{Q^{(n)}}(\mu)$ and $\Ij$ can be controlled as 
    \begin{equation}\label{eq:Delta_characterize}
    \Delta_n = f^{Q^{(n)}}(\mu_n)-\Ij \le c_1\sqrt{\EE{\chi^2\left(P_{X|Z}\,\|\, Q^{(n)}_{X|Z}\right)}}-c_2\,\EE{(  \bar{\mu}_n(X,Z)-\mustar(X,Z))^2}
    \end{equation}
    for some positive $c_1$ and $c_2$ that depend on $P$, where $\chi^2(\cdot\,\|\,\cdot)$ denotes the $\chi^2$ divergence.
    \end{theorem}
    The proof of Theorem~\ref{thm:robust_no_adj} can be found in Appendix~\ref{pf:thm:robust_no_adj}. Equation~\eqref{eq:Delta_characterize} formalizes that larger MSE of $\bar\mu_n$ actually \emph{improves} robustness, although we remind the reader once again that when $\Ij=0$, the MSE of $\bar\mu_n$ is always zero by construction in Equation~\eqref{eq:M_v}. Given the $n^{-1/2}$-rate half-width lower-bound for floodgate, a sufficient condition for asymptotically-exact coverage is
    \begin{equation}\label{eq:robustcondition}
    \sqrt{\EE{\chi^2\left(P_{X|Z}\,\|\, Q^{(n)}_{X|Z}\right)}} = o\left(n^{-1/2}+\EE{(  \bar{\mu}_n(X,Z)-\mustar(X,Z))^2}\right).
    \end{equation}
    When $Q^{(n)}_{X|Z}$ is a standard parametric estimator based on $N_n$ independent samples, the left-hand side has a $O(N_n^{-1/2})$ rate. Thus if $N_n\gg \min\{n,\EE{(  \bar{\mu}_n(X,Z)-\mustar(X,Z))^2}^{-2}\}$,
    %$n=o(N_n)$ or $n=O(N_n)$ and $\bar\mu_n$'s MSE shrinks at any rate slower than $n^{-1/2}$, 
    then floodgate's coverage will be asymptotically-exact. For certain parametric models for $X\mid Z$, Section~\ref{sec:relax} shows how to modify floodgate to attain asymptotically-exact inference without the need for estimation at all. }
%    We also note in passing that a weaker form of condition \eqref{eq:robustcondition} that replaces the $n^{-1/2}$ with $1$ is sufficient for a weaker guarantee of \emph{asymptotic non-overestimation}, i.e., the property that $\liminf_{n\rightarrow\infty}\mathbb{P}\left(L^{\alpha}_{n}(\mu_n) \le \Ij +  \epsilon\right)\ge 1-\alpha$ for any $\epsilon>0$.
   
    Theorem~\ref{thm:robust_no_adj} treats the sequence $Q^{(n)}_{X|Z}$ as fixed, which of course means $Q^{(n)}_{X|Z}$ can be estimated from any data that is independent of the data floodgate is applied to. This means the same data can be used to estimate $\mu_n$ and $Q^{(n)}_{X|Z}$. For $Q^{(n)}_{X|Z}$ however, this strict separation may not be necessary in practice, and in our simulations we found floodgate to be quite robust to estimating $Q^{(n)}_{X|Z}$ on samples that included those used as input to floodgate; see Section~\ref{sec:simul_robust}.
    
    Another layer of robustness beyond that addressed in this section can be injected by replacing $P_{X|Z}$ in floodgate with $P_{X|Z,T}$ for some random variable $T$. For instance, floodgate's model-X assumption can be formally relaxed to only needing to know a fixed-dimensional model for $P_{X|Z}$ by conditioning on $T$ that is a sufficient statistic for that model; see Section~\ref{sec:relax} for details. More generally, conditioning on $T$ that is a function of $\{(X,Z)\}_{i=1}^n$ may induce some degree of robustness, as conditioning on the order statistics of the $X_i$ can in conditional independence testing \citep{berrett2019permutation}. 
    
     \subsection{Proofs in Appendix~\ref{sec:robustness}}
%      In the true distribution case, the output $L_{\alpha}^n(\mu)$ from Algorithm \ref{alg:MOCK} is a asymptotic lower confidence bound for $\thetamu$. 
      In the case where the conditional distribution of $X$ given $Z$ is specified as $Q_{X\mid Z}$ (in the following, we often denote the true conditional distribution by $P:=P_{X\mid Z}$ and the specified conditional distribution by $Q:=Q_{X \mid Z}$ without causing confusion), the floodgate functional with input $Q_{X \mid Z}$ is denoted by $f^{Q}(\mu)$. Note that $\thetamu$ can be rewritten with explicit subscripts as below (here we use the equivalent expression of $\thetamu$ in \eqref{eq:equiv_fmu} and expand $h(W)$).
    \begin{equation}
        \thetamu= \frac{\Ep{P}{Y\left(\muX -\Epc{P}{\muX}{\Xnoj}\right)}}{\sqrt{\Ep{P_{Z}}{\Varpc{P}{\muX}{\Xnoj}}}}       
    \end{equation}
    Therefore, $f^{Q}(\mu)$ admits the following expression:
%    Clearly, $L_{\alpha}^{n,Q}(\mu)$ is a lower confidence bound for the following quantity:
    \begin{equation}\label{eq:thetaj_Q}
        \thetaQmu:= \frac{\Ep{P}{Y\left(\muX -\Epc{Q}{\muX}{\Xnoj}\right)}}{\sqrt{\Ep{P_{Z}}{\Varpc{Q}{\muX}{\Xnoj}}}}.
    \end{equation}    
    % Note the denominator is estimated only using the null samples from $Q_{X\mid Z}$, instead of using both the real samples (i.e. generated from $P_{X\mid Z}$) and the null samples. 
    Denote $\omega(x,z):=\frac{d P_{X|\Xnoj}(x|\xnoj)}{d Q_{X|\Xnoj}(x|\xnoj)}$. Note that $\omega(x,z)$ is the ratio of conditional densities if we are in the continuous case;  $\omega(x,z)$ is the ratio of conditional probability mass function in discrete case. %We also enforce the support of $Q$ must contain the support of $P$.
    Then we can quantify the difference between $\thetamu$ and $\thetaQmu$ as in Lemma \ref{lem:robust_gap}.
    \begin{lemma} \label{lem:robust_gap}
    Assuming $ \EE{Y^4} < \infty $, consider two joint distributions $P,Q$ over $(X,Z)$, defined as $P(x,z)=P_{X\mid Z}(\xj|\xnoj)P_{Z}(\xnoj),Q(x,z)=Q_{X\mid Z}(\xj|\xnoj)P_{Z}(\xnoj)$. If we denote $\calU$ to be the class of functions $\mu:\mathbb{R}^{p} \rightarrow \mathbb{R}$ satisfying one of the following conditions:
    \begin{itemize}
        \item $\mu(X,Z) \in \sigalg(Z)$;
        \item $\max\{\Ep{P}{ {\mu}^4(X,Z)}, \Ep{Q}{ {\mu}^4(X,Z)} \}/(\Ep{P_{Z}}{\Varpc{Q}{\mu(X,Z)}{Z}})^{2} \le c_0$.
    \end{itemize}
    % \begin{equation}
    %   \Ep{P}{Y^2} \le C_{0},~ \Ep{P}{\mu^4(X)},\Ep{Q}{\mu^4(X)} \le C_1
    % \end{equation}
    for some constants $c_0$, then we have the following bounds
    \begin{equation}\label{eq:robust_gap}
        \Delta(P,Q) := \sup_{\mu \in \calU} |{\theta}^{Q}(\mu)- \thetamu| \le C \sqrt{\Ep{P_{Z}}{\chi^2\left(P_{X\mid Z}\|Q_{X\mid Z}\right)}}
    \end{equation}
    for some constant $C$ only depending on $\EE{Y^4}$ and $c_0$, where the $\chi^2$ divergence between two distributions $P,Q$ on the probability space $\Omega$ is defined as $\chi^2\left(P\|Q\right):=\int_{\Omega}(\frac{d P}{d Q}-1)^2 dQ$.
    \end{lemma} 
    When the $X\mid Z$ model is misspecified, the inferential validity will not hold in general, without adjustment on the lower confidence bound. Lemma \ref{lem:robust_gap} gives a quantitative characterization about how much we need to adjust.
        \begin{proof}[Proof of Lemma \ref{lem:robust_gap}]\label{pf:lem:robust_gap}
        % support of Q must contain the support of P
        When the support of $Q$ does not contain the support of $P$, the $\chi^2$ divergence between $P$ and $Q$ is infinite, which immediately proves \eqref{eq:robust_gap}. From now, we work with the case where the support of $Q$ contains the support of $P$.
        % We also enforce the support of $Q$ must contain the support of $P$.
        When $\mu(X,Z) \in \sigalg(Z)$, $\thetamu = \thetaQmu =0$, thus the statement holds. Now we deal with the nontrivial case where $\Ep{P_{Z}}{\Varpc{Q}{\mu(X,Z)}{\Xnoj}}>0 $. Without loss of generality, we assume $\Ep{P_{Z}}{\Varpc{Q}{\mu(X,Z)}{\Xnoj}}=1$ for the following proof (since floodgate is invariate to positive scaling of $\mu$). Then the stated moment conditions on $\mu$ imply 
        \begin{equation}\label{eq:mmt_bound_robust_gap}
        \Ep{P}{ {\mu}^4(X,Z)}, \Ep{Q}{ {\mu}^4(X,Z)} \le c_0.
        \end{equation}
        First we simplify $\thetamu$ and $\thetaQmu$ into
        \begin{eqnarray} \nonumber
             \thetamu &=& \frac{\Ep{P}{\mustar(X,Z)\left(\muX -\Epc{P_{X\mid Z}}{\muX}{\Xnoj}\right)}}{\sqrt{\Ep{P_{Z}}{\Varpc{P_{X\mid Z}}{\muX}{\Xnoj}}}} = \frac{\Ep{P}{\mustar(W)\left(\mu(W) -\Epc{P}{\mu(W)}{\Xnoj}\right)}}{\sqrt{\Ep{P_{Z}}{\Varpc{P}{\mu( W)}{\Xnoj}}}}
            \\ \nonumber
             \thetaQmu &=& \frac{\Ep{P}{\mustar(X,Z)\left(\muX -\Epc{Q_{X\mid Z}}{\muX}{\Xnoj}\right)}}{\sqrt{\Ep{P_{Z}}{\Varpc{Q_{X\mid Z}}{\muX}{\Xnoj}}}} =
             \frac{\Ep{P}{\mustar(W)\left(\mu(W) -\Epc{Q}{\mu(W)}{\Xnoj}\right)}}{\sqrt{\Ep{P_{Z}}{\Varpc{Q}{\mu(W)}{\Xnoj}}}}
        \end{eqnarray}
    due to \eqref{eq:5_lem:max}, where we denote $W=(X,Z)$ (thus $w=(x,z)$). Noticing the following facts
    \begin{equation}\nonumber
    \left|
    \frac{a}{\sqrt{b}} - \frac{c}{\sqrt{d}}
    \right| 
     =  \left| \frac{a\sqrt{d} - c\sqrt{b}}{\sqrt{bd}}\right|\le \frac{a}{\sqrt{bd}} \left|\sqrt{b} - \sqrt{d}\right| + \frac{1}{\sqrt{d}}\left|a-c\right| \le \frac{a}{\sqrt{b}}\cdot \frac{1}{d}\left|{b} - {d}\right| + \frac{1}{\sqrt{d}}\left|a-c\right|,
    \end{equation}
     we let $a,c$ to be the numerators of $\thetamu$ and $\thetaQmu$ respectively and $\sqrt{b},\sqrt{d}$ to be their denominators. Before dealing with $|b-d|$ and $|c-d|$, we have the following bounds on the terms $a/\sqrt{b}$ and $1/d$.
    \begin{equation}
        a/\sqrt{b} = \thetamu \le \Ij  \le (\Ep{P}{Y^4})^{1/4}, ~~ 1/d = 1/\Ep{P_{Z}}{\Varpc{Q}{\muX}{\Xnoj}} = 1, 
    \end{equation}
    where the first equality is due to Lemma \ref{lem:max} and the second one is by applying Jensen's inequality ($\Ep{P_{Z}}{\Varpc{P}{\condmean}{\Xnoj}}\le \Ep{P_{Z}}{\Epc{P}{(\condmean)^2}{\Xnoj}}\le \EE{Y^2} \le\sqrt{\EE{Y^4}} $). The equality holds by assumption. Now it suffices to consider bounding $|b-d|$ and $|c-d|$ in terms of the expected $\chi^2$ divergence between $P_{X\mid Z}$ and $Q_{X\mid Z}$. We have the following equations for $|a-c|$:
    \begin{eqnarray}\nonumber 
    \left| a-c \right| 
    & = & \left| 
    \Ep{P}{\mustar(W)\left(\mu(W) -\Epc{P}{\mu(W)}{\Xnoj}\right)} -
    \Ep{P}{\mustar(W)\left(\mu(W)-\Epc{Q}{\mu(W)}{\Xnoj}\right)}    
    \right|\\ \nonumber
    & = & \left| 
    \Ep{P}{\mustar(W)\left(\Epc{P}{\mu(W)}{\Xnoj} -\Epc{Q}{\mu(W)}{\Xnoj} \right) } 
    \right|\\ 
    & = & \left| 
    \Ep{P_{Z}}{\Epc{P}{\mustar(W)}{\Xnoj} \left(\Epc{P}{\mu(W)}{\Xnoj} -\Epc{Q}{\mu(W)}{\Xnoj} \right)}
    \right|. \label{eq:1_lem:robust_gap}
    \end{eqnarray}
    Now we rewrite $|\Epc{P}{\mu(W)}{\Xnoj} -\Epc{Q}{\mu(W)}{\Xnoj}|$ in the form of integral then bound it as
    \begin{eqnarray}\nonumber
    |\Epc{P}{\mu(W)}{\Xnoj} -\Epc{Q}{\mu(W)}{\Xnoj}|
     &=&\left|\int \mu(x,Z)(1- \omega(x,Z))d Q_{X\mid Z}(x\mid Z)\right| \\ \nonumber
     &\le& \sqrt{\Epc{Q_{X\mid Z}}{\mu^2(X,Z)}{\Xnoj}} \sqrt{\int (1-w(x,Z))^2 d Q_{X\mid Z}(x\mid Z)}\\
     &=&  \sqrt{\Epc{Q_{X\mid Z}}{\mu^2(W)}{\Xnoj}} \sqrt{\chi^2\left(P_{X\mid Z}\|Q_{X\mid Z}\right)},
     \label{eq:2_lem:robust_gap}
    \end{eqnarray}
     where $\omega(x,Z) = \frac{d P_{X\mid Z}(x\mid Z)}{d Q_{X\mid Z}(x\mid Z)}$ and the above inequality is from the Cauchy--Schwarz inequality. Hence we can plug \eqref{eq:2_lem:robust_gap} into \eqref{eq:1_lem:robust_gap} and further bound $|a-c|$ by
    \begin{align}
         \left| a-c \right|  \nonumber
         \le ~&  \Ep{P_{Z}}{\Epc{P_{X\mid Z}}{\mustar(W)}{\Xnoj} 
         \sqrt{\Epc{Q_{X\mid Z}}{\mu^2(W)}{\Xnoj}} \sqrt{\chi^2\left(P_{X\mid Z}\|Q_{X\mid Z}\right)}
         } \\ 
         \le ~& \sqrt{\Ep{P_{Z}}{
         (\Epc{P_{X\mid Z}}{\mustar(W)}{\Xnoj} )^2 \Epc{Q_{X\mid Z}}{\mu^2(W)}{\Xnoj}
         }} \cdot
         \sqrt{\Ep{P_{Z}}{\chi^2\left(P_{X\mid Z}\|Q_{X\mid Z}\right)}}.
         \label{eq:3_lem:robust_gap}
     \end{align}
    For the first part of the product in \eqref{eq:3_lem:robust_gap}, we can apply the Cauchy--Schwarz inequality and Jensen's inequality and bound it by $(\Ep{P}{(\mustar)^4(W)}\Ep{Q}{\mu^4(W)})^{1/4}$, which is upper bounded by some constant under the stated condition $\EE{Y^4} < \infty$ and $\Ep{Q}{\mu^4(X,Z)} \le c_0$ (from \eqref{eq:mmt_bound_robust_gap}). Regarding $|b-d|$, we have
    \begin{eqnarray}
        |b-d| \nonumber
        & = & 
        \left| 
        \Ep{P_{Z}}{\Varpc{P}{\mu(W)}{\Xnoj}} - \Ep{P_{Z}}{\Varpc{Q}{\muX}{\Xnoj}}
        \right| \\ \nonumber
        &\le &  
        \left| 
        \Ep{P_{Z}}{
        (\Epc{P}{\mu(W)}{\Xnoj})^2 - (\Epc{Q}{\mu(W)}{\Xnoj})^2}
        \right| \\
        &~~& +
        \left| 
        \Ep{P_{Z}}{
        \Epc{P}{\mu^2(W)}{\Xnoj} - \Epc{Q}{\mu^2(W)}{\Xnoj} }
        \right|. 
        \label{eq:4_lem:robust_gap}
    \end{eqnarray}
    Similarly as \eqref{eq:2_lem:robust_gap}, we obtain
    \begin{equation} \nonumber
      \left|  
      \Epc{P}{\mu^2(W)}{\Xnoj} - \Epc{Q}{\mu^2(W)}{\Xnoj} 
      \right| \le \sqrt{\Epc{Q_{X\mid Z}}{\mu^4(W)}{\Xnoj}} \sqrt{\chi^2\left(P_{X\mid Z}\|Q_{X\mid Z}\right)}.
    \end{equation}
    Then under the moment bounds $\Ep{Q}{\mu^4(X,Z)} \le c_0$ in \eqref{eq:mmt_bound_robust_gap}, we show the second term in \eqref{eq:4_lem:robust_gap} is upper bounded by $ \sqrt{c_{0} \Ep{P_{Z}}{\chi^2\left(P_{X\mid Z}\|Q_{X\mid Z}\right)} }$. Regarding the first term in \eqref{eq:4_lem:robust_gap}, we can write 
    \begin{eqnarray}
        \nonumber
         (\Epc{P}{\mu(W)}{\Xnoj})^2 - (\Epc{Q}{\mu
         (W)}{\Xnoj})^2 = 
         \left(\Epc{P}{\mu(W)}{\Xnoj} - \Epc{Q}{\mu(W)}{\Xnoj} \right)
         \left(\Epc{P}{\mu(W)}{\Xnoj} + \Epc{Q}{\mu(W)}{\Xnoj}\right)
    \end{eqnarray}
    then apply similar strategies in deriving \eqref{eq:1_lem:robust_gap} and \eqref{eq:3_lem:robust_gap} to control the above term under $C\sqrt{\Ep{P_{Z}}{\chi^2\left(P_{X\mid Z}\|Q_{X\mid Z}\right)}}$ for some constant $C$. And this will make use of the moment bound conditions $\Ep{P}{\mu^4(X,Z)}$,$\Ep{Q}{\mu^4(X,Z)}\le c_0$ in \eqref{eq:mmt_bound_robust_gap}. Finally we establish
    the bound in \eqref{eq:robust_gap}.
    \end{proof}

\acc{
\begin{proof}[Proof of Theorem \ref{thm:robust_no_adj}]
\label{pf:thm:robust_no_adj}
%    When $\mu_n(X,Z)\in \sigalg(Z)$, we simply have $ L^{\alpha}_{n,Q^{(n)}}(\mu_n) =0$, thus 
%    \begin{equation}
%    \nonumber
%            \mathbb{P}\left(
%            L^{\alpha}_{n,Q^{(n)}}(\mu_n) 
%            \le \Ij
%           \right) = 1 \ge 1-\alpha - O(n^{-1/2}).
%          \end{equation}
% Otherwise we consider the nontrivial case where $\Ep{P_{Z}}{\Varpc{Q^{(n)}}{\mu(X,Z)}{\Xnoj}}>0 $. Similarly as in the proof of Theorem \ref{thm:accuracy}, when assuming $\ee{Y^{12}}<\infty$, $\Var{Y|X,Z} \ge \tau>0$, and a uniform moment condition $\max\left\{\EE{\mu_n^{12}(X,Z)}, \Ep{Q^{(n)}}{\mu_n^{12}(X,Z)}\right\}/(\EE{\Varpc{Q^{(n)}}{\mu_n(X,Z)}{Z}})^{6} \le C$, we have
%    \begin{equation}
%    \nonumber
%            \mathbb{P}\left(
%            L^{\alpha}_{n,Q^{(n)}}(\mu_n) 
%            \le \Ij +  \Delta_n
%           \right)\ge 1-\alpha - O(n^{-1/2}).
%    \end{equation}
%    where $ \Delta_n = f^{Q^{(n)}}(\mu_n) - \Ij$. Note that the constant in the rate of $n ^{-1/2}$ depends on $\tau$ and $C$. It is worth mentioning that when the specified conditional distribution is $Q^{(n)}$, in the proof of establishing the coverage rate of $n ^{-1/2}$, verifying the moment conditions \myrom{1} (described in the part immediately after \eqref{eq:L3V_bound}) is done similarly as the derivations in the part between \eqref{eq:L3V_bound} and \eqref{eq:2_thm:main}. These actually involve the term $\EE{\mu_n^{12}(X,Z)}$, in addition to $\Ep{Q^{(n)}}{\mu_n^{12}(X,Z)}$.
%    
    First notice that $\Delta_n$ can be decomposed into two parts:
     \begin{equation} \label{eq:Delta_gap}
    \Delta_n = f^{Q^{(n)}}(\mu_n) - \Ij  = (f^{Q^{(n)}}(\mu_n) - f(\mu_n)) - ( \Ij - f(\mu_n)  ).
     \end{equation}
     In the following, we will deal with $f^{Q^{(n)}}(\mu_n) - f(\mu_n)$ and $\Ij - f(\mu_n)$ separately.
    Applying Lemma~\ref{lem:robust_gap} to $P$, $Q^{(n)}$ and $\mu_n$ under the stated conditions gives
    \begin{equation} \label{eq:thetamu_Q_gap}
    (f^{Q^{(n)}}(\mu_n) - f(\mu_n)) \le c_1 \sqrt{\EE{\chi^2\left(P_{X|Z}\,\|\, Q^{(n)}_{X|Z}\right)}}
     \end{equation}
     for some constant $c_1$ only depending on $\EE{Y^4}$ and $c_0$.
%    for some constant depending on $\EE{Y^{12}}$ and $C$. 
    Regarding the term $  \Ij - f(\mu_n)  $, we recall the derivations in the proof of Theorem~\ref{thm:accuracy}, specifically \eqref{eq:movi_gap_case1} and \eqref{eq:movi_gap_case2}, then obtain    \begin{equation}\label{eq:thetamu_mustar_gap}
     \Ij - f(\mu_n)   \ge  \frac{
    \EE{(\bar{h}_n(W)-\hstar(W))^2}
    }{2\Ij} = \frac{
    \EE{(\bar{\mu}_n(W)-\mustar(W))^2}
    }{2\Ij},
     \end{equation}
     where the equality holds by the definition of $\hstar$, $\bar{\mu}_n$ and $\bar{h}_n$.
    Combining \eqref{eq:Delta_gap}, \eqref{eq:thetamu_Q_gap} and \eqref{eq:thetamu_mustar_gap} yields \eqref{eq:Delta_characterize}.
    % we have $\Delta_n = O(\delta_X^{(n)}-\delta_Y^{(n)})$ where $\delta_X^{(n)}$ and $\delta_Y^{(n)}$ are characterized as in \eqref{eq:Delta_characterize}.
    \end{proof}}

   \acc{
    \section{Details of extending the mMSE gap}
    \label{sec:iff_movi}
    \subsection{Taking the supremum over transformations}
    \label{sec:max_mmse}
        Drawing inspiration from the maximum correlation coefficient \citep{hirschfeld1935connection}, taking the supremum of the mMSE gap over transformations of $Y$ leads to other desirable properties. For a set $\mathcal{G}$ of functions $g$ mapping $Y$ to its sample space, let $\Ij_{\mathcal{G}} = \sup_{g\in\mathcal{G}} \Ij_{\text{sf}}(g(Y))$, where $\Ij_{\text{sf}}(g(Y))$ denotes the scale-free version of the mMSE gap when $Y$ is replaced by $g(Y)$. Then for any fixed function $g\in\mathcal{G}$, floodgate's LCB for $\Ij_{\text{sf}}(g(Y))$ is also an asymptotically valid LCB for $\Ij_{\mathcal{G}}$. And like $\mu$, $g$ can be chosen based on an independent split of the data to make the gap between $\Ij_{\text{sf}}(g(Y))$ and $\Ij_{\mathcal{G}}$ as small as possible. If $\mathcal{G}$ forms a group, then it is immediate that $\Ij_{\mathcal{G}}$ takes the same value when $g(Y)$ is used as the response, for any $g\in\mathcal{G}$, i.e., $\Ij_{\mathcal{G}}$ is invariant to any transformation $g\in\mathcal{G}$ of $Y$. For instance, we might choose $\mathcal{G}$ to be the group of all strictly monotone functions, or of all bijections. Regardless of whether $\mathcal{G}$ is a group or not, if it is large enough that it contains all bounded continuous functions then, by the Portmanteau Theorem, $\Ij_{\mathcal{G}}$ will be zero if \emph{and only if} $\jnull$. That is, for sufficiently large $\mathcal{G}$, $\Ij_{\mathcal{G}}$ satisfies the key property of the MOVI in \citet{azadkia2019simple} and floodgate provides asymptotically valid inference for it. A natural choice\footnote{We are grateful to an anonymous reviewer for suggesting this choice.} of $\mathcal{G}$ satisfying such property is $\{\indicat{y \le t}: t \in \mathbb{R} \}$ as    $$
    	\Ij_{\mathcal{G}} = \sup_{t \in \mathbb{R}} \frac{{\EE{\Varc{\Ec{\indicat{Y \le t}}{X,Z}}{\Xnoj}}}}{\Var{\indicat{Y \le t}}}.
    	$$
    	The above quantity is related to the measure of conditional dependence in \citet{azadkia2019simple} as both involve ${\EE{\Varc{\Ec{\indicat{Y \le t}}{X,Z}}{\Xnoj}}}$.
        \acc{	 
    \subsection{Extending via the RKHS framework}
    \label{sec:kernelfg} 
     
     A reviewer pointed out a very interesting work \citep{huang2020kernel} which came out after our arXiv preprint. To handle $X, Y, Z$ from general topological spaces, \citet{huang2020kernel} proposes the kernel partial correlation coefficient (KPC) to measure conditional dependence and provides consistent estimation methods. \citet{huang2020kernel} mentioned the numerator of KPC with a linear kernel equals to the mMSE gap considered in our paper. 
     %The reviewer also expressed interest in seeing whether the floodgate idea generalizes to the KPC measure. Therefore, 
     In this section, we discuss how to extend the floodgate inferential approach via reproducing kernel Hilbert spaces (RKHS) to apply to the KPC. For ease of exposition, we focus on the numerator of KPC and call it the average kernel maximum mean discrepancy (AKMMD). Note that the AKMMD with a characteristic kernel will be zero if \emph{and only if} $\jnull$}.     
     
%     with a ver question: whether the floodgate idea generalizes to the measures proposed in that paper. 
%      It considers
%        Using the theory of reproducing kernel Hilbert spaces (RKHSs), we can extend the mMSE gap to a class of variable importance measures, which will be called as average kernel maximum mean discrepancy (AKMMD). Such measures are able to handle data from general metric spaces. Accordingly, we have kernel floodgate procedures to make inference on the AKMMD.see Appendix \ref{sec:kernelfg} for details.
%        
%        
%        
Recall the equivalent expression of the mMSE gap in \eqref{eq:mMSEgap_mmd} 
$$
\Ij^2  = \EE{(\condmean - \condmeanj)^2},
$$
where $\condmean$ can be viewed as the kernel embedding of $P_{Y\mid X, Z}$ under a special linear kernel. Then $\Ij^2$ essentially quantifies the distance between $P_{Y\mid X, Z}$ and $P_{Y\mid Z}$ via the maximum mean discrepancy (MMD). To extend this idea using a general kernel, we introduce some new notations and preliminary concepts about RKHS. Suppose $(Y, X, Z)$ take values in some topological space $\cY \times \cX \times \cZ$ and let $P$ be the joint distribution over $(Y, X, Z)$. The marginal distribution of $Y$ is denoted by $P_{Y}$. Sometimes this subscript is dropped when doing so does not cause confusion. We use the bold $\bmu$ notation for kernel mean embeddings, which should be differentiated from the working regression function in the main text. Denote by $\cH_{\cY}$ an RKHS with kernel $\cK(\cdot, \cdot)$ on the space $\cY$, where $\cK:\cY \times \cY \rightarrow \mathbb{R}$ is a symmetric and positive semidefinite function such that $\cK(\cdot, y)$ is measurable function on $\cY, \forall ~y \in \cY$. The inner product and norm on the RKHS $\cH_{\cY}$ are denoted by $\inner{\cdot}{\cdot}_{\cH_\cY}$ and $\norm{\cdot}_{\cH_\cY}$, with the subscripts often dropped for simplicity. The kernel reproducing property implies that $h(y) = \inner{\cK(\cdot, y)}{h}_{\cH_\cY}$. First we introduce the definitions of the kernel mean embedding and the MMD \citep{deb2020measuring,huang2020kernel}.
%\ljmargin
\bdf[Kernel mean embedding]
Suppose $Y\sim P_{Y}$ and $\Ep{P}{\sqrt{\cK(Y,Y)}}< \infty$. There exists \citep{deb2020measuring,huang2020kernel} a unique $\bmu_{P} \in \cH_\cY$ satisfying 
\beq \nonumber
\inner{\bmu_P}{h}_{\cH_{\cY}} = \Ep{P}{h(Y)}, \quad \text{for all } h \in \cH_{\cY},
\eeq
which is called the kernel mean embedding of $P_{Y}$ into $\cH_{\cY}$.
\edf
\bdf[Maximum mean discrepancy]
We measure the distance between two distributions $P_1, P_2$ via the MMD (with respect to the kernel $\cK(\cdot, \cdot)$), defined as
\beq \nonumber
\mathrm{MMD}(P_1,P_2):= \norm{\bmu_{P_1} - \bmu_{P_2}}_{\cH_\cY}.
\eeq
It also has the following equivalent representation \citep{deb2020measuring,huang2020kernel}:
\beq  \nonumber
\mathrm{MMD}^2(P_1,P_2):= \EE{\cK(U,U')} + \EE{\cK(V,V')} - 2\EE{\cK(U,V)},
\eeq
where $U,U' \iid P_1$, $V,V' \iid P_2$ and $U \independent V$.
\edf
%}{both of these definitions seem to include some nontrivial mathematical statements beyond simply a definition (existence/uniqueness in the first, the equivalent representation in the second). Please add a reference for these results}
%The alternative representation of the squared MMD: 
%\beq
%\mathrm{MMD}^2(P_1,P_2):= \EE{\cK(U,U')} + \EE{\cK(V,V')} - 2\EE{\cK(U,V)}
%\eeq
%where $U,U' \iid P_1$, $V,V' \iid P_2$ and $U \independent V$.
Now we are ready to define the AKMMD.
\bdf[average kernel maximum mean discrepancy]
 The \emph{average kernel maximum mean discrepancy} for variable $X$ is defined as 
\beq
\label{{eq:kernelmgap}}
\cI^2_{\cK} = \EE{\mathrm{MMD}^2(P_{Y\mid X,Z}, P_{Y\mid Z})}
\eeq
whenever all the above expectations exist.
\edf
We also present its alternative expression in terms of the kernel:
\beq \nonumber
\cI^2_{\cK} = \EE{\cK(Y_2,\tilde Y_2)} - \EE{\cK(Y_1,\tilde Y_1)} = \EE{\Ec{\cK(Y_2,\tilde Y_2)}{X,Z}} - \EE{\Ec{\cK(Y_1,\tilde Y_1)}{Z}},
\eeq
where $Y_1, \tilde Y_1,Y_2, \tilde Y_2 $ are defined as below 
\beqa \nonumber
Y_1 \mid X \sim P_{Y|Z},\quad \tilde Y_1 \mid X \sim P_{Y|Z}, \quad \text{and } Y_1 \independent  \tilde Y_1 \mid X,\\ \nonumber
(X,Z) \sim P_{X, Z},\quad Y_2\mid X,Z \sim P_{Y\mid X,Z}, \quad  \tilde Y_2\mid X,Z \sim P_{Y\mid X,Z}, \quad \text{and } Y_2 \independent \tilde Y_2 \mid X,Z.
\eeqa
The floodgate functional constitutes a deterministic lower bound for the mMSE gap for any working regression function $\mu$. As we are now dealing with mean embeddings with a general kernel, we will replace the role of $\mu$ with $Q_{Y\mid X, Z}$, an estimate of the full conditional distribution of $Y\mid X, Z$ (as opposed to just its conditional mean). Let $Q = Q_{Y\mid X, Z} \times P_{X, Z}$ and the associated conditional distribution of $Y$ given $Z$ by $Q_{Y\mid Z}$. For notational simplicity, $Q_{Y\mid X, Z}$ and $Q_{Y\mid Z}$ are both sometimes abbreviated simply as $Q$. Given any non-random conditional distribution $Q_{Y\mid X,Z}$, we consider the kernel floodgate functional
\beq\label{eq:kfg_functional}
        f_{\cK}(Q) := 
        \frac{\EE{\cK(Y,Y_2^{\Q})} - \EE{\cK(Y,Y_1^{\Q})}  }{\sqrt{
        \EE{\cK(Y_2^{\Q},\tilde Y_2^{\Q})}  -  \EE{\cK(Y_2^{\Q},Y_1^{\Q})}  
        }}, 
        %=\frac{\sqrt{2}\EE{Y\left(\muX - \muXk\right)}}{\sqrt{\EE{(\muX - \muXk)^2}}}, 
        %~or~ \fmu := \frac{\eeup{Y - \muXk}{2} - \eeup{Y - \muX}{2}}{\sqrt{2\eeup{\muX - \muXk}{2}}}
        %\fmu := \frac{\bb{E}[Y - \muXk^2 - \bb{E}[Y - \muX]^2}{\sqrt{2 \bb{E}[\muX - \muXk]^2}}
\eeq
where the involved random variables are defined through
\begin{align} 
\label{eq:lotrvs} 
\begin{split}
(X,Z) \sim P_{X, Z},\quad Y\mid X,Z \sim P_{Y\mid X,Z}, \quad Y\mid Z \sim P_{Y\mid Z}\\
Y_2^{\Q},\tilde Y_2^{\Q} \mid X,Z \iid \Q_{Y\mid X,Z}, \quad  Y \independent (Y_2^{\Q},\tilde Y_2^{\Q}) \mid X,Z,  \\
Y_1^{\Q} \mid Z \iid Q_{Y|Z}, \quad  Y_1^{\Q} \independent (X,Y, Y_2^{\Q},\tilde Y_2^{\Q} ) \mid Z. \\
\end{split}
\end{align}
Lemma \ref{lem:max_kernel} shows $f_{\cK}$ tightly satisfies the lower-bounding property, as $f$ does in Lemma \ref{lem:max}. The proof can be found in Appendix \ref{sec:pf:lem:max_kernel}.
\blem
\label{lem:max_kernel}
    % Assuming the existence of $\I$, we have $\fmu\le \I$ hold for any nonrandom function $\mu$ such that $\EE{(\muX - \muXk)^2}$ exists, and the true regression function $\mustar$ attains the maximum i.e. $\fmus = \I$.
    For any $\Q$ such that $f_{\cK}(\Q)$ exists, we have $f_{\cK}(\Q) \le \cI_{\cK}$, with equality when $\Q=P_{Y\mid X, Z}$.
\elem
Therefore, we can provide an LCB for $ \cI_{\cK}$ via a LCB for $f_{\cK}(Q)$ with some choice of $Q$. Since the definition of $f_{\cK}(\Q)$ involves null $Y$ samples such as $Y_2^{\Q},\tilde Y_2^{\Q}, Y_1^{\Q}$, we will follow \eqref{eq:lotrvs} to generate null samples of $Y$ then construct i.i.d. unbiased estimates of the numerator and the denominator of $f_{\cK}(\Q)$ respectively. Based on the CLT and the delta method, we can derive asymptotically valid LCBs for $f_{\cK}(\Q)$. This idea is spelled out in Algorithm \ref{algo:kernerl_fg}. 
%Although in general we can generate the samples $Y_2^{\Q},\tilde Y_2^{\Q}, Y_1^{\Q}$ based on any $Q_{Y\mid X,Z}$ estimated from separate dataset, here we briefly describe a particular way using the idea of nearest neighbors. 
%Denote a separate dateset by $\cD_{tr} = \{X_s,Z_s,Y_s\}_{s=1}^{n_{tr}}$. For each $(X_i, Z_i), i \in [n]$, we define the augmented dataset $\cD_{i} := (X_i, Z_i) \cup \{X_s,Z_s\}_{s=1}^{n_{tr}}$. Consider the nearest neighbor of $(X_i, Z_i)$ among $\cD_{i}$ and denote it as $(X_{N(i)}, Z_{N(i)})$ then let $Y_{i2} = Y_{N(i)}$. 
\begin{algorithm*}[htbp]
\caption{Kernel floodgate}
\label{algo:kernerl_fg}
\begin{algorithmic}[1]
\REQUIRE Data $\{(Y_i,W_i)\}\nsubp$, a chosen kernel $\cK(\cdot, \cdot)$, a estimated conditional distribution of $P_{Y\mid X,Z}$, denoted by $Q_{Y\mid X,Z}$, resampling number $M$, $P_{X\mid Z}$, number of null replicates $K$, and a confidence level $\alpha \in (0,1)$.
\vspace{0.05cm}
\STATE For each $i \in [n]$, draw $\{Y_{2,i}^{(m)}\}_{m=1}^{M}$ from $Q_{Y\mid X,Z}$ given $(X_i,Z_i)$; given $Z_i$, draw i.i.d. null samples $\{\tilde X_i^{(k)}\}_{k=1}^{K}$ from $P_{X\mid Z}$, then draw $\{Y^{(k,m)}_{1,i}\}_{m=1}^{M}$ from $Q_{Y\mid X,Z}$ given $(X_i,\tilde Z_i^{(k)})$ for each $k\in [K]$. Denote $Y_{2,i}^{(m)} = Y^{(0,m)}_{1,i}$ for each $m\in [M]$.
\STATE Compute 
\beqa \nonumber
R_{i} &=& \frac{1}{M}\sum_{m=1}^M\cK\rbr{Y_i, Y_{2,i}^{(m)}} - \frac{1}{KM}\sum_{k=1}^K\sum_{m=1}^M \cK\rbr{Y_i, Y^{(k,m)}_{1,i}} \nline
V_{i} &=& \frac{2}{(K+1)M(M-1)}\sum_{k=0}^K~\sum_{1\le m_1 < m_2 \le M}\cK\rbr{Y_{1,i}^{(k,m_1)}, Y_{1,i}^{(k,m_2)}} \nline
&~~~~-&  \frac{2}{K(K+1)M^2}\sum_{ m_1, m_2 = 1}^M~\sum_{0\le k_1 < k_2 \le K} \cK\rbr{Y_{1,i}^{(k_1,m_1)}, Y_{1,i}^{(k_2,m_2)}}
\eeqa
for each $i\in [n]$, and their sample mean $(\bar{R},\bar{V})$ and sample covariance matrix $\hat{\Sigma}$, and compute $s^2 = \frac{ 1 }{ \bar{V} }\left[ \left(\frac{\bar{R} }{2 \bar{V} }\right)^2 \hat{\Sigma}_{22} +  \hat{\Sigma}_{11} - \frac{\bar{R} }{ \bar{V} } \hat{\Sigma}_{12} \right].$
\ENSURE Lower confidence bound $\LB(\mu)=\max\left\{\frac{\bar{R} }{ \sqrt{\bar{V}}}
- \frac{z_{\alpha}s}{\sqrt{n}},\,0\right\}$, with the convention that $0/0=0$.
\end{algorithmic}
\end{algorithm*}

%Remark we also have
%\beqa \nonumber
%\EE{\cK(Y_2^{\Q},\tilde Y_2^{\Q})} + \EE{\cK(Y_1^{\Q},\tilde Y_1^{\Q})} - 2 \EE{\cK(Y_2^{\Q},Y_1^{\Q})} = \EE{\cK(Y_2^{\Q},\tilde Y_2^{\Q})}  - \EE{\cK(Y_1^{\Q},Y_1^{\Q})}
%\eeqa
%  
    
     \subsection{Proofs in Appendix \ref{sec:kernelfg}}
     \label{sec:pf:lem:max_kernel}
\bpf[Proof of Lemma \ref{lem:max_kernel}]
Recall the form of the kernel floodgate functional in \eqref{eq:kfg_functional}
\beq\nonumber
        f_{\cK}(Q) = 
        \frac{\EE{\cK(Y,Y_2^{\Q})} - \EE{\cK(Y,Y_1^{\Q})}  }{\sqrt{
        \EE{\cK(Y_2^{\Q},\tilde Y_2^{\Q})} - \EE{\cK(Y_2^{\Q},Y_1^{\Q})}  
        }} := \frac{\mathrm{II}_1}{\sqrt{\mathrm{II}_2}}, 
        %=\frac{\sqrt{2}\EE{Y\left(\muX - \muXk\right)}}{\sqrt{\EE{(\muX - \muXk)^2}}}, 
        %~or~ \fmu := \frac{\eeup{Y - \muXk}{2} - \eeup{Y - \muX}{2}}{\sqrt{2\eeup{\muX - \muXk}{2}}}
        %\fmu := \frac{\bb{E}[Y - \muXk^2 - \bb{E}[Y - \muX]^2}{\sqrt{2 \bb{E}[\muX - \muXk]^2}}
\eeq
where $X, Z, Y, Y_2^{\Q}, \tilde Y_2^{\Q}, Y_1^{\Q}$ are defined as
\begin{align}
%\label{eq:lotrvs} 
%\begin{split}
\label{eq:rvxyz}
(X,Z) \sim P_{X, Z},\quad Y\mid X,Z \sim P_{Y\mid X,Z}, \quad Y\mid Z \sim P_{Y\mid Z}\\ \label{eq:rvy2Q}
Y_2^{\Q},\tilde Y_2^{\Q} \mid X,Z \iid \Q_{Y\mid X,Z}, \quad  Y \independent (Y_2^{\Q},\tilde Y_2^{\Q}) \mid X,Z,  \\ \label{eq:rvy1Q}
Y_1^{\Q} \mid Z \iid Q_{Y|Z}, \quad  Y_1^{\Q} \independent (X,Y, Y_2^{\Q},\tilde Y_2^{\Q} ) \mid Z.
%\end{split}
\end{align}
Denote the true conditional distributions $P_{Y\mid X,Z}, P_{Y\mid Z}$ by $F,G$ respectively, the estimated conditional distributions $Q_{Y\mid X,Z},Q_{Y\mid Z}$ by $F_{\myq}, G_{\myq}$ respectively,
% in the joint distribution $P_{Y\mid X,Z}^{\mu} \times P_{X,Z}$ we can define $P_{Y\mid Z}^{\mu}$ accordingly, and it will be denoted by $G_{\mu}$.
 and the kernel mean embeddings of those conditional distributions by $\bmu_{F},\bmu_{G}, \bmu_{F_\myq}, \bmu_{G_\myq}$. First notice  
 $$
 \inner{\bmu_{F}}{\bmu_{F_\myq}}_{\cH_\cY} = \inner{\Ec{\cK(\cdot,Y)}{X,Z}}{\Ec{\cK(\cdot,Y_2^{\Q})}{X,Z}}_{\cH_\cY} = \Ec{\cK(Y,Y_2^{\Q})}{X,Z}
 $$
 by \eqref{eq:rvxyz}, \eqref{eq:rvy2Q} and the definition of the kernel embedding. Similarly, we have the following equalities,
\beqa \label{eq:yy2Q}
\EE{\cK(Y,Y_2^{\Q})} = \EE{\Ec{\cK(Y,Y_2^{\Q})}{X,Z}} = \EE{ \inner{\bmu_{F}}{\bmu_{F_\myq}}_{\cH_\cY}}, 
\\ \label{eq:yy1Q}
\EE{\cK(Y,Y_1^{\Q})} = \EE{\Ec{\cK(Y,Y_1^{\Q})}{Z}} = \EE{ \inner{\bmu_{G}}{\bmu_{G_\myq}}_{\cH_\cY}}, 
\\ \label{eq:y2y2Q}
\EE{\cK(Y_2^{\Q}, \tilde Y_2^{\Q})} = \EE{\Ec{\cK(Y_2^{\Q}, \tilde Y_2^{\Q})}{X,Z}} = \EE{ \inner{\bmu_{F_\myq}}{\bmu_{F_\myq}}_{\cH_\cY}}, 
\\ \label{eq:y2y1Q}
%\EE{\cK(Y_1^{\Q}, \tilde Y_1^{\Q})} = \EE{\Ec{\cK(Y_1^{\Q}, \tilde Y_1^{\Q})}{Z}} = \EE{ \inner{\bmu_{G_\myq}}{\bmu_{G_\myq}}}, \nline
\EE{\cK(Y_2^{\Q},  Y_1^{\Q})} = \EE{\Ec{\cK(Y_2^{\Q},  Y_1^{\Q})}{X,Z}} = \EE{ \inner{\bmu_{F_\myq}}{\bmu_{G_\myq}}_{\cH_\cY}},
\eeqa
where we also apply the law of total expectation. Note that the subscripts for the expectation in the above equations are abbreviated. In addition to the these equalities, our derivation also relies on a key result $\EE{\inner{\bmu_{G}}{\bmu_{F_{\myq}}}} = \EE{\inner{\bmu_{G}}{\bmu_{G_{\myq}}}}$. Consider $\tilde Y$ satisfying $ \tilde Y\mid X,Z \sim P_{Y\mid Z}$, $\tilde Y\independent  Y_2^{\Q} \mid X,Z$, $\tilde Y\independent Y_1^{\Q} \mid Z$, then we prove the key result as below,
\begin{align}
\label{eq:kfgkey}
\begin{split}
\EE{\inner{\bmu_{G}}{\bmu_{F_{\myq}}}_{\cH_\cY}} 
&=~ \Ep{X,Z}{\Ec{\cK(\tilde Y, Y_2^{\Q})}{X,Z}} \\
&=~ \EE{\cK(\tilde Y, Y_2^{\Q})} \\
%&=~ \Ep{X,Z,\tilde Y}{\Ec{\cK(\tilde Y, Y_2^{\Q})}{X,Z,\tilde Y}} \\
%&=~ \Ep{X,Z,\tilde Y}{\Ec{\cK(\delta_{\tilde Y}, Y_2^{\Q})}{X,Z,\tilde Y}} \\
%&=~ \Ep{Z,\tilde Y}{\Epc{X\mid Z,\tilde Y}{\Ec{\cK(\delta_{\tilde Y}, Y_2^{\Q})}{X,Z,\tilde Y}}{Z,\tilde Y}} \\
%&=~ \Ep{Z,\tilde Y}{\Ec{\cK(\delta_{\tilde Y}, Y_1^{\Q})}{Z,\tilde Y}} \\
&=~ \Ep{Z}{\Ec{\cK(\tilde Y, Y_2^{\Q})}{Z}} \\
&=~ \Ep{Z}{\Ec{\cK(\tilde Y, Y_1^{\Q})}{Z}} = \EE{\inner{\bmu_{G}}{\bmu_{G_{\myq}}}_{\cH_\cY}},
\end{split}
\end{align}
where the first and the last equalities hold by the definition of the kernel mean embedding, the second and the third equalities hold by the law of total expectation, and the fourth equality holds by the definitions of $Y_1^{\Q}, Y_2^{\Q}, \tilde Y$.

Therefore we can rewrite the numerator of $f_\cK(Q)$ as
\beqa \nonumber
\mathrm{II}_1 &= &\EE{\cK(Y,Y_2^{\Q})} - \EE{\cK(Y,Y_1^{\Q})} \\ \nonumber
&=& \EE{\inner{\bmu_{F}}{\bmu_{F_\myq}}_{\cH_\cY}} - \EE{\inner{\bmu_{F}}{\bmu_{G_\myq}}_{\cH_\cY}} \nline
&=& \EE{\inner{\bmu_{F}}{\bmu_{F_\myq}-\bmu_{G_\myq}}_{\cH_\cY}} \nline
&=& \EE{\inner{\bmu_{F} - \bmu_{G}}{\bmu_{F_\myq}-\bmu_{G_\myq}}_{\cH_\cY}} + \EE{\inner{\bmu_{G}}{\bmu_{F_\myq}-\bmu_{G_\myq}}_{\cH_\cY}} \nline
&=& \EE{\inner{\bmu_{F} - \bmu_{G}}{\bmu_{F_\myq}-\bmu_{G_\myq}}_{\cH_\cY}} + \EE{\inner{\bmu_{G}}{\bmu_{F_\myq}}_{\cH_\cY}} - \EE{\inner{\bmu_{G}}{\bmu_{G_\myq}}_{\cH_\cY}} \nline
&=& \EE{\inner{\bmu_{F} - \bmu_{G}}{\bmu_{F_\myq}-\bmu_{G_\myq}}_{\cH_\cY}}  \nline
&\le& \EE{\norm{\bmu_{F} - \bmu_{G}}_{\cH_\cY} \norm{\bmu_{F_\myq}-\bmu_{G_\myq}}_{\cH_\cY}}  \\ \label{eq:kfgnum}
&\le& \sqrt{\EE{\norm{\bmu_{F} - \bmu_{G}}^2_{\cH_\cY} }} \sqrt{\EE{\norm{\bmu_{F_\myq}-\bmu_{G_\myq}}_{\cH_\cY}^2}},
\eeqa
where the first line holds due to \eqref{eq:yy2Q} and \eqref{eq:yy1Q}, the second to the fourth equalities hold by rearranging, the fifth equality holds due to \eqref{eq:kfgkey}, the last two inequalities hold by the Cauchy--Schwarz inequality. Regarding the denominator of $f_\cK(Q)$, we rewrite $\mathrm{II}_2$ in terms of the kernel embedding
\beqa \nonumber
\mathrm{II}_2 &=& \EE{\cK(Y_2^{\Q},\tilde Y_2^{\Q})} -  \EE{\cK(Y_2^{\Q},Y_1^{\Q})}  \\ \nonumber
&=& \EE{ \inner{\bmu_{F_\myq}}{\bmu_{F_\myq}}_{\cH_\cY}} - \EE{ \inner{\bmu_{G_\myq}}{\bmu_{F_\myq}}_{\cH_\cY}} \\ \nonumber
&=& \EE{ \inner{\bmu_{F_\myq}}{\bmu_{F_\myq}}_{\cH_\cY}} + \EE{ \inner{\bmu_{G_\myq}}{\bmu_{G_\myq}}_{\cH_\cY}} - 2\EE{ \inner{\bmu_{G_\myq}}{\bmu_{F_\myq}}_{\cH_\cY}} \\ \label{eq:kfgden}
&=& \EE{\norm{\bmu_{F_\myq}-\bmu_{G_\myq}}_{\cH_\cY}^2},
% = \EE{\mathrm{MMD}^2(F_{\mu},G_{\mu})}
\eeqa
where the second equality holds due to \eqref{eq:y2y2Q} and \eqref{eq:y2y1Q} and the third equality holds since $\EE{\inner{\bmu_{G_\myq}}{\bmu_{F_\myq}}_{\cH_\cY}} = \EE{\inner{\bmu_{G_\myq}}{\bmu_{G_\myq}}_{\cH_\cY}}$ can be similarly derived as \eqref{eq:kfgkey}. As $\cI^2_{\cK} = \EE{\mathrm{MMD}^2(P_{Y\mid X,Z}, P_{Y\mid Z})} =  \EE{\norm{\bmu_{F} - \bmu_{G}}_{\cH_\cY}^2 } $, we have $f_\cK(Q) \le \cI_{\cK}$ by combining \eqref{eq:kfg_functional}, \eqref{eq:kfgnum}, and \eqref{eq:kfgden}.
%\beqa \nonumber
%\cI^2 = \EE{\norm{\bmu_{F} - \bmu_{G}}^2 } = \EE{\mathrm{MMD}^2(P_{Y\mid X,Z}, P_{Y\mid Z})} 
%\eeqa
%therefore have $f(\mu)\le $.
\epf
    
	}

     \section{Transporting inference to other covariate distributions}
    \label{sec:transport}
    To present how to perform inference on a target population whose covariate distribution differs from the distribution the study samples are drawn from, let $Q$ denote the target distribution for all the random variables $(Y,X,Z)$, but assume that $Q_{Y|X,Z}=P_{Y|X,Z}$ and that $Q_{X|Z}$ and the likelihood ratio $Q_{Z}/P_Z$ are known (note this last requirement is trivially satisfied if only $X\mid Z$ changes between the study and target distributions, i.e., we know $Q_Z=P_Z$). Overloading notation slightly, let $Q$ and $P$ also denote the real-valued densities of random variables under their respective distributions (so, e.g., $P(Y=y\,|\,Z=z)$ denotes the density of $Y\mid Z=z$ under $P$ evaluated at the value $y$), which we assume to exist. We can now define a weighted analogue of the floodgate functional ~\eqref{eq:thetaj_mu}:
%    \begin{definition}[Normalized excess prediction error after weighting]\label{def:theta_j_w}
    \begin{equation}\label{eq:fw}
        f^{w}(\mu) = \frac{\mathbb{E}_{P}[(Y-{\mu}(\tilde{X},Z))^2 w(X,Z)w_{1}(\tilde{X},Z) -  (Y-{\mu}(X,Z))^2w(X,Z)] }{\sqrt{2\mathbb{E}_{P}[(\mu(X,Z)-{\mu}(\tilde{X},Z))^2 w(X,Z)w_{1}(\tilde{X},Z)]}},
    \end{equation}
    where $w(x,z)=w_{0}(z)w_{1}(x,z)$, $w_{0}(z)=\frac{Q(Z=z)}{P(Z=z)}$, $w_{1}(x,z)=\frac{Q(X=x\,|\,Z=z)}{P(X=x\,|\,Z=z)}$, %$w_{1}(\tilde{X})=\frac{Q(\tilde{X}_{j}|X_{\noj})}{P(\tilde{X}_{j}|X_{\noj})}$ 
    % (without loss of generality, we assume both $P_{X}$ and $Q_{X}$ are continuous).
    % \end{definition}
    and $\tilde{X}\sim P_{X|Z}$ conditionally independently of $Y$ and $X$. The following Lemma certifies that $f^w$ satisfies property (a) of a floodgate functional for $\Ij^2_Q=\mathbb{E}_{Q}\left[
    \mathrm{Var}_{Q}\left(\mathbb{E}_{Q}[Y\,|\,X,Z]\,|\,Z\right)\right]$, the mMSE gap with respect to $Q$.
    
    %Remark that we weight the prediction error evaluated at the original samples and that evaluated at the null samples in a different way. This can be understood well by noticing that when comparing the joint distribution of $(X,\tilde{X},Y)$ under different populations, extra weighting regarding the null samples needs to be considered. We have the following theorem, which guarantee the inferential results via our MOCK approach are transportable from the study population to the target population.
    \begin{lemma}\label{thm:transport}
    If $Q_{Y|X,Z}=P_{Y|X,Z}$, then for any $\mu$ such that $f^w(\mu)$ exists, $f^w(\mu)\le\Ij_Q$, with equality when $\mu=\mustar$.
    % If the target population satisfies the invariance assumption about the conditional model regarding the study population, i.e. $Q(Y|X)=P(Y|X)$, the weighted version of the normalized excess prediction error defined above is maximized at $\mu^{\star}$ with the maximum value
    % \begin{equation}
    % \theta_{j}^{w}(\mu^{\star})=\sqrt{\mathbb{E}_{Q}\left[
    % \mathrm{Var}_{Q}\left(\mathbb{E}_{Q}[Y\,|\,X,Z]\,|\,Z\right)\right]}  
    % \end{equation}
    % In particular, if $X_{j}$ is null under the study population, then we have $\theta_{j}^{w}(\mu^{\star})=0$
    \end{lemma}
The proof is immediate from Lemma~\ref{lem:max} if we notice that the ratio of the joint distribution of $(Y,X,\tilde{X},Z)$ under the two populations equals
    \begin{equation}\label{eq:ratio}
        \frac{Q(Y,X,Z)Q(\Xtil\,|\,Z)}{P(Y,X,Z)P(\Xtil\,|\,Z)}=\frac{Q(Y\,|\,X,Z)}{P(Y\,|\,X,Z)}\frac{Q(X,Z)}{P(X,Z)}\frac{Q(\Xtil\,|\,Z)}{P(\Xtil\,|\,Z)}=w_{1}(\tilde{X},Z)w(X,Z),
    \end{equation}
    where the last equality follows from $P_{Y|X,Z}=Q_{Y|X,Z}$. Floodgate property (b) of $f^w$ can be established in the same way as for $f$ by computing weighted versions of $R_i$ and $V_i$ from Algorithm~\ref{alg:MOCK} according to the weights in Equation~\eqref{eq:fw}, applying the central limit theorem, and combining them with the delta method.

     \section{Algorithm details for inference on the MACM gap}
     \label{app:class_computation}
     \acc{
    Recall the construction of the floodgate functional (\eqref{eq: kappaj_u} in Section \ref{sec:class}):
    \begin{equation}\nonumber
    \kappamu = 
    % 2\EE{\mathbbm{1}_{\{Y\cdot (\muXk - \Ec{\muX}{\Xnoj}) < 0\}} - \mathbbm{1}_{\left\{Y\cdot (\muX - \Ec{\muX}{\Xnoj}) < 0\right\}}},
    2\mathbb{P}\big(Y (\muXk - \Ec{\muX}{\Xnoj}) < 0\big)-2\mathbb{P}\big(Y (\muX - \Ec{\muX}{\Xnoj}) < 0\big).
    \end{equation}
    We can define random variables which are i.i.d. and unbiased for $\kappamu$ then construct CLT-based confidence bounds, as formalized in Algorithm \ref{alg:class_MOCK}.
%    Algorithm~\ref{alg:class_MOCK} formalizes the procedure and Theorem~\ref{thm:class_main} establishes its asymptotic coverage.
\begin{algorithm*}[h]
	\caption{Floodgate for the MACM gap}\label{alg:class_MOCK}
	\begin{algorithmic}
	\REQUIRE  Data $\{(Y_i,X_i,Z_i)\}\nsubp$, $P_{X\mid \Xnoj}$, a \revision{working regression function} $\mu:\mathbb{R}^p\rightarrow\mathbb{R}$, and a confidence level $\alpha \in (0,1)$.
	\STATE Let $U_i = \mu(X_i,Z_i)-\ec{\mu(X_i,Z_i)}{Z_i}$ and % and $\tilde{U}_i = 
	compute
	\vspace{-.2cm}\[ R_i = \left\{\begin{array}{rl} \Pc{U_i<0}{\Xinoj}-\indicat{U_i<0} & \text{ if }\;Y_i=1\\
	\Pc{U_i>0}{\Xinoj}-\indicat{U_i>0} & \text{ if }\;Y_i=-1
	\end{array}\right. \vspace{-.2cm}\]
	for $i\in [n]$,
	and compute its sample mean $\bar{R}$ and sample variance $s^2$.%by $(\bar{R}, s^2)$.
	\RETURN Lower confidence bound $L_{n}^{\alpha} (\mu)=2 \max\left\{\bar{R} - \frac{z_{\alpha}s}{\sqrt{n}},0\right\}$.
	\end{algorithmic}
\end{algorithm*}
	}
    Algorithm \ref{alg:class_MOCK} involves computing the terms $\ec{\mu(X_i,Z_i)}{Z_i}$ and evaluating the CDF of the conditional distribution $\muX\,|\,\Xnoj =\xnoj$ at the value $\ec{\mu(X_i,Z_i)}{Z_i}$, which is not analytically possible in general. Unlike in Section~\ref{sec:computation}, where users can replace $ \Ec{\mu(X,Z)}{\Xnoj}$ and ${\Varc{\mu(X,Z)}{\Xnoj}}$ by their Monte Carlo estimators without it impacting asymptotic normality, we need slightly more assumptions when inferring the MACM gap due to the discontinuous indicator functions in the definition of $\kappamu$. 
    Before stating the required assumptions, we introduce some notation, all of which is specific to a given \revision{working regression function} $\mu$. 
% 	 To establish the validity of the lower confidence bounds with the terms $R_{i}$ coming from Algorithm \ref{alg:class_MOCK_MC}, we make some extra assumptions. Before stating them, we introduce several notations. First recall the definitions in \eqref{eq:cdf_notation} and further denote
	  \begin{eqnarray}\nonumber
    &~& U:=\mu(X,Z),~~g(\xnoj):=\ec{\mu(X,Z)}{\Xnoj = \xnoj},\\ \nonumber
    &~& G_{z}(u):=\Pc{U < u}{\Xnoj =\xnoj},~~F_{z}(u):=\Pc{U \le u}{\Xnoj =\xnoj}. \\ \nonumber
    % \end{eqnarray}
    % \begin{eqnarray}
    &~&
    \varsigma(\xnoj) := \sqrt{\Varc{\muX}{\Xnoj=\xnoj}},\\ 
    % &~& 
    % G_{\xnoj,y}(u):=\Pc{U < u}{\Xnoj =\xnoj,Y=y},~~F_{\xnoj,y}(u):=\Pc{U \le u}{\Xnoj =\xnoj,Y=y}. \\ \nonumber
    &~&
    C_{u,\xnoj,y}: = 
    \frac{\max\{
            \left|G_{\xnoj,y}(u) - G_{\xnoj,y}(g(\xnoj)) \right|,
            \left|F_{\xnoj,y}(u) - F_{\xnoj,y}(g(\xnoj)) \right|\}
            }{\left|u - g(\xnoj)\right|} \label{eq:cdf_notation} 
    \end{eqnarray}
     where $ F_{\xnoj,y}(u)$ is the CDF of $\mu(X,Z)\mid\Xnoj=\xnoj,Y=y$ evaluated at $u$, $G_{\xnoj,y}(u)$ is the limit from the left of the same CDF at $u$, and with the convention for $C_{u,\xnoj,y}$ that $0/0=0$ (so it is well-defined when $u=g(z)$). Now we are ready to state Assumption \ref{asmp:class_smooth}.
    % then we consider the conditional distribution of $Z$ given $\Xnoj,Y$ and denote $\Pc{Z\le z}{\Xnoj=\xnoj,Y=y}$ by $F_{\xnoj,y}(z)$ and its left limit by $G_{\xnoj,y}(z)$.
    \begin{assump}\label{asmp:class_smooth}
    Assume the joint distribution over $(Y,X,Z)$ and the nonrandom function $\mu:\mathbb{R}^{p} \rightarrow \mathbb{R}$ satisfy the following on a set of values of $Y=y,Z=z$ of probability 1:
     %Denote $\calP$ to be the class of the joint distributions over $(X,Y)$ and $\calU$ to be the class of functions $\mu:\mathbb{R}^{p} \rightarrow \mathbb{R}$ satisfying the following:
        \begin{enumerate}[(a)]
            % \item There exists $\delta_{\xnoj,y}>0$, such that for any $z$ satisfying $\frac{|z-g(\xnoj)|}{\sigalg(\xnoj)} \le \delta_{\xnoj,y}$, we have
            % $$
            % \max\{
            % \left|F_{\xnoj,y}(z) - F_{\xnoj,y}(g(\xnoj)) \right|,
            % \left|G_{\xnoj,y}(z) - G_{\xnoj,y}(g(\xnoj)) \right|\}
            % \le C_{z,\xnoj,y} \left|z - g(\xnoj)\right|
            % $$
            % where $C_{z,\xnoj,y}$ is finite and can depend on $z$.
            \item There exists a $\delta_{\xnoj,y}>0$ and finite $C_{\xnoj,y}$ such that
            $$
            C_{u,\xnoj,y} \le C_{\xnoj,y}~~\text{when}~ |u-g(\xnoj)|\le \varsigma(\xnoj) \delta_{\xnoj,y}.
            $$
            %for some finite $C_{\xnoj,y}$.
            \item The above $C_{\xnoj,y}$ and $\delta_{\xnoj,y}$ satisfy
            $$
            \EE{C^2_{\Xnoj,Y}} < \infty,~ %\le C_{0},~
            \EE{\frac{1}{\delta_{\Xnoj,Y}}} < \infty. %\le C_{1}.
            $$ 
            % for some constant $C_{0},C_{1}>0$.
            % \item The above $C_{\xnoj,y}$ and $\delta_{\xnoj,y}$ satisfy
            % $$
            % \EE{C^2_{\Xnoj,Y}}\le C_{0},~
            % \EE{\frac{\exp\{-M_{0}\delta^2_{\Xnoj,Y}\}}{\delta_{\Xnoj,Y}}}\le C_{1}.
            % $$ 
            % for some constant $C_{0},C_{1}>0$ and some large enough number $M_{0}$.
            % $$
            % \EE{C_{\muX,\Xnoj,Y} \mathbbm{1}_{{\left\{\frac{|\muX - g(\Xnoj)|}{\sigalg(\Xnoj)}\le \delta_{\Xnoj,Y}\right\}}}}\le C_{0},~
            % \EE{\frac{\exp\{-M_{0}\delta^2_{\Xnoj,Y}\}}{\delta_{\Xnoj,Y}}}\le C_{1}.
            % $$
            \item 
            $
            \EE{\varsigma^2(\Xnoj)} < \infty, % \le C_{2}
             ~ \EE{\frac{\Ec{|\mu(X,Z) - \Ec{\mu(X,Z)}{\Xnoj}|^3}{\Xnoj}}{\varsigma^3(\Xnoj)}} < \infty. % \le C_{3}.
            $
        \end{enumerate}
        %for some constants $C_{0},C_{1},C_{2},C_{3}>0$.
        %uniformly assumed over $(\calP,\calU)$.
        % Remark that the constants $C_{0},C_{1},C_{2},C_{3}>0$ the large number $M_{0}$ are uniformly assumed over $(\calP,\calU)$.
    \end{assump}
    These assumptions are placed because we have to construct the Monte Carlo estimator of $\Ec{\muX}{\Xnoj}$ then plug it into the discontinuous indicator functions in $\kappamu$. Assumptions \ref{asmp:class_smooth}$(a)$ and \ref{asmp:class_smooth}$(b)$ are smoothness requirements on the the CDF of $\muX\mid\Xnoj,Y$ around $\Ec{\muX}{\Xnoj}$. Assumption \ref{asmp:class_smooth}$(c)$ specifies mild moment bound conditions on $\muX$. To see that they are actually sensible, we consider the example of logistic regression and walk through those assumptions in Appendix \ref{sec:eg_asmp:class_smooth}. 
    
    Assume that we can sample $(M+K)$ copies of $X_i$ from $P_{X_i|Z_i}$ conditionally independently of $X_i$ and $Y_i$, which are denoted by $\{\Xtil_i^{(m)}\}_{m=1}^{M}$, $\{\Xtil_i^{(k)}\}_{k=1}^{K}$, and thus
    replace $g(Z_i)$ (i.e. $\ec{\mu(X_i,Z_i)}{Z_i}$) and $R_i$, respectively, by the sample estimators
    \[
    g^M(Z_i) =  \frac{1}{M}\sum_{m=1}^{M}\mu(\Xtil_i^{(m)},Z_i),~ R_{i}^{M,K}=\frac{1}{K}\sum_{k=1}^{K}\left(\indicat{Y_i(\mu(\Xtil_i^{(k)},Z_i)-   g^M(Z_i) )< 0} \right)- \indicat{Y_i(\mu(X_i,Z_i)) -  g^M(Z_i) )< 0}
    \]
    % Hence we require the CDF of the conditional distribution of $\muX\mid\Xnoj,Y$ to have sort of smoothness around $\Ec{\muX}{\Xnoj}$. 
     \begin{theorem}\label{thm:class_main_MC}
     Under the same setting as in Theorem \ref{thm:class_main}, if either \myrom{1} $\EE{\varc{\mu(X,Z)}{Z}}=0$ or \myrom{2} $\EE{\Varc{\indicat{Y\cdot [\muX - \Ec{\muX}{\Xnoj} ]< 0 }}{\Xnoj,Y}}> 0$ holds together with Assumption \ref{asmp:class_smooth} and
        % Denote $\calP$ to be the class of the joint distributions over $(X,Y)$ and $\calU$ to be the class of functions $\mu:\mathbb{R}^{p} \rightarrow \mathbb{R}$ satisfying either of the following:
        % \begin{itemize}
        %     \item $\mu(X)\in \sigalg(\Xnoj)$.
        %     \item Assumption \ref{asmp:class_smooth} holds and $\EE{\Varc{\indicat{Y\cdot [\muX - \Ec{\muX}{\Xnoj} ]\le 0 }}{\Xnoj,Y}}\ge \tau_{0}$.
        % \end{itemize}
        % where the constants $C_{0},C_{1},C_{2},C_{3}, \tau_{0}>0$ are uniformly assumed over $(\calP,\calU)$.
        % % Under Assumption \ref{asmp:class_smooth} and $\EE{\Varc{\indicat{Y\cdot [\muX - \Ec{\muX}{\Xnoj} ]\le 0 }}{\Xnoj,Y}}\ge \tau_{0}$, 
        % Then if the number of null samples $M$ is large enough such that 
        ${n}/{M}=o(1)$, then $L^{\alpha}_{n,M,K}(\mu)$ computed by replacing $g(Z_i)$ and $R_i$ with $g^{M}(Z_i)$ and $R_i^{M,K}$, respectively, in Algorithm \ref{alg:class_MOCK} satisfies
        \begin{equation} 
    	\label{eq:unif_valid_class}\nonumber
	   % \inf_{P\in \calP,~ \mu \in \calU}\mathbb{P}_{{P}}\left(
	   %  \Ljc \le \Ijc  \right) 
	    \PP{  L^{\alpha}_{n,M,K}(\mu)  \le \Ijc   }
	     \ge 1 - \alpha + o(1).
	    \end{equation}
	    %where $\Ljc$ is the lower confidence bounds with the $R_{ij}$ terms coming from Algorithm \ref{alg:class_MOCK_MC}.
    \end{theorem}
     The proof can be found in Appendix~\ref{pf:thm:class_main_MC}. Intuitively when we construct a lot more null samples to estimate the term $g(Z_i)$, our inferential validity improves. Formally, when ${n^2}/{M}=O(1)$, we can improve the asymptotic miscoverage to $O(n^{-1/2})$.
     Note that we only place a rate assumption on $M$ (but put no requirement on $K$).

      \subsection{Illustration of assumption \ref{asmp:class_smooth}}
    \label{sec:eg_asmp:class_smooth}
    We consider the joint distribution over $W$ to be $p$-dimensional multivariate Gaussian with $X=W_j, Z = W_{\noj}$ for some $1\le j \le p$, and $Y$ follows a generalized linear model with logistic link. That is,
    \begin{equation}
        \nonumber
       W  \sim \gauss{\bm{0}}{\bm{\Sigma}},~~
        \mustar(W) = 2 \Pc{Y=1}{W} -1, ~~~\text{where}~\Pc{Y=1}{W}  = \frac{\exp{(W\beta^{\star})}}{1 + \exp{(W\beta^{\star}})}, ~\beta^{\star} \in \mathbb{R}^{p}.
    \end{equation}
    Choosing logistic regression as the fitting algorithm, we have $U:=\muX$ takes the following form
    \begin{equation}
        \nonumber
        U:= \mu(W) = \frac{2\exp{(W\beta)}}{1 + \exp{(W\beta)}} - 1
    \end{equation}
    where $\beta \in \mathbb{R}^{p}$ is the fitted regression coefficient vector and $\beta_{j}\ne0$ whenever $\EE{\varc{\mu(X,Z)}{Z}}>0$. Conditional on $\Xnoj$, $U$ follows a logit-normal distribution (defined as the logistic function transformation of normal random variable) up to constant shift and scaling. Note that the probability density function (PDF) of logit-normal distribution with parameters $a,\sigma$ is
    \begin{equation}
        h_{\text{logit}}(u)=\frac{1}{\sigma \sqrt{2\pi}}\exp{\left(-\frac{(\text{logit}(u) - a)^2}{2\sigma^2}\right)}\frac{1}{u(1-u)},~~ u \in (0,1)
    \end{equation}
    where $\text{logit}(u)=\log(u/(1-u))$ is the logit function. 
    Note $h_{\text{logit}}(u)$ is bounded over its support. Regarding the PDF of $U \mid \Xnoj=\xnoj,Y=1$, which is denoted as $h_{\xnoj,1}(u)$, we first notice the following expression
    \begin{equation}\label{eq:1_eg_asmp:class_smooth}
            h(\xj \mid \Xnoj =\xnoj,Y=1) = \frac{h(\xj \mid \Xnoj=\xnoj) \Pc{Y=1}{W=w} }{\int h(\xj\mid \Xnoj =\xnoj) \Pc{Y=1}{W=w} d\xj}
    \end{equation}
    where  $w_j =x, w_{\noj} = z$, $ h(\xj \mid \Xnoj =\xnoj,Y=1)$ and $
            h(\xj \mid \Xnoj =\xnoj,Y=1) $ denote the density functions of $X\mid Z=z, Y=1$ and $X \mid Z=z$.
    Since $\text{logit}(z)$ is one-to-one mapping, we have $f_{\xnoj,1}(z)$ (up to constant shift and scaling) takes the form similar to \eqref{eq:1_eg_asmp:class_smooth}
    \begin{equation}
           h_{\xnoj,1}(u) = \frac{h_{\text{logit}}(u) \Pc{Y=1}{W=w} }{\int h_{\text{logit}}(u) \Pc{Y=1}{W=w } d\xj}
    \end{equation}
    where $w=(x,z)=\mu^{-1}(u)$, and we denote the PDF of $ U \mid \Xnoj = \xnoj$ as $h_{\text{logit}}(u)$ without causing confusion (the parameters of $h_{\text{logit}}(u)$ depend on $\xnoj,\beta$). Therefore we can show $h_{\xnoj,1}(z)$ is bounded (similarly for $h_{\xnoj,-1}(z)$).
    
    The boundedness of $h_{\xnoj,y}(u)$ implies that the corresponding CDF $F_{\xnoj,y}$ ($F_{\xnoj,y} = G_{\xnoj,y}$ in this case) satisfies a Lipschitz condition over its support. Hence $\delta_{\xnoj,y}$ can be chosen to be greater than some positive constant uniformly, so that $\EE{\frac{1}{\delta_{\Xnoj,Y}}} < \infty$ holds. Though the Lipschitz constant does depend on $\xnoj,\beta$, it is easy to verify $ \EE{C^2_{\Xnoj,Y}} < \infty$, thus assumption (b) holds. And assumption (c) is just a regular moment condition.
    
     \subsection{Proofs in Appendix \ref{app:class_computation}}
     \begin{proof}[Proof of Theorem \ref{thm:class_main_MC}]\label{pf:thm:class_main_MC}
    Similar to the proof of Theorem \ref{thm:class_main}, it suffices to deal with the case where $\muX \notin \sigalg(\Xnoj)$ and prove
    \begin{equation} 
	\label{eq:1_thm:class_main_MC}
	   % \inf_{P\in \calP,~ \mu \in \calU}\mathbb{P}_{{P}}\left(\kappamu  \ge \Ljc \right) 
	   \PP{  L^{\alpha}_{n,M,K}(\mu) \le \kappamu  }
	    \ge 1 - \alpha + o(1). % - O(1/\sqrt{n}).
	\end{equation}    
    Note that in Algorithm \ref{alg:class_MOCK}, $\EE{R_i} = f_{\ell_1}(\mu)/2$. But when $g(Z_i)$ (i.e., $\ec{\mu(X_i,Z_i)}{Z_i}$) and $R_i$ are replaced by $g^{M}(Z_i)$ and $R_i^{M,K}$, respectively, in Algorithm \ref{alg:class_MOCK}, we do not have $\EE{R_{i}^{M,K}}$ equal to $\kappamu/2$ anymore. Note that $f_{\ell_1}(\mu)/2$ equals the following
    \begin{equation}\label{eq:2_thm:class_main_MC}
    f_{\ell_1}(\mu)/2 = 
    \EE{\mathbbm{1}_{\{
    Y\cdot [\muXk - \Ec{\muX}{\Xnoj}] < 0
    \}}} - 
    \EE{\mathbbm{1}_{\left\{
    Y\cdot [\muX - \Ec{\muX}{\Xnoj} ] < 0
    \right\}}},
    \end{equation} 
    and $R_i^{M,K}$ is defined as
   \begin{equation}\label{eq:Ri_MK_thm:class_main_MC}
    R_{i}^{M,K}=\frac{1}{K}\sum_{k=1}^{K}\left(\indicat{Y_i(\mu(\Xtil_{i}^{(k)}, Z_i)- g^{M}(Z_i)) < 0} \right)- \indicat{Y_i(\mu(X_i, Z_i)) - g^{M}(Z_i) ) < 0}
    \end{equation}
    Remark the value of $\EE{R_{i}^{M,K}}$ does not depend on $K$, hence we simplify the notation into $R_{i}^{M}$ without causing confusion. Actually we can show as $M\rightarrow \infty$, $\EE{R_{i}^{M}}\rightarrow \kappamu/2$.  Indeed, we need to show $\sqrt{n}|\EE{R_{i}^{M}} - \kappamu/2| = o(1)$ in order to prove \eqref{eq:1_thm:class_main_MC}. Also remark that in Section~\ref{sec:class}, it is mentioned that under a stronger condition $n^2/M = O(1)$ (which will imply $\sqrt{n}|\EE{R_{i}^{M}} - \kappamu/2| = O(1/\sqrt{n})$), we can additionally establish a rate for $n^{-1/2}$ for the asymptotic coverage validity in Theorem~\ref{thm:class_main_MC}. 
    % Under the stated moment conditions, we can similarly establish the convergence by central limit theorem
    % $$
    % \frac{\sqrt{n}(\bar{R} - |\EE{R_{i}})}{s} \stackrel{d}{\rightarrow} \calN(0,1)
    % $$
    % and uniformly bound the Kolmogorov distance between the original CDF and the normal CDF with a rate $O(1/\sqrt{n})$. Hence in order to establish \eqref{eq:1_thm:class_main_MC}, it suffices to prove $\sqrt{n}(\EE{R_{i}} -\kappamu/2 )=o(1)$. %uniformly over $(\calP,\calU)$. 
    In either cases, it is reduced to prove 
      \begin{equation}\label{eq:kappamu_gap_unif}
    %\inf_{P\in \calP,~ \mu \in \calU}
    \left|
    \EE{R_{i}^M} - \frac{\kappamu}{2}
    \right| = O\left(\frac{1}{\sqrt{M}}\right)
    \end{equation}
    %According to Algorithm \ref{alg:class_MOCK_MC}, 
    First we ignore the $i$ subscripts and get rid of the average over $K$ null samples in the definition of $R_{i}^{M,K}$, then $\EE{R_{i}^M}$ can be simplified into
    \begin{equation}
        \EE{ \indicat{Y(\mu\Xk- g^{M}(Z)) < 0} - \indicat{Y(\mu(X,Z) -g^{M}(Z))< 0}}
    \end{equation}
    where $g^{M}(Z) = \frac{1}{M}\sum_{m=1}^{M}\mu(\Xtil^{(m)}, Z)$. To bound $\left|
    \EE{R_{i}^M} -\kappamu/2
    \right|$, we consider the two terms in \eqref{eq:2_thm:class_main_MC} and separately bound
     \begin{eqnarray}\nonumber
     \mathrm{II}_1
     &:=&
     \left|\EE{ \indicat{Y(\mu(\Xtil, Z) -  g^{M}(Z) ) < 0} - \mathbbm{1}_{\{
    Y\cdot [\muXk - \Ec{\muX}{\Xnoj}] < 0}} \right|, \\ \nonumber
    \mathrm{II}_2
    &:=&
    \left|\EE{ \indicat{Y(\mu(X,Z) - g^{M}(Z) < 0}   -\mathbbm{1}_{\left\{
    Y\cdot [\muX - \Ec{\muX}{\Xnoj} ] < 0
    \right\}}} \right|.
    \end{eqnarray}
    Starting from the second term above, we rewrite it as
    \begin{align}\nonumber
     \mathrm{II}_2  
     &=~
     \left|\EE{ \Ec{\indicat{Y(\mu(X,Z) - g^{M}(Z)) < 0}   -\mathbbm{1}_{\left\{
    Y\cdot [\muX - \Ec{\muX}{\Xnoj} ] < 0
    \right\}}}{\Xnoj,Y,  \{\Xtil^{(m)}\}_{m=1}^M }} \right|
    \\\nonumber
    &\le~ \left|\EE{ \indicat{Y=1}
    \Ec{\indicat{\mu(X,Z) <  g^{M}(Z) }   -\mathbbm{1}_{\left\{
    \muX  < \Ec{\muX}{\Xnoj}
    \right\}}}{\Xnoj,Y, \{\Xtil^{(m)}\}_{m=1}^M  }}\right| \\ \nonumber
    &~~
    +\left|\EE{\indicat{Y=-1}
    \Ec{
    \indicat{\mu(X,Z) >  g^{M}(Z) }   -\mathbbm{1}_{\left\{
    \muX  > \Ec{\muX}{\Xnoj}
    \right\}}
    }{\Xnoj,Y, \{\Xtil^{(m)}\}_{m=1}^M } } \right| \\ \nonumber
    &\le~ \EE{
    \max\{
            \left|G_{\Xnoj,Y}( g^{M}(Z)) - G_{\Xnoj,Y}(g(\Xnoj)) \right|,
            \left|F_{\Xnoj,Y}(  g^{M}(Z) ) - F_{\Xnoj,Y}(g(\Xnoj)) \right|\}
    }\\  \label{eq:3_thm:class_main_MC}
    &:=~ \EE{
    A}
    \end{align}
    where the first equality is by the law of total expectation, the first and the second inequality are simply expanding and rearranging. By construction, $\mu(\Xtil^{(m)},Z),m\in[M]$ are i.i.d. random variables conditioning on $\Xnoj,Y$, then by central limit theorem we have
    $$
    \frac{\sqrt{M}( g^{M}(Z) - g(\Xnoj))}{\varsigma(\Xnoj)} \stackrel{d}{\rightarrow} \gauss{0}{1} 
    $$
    conditioning on $\Xnoj,Y$. Further we obtain the following from the Berry--Esseen bound i.e. Lemma \ref{lem:berry}:
    \begin{equation}\label{eq:4_thm:class_main_MC}
        \left|\Pcmid{\left|\frac{\sqrt{M}|g^{M}(Z) - g(\Xnoj)|}{\varsigma(\Xnoj)}\right|> \sqrt{M}\delta_{\Xnoj,Y}}{\Xnoj,Y} - \widebar{\Phi}(\left|\sqrt{M}\delta_{\Xnoj,Y}\right|) \right| \le \frac{C}{\sqrt{M}}\cdot \frac{\Ec{|\mu^3(X,Z)|}{\Xnoj}}{\varsigma^3(\Xnoj)}
    \end{equation}
        for any $\delta_{\Xnoj,Y}$  when conditioning on $\Xnoj,Y$, where $\widebar{\Phi}(x) = 1 - \Phi(x)$ and $C$ is some constant which does not depend on the distribution of $(Y,X,Z)$. Regarding \eqref{eq:3_thm:class_main_MC}, by considering the event $B:=\{| g^{M}(Z)  - g(\Xnoj)|/\varsigma(\Xnoj) \le  \delta_{\Xnoj,Y}\}$, we can decompose \eqref{eq:3_thm:class_main_MC} into
    \begin{equation}
    \EE{A} =
    \EE{A\indicat{B}} +   \EE{A\indicat{B^{c}}} 
    \end{equation}
    For the first term, we have
    \begin{eqnarray}
    \EE{A\indicat{B}} \nonumber
    &\le&
     \EE{C_{ g^{M}(Z) ,\Xnoj,Y}\left| g^{M}(Z) -  g(\Xnoj)\right|\indicat{B}}  \\ \nonumber
    &=& \EE{\Ecmid{C_{ g^{M}(Z) ,\Xnoj,Y}\left|g^{M}(Z)-  g(\Xnoj)\right|\indicat{B}}{\Xnoj,Y}}\\ \nonumber
    &\le& \EE{C_{\Xnoj,Y}\Ecmid{\left|g^{M}(Z) -  g(\Xnoj)\right|}{\Xnoj,Y}}\\ \label{eq:5_thm:class_main_MC}
    &\le& \EE{C_{\Xnoj,Y}\sqrt{\Ecmid{\left|g^{M}(Z) -  g(\Xnoj)\right|^2}{\Xnoj,Y}}}
    \end{eqnarray}
    where the first inequality is by the definition of $C_{u,\xnoj,y}$, the first equality is from the law of total expectation, the second inequality holds by (a) in Assumption \ref{asmp:class_smooth} and the last inequality holds due to the Cauchy--Schwarz inequality. Remember we have $g^{M}(Z) = \frac{1}{M}\sum_{m=1}^{M}\mu(\Xtil^{(m)},Z)$ where $\mu(\Xtil^{(m)},Z ),m\in [M]$ are i.i.d. random variables with mean $g(\Xnoj)$ when conditioning on $\Xnoj,Y$, hence \eqref{eq:5_thm:class_main_MC} equals 
    \begin{equation} \nonumber
    \EE{C_{\Xnoj,Y}\sqrt{\frac{\varsigma^2(\Xnoj)}{M}}}
    \le \frac{1}{\sqrt{M}}\sqrt{\EE{C^2_{\Xnoj,Y}}}\sqrt{\EE{{\varsigma^2(\Xnoj)}}}  =  O\left(\frac{1}{\sqrt{M}}\right)
    %\le \frac{1}{\sqrt{M}}\sqrt{C_{0}C_{2}}
    % &\le& \frac{1}{\sqrt{M}}\sqrt{\EE{C^2_{\Xnoj,Y}}}\sqrt{\EE{\frac{\Ec{|\mu^3(X)|}{\Xnoj}}{\sigma^3(\Xnoj)}}}\\ \nonumber
    \end{equation}
    where the first inequality is from the Cauchy--Schwarz inequality and the second one holds by (b) and (c) in Assumption \ref{asmp:class_smooth}. Now we have showed 
     \begin{equation} \label{eq:A1_B_bound}
    \EE{A\indicat{B}}= O\left(\frac{1}{\sqrt{M}}\right), 
    \end{equation} % \sqrt{C_{0}C_{2}}/\sqrt{M}$,
    it suffices to prove the same rate for $\EE{A\indicat{B^{c}}}$:
    \begin{eqnarray} \nonumber
    \EE{A\indicat{B^{c}}} 
    &\le& 2~ \PP{B^{c}} \\ \ \nonumber
    &=& 2~ \EE{\Pc{B^{c}}{\Xnoj}} \\ \nonumber
    &=& 2~ \EE{\Pc{\sqrt{M}| g^{M}(Z) - g(\Xnoj)|/\varsigma(\Xnoj) >  \sqrt{M}\delta_{\Xnoj,Y}}{\Xnoj}} \\ \nonumber
    &\le& 2 \EE{
    \widebar{\Phi}(\left|\sqrt{M}\delta_{\Xnoj,Y}\right|) + \frac{C}{\sqrt{M}}\cdot {\frac{\Ec{|\mu^3(X,Z)|}{\Xnoj}}{\varsigma^3(\Xnoj)}}    } \\ \nonumber
    &\le& 2 \EE{
    \frac{2}{\sqrt{2\pi }}\frac{\exp\{-M\delta^2_{\Xnoj,Y}\}}{\sqrt{M}\delta_{\Xnoj,Y}} + \frac{C}{\sqrt{M}}\cdot {\frac{\Ec{|\mu^3(X,Z)|}{\Xnoj}}{\varsigma^3(\Xnoj)}} } 
    \end{eqnarray}
    where the first inequality holds since $ F_{\xnoj,y}(u), G_{\xnoj,y}(u)$ are bounded between $0$ and $1$, the first equality  is due to the law of total expectation, the second equality is from the definition of the event B, the second inequality holds due to \eqref{eq:4_thm:class_main_MC} and the last inequality is a result of Mill's Ratio, see Proposition 2.1.2 in \citet{vershynin2018high}. Under (b) and (c) in Assumption \ref{asmp:class_smooth}, the following holds
     \begin{equation} 
    \label{eq:A1_Bc_bound}
     \EE{A\indicat{B^{c}}}  =  O\left(\frac{1}{\sqrt{M}}\right). %\le  \frac{1}{\sqrt{M}}(4C_{1}/\sqrt{2\pi}+2 CC_{3}).
    \end{equation}
    Finally we prove 
    $$
      \left|\EE{ \indicat{Y(\mu(X,Z) -  g^{M}(Z)  ) < 0}   -\mathbbm{1}_{\left\{
    Y\cdot [\muX - \Ec{\muX}{\Xnoj} ] < 0
    \right\}}} \right| =  O\left(\frac{1}{\sqrt{M}}\right).
    %\le \frac{\sqrt{C_{0}C_{2}} + 4C_{1}/\sqrt{2\pi}+2 CC_{3} }{\sqrt{M}}
    $$
    Regarding the term 
    $$
     \mathrm{II}_1 = \left|\EE{ \indicat{Y(\mu\Xk-   g^{M}(Z)  ) < 0} - \mathbbm{1}_{\{
    Y\cdot [\muXk - \Ec{\muX}{\Xnoj}] < 0}} \right|
    $$
    All of the steps are the same except that the CDF (and its limit) of the conditional distribution $\Xj\mid \Xnoj,Y$ are replaced by those of $\Xj\mid \Xnoj$, i.e. $F_{\xnoj}(u)$ and $G_{\xnoj}(u)$ as defined in \eqref{eq:cdf_notation}. Hence it suffices to notice the following derivations for $F_{\xnoj}(u )$:
    \begin{eqnarray} \nonumber
        F_{\xnoj}(u) = \Pc{U \le u}{\Xnoj =\xnoj} &=&\Epc{Y|\Xnoj=\xnoj}{\Pc{U \le u}{\Xnoj =\xnoj,Y}}{\Xnoj=\xnoj}\\ \nonumber
        &=&\Epc{Y|\Xnoj=\xnoj}{F_{\xnoj,Y}(u)}{\Xnoj=\xnoj},
    \end{eqnarray}
    and similarly for $G_{\xnoj}(u)$. Together with the definition of $C_{u,\xnoj,y}$ and (a) in Assumption \ref{asmp:class_smooth}, the above equations yield
    $$
    \max\{
            \left|F_{\xnoj}(u) - F_{\xnoj}(g(\xnoj)) \right|,
            \left|G_{\xnoj}(u) - G_{\xnoj}(g(\xnoj)) \right|\} \le C_{\xnoj,y}|u-g(\xnoj)|
    $$
    over the region $|u-g(\xnoj)|\le \varsigma(\xnoj) \delta_{\xnoj,y}$. Then the other steps follow as those of proving the term $\mathrm{II}_2$. Finally, we obtain a rate of $ O\left(\frac{1}{\sqrt{M}}\right)$ for $\left|
    \EE{R_{i}^M} -\kappamu/2.
    \right|$.% Remark the involving constants are uniform over the class $(\calP,\calU)$. 

    In the following, we prove the stronger version of \eqref{eq:1_thm:class_main_MC}, i.e.,
      \begin{equation} 
	\label{eq:withrate_class_main_MC}
	   % \inf_{P\in \calP,~ \mu \in \calU}\mathbb{P}_{{P}}\left(\kappamu  \ge \Ljc \right) 
	    \PP{  L^{\alpha}_{n,M,K}(\mu) \le \kappamu  }
	    \ge 1 - \alpha - O\left(\frac{1}{\sqrt{n}}\right),
	  \end{equation}
	when assuming $n^2/M =O(1)$.
% 	Similarly as the proof of Theorem \ref{thm:main}, the above holds when applying Lemma~\ref{lem:berry_t} and verifying the moment conditions. Since $R_i$ are bounded by definition, the moment condition can be easily verified using the similar strategy as in the derivations for \eqref{eq:variance_bound_rate}. 
	For this it suffices to establish the following Berry--Esseen bound:
    \begin{equation}\nonumber
       \Delta:= \sup_{t \in \mathbb{R}} \left|
        \PP{\sqrt{n} 
        \left(
        \frac{\bar{R} - \kappamu/2  }{{s}}
        \right) \le t  
        } - \Phi(t)  
        \right| = O\left(\frac{1}{\sqrt{n}}\right),
    \end{equation}
    where $\bar{R}$ and $s$ are defined similarly as in Algorithm~\ref{alg:class_MOCK} except that $g(Z_i)$ and $R_i$ are replaced with $g^{M}(Z_i)$ and $R_i^{M,K}$, respectively. Notice that 
    \begin{eqnarray} \nonumber
       \Delta 
       &=&  \sup_{t \in \mathbb{R}} \left|
        \PP{\sqrt{n} 
        \left(
        \frac{\bar{R} - \EE{R_i^M}}{s }
        \right) \le t + \sqrt{n}\frac{ (\EE{R_i^M} - \kappamu/2)  }{s}
        } - \Phi(t)  
        \right|\\ \nonumber
        &\le &  \sup_{t \in \mathbb{R}} \left|
        \PP{\sqrt{n} 
        \left(
        \frac{\bar{R} - \EE{R_i^M}}{s}
        \right) \le t 
        } - \Phi(t)  \right|  +  
        \sup_{t \in \mathbb{R}} \left| \Phi\left(t + \sqrt{n}\frac{ (\EE{R_i^M} - \kappamu/2)  }{s}\right) -  \Phi(t) 
        \right|\\ \nonumber
        &:=& \Delta_1 + \Delta_2
    \end{eqnarray}
    Since the first derivative of $ \Phi(t)$ is bounded by $1/\sqrt{2\pi}$ over $\mathbb{R}$, we have
     \begin{eqnarray} \nonumber
    \Delta_2 
    &\le& \frac{\sqrt{n}}{\sqrt{2\pi}} \frac{ |\kappamu/2  - \EE{R_i^M}| }{\sqrt{\Var{R_i^M}}}\cdot ({\sqrt{\Var{R_i^M}}}/{s})
     \end{eqnarray}
    by Taylor expansion. Note that as a result of \eqref{eq:kappamu_gap_unif}, we have
	\begin{equation}\label{eq:gap_withrate_class_MC}
	    \sqrt{n}|\EE{R_{i}^M} - \kappamu/2| = O(1/\sqrt{n}).
	\end{equation}
	Then it suffices to prove $\Delta_1 = O(1/\sqrt{n}) $ and $\Var{R_i^M}>0$ (since $s$ is simply the sample mean estimator of $\Var{R_i^M}$ thus consistent).
	%Remark this is non-trivial since $\Var{R_i^M}$ varies with $M$ grows. 
	$\Delta_1 = O(1/\sqrt{n}) $ holds when applying the triangular array version of the Berry--Esseen bound in Lemma~\ref{lem:berry_t} (note that the result is stated in a way such that the bound clearly applies to the triangular array with i.i.d. rows $\{R_i^{M,K}\}\nsubp$ for each $M$). The only thing we need to deal with is to verify the following uniform moment conditions:
	\begin{enumerate}[(i)]
	    \item $\sup_{M,K}\EE{
	    \left|R_i^{M,K} - \EE{R_i^{M,K}}\right|^3} < \infty$,
	    \item $\inf_{M,K}\Var{R_i^{M,K}} > 0$.
	\end{enumerate}
	where we go back to the original notation $R_i^{M,K}$ from the simplified one $R_i^{M}$ since the above moments do depend on both $M$ and $K$. Since $R_i^{M,K}$ is always bounded, \myrom{1} holds. Regarding \myrom{2}, notice that we have the following
	\begin{align}\nonumber
		    &~\Var{R_{i}^{M,K}} \\ \nonumber
		    &=
		    \EE{\Varc{R_{i}^{M,K}}{\Xinoj,Y_i, \{\Xtil_i^{(m)}\}_{m=1}^M }} +\Var{\Ec{R_{i}^{M,K}}{\Xinoj,Y_i, \{\Xtil_i^{(m)}\}_{m=1}^M  }} \\ \nonumber
		    &\ge
		     \EE{\Varc{R_{i}^{M,K}}{\Xinoj,Y_i, \{\Xtil_i^{(m)}\}_{m=1}^M }} \\ \nonumber
		    &=
		    \EE{\Varcmid{ \frac{1}{K}\sum_{k=1}^{K}\left(\indicat{Y_i(\mu(\Xtil_{i}^{(k)}, Z_i)- g^{M}(Z_i)) < 0} \right)- \indicat{Y_i(\mu(X_i, Z_i)) - g^{M}(Z_i) ) < 0} }{\Xinoj, Y_i, \{\Xtil_i^{(m)}\}_{m=1}^M  }}\\
		    &\ge
		    \EE{\Varcmid{ \indicat{Y_i(\mu(X_i, Z_i)) - g^{M}(Z_i) ) < 0} }{\Xinoj,Y_i, \{\Xtil_i^{(m)}\}_{m=1}^M  }}:=	\sigma^2_M \label{eq:VarRi_MK_thm:class_main}
		  %  &\ge&
		  %  \tau_0
		\end{align}
		where the first equality is due to the law of total expectation, the second equality is by the definition of $R_{i}^{M,K}$, the second inequality holds since $ \{\Xtil_{i}^{(k)}\}_{k=1}^{K} \independent X_i\mid \Xinoj,Y_i, \{\Xtil_i^{(m)}\}_{m=1}^M $ due to the construction of $ \{\Xtil_{i}^{(k)} \}_{k=1}^{K} $ and the variance of first term is non-negative. Before dealing with \eqref{eq:VarRi_MK_thm:class_main}, notice the stated condition
		$$
		\sigma^2_0 := \EE{\Varcmid{ \indicat{Y_i(\mu(X_i, Z_i)) - g(Z_i) ) < 0} }{\Xinoj,Y_i }} >0
		$$
		Thus to establish \myrom{2}, it suffices to show $ 	\sigma^2_M \rightarrow  \sigma^2_0$ as $M\rightarrow \infty$. Recall the derivations in \eqref{eq:3_thm:class_main_MC} for bounding the term $\mathrm{II}_2$, we can similarly bound $	|\sigma^2_M - \sigma^2_0|$ by the following quantity:
		\begin{eqnarray}  \nonumber
			|\sigma^2_M - \sigma^2_0 |
			&\le& \EE{ 3  \max\{
            \left|G_{\Xnoj,Y}( g^{M}(Z)) - G_{\Xnoj,Y}(g(\Xnoj)) \right|,
            \left|F_{\Xnoj,Y}(  g^{M}(Z) ) - F_{\Xnoj,Y}(g(\Xnoj)) \right|\}
               }\\ \nonumber
              &=& 3\EE{A} = 3 (\EE{A\indicat{B}} +   \EE{A\indicat{B^{c}}} ) =  O\left(\frac{1}{\sqrt{M}}\right).
		\end{eqnarray}
		where the last equality holds due to the results \eqref{eq:A1_B_bound} and \eqref{eq:A1_Bc_bound} from previous derivations for the term $\mathrm{II}_2$. Finally we conclude \eqref{eq:withrate_class_main_MC}, which immediately implies a weaker version of the result, i.e.the statement of Theorem \ref{thm:class_main_MC}.
% 		$ F_{\Xinoj}(g_{ij})\indicat{Y_i=1} + (1-G_{\Xinoj}(g_{ij})\indicat{Y_i=-1} \in \sigalg(\Xnoj,Y)$ and the final inequality is by the moment assumption.
    % \eqref{eq:Ri_MK_thm:class_main_MC}
    %We then conclude \eqref{eq:kappamu_gap_unif} and complete the proof of Theorem \ref{thm:class_main_MC}.
    \end{proof}
    
	\section{Co-sufficient floodgate details}
	\label{sec:relax_details}
	\acc{
	The strategy described in Section \ref{sec:relax} is formalized in Algorithm \ref{alg:batch_cond} (under the simplifying assumption that the number of batches, $n_2$, evenly divides the sample size $n$). 
%	We call this procedure \emph{co-sufficient} floodgate
      \begin{algorithm*}[h!]
			\caption{Co-sufficient floodgate} \label{alg:batch_cond}
		\begin{algorithmic}[1]
			\REQUIRE The inputs of Algorithm \ref{alg:MOCK}, a sufficient statistic functional $\calT$, and a batch size $n_{2}$.
			%a batching rule (i.e. $n=n_{1} n_{2}$, with $n_{1}$ batches, each of size $n_{2}$).
			\STATE Let $n_1=n/n_2$ and for $m\in[n_1]$, denote $(\bX_m,\bZ_m) = \{X_i,Z_i\}_{i=(m-1)n_2+1}^{mn_2}$, and let $\bT_m=\calT(\bX_m,\bZ_m)$.
      	            \STATE For $m\in[n_1]$, compute
      	          \begin{align*}
      	          R_m &= \frac{1}{n_2}\sum_{i=(m-1)n_2+1}^{mn_2} Y_i\, (\mu(X_i,Z_i) - \Ec{\mu(X_i,Z_i)}{\bs{Z}_m, \bT_m} ),\\
      	          V_m &= \frac{1}{n_2}{\sum_{i=(m-1)n_2+1}^{mn_2} \Varc{\mu(X_i,Z_i)}{\bs{Z}_m,\bT_m}},
      	          \end{align*}
		  their sample mean $(\bar{R},\bar{V})$, their sample covariance matrix $\hat{\Sigma}$, and $s^2 = \frac{ 1 }{ \bar{V} }\left[ \left(\frac{\bar{R} }{2 \bar{V} }\right)^2 \hat{\Sigma}_{22} +  \hat{\Sigma}_{11} - \frac{\bar{R} }{ \bar{V} } \hat{\Sigma}_{12} \right].$
		    \RETURN Lower confidence bound $\LjT (\mu)=\max\left\{\frac{\bar{R} }{ \sqrt{\bar{V}}}
        - \frac{z_{\alpha}s}{\sqrt{n_1}},\,0\right\}$, with the convention that $0/0=0$. 
		\end{algorithmic}
	\end{algorithm*}
	}
	\subsection{Monte Carlo analogue of co-sufficient floodgate}
	\label{sec:mccf}
	Similarly as in Section~\ref{sec:method}, when the conditional expectations in Algorithm~\ref{alg:batch_cond} do not have closed-form expressions, Monte Carlo provides a general approach: within each batch, we can sample $K$ copies $\tilde{\bs{X}}_m^{(k)}$ of $\bs{X}_m$ from the conditional distribution $\bs{X}_{m}\,|\,\bs{Z}_{m},\bT_{m}$, conditionally independently of $\bs{X}_m$ and $\bs{y}$ and thus
    replace $R_m$ and $V_m$, respectively, by the sample estimators
    \begin{align*}
    (R_m^K,   V^{K}_{m}) &=  \frac{1}{n_2}
    \left(
      \sum_{i=(m-1)n_2+1}^{mn_2} Y_i\left(\mu(X_i,Z_i) - \frac{1}{K}\sum_{k=1}^{K} \mu(\Xtil_i^{(k)} ,Z_i)\right)\right.
    , \\
    &\hspace{1.4cm}\left.\sum_{i=(m-1)n_2+1}^{mn_2} \frac{1}{K-1}\sum_{k=1}^{K}\left( \mu(X_i^{(k)} ,Z_i)-  \frac{1}{K}\sum_{k=1}^{K} \mu(\Xtil_i^{(k)} ,Z_i) \right)^2
    \right)
    % \sum_{i\in\calB_m} Y_i\left(\mu(X_i,Z_i) - \frac{1}{K}\sum_{k=1}^{K} \mu(\Xtil_i^{(k)} ,Z_i)\right),\\ 
    % V^{K}_{m} =  \frac{1}{n_2(K-1)}\sum_{i\in\calB_m} \sum_{k=1}^{K}\left( \mu(X_i^{(k)} ,Z_i)-  \frac{1}{K}\sum_{k=1}^{K} \mu(\Xtil_i^{(k)} ,Z_i) \right)^2. 
    \end{align*}
    % where $\bigcup_{m=1}^{n_{1}}\calB_{m}=[n]$
			 %       with $|\calB_{m}|=n_{2}$.
			        %and denote $( \bs{X}_m, \bs{Z}_m) = \{(X_i,Z_i)\}_{i \in \calB_m}$.
			        %$ \bs{X}_m = \{X_{i}\}_{i \in \calB_m}$,  $ \bs{Z}_m = \{Z_i\}_{i \in \calB_m}$. 
    We defer to future work a proof of validity of the Monte Carlo analogue of co-sufficient floodgate following similar techniques as Theorem~\ref{thm:main_general}.

	          \subsection{Proofs in Appendix~\ref{sec:relax_details}}
       \begin{lemma}\label{lem:gap_term_general}
         Under the moment conditions $\EE{\mu^2(X,Z)},\EE{(\mustar)^2(X,Z)} <\infty$, we can quantify the gap between $\thetamu$ and $\thetamuT$ as below.
	           \begin{equation}\label{eq:gap_term_general}
	                \thetamu - \thetamuT = O\left(\max\{\mathrm{II}(\mu), \mathrm{II}(\mustar)\}\right)
	           \end{equation}
	           where $ \mathrm{II}(\mu) =  \Ep{ \bs{Z}}{ \Varp{ \Tcondlaw}{ \Ec{\mu(X_i,Z_i) }{\condedT}} }$.
	   \end{lemma}
	   When this lemma is used in the proof of Proposition \ref{prop:gap_rate_gaussian} and \ref{prop:gap_rate_DMC}, the natural sufficient statistic and $\thetamuT$ are actually defined based on the batch $\calB_m$ whose sample size is $n_{2}$. We do not carry these in the above notation, but use generic $(\bs{X},\bs{Z})$ instead, where $(\bs{X},\bs{Z})= \{(X_i,Z_i)\}\nsubp$.
	   \begin{proof}[Proof of Lemma~\ref{lem:gap_term_general}]\label{pf:lem:gap_term_general}
% 	  {\color{red}{
% 	  Recall the derivations for Lemma \ref{lem:max} in Appendix \ref{pf:lem:max}, we can simplify $\thetamu$ into the following}}
% 	  \begin{equation}\nonumber
% 	     \thetamu =  \frac{\EE{\hstar(W_i) h(W_i)}}{\sqrt{\EE{h^2(W_i)}}},
% 	  \end{equation}
% 	  where $\hstar(W_i):=\mustar(W_i) - \Ec{\mustar(W_i) }{Z_i}$ and $h(W_i)$ is similarly defined. Analogous simplifications for $\thetamuT$ yields 

% 	  where $h^{\calT}(X_i):= \mustar(W_i)- \Ec{\mustar(W_i) }{ \condedT}$. The derivation basically makes use of the following facts: 
% 	  \begin{itemize}
% 	    \item $\tilde{X}_{ij}^{\calT}\mid \condedT \deq X_{i} \mid \condedT$
% 	    \item $\tilde{X}_{ij}^{\calT}\independent (X_i,Y_i) \mid \condedT$
% 	    \item $\EE{\eps(X_i,Y_i)(
% 	     \mu(X_{i})-{\mu}(\tilde{X}_{i}^{\calT}))  }=0$
% 	    \item $\Ec{h^{\calT}(X_i)}{\bs{Z}}=0$.
% 	  \end{itemize}
	  Recall the definition of $\thetamu$ and $ \thetamuT$,
	    \begin{eqnarray}\label{eq:f_mu}
	   \thetamu 
	    &=& 
        \frac{\EE{\covc{\mustar(X,Z)}{\muX}{Z}}}{\sqrt{\EE{\varc{\muX}{Z}}}}, \\ \label{eq:fT_mu}
        \thetamuT
        &=&
        \frac{\EE{\covc{\mustar(X_i,Z_i)}{\mu(X_i,Z_i)}{\bs{Z},\bs{T}}}}{\sqrt{\EE{\varc{\mu(X_i,Z_i)}{\bs{Z},\bs{T}}}}}, 
        % \frac{
        % \EE{(Y_i - \mu(\XikT,Z_i))^2} -  \EE{(Y_i - \mu(X_i,Z_i))^2} 
        % }{
        % \sqrt{2 \EE{(\mu(X_i,Z_i) - \mu(\XikT,Z_i))^2 }}
        % }
        % {\theta}_{j}^{\calT}(\mu) :=  \frac{\mathbb{E}[|Y_{i}-{\mu}(\tilde{X}_{i}^{\calT})|^2 -  |Y_{i}-{\mu}(X_{i})|^2] }{\sqrt{2 \mathbb{E}|\mu(X_{i})-{\mu}(\tilde{X}_{i}^{\calT})|^2}} 
        \end{eqnarray}
        then denote $W_i = (X_i,Z_i),~ h(W_i):=\mu(W_i) - \Ec{\mu(W_i) }{Z_i}, h^{\calT}(W_i):= \mustar(W_i)- \Ec{\mustar(W_i) }{ \condedT}$ and assume $\EE{h^2(W_i)} = 1$ without loss of generality. First notice a simple fact $|\frac{a}{b} - \frac{c}{d}|
	  = \frac{|ad -bc|}{bd} = \frac{|ad -cd +cd -bc|}{bd}
	  \le \frac{|a-c|}{b} + \frac{c|b-d|}{bd}$ for $a,b,c,d>0$, then let the numerator and denominator of $ \thetamu $ in \eqref{eq:f_mu} to be $a,b$ respectively (similarly denote $c,d$ for $\thetamuT$ in \eqref{eq:fT_mu}).
% 	  $$ 
% 	   \thetamu  = \frac{\EE{\hstar(X_i) h(X_i)}}{\sqrt{\EE{h^2(W_i)}}} :=\frac{a}{b}, ~~ \thetamuT = \frac{\EE{\hstar(X_i) h^{\calT}(X_i)}}{\sqrt{\EE{(h^{\calT})^2(W_i)}}}:= \frac{c}{d}.
% 	  $$
	  And we have
	  $$
	  \max\{\frac{1}{b},~\frac{c}{bd}\} \le 1 + \thetamuT \le 1+ \thetamusT \le 1 + \thetamus \le 1 + \EE{(\mustar)^2(X,Z)} < \infty,
	  $$
	  hence it suffices to bound $|a-c|$ and $|b-d|$. First we have the following

	   \begin{eqnarray}
     a-c 
     &=& \EE{\Covc{ \mustar(W_i) }{\mu(W_i)}{\bs{Z}  }} -    
     \EE{\Covc{\mustar(W_i)}{\mu(W_i)}{\condedT }}\\ \nonumber
     &=& \EE{\Covc{ \Ec{\mustar(W_i) }{\condedT}
     }{ \Ec{\mu(W_i) }{\condedT}}{\bs{Z} } } \\ \nonumber
     &=& \Ep{\bs{Z}}{\Covp{\Tcondlaw}{
     \Ec{\mustar(W_i) }{\condedT}}{
     \Ec{\mu(W_i) }{\condedT}}}.
    %  &=& \Ep{\condedT}{\Covc{\hstar(X_i) }{h(X_i) -  h^{\calT}(X_i) }{ \condedT} } + \\ \nonumber
    %  &~~& \Cov{ \Ec{ \hstar(X_i)}{\condedT}
    %  }{\Ec{h(X_i) -  h^{\calT}(X_i)}{\condedT} } \\ \nonumber
    %  & =& \Ep{\condedT}{\Covc{\hstar(X_i) }{h(X_i)  }{ \condedT} } \\ \nonumber
    %  &\le &  \sqrt{\Ep{\condedT}{\Varc{\hstar(X_i) }{ \condedT}}}   \sqrt{\Ep{\condedT}{\Varc{h(X_i) }{ \condedT}}}  
     \end{eqnarray}
     where the first equality holds due to the independence among $i.i.d.$ samples $ (\bs{X},\bs{Z})= \{(X_i,Z_i)\}\nsubp$.
     For the second equality, we apply the law of total covariance to the covariance term $\Covc{ \mustar(W_i) }{\mu(W_i)}{\bs{Z}  }$ then cancel out the second term of the first line, leading to the term in the second line. Finally we spell out the randomness of the expectation and covariance through explicit subscripts in the last inequality. They by applying Cauchy--Schwarz inequality, we obtain
	 \begin{equation} \label{eq:|a-c|}
	     	|a-c| \le \sqrt{\Ep{\bs{Z}}{\Varp{\Tcondlaw}{ \Ec{\mustar(W_i) }{\condedT}}}} \sqrt{\Ep{\bs{Z}}{\Varp{\Tcondlaw}{ \Ec{\mu(W_i) }{\condedT}}}} 
	  \end{equation}
	 Regarding the term $|b-d|$, we have
	 \begin{eqnarray}\label{eq:|b-d|}
	    |b - d | \nonumber
	    &=& \left| \sqrt{\EE{h^2(W_i)}} -  \sqrt{\EE{(h^{\calT})^2(W_i)}} \right| \\ \nonumber
	    &=& \frac{\left| \EE{h^2(W_i)} - \EE{(h^{\calT})^2(W_i)} \right|
	    }{ \sqrt{\EE{h^2(W_i)}} +  \sqrt{ \EE{(h^{\calT})^2(W_i)}}}  \\ \nonumber
	    &\le & \frac{\left| \EE{h^2(W_i)} - \EE{(h^{\calT})^2(W_i)} \right|
	    }{ \sqrt{\EE{h^2(W_i)}}}  \\ \nonumber
	    &\le& \EE{ \Varc{\mu(W_i)}{\bs{Z}
	    }} - \EE{ \Varc{\mu(W_i)}{\condedT
	    }}  \\
	    &=& \Ep{ \bs{Z}}{ \Varp{ \Tcondlaw}{ \Ec{\mu(W_i) }{\condedT}} }
	 \end{eqnarray}
	 where we use the assumption $\EE{h^2(W_i)}=1$ and the definition of $h, h^{\calT}$ in the second inequality. The last equality holds as a result of applying the law of total variance to the variance term $\Varc{\mu(W_i)}{\bs{Z}}$ then getting the second term of line $4$ cancelled out. Finally, combining \eqref{eq:|a-c|} and \eqref{eq:|b-d|} establishes the bound in \eqref{eq:gap_term_general}.
	   \end{proof}
	   	   \subsubsection{Proposition~\ref{prop:gap_rate_gaussian}}
	  \begin{proof}[Proof of Proposition~\ref{prop:gap_rate_gaussian}]\label{pf:prop:gap_rate_gaussian}
	  Throughout the proof, the natural sufficient statistic and $\thetamuT$ are defined based on the batch $\calB_m$ whose sample size is $n_{2}$. But we will abbreviate the notation dependence on it for simplicity and use a generic $n$ instead of $n_{2}$ to avoid carrying too many subscripts, without causing any confusion. Now we present a roadmap of this proof.
	  \begin{enumerate}[(i)]
	      \item due to Lemma \ref{lem:gap_term_general}, it suffices to bound the term $\mathrm{II}(\mu)$, $\mathrm{II}(\mustar)$ in \eqref{eq:gap_term_general}.
	      \item we bound $\mathrm{II}(\mu)$, $\mathrm{II}(\mustar)$ with the same strategy. Specifically, we will show 
	      $$
	      \mathrm{II}(\mu) = O\left(\Ep{\Xinoj}{
	      \Ep{F}{\mu^2(W_i)}
	      {\Ec{h_{ii}}{\Xinoj}}
	      } \right)$$
	      and similarly for $\mathrm{II}(\mustar)$ under the stated model, where $F$ denotes the conditional distribution of $X_{i}|\bs{Z}$, and $h_{ii}$ is the $i$th diagonal term of the hat matrix $\bH$, which is defined later. This terminology comes from the fact that we can treat $X_j$ as response variable, $(1,\Xnoj)$ as predictors, the natural sufficient statistic for this low dimensional multivariate Gaussian distribution is equivalent to the OLS estimator.
	      \item Regarding the term $\Ec{h_{ii}}{\Xinoj}$ above, we can carefully bound it by $ 1/(n-1) + \Ec{\bm{\Xi}}{Z_i}$, where $\bm{\Xi}$ is defined in \eqref{eq:Xi_def}.
	      \item Simply expanding $\Ec{\bm{\Xi}}{Z_i}$ into three terms: $\mathrm{III}_1,\mathrm{III}_2,\mathrm{III}_3$, which are defined in \eqref{eq:III_1}, \eqref{eq:III_2} and \eqref{eq:III_2}, we will show $\mathrm{III}_2=0$ and figure out the stochastic representation of $\mathrm{III}_1,\mathrm{III}_3$, which turns out to be related to chi-squared, Wishart and inverse-Wishart random variables.
	      \item Cauchy--Schwarz inequalities together with some properties of those random variables (chi-squared, Wishart and inverse-Wishart) and the stated moment conditions finally gives us the result in \eqref{eq:gap_rate_gaussian}.
	  \end{enumerate}
% 	  Due to the result in Lemma \ref{lem:gap_term_general}, it suffices to bound the following term
% 	  then if the gap terms for the numerators and denominators in $\thetamu$ and $\thetamuT$ can be bounded as $O(1/n)$, we immediately have $\thetamu - \thetamuT\le O(1/n) + \frac{O(1/n)}{1-O(1/n)}\EE{\hstar h} = O(1/n)$. 
% 	  The bounding strategy will be the same for the numerators and denominators. Here we prove the case of numerators and denote them by $\eta_j(\mu)$ and $\eta_{j}^{\calT}(\mu)$ respectively. First, we have the following derivations for the gap term: 
	  %the gap term between the unnormalized excess prediction error
	 Having proved Lemma \ref{lem:gap_term_general}, now we directly start with step \myrom{2}. Notice the following
	  \begin{eqnarray}\nonumber
% 	  \eta_j(\mu) - \eta_{j}^{\calT}(\mu)
%         b-d
% 	  &=& \mathbb{E}\mathrm{Var}_{G}(\mu(X_{i})|\{ Z_i\}) - \mathbb{E}\mathrm{Var}_{F}(\mu(X_{i})|\bs{Z}) \\ \nonumber
      \mathrm{II}(\mu)
      &=&  
    \Ep{ \bs{Z}}{ \Varp{ \Tcondlaw}{ \Ec{\mu(W_i) }{\condedT}}} \\  \nonumber
      &=& \Ep{ \bs{Z}}{ \Ep{ \Tcondlaw}{ 
      ( \Ep{F}{\mu(W_i)} - \Ep{F_{\bT}}{\mu(W_i)}  )^2
      }}\\ \nonumber
% 	  &=& \mathbb{E}_{\bs{Z}}\mathbb{E}_{\bm{T}|\bs{Z}}(\mathbb{E}_{G}\mu(X_{i}) - \mathbb{E}_{F}\mu(X_{i}))^{2} \\ \nonumber
      &=& \Ep{ \bs{Z}}{ \Varp{F}{\mu(W_i)} \Ep{ \Tcondlaw}{ 
      \frac{( \Ep{F}{\mu(W_i)} - \Ep{F_{\bT}}{\mu(W_i)}  )^2}{\Varp{F}{\mu(W_i)} }
      }} \\  
% 	  &=& \mathbb{E}\left\{\mathrm{Var}_{G}(\mu)\mathbb{E}_{\bm{T}|\bs{Z}}\left[
% 	  	\frac{(\mathbb{E}_{G}\mu(X_{i}) - \mathbb{E}_{F}\mu(X_{i}))^{2}}{\mathrm{Var}_{G}(\mu)}
% 	   \right]\right\} \\ \nonumber
% 	  &\le&   \mathbb{E}[\mathrm{Var}_{G}(\mu) \mathbb{E}_{\bm{T}|\bs{Z}}
% 	  	\chi^{2}(F\|G)]
	  &\le&  \Ep{ \bs{Z}}{ \Varp{F}{\mu(W_i)} \min \left \{ \Ep{ \Tcondlaw}{ 
        \chi^2(F_{\bT}\|F)
      },2\right\}
      }	 \label{eq:term_II}
	  \end{eqnarray}
	  where the second equality is just rewriting the conditional variance, with $F$ denoting the conditional distribution $X_{i}|\bs{Z}$ and $F_{\bT}$ denoting the conditional distribution $X_{i}|\bs{Z},\bT$. Here we abbreviate the subscript dependence on $i$ for notation simplicity. The third equality holds since $\Varp{F}{\mu(W_i)} \in \sigalg(\bs{Z})$. Regarding the last inequality, we make use of the variational representation of $\chi^2$-divergence:
	  \begin{equation} \nonumber
		\chi^{2}(P\|Q) = \sup_{\mu} \frac{(\mathbb{E}_{P}(\mu) - \mathbb{E}_{Q}(\mu))^{2}}{\mathrm{Var}_{Q}(\mu)}
	  \end{equation}
	  and the fact that
	  \begin{eqnarray} \nonumber
	  &~~& \Ep{ \Tcondlaw}{ 
      \frac{( \Ep{F}{\mu(W_i)} - \Ep{F_{\bT}}{\mu(W_i)}  )^2}{\Varp{F}{\mu(W_i)} }}\\ \nonumber
       &\le& \frac{
      \Ep{ \Tcondlaw}{  \Ep{F}{\mu^2(W_i)} } + \Ep{ \Tcondlaw}{ \Ep{F_{\bT}}{\mu^2(W_i)}} - 2  \Ep{ \Tcondlaw}{ \Ep{F_{\bT}}{\mu(W_i)}\Ep{F}{\mu(W_i)}
      }
      }{ \Varp{F}{\mu(W_i)} 
      } \\ \nonumber
      &=& \frac{
      \Ep{F}{\mu^2(W_i)} +  \Ep{F}{\mu^2(W_i)} - 2(\Ep{F}{\mu(W_i)})^2 
      }{ \Varp{F}{\mu(W_i)} 
      }  \\ \nonumber
      &= & 
      \frac{
      2\Varp{F}{\mu(W_i)}
      }{ \Varp{F}{\mu(W_i)} 
      }  = 2
	  \end{eqnarray}
	  where the first inequality is from expanding the quadratic term and the fact $ (\Ep{F}{\mu(W_i)})^2 \le\Ep{F}{\mu^2(W_i)}$, $(\Ep{F_{\bT}}{\mu(W_i)})^2 \le\Ep{F_{\bT}}{\mu^2(W_i)}$, the first equality holds as a result of the tower property of conditional expectation and $\Ep{F}{\mu(W_i)} \in \sigalg(\bs{Z})$.
% 	  Below we present a result on the $\chi^{2}$-divergence between two Gaussian random variables. 
% 	  \begin{lemma}\label{lem:chisq_div_gaussians}
% 	  The $\chi^{2}$-divergence between $P:\gauss{a_1}{\sigma_1^2}$ and $Q:\gauss{a_2}{\sigma_2^2}$ equals the following whenever $2\sigma_{2}^{2}>\sigma_{1}^{2}>0$:
% 	  \begin{equation}
% 	  \frac{1}{2}\left[
% 	  \frac{\sigma^2_{2}}{\sigma_1\sqrt{2\sigma_{2}^{2} - \sigma_{1}^{2}}}\exp\left\{\frac{(a_{1}-a_{2})^{2}}{2\sigma_{2}^{2} - \sigma_{1}^{2}}\right\} -1 
% 	  \right]
% 	  \nonumber
% 	  \end{equation}
% 	  \end{lemma}
	   Denote $u_{i} = (1,Z_i)^{\top}$ and the following $n$ by $p$ matrix by $\bU$:
	  \begin{equation} \label{eq:def_W}
	      \bU = \left(
	      \begin{array}{c}
	           u_{1}^{\top}  \\
	           \vdots \\
	           u_{n}^{\top}
	      \end{array}
	      \right) = (\bm{1}, \bZ)
	  \end{equation}
	  Recall that the sufficient statistic (here we ignore the batching index)
	  $$
	  \bT = (\sum_{i\in[n]}X_{i}, \sum_{i\in[n]}X_{i}Z_i)= \bU^{\top} \bX,
	  $$
	  under the stated multivariate Gaussian model, we know $ \bX \mid \bZ \sim \gauss{\bU \gamma }{\sigma^2 \bs{I}_n}$, then the conditional distribution of $(X_{i},\bT)\mid \bZ$ can be specified as below
    \begin{equation} \label{eq:X_T_joint_gaussian}
	  \left(
	  \begin{array}{c}
	       X_{i} \\
	       \bT 
	  \end{array}
	  \right)
	  \sim 
	  \mathcal{N}
	  	  \left(
	  	  \left[
	  \begin{array}{c}
	       (1,Z_i)\gamma \\
	       \bU^{\top}\bU \gamma
	  \end{array}
	  \right],
	  	  \sigma^2 \left[
	  \begin{array}{cc}
	       1 & e_i^{\top} \bU \\
	       \bU^{\top} e_i^{\top} &  \bU^{\top}\bU
	  \end{array}
	  \right]
	  \right)
	 \end{equation}
	  where $e_i \in \mathbb{R}^n$, $(e_1,\cdots, e_{n})$ forms the standard orthogonal basis. Noticing the above joint distribution is multivariate Gaussian, we can immediately derive the conditional distribution as below,
	  $$
	   X_{i} \mid \bs{Z}, \bm{T} \sim \gauss{e_{i}^{\top}\bU(\bU^{\top} \bU)^{-1} \bU^{\top} \bm{X}}{\sigma^{2} (1-   e_i^{\top}\bU(\bU^{\top} \bU)^{-1} \bU^{\top} e_i )}.
	  $$
	  Denote $\bH =  \bU(\bU^{\top} \bU)^{-1} \bU^{\top}$, which is the ``hat" matrix. Now we compactly write down the following two conditional distributions:
	   \begin{eqnarray} \nonumber
	  &F_{\bT}&: X_{i} \mid \bs{Z}, \bm{T} \sim \gauss{e_{i}^{\top}\bH \bm{X}}{\sigma^{2}(1-h_{ii})} \\ \nonumber 
	  &F&: X_{i} \mid \bs{Z} \sim \gauss{(1,Z_i)\gamma}{\sigma^{2}}
	   \end{eqnarray}
	  Note the sufficient statistic $\bT$ is equivalent to 
	  $$
	  \hat{\gamma}^{OLS} = (\bU^{\top} \bU)^{-1} \bU^{\top} \bX 
	  $$
	  whenever $\bU^{\top} \bU$ is nonsingular. Here $\hat{\gamma}^{OLS}$ is the OLS estimator for $\gamma$ (when treating $X$ as response variable, $(1,\Xnoj)$ as predictors). Simply, we have
	  $$
	   \hat{\gamma}^{OLS} \sim \gauss{\gamma }{ \sigma^2 (\bU^{\top} \bU)^{-1} }
	  $$
	   Now we are ready to calculate $\chi^{2}(F_{\bT}\|F)$. First,
	  \begin{eqnarray}\nonumber
	  e_{i}^{\top}\bH \bX - (1,Z_i)\gamma 
	  &=& e_{i}^{\top}\bU \hat{\gamma}^{OLS} - (1,Z_i)\gamma 
	  \\
	  &=& e_{i}^{\top}\bU (\hat{\gamma}^{OLS} - \gamma) \sim \calN(0,\sigma^{2}h_{ii})  \label{eq:mean_diff_dist}
	  \end{eqnarray} 
% 	  where $\hat{\gamma}^{OLS} = (\bm{X}_{\noj}^{\top}\bZ)^{-1}\bZ^{\top}\bX$ is the OLS estimator for $\gamma$.
	  Since $2\sigma^{2}>\sigma^{2}(1-h_{ii})$, applying Lemma \ref{lem:chisq_div_gaussians} yields the following
	  \begin{eqnarray}\nonumber
	  \chi^{2}(F_{\bT}\|F)
	  &= & \frac{1}{2}\left[\frac{1}{\sqrt{1-h^2_{ii}}} \exp{\left\{
	  \frac{(e_{i}^{\top} \bH \bm{X} - (1,Z_i)\gamma)^{2}}{\sigma^{2}(1+h_{ii})}
	  \right\}} -1  \right]\\ \nonumber
	   &\le & \frac{1}{\sqrt{1-h_{ii}}}  \exp{\left\{
	  \frac{(e_{i}^{\top} \bH \bm{X} - (1,Z_i)\gamma)^{2}}{\sigma^{2}(1+h_{ii})}
	  \right\}} -1 \\ 
	  &=& \frac{1}{\sqrt{1-h_{ii}}}
	  \exp{\left\{
	  \frac{h_{ii}G^2}{1+h_{ii}}
	  \right\}} - 1  \label{eq:chisq_upperbound}
	  \end{eqnarray}
	  where $G \sim\calN(0,1)$ is independent from $\bX$ and the last equality holds due to \eqref{eq:mean_diff_dist}. Plugin \eqref{eq:chisq_upperbound} back to \eqref{eq:term_II}, we have
    \begin{eqnarray}
        \mathrm{II}(\mu)  \nonumber
        &\le&  \Ep{ \bs{Z}}{ \Varp{F}{\mu(W_i)} \min \left \{ \Ep{ \Tcondlaw}{ 
        \chi^2(F_{\bT}\|F)
        },2\right\}
        }  \\ \nonumber
        &\le & \Ep{ \bs{Z}}{ \Varp{F}{\mu(W_i)}
        \min\left\{
        \Ep{ \Tcondlaw}{\frac{1}{\sqrt{1-h_{ii}}} 
          \exp{\left\{
	    \frac{h_{ii}G^2}{1+h_{ii}}
	    \right\}} - 1}, 2\right\}
        } 
    \end{eqnarray}
	  Note the moment generating function for $\chi^{2}_{1}$ random variable is $\frac{1}{\sqrt{1-2t}}$ when $t<1/2$. Since the expectation of $\exp{\left\{
	    \frac{h_{ii}G^2}{1+h_{ii}}
	    \right\}}$ does not always exist, we consider two events $E$ and $E^c$ such that conditional on the event $E$, the expectation exists and the probability of event $E^c$ is small. More specifically, define the event $E=\{h_{ii} < \frac{1}{2} \}$, which implies 
        \begin{eqnarray}\label{eq:bound_expected_chisq} \nonumber
	   \Ep{ \Tcondlaw}{ \frac{1}{\sqrt{1-h_{ii}}}
        \exp{\left\{
	    \frac{h_{ii}G^2}{1+h_{ii}}
	  \right\}}} - 1   \nonumber
	  &= & \frac{1}{\sqrt{1-h_{ii}}\sqrt{1 - {2h_{ii}}/{(1+ h_{ii})}}} - 1\\ \nonumber
	  &= & \frac{\sqrt{1+h_{ii}} }{{1-h_{ii}}}-1 \\ \nonumber 
	  &\le& \frac{1+h_{ii}}{1-h_{ii}} -1 
    %   &= & \frac{2h_{ii}}{(\sqrt{1+h_{ii}} + \sqrt{1-h_{ii}}) \sqrt{1-h_{ii}} } 
      \\ \nonumber
	  &\le& 4 h_{ii} 
	  \end{eqnarray}
	  hence we can bound $\mathrm{II}(\mu)$ by the summation of the following two terms:
	   $$
	    \mathrm{II}_{1}:=  \Ep{\bs{Z}}{
	      \Varp{F}{\mu(W_i)} \indicat{E} \cdot 4h_{ii}},~~
	     \mathrm{II}_{2}:=  \Ep{\bs{Z}}{
	      \Varp{F}{\mu(W_i)} \indicat{E^c} \cdot 2
	      }
	   $$
	   Regarding $ \mathrm{II}_{1}$, the following holds:
	  	 \begin{eqnarray}\nonumber
	   %   \mathrm{II}_{1} \preceq
	   %   \Ep{\bs{Z}}{
	   %   \Varp{G}{\mu(W_i)} \indicat{E^c}
	   %   }
	   %   \bb{E}_{\bs{Z}}[\Ep{G}{\mu^2}E^{c}] 
	   %   \le
	   %   \bb{E}_{Z_i}[\bb{E}(\mu^2(X_{i})|Z_i)\bb{E}(h_{ii}|Z_i)]/\varepsilon
	   %   \\
	    \mathrm{II}_{1}  
	     \le \nonumber 4~
	      \Ep{\Xinoj}{
	      \Ep{F}{\mu^2(W_i)}
	      {\Ec{h_{ii}}{\Xinoj}}
	      },
	   %  \bb{E}_{\bs{Z}}[\Varp{G}{\mu}h_{ii}] \le \bb{E}_{Z_i}[\bb{E}(\mu^2(X_{i})|Z_i)\bb{E}(h_{ii}|Z_i)]
	  \end{eqnarray}
	  where we apply the tower property of conditional expectation and $ \Varp{F}{\mu(W_i)} \le  \Ep{F}{\mu^2(W_i)} \in \sigalg(\Xinoj) $
	   % $h_{ii} < 1/2$ and $  \frac{1}{\sqrt{1-2h_{ii}}}-1 \le 8 h_{ii}$,  
	   % for some constant c (for example $\varepsilon = 3/8$ and $c=8$),
% 	    hence we have, conditional on $E$
% 	  \begin{equation}\label{eq:bound_expected_chisq}
% 	   \Ep{ \Tcondlaw}{ 
%         \exp{\left\{
% 	\frac{h_{ii}Z^2}{1+h_{ii}}
% 	  \right\}}} - 1  = \sqrt{}
% 	  \preceq h_{ii} 
% 	  \end{equation}
% 	  And by Markov's inequality, we have
% 	  $$
% 	  \Pc{E^{c}}{\Xinoj} \le \frac{\Ec{h_{ii}}{\Xinoj }}{\varepsilon}
% 	  $$
% 	  And we have $\bb{P}(E^c | Z_i)\le \frac{\bb{E}(h_{ii}|Z_i)}{\varepsilon}$ by Markov's inequality, 
% 	  From \eqref{eq:term_II} and separately consider two cases $E$ and $E^c$, we can bound $\mathrm{II}$ by the following two terms: $\mathrm{II}_{1}$ and $ \mathrm{II}_{2}$.
    Regarding $\mathrm{II}_{2}$, we have
	  \begin{eqnarray} \nonumber
	      \mathrm{II}_{2} 
	      &=&
	      2~\Ep{\bs{Z}}{
	      \Varp{F}{\mu(W_i)} \indicat{E^c}
	      }\\ \nonumber
	   %   \bb{E}_{\bs{Z}}[\Ep{G}{\mu^2}E^{c}] 
          &=&
	      2~\Ep{\bs{Z}}{
	      \Varp{F}{\mu(W_i)} \Ec{ \indicat{E^c}}{\Xinoj}
	      }\\ \nonumber
	      &\le&
	       2~ \Ep{\Xinoj}{
	      \Ep{F}{\mu^2(W_i)} \Pc{h_{ii} \ge \frac{1}{2}}{\Xinoj} 
	      } \\ \nonumber
	      &\le& 4 ~
	      \Ep{\Xinoj}{
	      \Ep{F}{\mu^2(W_i)}
	      {\Ec{h_{ii}}{\Xinoj}}
	      }
	   %   \bb{E}_{Z_i}[\bb{E}(\mu^2(X_{i})|Z_i)\bb{E}(h_{ii}|Z_i)]/\varepsilon
	   %   \\
	  \end{eqnarray}
	  where the second equality comes from the tower property of conditional expectation and $ \Varp{F}{\mu(W_i)} \in \sigalg(\Xinoj) $ and the last inequality holds due to Markov's inequality. Now we can compactly write down the following bound for $\mathrm{II}(\mu)$,
	   \begin{equation} \label{eq:II_result}
	  \mathrm{II}(\mu) \le \mathrm{II}_1 + \mathrm{II}_2 \le 8 ~
	      \Ep{\Xinoj}{
	      \Ep{F}{\mu^2(W_i)}
	      {\Ec{h_{ii}}{\Xinoj}}
	      },
	  \end{equation}
	  Similarly we obtain $\mathrm{II}(\mustar)= O\left( \Ep{\Xinoj}{
	      \Ep{F}{(\mustar)^2(W_i)}
	      {\Ec{h_{ii}}{\Xinoj}}
	      }\right)$. Now we proceed step \myrom{3}, i.e. calculating $\Ec{h_{ii}}{Z_i}$.
% 	  Remark here $\epsilon$ is bounded away from $0$.
% 	  where the first line holds due to \eqref{eq:bound_expected_chisq} and the second line again comes from the tower property of total expectation. 
% 	  hence it suffices to calculate $\Ec{h_{ii}}{Z_i}$. 
	  Notice $h_{ii}$ is the $i$th diagonal term of the ``hat" matrix, which involves $\{w_{i}\}\nsubp$. In order to bound the conditional expectation of $h_{ii}$ given $\Xinoj$ in a sharp way, we carefully expand $h_{ii}$ and try to get $w_{i}$ separated from $\{w_m\}_{m\ne i}$. 
% 	  Denote $w_{i} = (1,Z_i)^{\top}$ and the following $n$ by $p$ matrix by $\bU$:
% 	  \begin{equation}
% 	      \bU = \left(
% 	      \begin{array}{c}
% 	           w_{1}^{\top}  \\
% 	           \vdots \\
% 	           w_{n}^{\top}
% 	      \end{array}
% 	      \right)
% 	  \end{equation}
% 	  then we have $\bH =  \bU(\bU^{\top} \bU)^{-1} \bU^{\top}$. 
    Recall the definition of $\bU = (\bm{1}, \bZ)$ in \eqref{eq:def_W}, we can rewrite 
    $$
    \bU^{\top} \bU = \sum_{m\ne i}u_m u_m^{\top} + u_i u_i^{\top},~~~ \bA:= \sum_{m\ne i}u_m u_m^{\top}
    $$
    Note that $h_{ii}=u_i^{\top} (\bU^{\top} \bU)^{-1}u_i$ since $\bH =\bU (\bU^{\top} \bU)^{-1} \bU^{\top} $, hence we have
	  \begin{equation} \nonumber
	      {h_{ii}} = u_i^{\top}(\bA + u_i u_i^{\top})^{-1} u_i
	  \end{equation}
	  As $n>p$, $\bA$ is almost surely positive definite thus invertible, then applying Sherman--Morrison formula to $\bA$ and $u_i u_i^{\top}$ yields the following
	  \begin{equation} \label{eq:hii_simplebound}
	      h_{ii} = u_i^{\top} \bA^{-1} u_i - \frac{( u_i^{\top} \bA^{-1} u_i )^{2}}{1 +  u_i^{\top} \bA^{-1} u_i } \le  u_i^{\top} \bA^{-1} u_i.
	  \end{equation}
     Since $\bA$ also involves the unit vector $\bm{1}_{n-1}$, it is easier when we first project $\bZ_{\noi}$ on $\bm{1}_{n-1}$ then work with the orthogonal complement.	Bearing this idea in mind, we denote $\bOme = (\bm{1}_{n-1}, \bZ_{\noi})$ which is a $n-1$ by $p$ matrix, then rewrite $\bA$ as
	  \begin{equation} \nonumber
	     \bA = \bOme^{\top}\bOme = \left(
	      \begin{array}{cc}
	        \bm{1}_{n-1}^{\top} \bm{1}_{n-1}  &  \bm{1}_{n-1}^{\top} \bZ_{\noi}\\
	         \bZ_{\noi}^{\top} \bm{1}_{n-1}   & \bZ_{\noi}^{\top} \bZ_{\noi}
	      \end{array}
	      \right)
	  \end{equation}
	  	where $\Ib_{n-1}$ is the $(n-1)$ dimensional identity matrix. Denote
	   \begin{equation} \label{eq:Omega_Gamma}
	   \widebar{\bZ_{\noi}} :=  \frac{1}{n-1}\sum_{m\ne i}Z_{m}=\frac{1}{n-1} \bm{1}_{n-1}^{\top}\bZ_{\noi}~~~~
	        \bUam:= \left(
	        \begin{array}{cc}
	            1 &   - \widebar{\bZ_{\noi}}  \\
	            \bm{0} & \Ib_{n-1}
	        \end{array}
	        \right),
	   \end{equation}
	   	we have 
	    \begin{eqnarray}\nonumber
	   	\bOme \bUam  = (\bm{1}_{n-1}, \bZ_{\noi})\bUam
	   	&=& (\bm{1}_{n-1}, \bZ_{\noi} - \bm{1}_{n-1} \widebar{\bZ_{\noi}})\\ \nonumber
	   	&=& (\bm{1}_{n-1}, (\Ib_{n-1}-\projmat)\bZ_{\noi}).
	    \end{eqnarray}
	   	where $\projmat = \bm{1}_{n-1} \bm{1}_{n-1}^{\top}/(n-1)$ is the projection matrix onto $\bm{1}_{n-1}$. Then we immediately have
	   	\begin{equation}\nonumber
	   	  (\bOme \bUam)^{\top}\bOme \bUam = \left(
	   	  \begin{array}{cc}
	   	     n-1 & \bm{0} \\
	   	    \bm{0}    & \bZ_{\noi}^{\top} (\Ib_{n-1} - \projmat) \bZ_{\noi}
	   	  \end{array}
	   	  \right)
	   	\end{equation}
	    since $\projmat \bm{1}_{n-1} =  \bm{1}_{n-1}, (\Ib_{n-1}-\projmat) \bm{1}_{n-1} = \bm{0}$ and 
	    \begin{equation}
        \label{eq:w_Gamma}
	    u_i^{\top} \bUam = (1,Z_i)  \bUam = (1, Z_i -\widebar{\bZ_{\noi}} ).
	    \end{equation}
	    Combining \eqref{eq:Omega_Gamma} with \eqref{eq:w_Gamma} yields the following 
	    \begin{eqnarray} \nonumber
	        u_i^{\top} \bA^{-1} u_i 
	        &=&   u_i^{\top} (\bOme^{\top} \bOme )^{-1} u_i  \\ \nonumber
	        &=& u_i^{\top}  \bUam  ((\bOme \bUam)^{\top}\bOme \bUam)^{-1} \bUam^{\top} w_i\\ \nonumber 
	        &=&
	        \frac{1}{n-1} + (Z_i -\widebar{\bZ_{\noi}} ) (\bZ^{\top}_{\noi} (\Ib_{n-1}-\projmat) \bZ_{\noi})^{-1}(Z_i -\widebar{\bZ_{\noi}} )^{\top},
	    \end{eqnarray}
	   which together with \eqref{eq:hii_simplebound} implies $\Ec{h_{ii}}{Z_i}  \le
	   \Ec{ u_i^{\top} \bA^{-1} u_i }{Z_i} =
	   1/(n-1) + \Ec{\bm{\Xi}}{Z_i}$, where
        \begin{equation}\label{eq:Xi_def}
                \bs{\Xi} = (Z_i -\widebar{\bZ_{\noi}} ) (\bZ^{\top}_{\noi} (\Ib_{n-1}-\projmat) \bZ_{\noi})^{-1}(Z_i -\widebar{\bZ_{\noi}} )^{\top}. 
         \end{equation}
	    As the problem has been reduced to calculating $\Ec{\bm{\Xi}}{Z_i}$, we arrive at the step \myrom{4} now. Write $(\bZ_i -\widebar{\bZ_{\noi}} ) = (\bZ_i - \bv_0) - (\widebar{\bZ_{\noi}} -  \bv_0)$, where $ \bv_0$ is the mean of Gaussian random variable $\Xnoj$, we can expand $
	       \Ec{\bm{\Xi}}{Z_i} = \mathrm{III}_1 + \mathrm{III}_2 + \mathrm{III}_3 $, where
	   \begin{eqnarray} \label{eq:III_1}
	       \mathrm{III}_1 &=& (\bZ_i - \bv_0)  \Ec{(\bZ_{\noi}^{\top} (\Ib_{n-1}-\projmat)  \bZ_{\noi})^{-1}}{\bZ_i} (\bZ_i - \bv_0)^{\top} \\ \label{eq:III_2}
	       \mathrm{III}_2 &=& -2 (\bZ_i - \bv_0) \Ec{(\bZ_{\noi}^{\top} (\Ib_{n-1}-\projmat)  \bZ_{\noi})^{-1}(\widebar{\bZ_{\noi}} -  \bv_0)^{\top}}{\bm{Z}_{i}} \\
	       \label{eq:III_3}
		  \mathrm{III}_3 &=&   \Ec{(\widebar{\bZ_{\noi}} -  \bv_0)(\bZ_{\noi}^{\top} (\Ib_{n-1}-\projmat)  \bZ_{\noi})^{-1}(\widebar{\bZ_{\noi}} -  \bv_0)^{\top}}{\bm{Z}_{i}} 
	   \end{eqnarray}
	   Below we are going to show $\mathrm{III}_2=0$ and derive $\mathrm{III}_1, \mathrm{III}_3$ carefully.
	   %For the above three terms, we make use of the fact $ $
	   %$\mathrm{III}_3$ is defined similarly for $(\widebar{\bZ_{\noi}} -  \bv_0)$ and $\mathrm{III}_2$ is the cross term.
	   Regarding the term $\mathrm{III}_1$, we exactly write down its stochastic representation. Under the state Gaussian model, we have $\bZ_{\noi}^{\top} \sim \gauss{\bv_0  \bm{1}_{n-1}^{\top} }{\Ib_{n-1}\otimes\bSig_0}$, then $(\bZ_{\noi}^{\top} (\Ib_{n-1}-\projmat)  \bZ_{\noi})^{-1}$ follows an inverse Wishart distribution i.e.
	   \begin{eqnarray}\nonumber
	     %~~ &~& \bX_{-i,\noj}^{\top} (\Ib_{n-1}-\projmat)  \bX_{-i,\noj} \sim \calW_{p-1}(\bSig_0 , n-2)\\
	       %&~& 
	       (\bZ_{\noi}^{\top} (\Ib_{n-1}-\projmat)  \bZ_{\noi})^{-1} \sim  \calW^{-1}_{p-1}(\bSig_0^{-1} , n-2)
	   \end{eqnarray}
	   and $ \bZ_{\noi} \independent \bZ_i $, hence we can calculate 
	   $$
	   \Ec{(\bZ_{\noi}^{\top} (\Ib_{n-1}-\projmat)   \bZ_{\noi})^{-1}}{\bZ_i}= \frac{\bSig_0^{-1}}{n-p-2}.
	   $$
	   Plug in the above equation into \eqref{eq:III_1}, we have
	   \begin{equation}\label{eq:III1_result}
	       \mathrm{III}_1 = (\bZ_i - \bv_0) \bSig_0^{-1} (\bZ_i - \bv_0)^{\top} = \frac{\bm{\Phi}}{n-p-2},~~~~\text{where}~~\bm{\Phi}\sim \chi_{p-1}^2,~ \bm{\Phi} \independent \bZ_{\noi}.
	   \end{equation}
	   	   %Further we obtain $\mathrm{III}_1 = {\bm{\Phi}}/(n-p-2)$, where $\bm{\Phi}\sim \chi_{p-1}^2$ and is independent from $\bZ_{\noi}$. 
	   	  Regarding the term $\mathrm{III}_2 $ in \eqref{eq:III_2}, 
	   %	  we first write down its expression 
	   %\begin{equation}
	   %    \mathrm{III}_2 = -2 (\bZ_i - \bv_0) \Ec{(\bZ_{\noi}^{\top} (\Ib_{n-1}-\projmat)  \bZ_{\noi})^{-1}(\widebar{\bZ_{\noi}} -  \bv_0)^{\top}}{\bm{Z}_{i}}
	   %\end{equation}
	   we first denote $\bZ = \bZ_{\noi} - \bm{1}_{n-1}\bv_0$ and notice
	   \begin{equation}\label{eq:Z_gaussian}
	   	\bZ \sim \gauss{\bs{0}}{\Ib_{n-1} \otimes \bSig_0},
	   ~~~
	    \bm{1}_{n-1}^{\top} \bZ = (n-1) (\widebar{\bZ_{\noi}} -  \bv_0),
	   	 \end{equation}
	   then rewrite $\mathrm{III}_2$ as below
	   $$
	   \mathrm{III}_2 = -2 (\bZ_i - \bv_0) \EE{(( \bZ +  \bm{1}_{n-1}\bv_0)^{\top} (\Ib_{n-1}-\projmat)  (\bZ + \bm{1}_{n-1}\bv_0) )^{-1} \frac{(\bm{1}_{n-1}^{\top} \bZ)^{\top} }{n-1}
	   }
	   $$
	   where we also makes use of the fact that
	   $$
	   (\bZ_{\noi}^{\top} (\Ib_{n-1}-\projmat)  \bZ_{\noi})^{-1}(\widebar{\bZ_{\noi}} -  \bv_0)^{\top} \independent \bm{Z}_{i}
	   $$
	   Noticing that $(\bm{1}_{n-1}\bv_0)^{\top}(\Ib_{n-1}-\projmat)=\bm{0}$, we can simplify further  
	  \begin{equation}\label{eq:III2_simple}
	     \mathrm{III}_2 =  -\frac{2}{n-1} (\bZ_i - \bv_0) \EE{(\bZ^{\top}(\Ib_{n-1}-\projmat) \bZ)^{-1}( \bm{1}_{n-1}^{\top} \bZ)^{\top}}
	  \end{equation}
	  Notice in the above equation, $\bZ^{\top}(\Ib_{n-1}-\projmat)$ is the orthogonal complement of $\bZ^{\top} \bm{1}_{n-1}$, which implies independence under the Gaussian distribution assumption, which we will now use to prove the expectation in \eqref{eq:III2_simple} equals zero.
	 Formally, we first have $(\bZ^{\top}(\Ib_{n-1}-\projmat),\bZ^{\top} \bm{1}_{n-1} )$ are multivariate Gaussian. Introducing the vectorization of matrix and the Kronecker product, we can express in the following way: 
	 $$
	 \text{vec}(\bZ^{\top}(\Ib_{n-1}-\projmat)) =(\Ib_{n-1}-\projmat)\otimes  \Ib_{p-1} \text{vec}(\bZ^{\top}),~~
	 \text{vec}(\bZ^{\top}) = \bs{1}_{n-1} \otimes  \Ib_{p-1} \text{vec}(\bZ^{\top}).
	 $$
	 Now we are ready to calculate the covariance
	 \begin{eqnarray}
	    &~~& \nonumber \Cov{\text{vec}(\bZ^{\top}(\Ib_{n-1}-\projmat))}{\text{vec}(\bZ^{\top} \bm{1}_{n-1})} \\ \nonumber
	    &=&( (\Ib_{n-1}-\projmat)\otimes  \Ib_{p-1}) (\Ib_{n-1} \otimes \bSig_0)  (\bs{1}_{n-1} \otimes  \Ib_{p-1})^{\top}  \\ \nonumber
	    &=& ((\Ib_{n-1}-\projmat)  \Ib_{n-1} \bs{1}_{n-1}) \otimes ( \Ib_{p-1} \bSig_0 \Ib_{p-1}) = \bs{0}
	 \end{eqnarray}
	 where in above equalities we use the fact $\Var{\text{vec}(\bZ^{\top})} = \Ib_{n-1} \otimes \bSig_0$ in \eqref{eq:Z_gaussian} and the mixed-product property of the Kronecker product.
	  Therefore
	  \begin{equation}\label{eq:III2_result}
	       \bZ^{\top}(\Ib_{n-1}-\projmat) \independent \bZ^{\top} \bm{1}_{n-1}   ~~\Longrightarrow~~ \mathrm{III}_2 = 0
	  \end{equation}
	  Regarding the term $\mathrm{III}_{3}$, first denote $\bm{\Psi}_1 = \bZ^{\top} \projmat \bZ $ and $\bm{\Psi}_2 = \bZ^{\top} (\Ib_{n-1}-\projmat) \bZ$, we obtain two independent Wishart random variables i.e.
	  $$
	  \bm{\Psi}_1  \sim \calW_{p-1}(\bSig_0, 1),~~ \bm{\Psi}_2  \sim \calW_{p-1}(\bSig_0,  n-2),~~
	  \bm{\Psi}_1 \independent \bm{\Psi}_2
	  . 
	  $$
	  Then $\mathrm{III}_{3}$ can be calculated as below 
	  \begin{eqnarray} \nonumber
	        \mathrm{III}_{3} 
	        &=& \Ec{(\widebar{\bZ_{\noi}} -  \bv_0)(\bZ_{\noi}^{\top} (\Ib_{n-1}-\projmat)  \bZ_{\noi})^{-1}(\widebar{\bZ_{\noi}} -  \bv_0)^{\top}}{\bm{Z}_{i}}  \\ \nonumber
	        &=& \EE{ \bm{1}_{n-1}^{\top} \bZ (\bZ^{\top}(\Ib_{n-1}-\projmat) \bZ)^{-1} \bZ^{\top} \bm{1}_{n-1}  }/(n-1)^2 \\ \nonumber	        
	        &=& \EE{\Tr\left( \bm{1}_{n-1}^{\top} \bZ (\bZ^{\top}(\Ib_{n-1}-\projmat) \bZ)^{-1} \bZ^{\top} \bm{1}_{n-1}
	        \right)}/(n-1)^2 \\ \nonumber
	        &=& \EE{ \Tr (\bm{\Psi}_1  \bm{\Psi}_2^{-1}   )}/(n-1)\\ \nonumber
	        &=&  \Tr \EE{\bm{\Psi}_1 \bm{\Psi}_2^{-1}}/(n-1) \\ \nonumber
	        &=& \Tr ( \EE{\bm{\Psi}_1} \EE{\bm{\Psi}_2^{-1}})/(n-1) \\ \nonumber
	        &=& \Tr(\bSig_0 \frac{\bSig_0^{-1}}{n-p-2} )/(n-1) \\ \label{eq:III3_result}
	        &=& \frac{p}{(n-1)(n-p-2)} 
	  \end{eqnarray}
	  where the first equality is from \eqref{eq:III_3}, the second equality is similarly obtained as \eqref{eq:III2_simple}, the fourth equality holds
        by the fact $\Tr(AB) = \Tr(BA)$ and the definition of $\bm{\Psi}_1$ and $\bm{\Psi}_2$, the sixth equality holds due to $\bm{\Psi}_1 \independent \bm{\Psi}_2$. So far we have shown $\mathrm{III}_2=0$ and figured out the stochastic representation of $\mathrm{III}_2, \mathrm{III}_3$, which are also further simplified using the properties of Wishart and inverse-Wishart random variables. These bring us to the final stage i.e. step \myrom{5}. Combining \eqref{eq:hii_simplebound}, \eqref{eq:III1_result}, \eqref{eq:III2_result} and \eqref{eq:III3_result}, we finally obtain 
	   \begin{eqnarray} \nonumber
	   \Ec{h_{ii}}{Z_i} 
	   &\le &   \Ec{   u_i^{\top} \bA^{-1} u_i }{Z_i} \\ \nonumber
	   &\le& \frac{1}{n-1} + \Ec{\bm{\Xi}}{Z_i} \\ \nonumber
	    &=& \frac{1}{n-1} +   \mathrm{III}_{1}  +  \mathrm{III}_{2} +  \mathrm{III}_{3}
	    \\ \label{eq:hii_given_ithsample}
	   & \le &
	   \frac{1}{n-1}\cdot \frac{n-2}{n-p-2} + \frac{\bm{\Phi}}{n-p-2}  
	   \end{eqnarray}
	   Recall the bound for $ \mathrm{II}(\mu)$ in \eqref{eq:II_result}, then we apply the Cauchy--Schwarz inequality to $ \Ec{\mu^2(W_i)}{\bZ_i}$ and $\Ec{h_{ii}}{\bZ_i}$, which yields
	   \begin{eqnarray}\nonumber
	       \mathrm{II}(\mu) 
	       &\le & 8 ~
	      \Ep{\Xinoj}{
	      \Ep{F}{\mu^2(W_i)}
	      {\Ec{h_{ii}}{\Xinoj}}
	      } \\ \nonumber
	       &\le& \frac{8(n-2)\EE{\mu^2(W_i)}}{(n-1)(n-p-2)} + \frac{8\sqrt{\EE{\bs{\Phi}^2}}}{n-p-2} \sqrt{\Ep{\bZ_i}{\Ec{\mu^{4}(W_i)}{\bZ_i}}} \\
	       &\le &  \frac{8\sqrt{\EE{\mu^4(X,Z)}}}{n-p-2}\left(
	       1 + \sqrt{\EE{\bs{\Phi}^2}}
	       \right )
	      \end{eqnarray}
	   where in the above equality, $\bm{\Phi}\sim \chi_{p-1}^2$ and is independent from $\bZ_{\noi}$.
	   Since $\EE{\bs{\Phi}^2} \le  p^2$, under the assumption $\EE{\mu^4(X,Z)} < \infty$, we obtain the following bound on $\mathrm{II}(\mu)$,
	   \begin{equation}
	       	   \mathrm{II}(\mu) = O\left(\frac{p}{n-p-2}\right).
	       	   %\le \frac{n-2}{(n-1)(n-p-2)} + \frac{p-1}{n-p-2} \le   \frac{p}{n-p-2}
	   \end{equation}
	   Replacing the $\mu$ function by $\mustar$ and applying the assumption $\EE{(\mustar)^4(X,Z)} < \infty$, we can establish the same rate for $\mathrm{II}(\mustar)$.
	   %the terms $|a-b|,|c-d|$ in \eqref{eq:|a-c|}, \eqref{eq:|b-d|}.
	   %\begin{equation}
	   %   \bX_{m\noj} \sim \gauss{\bv_0}{\bSig_0},~~~ \widebar{\bX}_{-i,\noj} \sim \gauss{\bv_0}{\frac{1}{n-1}\bSig_0}
	   %\end{equation}
    % which is also the rate of $\mathrm{II}_{1}$.
    Shifting back to the $n_{2}$ notation, we finally establish \eqref{eq:gap_rate_gaussian}, i.e.
    $$
     \thetamu - \thetamuT = O\left(\frac{p}{n_{2} - p -2}\right).
    $$
    \end{proof}
     	   \subsubsection{Proposition~\ref{prop:gap_rate_DMC}}
         \begin{proof}[Proof of Proposition~\ref{prop:gap_rate_DMC}]\label{pf:prop:gap_rate_DMC}
     From the proposition statement, we know the sufficient statistic $\bT_m$ and $\thetamuT$ are defined based on the batch $\calB_m$ whose sample size is $n_{2}$. Again, we will abbreviate the notation dependence for simplicity, i.e. use a generic $n$ instead of $n_{2}$, use $\bT$ and $\bs{Z}$ instead of $\bT_{m}$ and $\bZ_{m}$, as we did in the proof of Proposition \ref{prop:gap_rate_gaussian}.
    Following the derivations up to \eqref{eq:term_II} in the proof of Proposition \ref{prop:gap_rate_gaussian}, it suffices to deal with the following term:
    $$
    \Pi(\mu):=\Ep{ \bs{Z}}{ \Varp{F}{\mu(W_i)} \Ep{ \Tcondlaw}{ 
        \chi^2(F^{\bT}\|F)}}.
    $$
    where $F$ denotes the conditional distribution $X_{i}|\bs{Z}$ and $F_{\bT}$ denotes the conditional distribution $X_{i}|\bs{Z},\bT$. Below we will consider quantifying the $\chi^2$ divergence between $F_{\bT}$ and $F$, 
    % First denote $q({k,k_{1},k_{2}}) = \PP{X_j=k|X_{j-1}=k_1, X_{j+1} =k_2}$, then we have
    % \[
    % q({k,k_{1},k_{2}}) = \frac{  \Pi^{(j)}_{k_1,k} \Pi^{(j+1)}_{k,k_2}  }{\sum_{k=1}^{K} \Pi^{(j)}_{k_1,k} \Pi^{(j+1)}_{k,k_2}}.
    % \]
    % Writing down the conditional distribution of $\bm{X}$ given $\bm{X}_{\noj}$ compactly as
    % \begin{equation}
    %     \Pc{ \bm{X}}{\bm{X}_{\noj}} =  %\prod_{k=1}^{K}
    %     \prod_{k,k_{1},k_{2} \in [K]} 
    %     (q({k,k_{1},k_{2}}))^{N(k,k_{1},k_{2})}
    % \end{equation}
    % where $N(k,k_{1},k_{2}) = \sum_{i\in[n]} \indicat{X_{i}=k,X_{i,j-1} =k_{1},X_{i,j+1}=k_{1} }$, we immediately have $\{N(k,k_{1},k_{2})\}_{(k,k_{1},k_{2}\in[K] )}$ is sufficient, which will be denoted by $\bT:= \calT(\bm{X}_j, \bm{X}_{\noj})$.
    % $\bm{T}:= \bm{T}(\bm{X}_j| \bm{X}_{\noj})$. 
    % Given $\{X_{i},Z_i\}$ and 
    Let $k_{1},k_{2}$ be $W_{i,j-1},W_{i,j+1}$ respectively, we can write down the probability mass function of $F_{\bT}$ and $F$:
    \begin{eqnarray}
     F:\Pc{X_{i}}{\bs{Z} } 
    % = \Pc{X_{i}}{X_{i,j-1} =k_{1},X_{i,j+1}=k_{2}} 
    = \prod_{k=1}^{K} (q(k,k_{1},k_{2}))^{\indicat{X_{i}=k,W_{i,j-1} =k_{1},W_{i,j+1}=k_{1}}  } \\
     F_{\bT}: \Pc{X_{i}}{\bs{Z},\bm{T}}= \prod_{k=1}^{K} (\hat{q}(k,k_{1},k_{2}))^{\indicat{X_{i}=k,W_{i,j-1} =k_{1},W_{i,j+1}=k_{1}}  }
    \end{eqnarray}
    % \[
    % F:\Pc{X_{i}}{\bs{Z} } 
    % % = \Pc{X_{i}}{X_{i,j-1} =k_{1},X_{i,j+1}=k_{2}} 
    % = \prod_{k=1}^{K} (q(k,k_{1},k_{2}))^{\indicat{X_{i}=k,X_{i,j-1} =k_{1},X_{i,j+1}=k_{1}}  },
    % \]
    % \[
    % F_{\bT}: \Pc{X_{i}}{\bs{Z},\bm{T}}= \prod_{k=1}^{K} (\hat{q}(k,k_{1},k_{2}))^{\indicat{X_{i}=k,X_{i,j-1} =k_{1},X_{i,j+1}=k_{1}}  }
  % \]
    where $\hat{q}(k,k_{1},k_{2})=  N(k,k_{1},k_{2})/N(:,k_{1},k_{2})$ and  $N(:,k_{1},k_{2}) = \sum_{i=1}^{n}\indicat{W_{i,j-1}=k_{1},W_{i,j+1}=k_{2}}$.
    Recall the definition of $\chi^2$ divergence between two discrete distributions, we have
    \[
        \chi^2(F_{\bT}\|F) = \sum_{k=1}^{K}\frac{(\hat{q}(k,k_{1},k_{2}) - q(k,k_{1},k_{2}) )^2}{q(k,k_{1},k_{2})}~
    \]
    Notice that 
	$$
	\Ep{\bm{T}\mid \bs{Z}}{\hat{q}(k,k_{1},k_{2}) } = q(k,k_{1},k_{2}) ,~~~
	\Varp{\bm{T}\mid \bs{Z}}{ \hat{q}(k,k_{1},k_{2})} =\frac{q(k,k_{1},k_{2})(1-q(k,k_{1},k_{2}))}{ N(:,k_{1},k_{2}) }
	$$
    hence we can calculate the following conditional expectation,
    % Taking the expectation with respect to $\bm{T}$, conditional on $\bs{Z}$, we have
    \begin{eqnarray}\nonumber
    \Ep{\bm{T}\mid \bs{Z}}{ \chi^2(F_{\bT}\|F)}
    &=&  \sum_{k=1}^{K} \Ep{\bm{T}\mid \bs{Z}}{
    \frac{(\hat{q}(k,k_{1},k_{2}) - q(k,k_{1},k_{2}) )^2}{q(k,k_{1},k_{2})}
    } \\ \nonumber
    &=& \sum_{k=1}^{K} \frac{q(k,k_{1},k_{2})(1-  q(k,k_{1},k_{2}))}{
     N(:,k_{1},k_{2})q(k,k_{1},k_{2})}\\
    &=&
    \sum_{k=1}^{K} \frac{K-1}{ N(:,k_{1},k_{2}) } 
    \label{eq:DMC_expected_chisq}
    \end{eqnarray}
    where we use the fact $\sum_{k=1}^{K}q(k,k_{1},k_{2}) = 1$ in the last equality.
% 	since
% 	$$
% 	\Ep{\bm{T}\mid \bs{Z}}{\hat{q}(k,k_{1},k_{2}) } = q(k,k_{1},k_{2}) ,~~~
% 	\Varp{\bm{T}\mid \bs{Z}}{ \hat{q}(k,k_{1},k_{2})} =\frac{q(k,k_{1},k_{2})(1-q(k,k_{1},k_{2}))}{ N(:,k_{1},k_{2}) }
% 	$$
% 	Following the derivations in the proof of \ref{prop:gap_rate_gaussian}, it suffices to deal with the following term    
    Now $\Pi(\mu)$ can be calculated as below.
		\begin{eqnarray}\nonumber
	   % \mathbb{E}[\mathrm{Var}_{G}(\mu) \mathbb{E}_{\bm{T}|\bs{Z}}
	  	% \chi^{2}(F\|G)]  
	  	\Pi(\mu)
	  	&=&
	  	\Ep{ \bs{Z}}{ \Varp{F}{\mu(W_i)}  \Ep{ \Tcondlaw}{ 
         \chi^2(F_{\bT}\|F)}}	\\ \nonumber
	  	&= &\Ep{Z_i}{  \Varp{F}{\mu(W_i)} \Ec{
	  	 \Ep{\bm{T}\mid \bs{Z}}{ \chi^2(F_{\bT}\|F)}
	  	}{Z_i} }\\ \nonumber
	  	&=& 
	  	\Ep{Z_i}{  \Varp{F}{\mu(W_i)}  \Ec{
	  	 \frac{K-1}{N(:,W_{i,j-1},W_{i,j+1})}
	  	}{Z_i} } \\  \label{eq:Pi_simple}
	  	&=& 
	  	\Ep{Z_i}{ \Varp{F}{\mu(W_i)} \Ec{
	  	 \frac{K-1}{1+ N_{n-1}(W_{i,j-1},W_{i,j+1})}}{Z_i}
	  	} 
	\end{eqnarray}
% 		\begin{eqnarray}
% 	    \mathbb{E}[\mathrm{Var}_{G}(\mu) \mathbb{E}_{\bm{T}|\bs{Z}}
% 	  	\chi^{2}(F\|G)] 
% 	  	&\le & \sqrt{\EE{(\Varp{G}{\mu})^2}}\sqrt{\EE{
% 	  	 (\Ep{\bm{T}\mid \bs{Z}}{ \chi^2(F_{\bT}\|F)}})^2 } \\
% 	  	\\
% 	  	&\le& \sqrt{C_{0}} \sqrt{ \Ep{{Z_i}}{\Ec{
% 	  	 (\Ep{\bm{T}\mid \bs{Z}}{ \chi^2(F_{\bT}\|F))^2 }
% 	  	}{Z_i} }}
% 	  	\\ 
% 	  	&\le & 
% 	   \sqrt{C_{0}}K\sqrt{ \Ep{{Z_i}}{\Ec{
% 	  	\frac{1}{(N(:,X_{i,j-1},X_{i,j+1}))^2}
% 	  	}{Z_i} }}
% 	\end{eqnarray}
    where the second equality comes from the tower property of conditional expectation,
    the third equality holds due to \eqref{eq:DMC_expected_chisq} and $k_1 = W_{i,j-1}, k_2 = W_{i,j+1}$. In term of the fourth equality, we simply use the new notation that $N_{n-1}(W_{i,j-1},W_{i,j+1}) =\sum_{m\ne i}^{n} \indicat{W_{m,j-1}= W_{i,j-1},W_{m,j+1}=W_{i,j+1}}$. Due to the independence among $i.i.d.$ samples $\{W_i\}\nsubp$, we have, when conditioning on $\Xinoj =W_{i,\noj}$
    \[
    \indicat{W_{m,j-1}= W_{i,j-1},~W_{m,j+1}=W_{i,j+1}} \stackrel{i.i.d.}{\sim} \text{Bern}(
    q(W_{i,j-1}, W_{i,j+1})
    ),~~~~ m \in [n],~ m\ne i.
    \]
    where $q(W_{i,j-1}, W_{i,j+1}) = \Pc{ W_{j-1} = W_{i,j-1} ,W_{j+1} = W_{i,j+1}}{\Xinoj}$. Given a binomial random variable $B\sim \text{Bin}(n,q)$, we have the following fact by elementary calculus,
    % $$
    % \EE{\frac{1}{(1+B)^2}} \le \EE{\frac{2}{(1+B)(2+B)}} = \frac{1}{(n+1)(n+2)q^2} \cdot (1- (1-q)^{n+2} - q(1-q)^{n+1}).
    % $$ 
    \begin{equation}\label{eq:Binomial_fact}
    \EE{\frac{1}{1+B}} = \frac{1}{(n+1)q} \cdot (1- (1-q)^{n+1}).
    \end{equation} 
    hence we can bound the term $\Pi(\mu)$ as below
    \begin{eqnarray}
    \Pi(\mu)
    &=& \frac{K-1}{n}\Ep{Z_i}{ 
       \Varp{F}{\mu(W_i)} \frac{1- (1-q(W_{i,j-1}, W_{i,j+1}))^n}{q(W_{i,j-1}, W_{i,j+1}) } 
    } \\
    &\le & \frac{K-1}{n} \Ep{Z_i}{ \Varp{F}{\mu(W_i)} } \frac{K^2}{ K^2 \min \{q(k_{1},k_{2})\} } \\
    &\le&  \frac{K^3}{n} \frac{\EE{\mu^2(X,Z)}}{q_{0}}
    % \preceq \frac{K^3}{n}.
    \end{eqnarray}
    where the equality holds as a result of \eqref{eq:Pi_simple} and \eqref{eq:Binomial_fact}. And in the second line, we lower bound $q(W_{i,j-1}, W_{i,j+1})$ by $ \min \{q(k_{1},k_{2})\} $. Assuming $K^2\min\{ 
        \PP{W_{j-1}=k_{1},W_{j+1}=k_{2}}\}_{k_{1},k_{2}\in [K]}
        \} \ge q_{0}>0$ gives us the third line. Then we can establish $ \Pi(\mu)= O\left(\frac{K^3}{n}\right)$ (and similarly for $ \Pi(\mustar)$) under the stated moment condition $\EE{(\mu)^2(X,Z)}, \EE{(\mustar)^2(X,Z)}< \infty$. Finally, making use of the rate result about $\Pi(\mu),\Pi(\mustar)$ and following the same derivation as in Proposition \ref{prop:gap_rate_gaussian}, we have $ \thetamu - \thetamuT  = O\left(\frac{K^3}{n_{2}}\right)$, where we shift back to the $n_{2}$ notation.
            \end{proof}
    \subsubsection{Ancillary lemmas}
    Lemma \ref{lem:chisq_div_gaussians} can be similarly derived as the expression for the R\'{e}nyi divergence between two multivariate Gaussian distributions in Section 2.2.4 of \citet{gil2011renyi}. For completeness, we still present our proof below.
    \begin{lemma}\label{lem:chisq_div_gaussians}
	  The $\chi^{2}$-divergence between $P:\gauss{\bm{a}_1}{\Sigma_1}$ and $Q:\gauss{\bm{a}_2}{\Sigma_2}$ equals the following whenever $2 \Sigma_2 - \Sigma_1 \succ 0$:
	  \begin{equation}\nonumber
	  \frac{|\Sigma_{2}|}{|\Sigma_1|^{\frac{1}{2}}  | 2 \Sigma_2 - \Sigma_1|^{\frac{1}{2}} }\exp\left\{
	  (\bm{a}_1 - \bm{a}_2 )^{\top} ( 2 \Sigma_2 - \Sigma_1 )^{-1} (\bm{a}_1- \bm{a}_2)
	  \right\} -1 .
	  \end{equation}
	  where $\bs{a}_1, \bs{a}_2  \in \mathbb{R}^d$, $\Sigma_1, \Sigma_2 \in \mathbb{R}^{d \times d}$, $\Sigma \succ 0$ means a matrix $\Sigma$ is positive definite and $|\Sigma|$ denotes its determinant.
	  \end{lemma}
	  \begin{proof}[Proof of Lemma \ref{lem:chisq_div_gaussians}]
	  \label{pf:lem:chisq_div_gaussians}
	  According to the definition of the $\chi^{2}$-divergence, we have
	  \begin{equation}\label{eq:chi_sq_def}
	      \chi^2 (P\|Q) := \int \left( 
	      \frac{dP}{dQ} 
	      \right)^2 dQ -1  = \int \frac{p^2(x)}{q(x)} dx -1 ,
	  \end{equation}
	  where $p(x),q(x)$ are the Gaussian density functions. For multivariate Gaussian random variable with mean $\bs{a} \in \mathbb{R}^d$ and covariance matrix $\Sigma \in \mathbb{R}^{d \times d}$, the density function equals the following
	  \begin{equation}
	      f(x) = \frac{1}{(2\pi)^{\frac{d}{2}} |\Sigma|^{\frac{1}{2}} } \exp\left\{-\frac{1}{2}(x -
	  \bm{a})^{\top} \Sigma^{-1}(x -
	  \bm{a})
	  \right\},~~~ x \in \mathbb{R}^d.
	  \end{equation}
	 Hence we can calculate the $\chi^{2}$-divergence as below,
	 \begin{align} \nonumber
	  \chi^2 (P\|Q) 
	   &=  \frac{|\Sigma_2|^{\frac{1}{2}} }{ |\Sigma_1|}\int_{\mathbb{R}^d} \frac{1}{(2\pi)^{\frac{d}{2}}} \exp \left\{-\frac{1}{2}(x -
	  \bm{a}_1)^{\top}  (2\Sigma_1^{-1})(x -
	  \bm{a}_1) + \frac{1}{2}(x -
	  \bm{a}_2)^{\top}  \Sigma_2^{-1}(x -
	  \bm{a}_2)
	  \right\} dx -1  \\ \label{eq:chi_sq_expand}
	  &:= \frac{|\Sigma_2|^{\frac{1}{2}} }{ |\Sigma_1|}\int_{\mathbb{R}^d} \frac{1}{(2\pi)^{\frac{d}{2}}} \exp \left\{
     \mathrm{II}_1 + \mathrm{II}_2 + 	\mathrm{II}_3
	  \right\} dx -1, 
	 \end{align}
	  where the first equality holds following the definition in \eqref{eq:chi_sq_def} and the second equality comes from expanding the term in the exponent and combining, together with the following new notations:
	\begin{eqnarray} \label{eq:chi_sq_termII1}
	\mathrm{II}_1 &:=& -\frac{1}{2} x^{\top}(2 \Sigma_1^{-1} - \Sigma_2^{-1})x \\  \label{eq:chi_sq_termII2}
	\mathrm{II}_2 &:=& -\frac{1}{2}\cdot (-2 x^{\top}) (2 \Sigma_1^{-1} \bs{a}_1 -  \Sigma_2^{-1} \bs{a}_2   ) \\  \label{eq:chi_sq_termII3}
    \mathrm{II}_3 &:=& -\frac{1}{2}(2\bs{a}_1^{\top} \Sigma_1^{-1} \bs{a}_1  - \bs{a}_2 \Sigma_2^{-1} \bs{a}_2  ) 
	 \end{eqnarray}  
	 Let $\Sigma_{\star}^{-1}= 2 \Sigma_1^{-1} - \Sigma_2^{-1} $, $ \Sigma_{\star}^{-1}\bs{a}_{\star}= 2 \Sigma_1^{-1} \bs{a}_1 -  \Sigma_2^{-1} \bs{a}_2$ (since we assume the positive definiteness of $2 \Sigma_2 - \Sigma_1$, which implies $2 \Sigma_1^{-1} - \Sigma_2^{-1} \succ 0 $, hence $\Sigma_{\star}$ and $\bs{a}_{\star}$ are well-defined), then we have
	 \begin{align} 
	 \label{eq:term1_Sigma_star}
	  (\Sigma_1^{-1}  \Sigma_{\star} \Sigma_2^{-1})^{-1} &= \Sigma_2  \Sigma_{\star}^{-1} \Sigma_1 =  2 \Sigma_2 - \Sigma_1  \\ \label{eq:term2_Sigma_star}
	  2  \Sigma_{\star} \Sigma_1^{-1} - \mathrm{I}_{d} &=  \Sigma_{\star}(2 \Sigma_1^{-1} -  \Sigma_{\star}^{-1} ) =
	  \Sigma_{\star} \Sigma_2^{-1} \\ \nonumber
	  \frac{1}{2} \bs{a}_{\star}^{\top} \Sigma_{\star}^{-1}\bs{a}_{\star} 
	   &= \frac{1}{2} (2 \Sigma_1^{-1} \bs{a}_1 -  \Sigma_2^{-1} \bs{a}_2)^{\top}  \Sigma_{\star} (2 \Sigma_1^{-1} \bs{a}_1 -  \Sigma_2^{-1} \bs{a}_2)  \\ \nonumber
	   &=  2\bs{a}_1^{\top} \Sigma_1^{-1} \Sigma_{\star}  \Sigma_1^{-1}  \bs{a}_1 - 2\bs{a}_1^{\top} \Sigma_1^{-1} \Sigma_{\star}  \Sigma_2^{-1}  \bs{a}_2    + \frac{1}{2}\bs{a}_2^{\top} \Sigma_2^{-1} \Sigma_{\star}  \Sigma_2^{-1}  \bs{a}_2  \\ \label{eq:term3_Sigma_star}
	   &= 2\bs{a}_1^{\top} \Sigma_1^{-1} \Sigma_{\star}  \Sigma_1^{-1}  \bs{a}_1 - 
	   2\bs{a}_1^{\top} ( 2 \Sigma_2 - \Sigma_1 )^{-1}  \bs{a}_2   
	   + \frac{1}{2}\bs{a}_2^{\top} \Sigma_2^{-1} \Sigma_{\star}  \Sigma_2^{-1}  \bs{a}_2 
	 \end{align}
	 where the first and the second line hold by the definition of $\Sigma_{\star}$, the second equality holds since $\Sigma_{\star}^{-1} = \Sigma_{\star}^{-1} \Sigma_{\star} \Sigma_{\star}^{-1} $, the third line is simply from expanding and the last equality comes from \eqref{eq:term1_Sigma_star}. The above equations will be used a lot for the incoming derivations. Now the term in the exponent can be written as 
	 \begin{align} \nonumber
	 &~~ \mathrm{II}_1  + \mathrm{II}_2 + \mathrm{II}_3 \\ \nonumber
	 &=
	  -\frac{1}{2}(x^{\top}  \Sigma_{\star}^{-1} x - 2x^{\top}    \Sigma_{\star}^{-1}\bs{a}_{\star}) + \mathrm{II}_{3}  \\ \nonumber
	 &=  -\frac{1}{2}(x - \bs{a}_{\star})^{\top} \Sigma_{\star}^{-1}(x - \bs{a}_{\star}) + \frac{1}{2}  \bs{a}_{\star}^{\top} \Sigma_{\star}^{-1} \bs{a}_{\star}  -\frac{1}{2}(2\bs{a}_1^{\top} \Sigma_1^{-1} \bs{a}_1  - \bs{a}_2 \Sigma_2^{-1} \bs{a}_2  )   \\ \nonumber
	 &= \lambda(x) + \bs{a}_1^{\top} \Sigma_1^{-1} (2 \Sigma_{\star}  \Sigma_1^{-1} -  \mathrm{I}_{d} ) \bs{a}_1 -
	  2\bs{a}_1^{\top} ( 2 \Sigma_2 - \Sigma_1 )^{-1}  \bs{a}_2  +
	  \frac{1}{2}\bs{a}_2^{\top} \Sigma_2^{-1}  (  \Sigma_{\star}  \Sigma_2^{-1} +   \mathrm{I}_{d})\bs{a}_2\\ \nonumber
	  &= \lambda(x) + \bs{a}_1^{\top} \Sigma_1^{-1} \Sigma_{\star} \Sigma_2^{-1}  \bs{a}_1 - 
	    2\bs{a}_1^{\top} ( 2 \Sigma_2 - \Sigma_1 )^{-1}  \bs{a}_2  +
	   \bs{a}_2^{\top} \Sigma_2^{-1}  \Sigma_{\star} \Sigma_1^{-1}  \bs{a}_2 \\ \nonumber
	  &= \lambda(x) + \bs{a}_1^{\top} ( 2 \Sigma_2 - \Sigma_1)^{-1} \bs{a}_1 - 
	    2\bs{a}_1^{\top} ( 2 \Sigma_2 - \Sigma_1 )^{-1}  \bs{a}_2  +
	   \bs{a}_2^{\top}  ( 2 \Sigma_2 - \Sigma_1)^{-1}  \bs{a}_2 \\ \label{eq:exponent_simple}
	   &= \lambda(x) + ( \bs{a}_1- \bs{a}_2)^{\top}  ( 2 \Sigma_2 - \Sigma_1 )^{-1} ( \bs{a}_1- \bs{a}_2):= \lambda(x) + Q( \bs{a}_1, \bs{a}_2,  \Sigma_1,  \Sigma_2 )
	 \end{align}	 
	 where the first equality holds by the definition of $\Sigma_{\star}$, $\bs{a}_{\star}$ and \eqref{eq:chi_sq_termII1}, \eqref{eq:chi_sq_termII2}, and the second equality holds due to \eqref{eq:chi_sq_termII3}. Regarding the third equality, we denote the term which depends on $x$ by $\lambda(x):= -\frac{1}{2}(x - \bs{a}_{\star})^{\top} \Sigma_{\star}^{-1}(x - \bs{a}_{\star})$. As for the other constant terms in the third line, we simply combine \eqref{eq:term3_Sigma_star} with the expansion of the term $\mathrm{II}_{3}$ and rearrange them into three terms: $ \bs{a}_1^{\top}(\cdot) \bs{a}_1 $, $ \bs{a}_1^{\top}(\cdot) \bs{a}_2 $ and $ \bs{a}_2^{\top}(\cdot) \bs{a}_2 $. The fourth equality holds as a result of applying \eqref{eq:term2_Sigma_star} twice and the last equality is simply from rearranging. Since only the term $\lambda(x)$ depends on $x$, we can simplify the $\chi^{2}$-divergence into the following
	 \begin{eqnarray} \nonumber
	 \chi^2 (P\|Q) 
	   &=&  \frac{|\Sigma_2|^{\frac{1}{2}} }{ |\Sigma_1|} 
	   \exp\left\{ 
	  Q( \bs{a}_1, \bs{a}_2,  \Sigma_1,  \Sigma_2 )
	   \right\}  \int_{\mathbb{R}^d} \frac{1}{(2\pi)^{\frac{d}{2}}} 
	   \exp \left\{
	   \lambda(x)
	  \right\} dx -1 \\ \nonumber
	  &=& \frac{|\Sigma_2|^{\frac{1}{2}} }{ |\Sigma_1|} \exp\left\{
      Q( \bs{a}_1, \bs{a}_2,  \Sigma_1,  \Sigma_2 )
	  \right\} 
	  \int_{\mathbb{R}^d} \frac{|\Sigma_{\star}|^{\frac{1}{2}} }{(2\pi)^{\frac{d}{2}} |\Sigma_{\star}|^{\frac{1}{2}} } 
	   \exp 
	   \left\{
	  \lambda(x)
	  \right\} dx -1 \\ \nonumber
	  &=& \frac{|\Sigma_2|^{\frac{1}{2}} }{ |\Sigma_1|} |\Sigma_{\star}|^{\frac{1}{2}} \exp\left\{
        Q( \bs{a}_1, \bs{a}_2,  \Sigma_1,  \Sigma_2 )
	  \right\}  -1 \\ \nonumber
	   &=& \frac{|\Sigma_2| }{ |\Sigma_1|^{\frac{1}{2}}} | \Sigma_1^{-1}  \Sigma_{\star} \Sigma_2^{-1} |^{\frac{1}{2}}  \exp\left\{
	 Q( \bs{a}_1, \bs{a}_2,  \Sigma_1,  \Sigma_2 )
	   \right\}  -1 \\  \nonumber
	   &=&   \frac{|\Sigma_{2}|}{|\Sigma_1|^{\frac{1}{2}}  | 2 \Sigma_2 - \Sigma_1|^{\frac{1}{2}} } \exp\left\{
	  (\bm{a}_1 - \bm{a}_2 )^{\top} ( 2 \Sigma_2 - \Sigma_1 )^{-1} (\bm{a}_1- \bm{a}_2)
	  \right\} -1 
	 \end{eqnarray}	 
	 where the first equality comes from \eqref{eq:chi_sq_expand} and \eqref{eq:exponent_simple}, the third equality holds due to the definition of $\lambda(x)$ and the fact that $\int f(x)dx =1$, where $f(x)$ is the Gaussian density function with the mean $ \bs{a}_{\star}$ and covariance matrix $\Sigma_{\star}$), the fourth equality holds by making use of the properties of determinant and the last equality holds as a result of \eqref{eq:term1_Sigma_star}. 
	 \end{proof}

\section{Further simulation details}
Source code for conducting floodgate in our simulation studies can be found at \url{https://github.com/LuZhangH/floodgate}.
\subsection{Nonlinear model setup}
    \label{app:sim:nonlinear}
        % Our method essentially makes no assumption about the conditional model except for some moment bounds. It is of great interest how it performs under general nonlinear models. In this section, we consider $X$ follows a Gaussian copula distribution, i.e.
    Consider $W$ which follows a Gaussian copula distribution with $X=W_{j_0}, Z = W_{\text{-}j_0}$ for some $j_0$ ($1\le j_0 \le p$), i.e.,
    \begin{equation}
        W^\text{latent}\sim AR(1),~ W_j = 2\varphi(X_j^{\text{latent}}) -1, ~~\forall~1\le j \le p. 
    \end{equation}
    Hence the marginal distribution for $W_j$ is $\text{Unif}[-1,1]$ (in fact, these are the inputs to the fitting methods we use in floodgate, not the AR(1) latent variables $W^{\text{latent}}$). We consider the following conditional model for $Y$ given $W$, with standard Gaussian noise,
    \begin{equation}\label{eq:nonlinear_mustar}
        \mustar(x,z)=\mustar(w):=\sum_{j \in S^1}g_j(w_{j}) + \sum_{(j,l)\in S^{2}} g_j(w_j)g_l(w_l) + \sum_{(j,l,m)\in S^{3}} g_j(w_j)g_l(w_l)g_m(w_m)
    \end{equation}
    where each function $g_j(x)$ is randomly chosen from the following: 
    \begin{eqnarray}
        ~~~~~~~~\sin(\pi x),\cos(\pi x),\sin(\pi x/2),\cos(\pi x)I(x>0),x\sin(\pi x),x,|x|,x^2, x^3,\exp(x)-1.
    \end{eqnarray}
    $S^{1}$ basically contains the main effect terms, while $S^{2}$ contain the pairs of variables with first order interactions. Tuples of variables involving second order interaction are denoted by $S^{3}$.  For a given amplitude, \eqref{eq:nonlinear_mustar} is scaled by the amplitude value divided by $\sqrt{n}$.
    
    Now we describe the construction of $S^{1},S^{2},S^{3}$. First we randomly pick $30$ variables into $S_{\star}$ and initialize $S_{\text{wl}}=S_{\star}$. $15$ of them will be randomly assigned into $S^{1}$ and removed from $S_{\text{wl}}$. Among these $15$ variables in $S^{1}$, we further choose $10$ variables into $5$ pairs randomly, which will be included in $S^{2}$. Regarding the other pairs in $S^{2}$, each time we randomly pick $2$ variables from $S_{\star}$ with the unscaled weight being $2|S_{\text{wl}}|/|S_{\star}|$ for variables in $S_{\text{wl}}$, $|S_{\star}\setminus S_{\text{wl}}|/|S_{\star}|$ for the others, then add them as a pair into $S^{2}$. Once picked, the variables will be removed from $S_{\text{wl}}$. This process iterates until $|S_{\text{wl}}|\le 5$. Regarding the construction of $S^{3}$, each time we randomly pick $3$ variables from $S_{\star}$ with the unscaled weight being $1.5|S_{\text{wl}}|/|S_{\star}|$ for variables in $S_{\text{wl}}$, $|S_{\star}\setminus S_{\text{wl}}|/|S_{\star}|$ for the others, then add them as a tuple into $S^{3}$. Once picked, the variables will be removed from $S_{\text{wl}}$. This process iterates until $|S_{\text{wl}}|=0$.

    \subsection{Implementation details of fitting algorithms}
     \label{app:sim:algorithms}
    Regarding how to obtain the \revision{working regression function}, there will be four different fitting algorithms for non-binary responses: 
    \begin{itemize}
        \item \textit{LASSO}: We fit a linear model by 10-fold cross-validated LASSO and output a \revision{working regression function}. The subsequent inference step will be quite fast. First, as implied by Algorithm \ref{alg:MOCK}, $L_{n}^{\alpha} (\mu)$ will be set to zero for unselected variables, without any computation. Second, as alluded to in Section~\ref{sec:computation}, we can analytically compute the conditional quantities in Algorithm \ref{alg:MOCK}.
        \item \textit{Ridge}: We again use 10-fold cross-validation to choose the penalty parameter for Ridge regression. It is also fast to perform floodgate on, due to the second point mentioned above.
        \item \textit{SAM}: We consider additive modelling, for example the sparse additive models (SAM) proposed in \citet{ravikumar2009sparse}. As suggested by the name, it carries out sparse penalization and our method will assign $L_{n}^{\alpha} (\mu)=0$ to unselected variables, as in \textit{lasso}.
        \item \textit{Random Forest}: Random forest \citep{breiman2001random} is included as a purely nonlinear machine learning algorithm. While random forest do not generally conduct variable selection, we rank variables based on the heuristic importance measure and use the top $50$ variables to run Algorithm \ref{alg:MOCK} and set $L_{n}^{\alpha} (\mu)=0$ for the remaining ones. Remark this is only for the concern of speed and does not have any negative impact on the inferential validity.
    \end{itemize}
    There are two additional fitting algorithms for binary responses: logistic regression with L1 regularization and L2 regularization, denoted by \textit{Binom\_LASSO} and \textit{Binom\_Ridge} respectively. Both use 10-fold cross-validation to choose the penalty parameter. 
    
    \subsection{Implementation details of ordinary least squares}
    \label{app:sim:ols}
    When the conditional model of $Y\mid X,Z$ is linear, i.e., $\condmean = X\beta + Z\theta$ with $(\beta,\theta)\in \mathbb{R}^{p}$ the coefficients, the mMSE gap for $X$ is closely related to its linear coefficient, formally
    \begin{equation}\nonumber
            \Ij = \left|\beta\right| \sqrt{\ee{\varc{X}{Z}}}.
    \end{equation}
    When the sample size $n$ is greater than the number of variables $p$, ordinary least squares (OLS) can provide valid confidence intervals for $\beta$. 
    % As for the factor $\sqrt{\ee{\varc{\Xj}{\Xnoj}}}$, it is straightforward to scale up the OLS confidence intervals. 
    However, there does not seem to exist a non-conservative way to transform the OLS confidence interval for $\beta$ into a confidence bound for $|\beta|$. So instead, we provide OLS with further oracle information: the sign of $\beta$ (we only compare half-widths of non-null covariates, and hence never construct OLS LCBs when $\beta=0$). In particular, if [LCI, UCI] denotes a standard OLS 2-sided, equal-tailed $1-2\alpha$ confidence interval for $\beta$, then the OLS LCB for $\Ij$ we use is
    \begin{equation}\label{eq:ols_transform}
    % [\text{LCI},\text{UCI}]  \rightarrow 
    \text{LCB}_{\text{OLS}} = \left\{\begin{array}{rl} \text{LCI}\sqrt{\ee{\varc{\Xj}{\Xnoj}}} & \text{ if }\beta>0\\
    -\text{UCI}\sqrt{\ee{\varc{\Xj}{\Xnoj}}} & \text{ if }\beta<0 \end{array}\right.
    % \max\{\text{LCI},-\text{UCI},0\} \sqrt{\ee{\varc{\Xj}{\Xnoj}}},
    \end{equation}
    which guarantees exact $1-\alpha$ coverage of $\Ij$ for any nonzero value of $\beta$. We again emphasize that, in order to construct this interval, OLS uses the oracle information of the sign of $\beta$ (this information is not available to floodgate in our simulations).
    % with probability at least $1-\alpha$ for any $\beta$. It has coverage exactly $1-\alpha$ when $\beta=0$ although its coverage approaches $1-\alpha/2$ as $|\beta|\rightarrow\infty$, but we were not able to come up with a way to use OLS to construct a LCB for $\Ij$ with exactly $1-\alpha$ coverage for all $\beta$. 
    % While our method works when $n<p$, our competitor does require $n>p$, hence we choose $n=1100,~p=1000$ to run the following simulations.

\subsection{Plots deferred from the main paper}
\label{app:sim:cover}
\subsubsection{Effect of sample splitting proportion}
{The corresponding coverage plots of Figure~\ref{fig:split} are given in Figure~\ref{fig:split_cover}. Figures~\ref{fig:split_app} and \ref{fig:split_app_cover} are additional plots with different simulation parameters specified in the captions. Figures~\ref{fig:split_cover} and \ref{fig:split_app_cover} show that in the simulations in Section~\ref{sec:simul_split}, the coverage of floodgate is consistently at or above the nominal 95\% level.}
    \begin{figure}[tb]
        \centering
        \includegraphics[width = 1\linewidth]{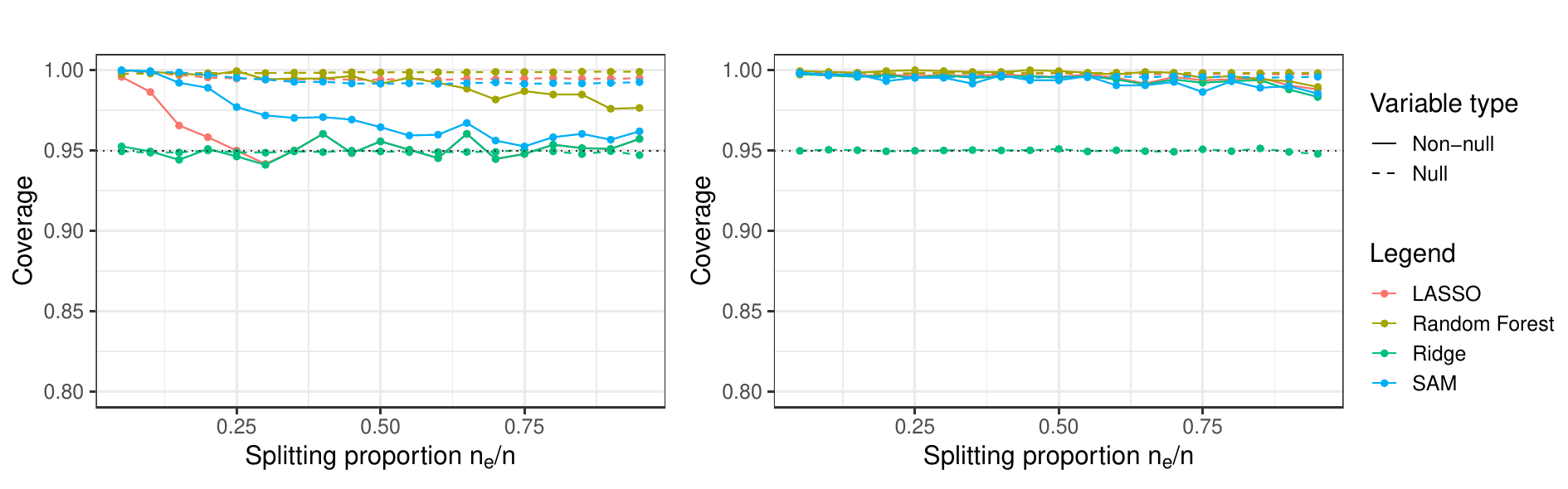}
        \caption{Coverage for the the linear-$\mustar$ (left) and nonlinear-$\mustar$ (right) simulations of Section~\ref{sec:simul_split}. The coefficient amplitude is chosen to be 10 for the left panel and the sample size $n$ equals 3000 in the right panel; see Section~\ref{sec:simul_setup} for remaining details. Standard errors are below 0.007 (left) and 0.003 (right).
        }
        \label{fig:split_cover}
    \end{figure}
    
 \begin{figure}[tb]
        \centering
        \includegraphics[width = 1\linewidth]{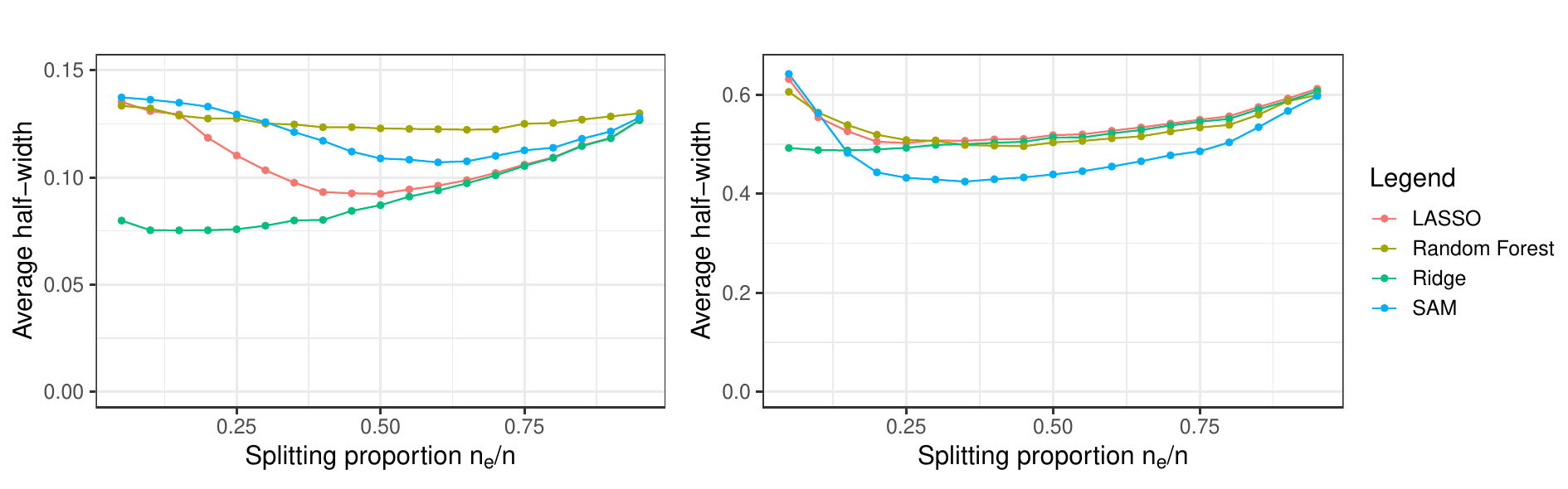}
        \caption{Average half-widths for the linear-$\mustar$ (left) and nonlinear-$\mustar$ (right) simulations of Section~\ref{sec:simul_split}. The coefficient amplitude is chosen to be 5 for the left panel and the sample size $n$ equals 1000 in the right panel; see Section~\ref{sec:simul_setup} for remaining details. Standard errors are below 0.002 (left) and 0.01 (right). }
        %Half width plots: design matrix with $n$ i.i.d. rows follows the Gaussian copula model in  \ref{sec:nonlinear_setup} with the auto-correlation coefficient being $0.3$; $Y|X$ follows the nonlinear model described in \ref{sec:nonlinear_setup} with amplitude being $50$ and $30$ non-null variables; $p = 1000$; number of null copies $K = 500$ for nonlinear fitting algorithms.
        % }
        \label{fig:split_app}
    \end{figure}  
    
     \begin{figure}[tb]
        \centering
        \includegraphics[width = 1\linewidth]{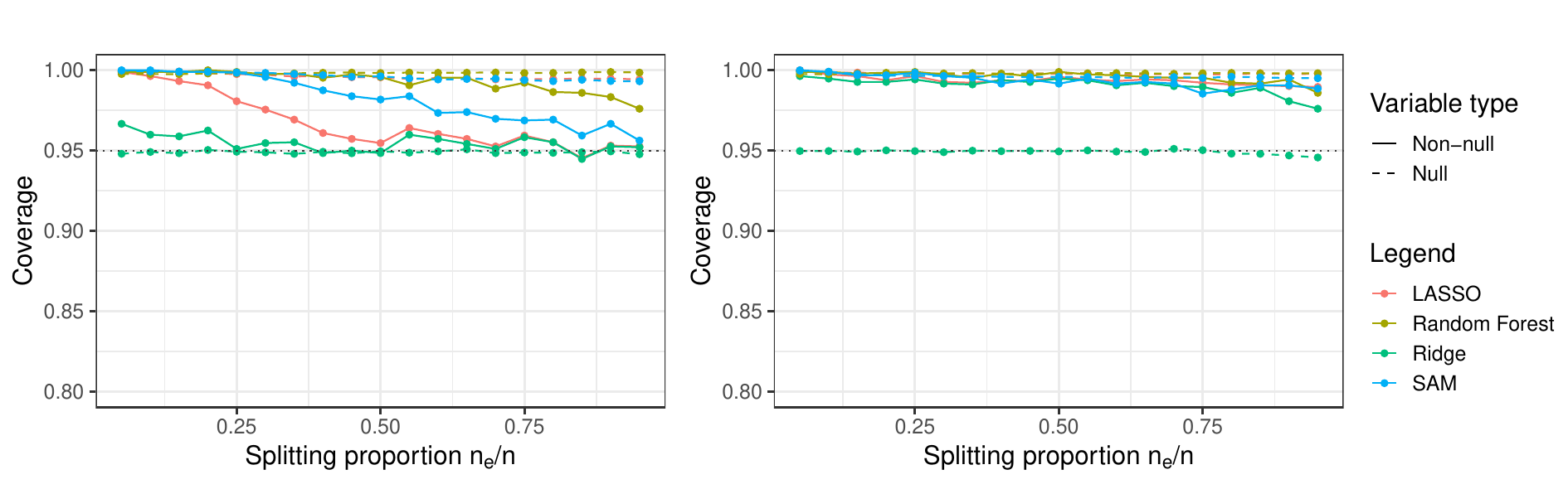}
        \caption{Coverage for the  the linear-$\mustar$ (left) and nonlinear-$\mustar$ (right) simulations of Section~\ref{sec:simul_split}. The coefficient amplitude is chosen to be 5 for the left panel and the sample size $n$ equals 1000 in the right panel; see Section~\ref{sec:simul_setup} for remaining details. Standard errors are below 0.006 (left) and 0.004 (right).}
        %Half width plots: design matrix with $n$ i.i.d. rows follows the Gaussian copula model in  \ref{sec:nonlinear_setup} with the auto-correlation coefficient being $0.3$; $Y|X$ follows the nonlinear model described in \ref{sec:nonlinear_setup} with amplitude being $50$ and $30$ non-null variables; $p = 1000$; number of null copies $K = 500$ for nonlinear fitting algorithms.
        % }
        \label{fig:split_app_cover}
    \end{figure}
\subsubsection{Effect of covariate dimension}

% Figures~\ref{fig:vary_p_cover} and \ref{fig:sine_p_cover} show that in the simulations in Section~\ref{sec:simul_setup}, the coverage of floodgate is consistently at or above the nominal 95\% level.

{The corresponding coverage plots of Figure~\ref{fig:vary_p_comb} are given in Figure~\ref{fig:vary_p_comb_cover}. Figures~\ref{fig:vary_p_app} and \ref{fig:vary_p_app_cover} are additional plots with different simulation parameters specified in the captions. Figures~\ref{fig:vary_p_comb_cover} and \ref{fig:vary_p_app_cover} show that in these simulations, the coverage of floodgate is consistently at or above the nominal 95\% level.}
    \begin{figure}[tb]
        \centering
        \includegraphics[width = 1\linewidth]{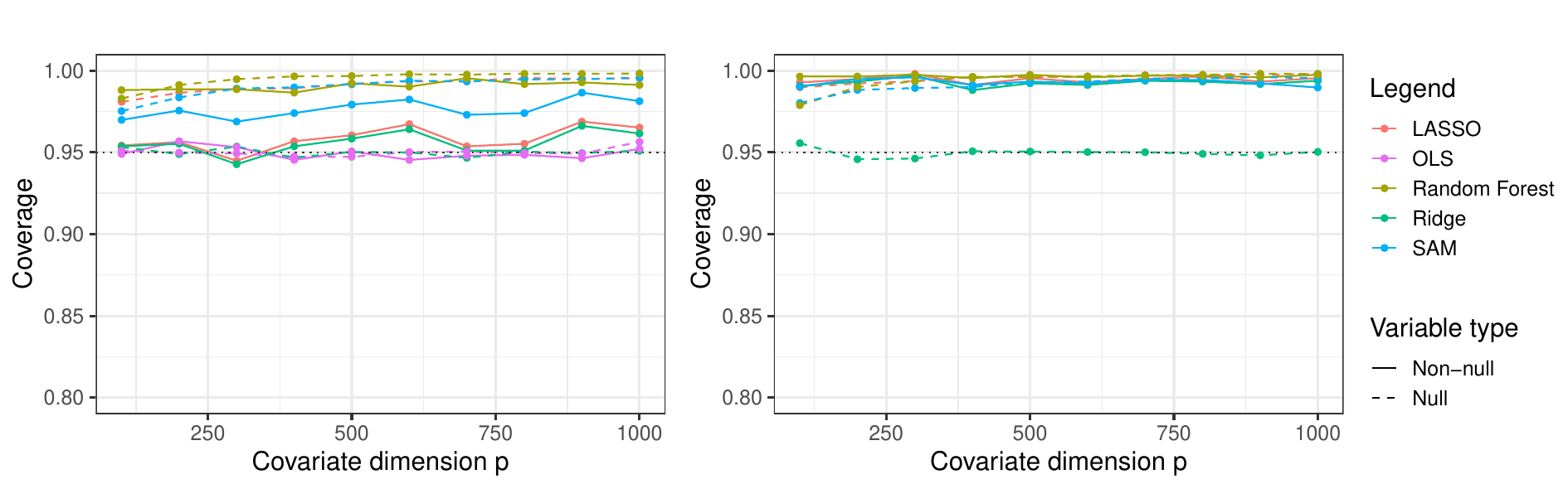}
        \caption{Coverage for the linear-$\mustar$ (left) and nonlinear-$\mustar$ (right) simulations of Section~\ref{sec:simul_vary_p}. OLS is run on the full sample. $p$ is varied on the x-axis; see Section~\ref{sec:simul_setup} for remaining details. Standard errors are below 0.006 (left) and 0.004 (right).}
        % Half width plots: design matrix with $n$ i.i.d. rows from an AR(1) model with autocorrelation $0.3$; $Y|X \sim \calN(X\beta,1)$, where $\beta$ has non-zero entries with random signs and equal absolute values (which will be the amplitude value divided by $\sqrt{n}$); amplitude equals $5$; there are $30$ non-null variables; $n = 1100$; number of null copies $K = 500$ for nonlinear fitting algorithms.
        
        \label{fig:vary_p_comb_cover}
    \end{figure}
    
        \begin{figure}[tb]
        \centering
        \includegraphics[width = 1\linewidth]{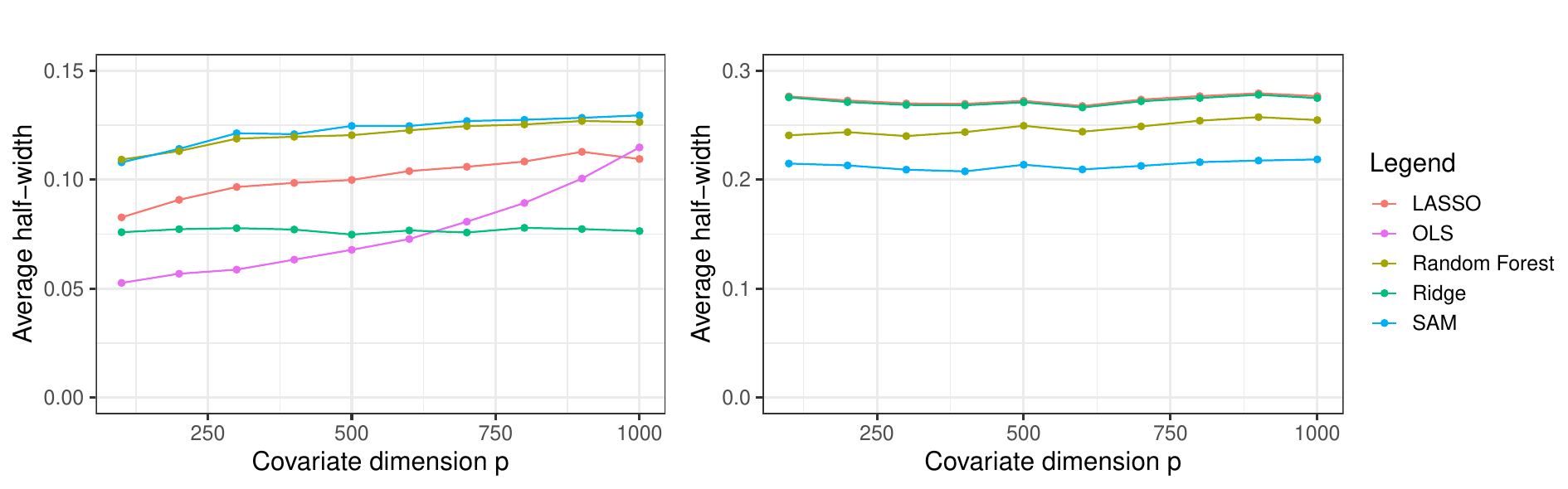}
        \caption{Average half-widths for the linear-$\mustar$ (left) and nonlinear-$\mustar$ (right) simulations of Section~\ref{sec:simul_vary_p}. The splitting proportion is chosen to be 0.25 for the left panel and the sample size n equals 3000 in the right panel. $p$ is varied on the x-axis; see Section~\ref{sec:simul_setup} for remaining details. Standard errors are below 0.002 (left) and 0.005 (right).}
        % Half width plots: design matrix with $n$ i.i.d. rows follows the Gaussian copula model in  \ref{sec:nonlinear_setup} with the auto-correlation coefficient being $0.3$; $Y|X$ follows the nonlinear model described in \ref{sec:nonlinear_setup} with amplitude being $50$ and $30$ non-null variables; number of null copies $K = 500$ for nonlinear fitting algorithms.
        
        \label{fig:vary_p_app}
    \end{figure}  
    
     \begin{figure}[tb]
        \centering
        \includegraphics[width = 1\linewidth]{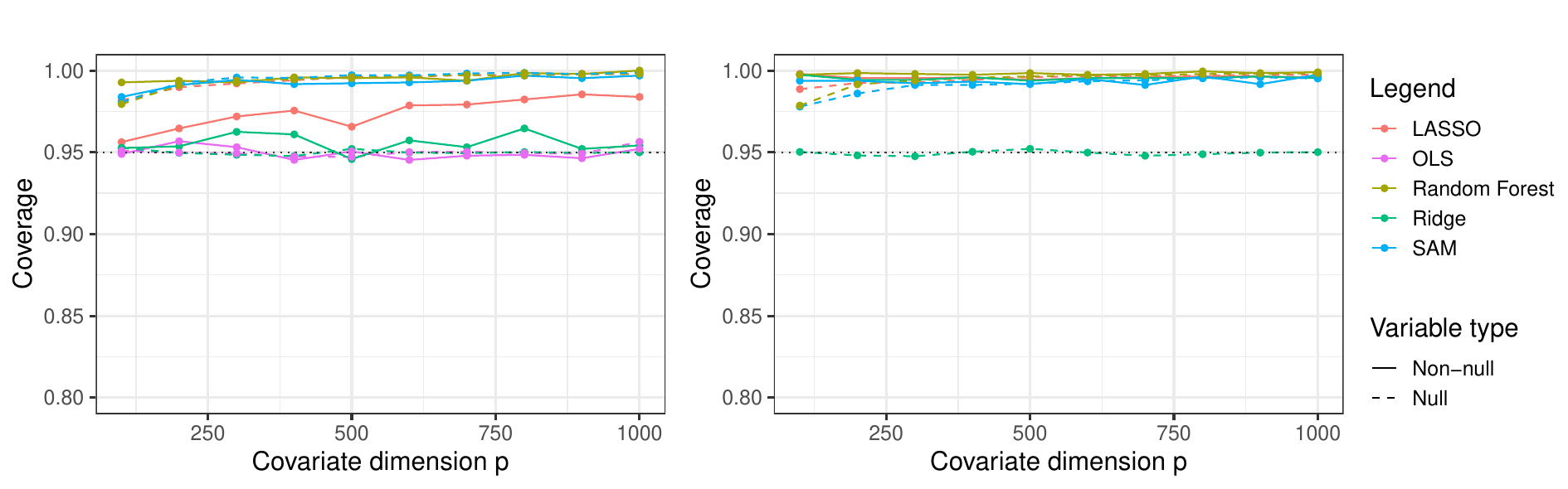}
        \caption{Coverage for the linear-$\mustar$ (left) and nonlinear-$\mustar$ (right) simulations of Section~\ref{sec:simul_vary_p}. The splitting proportion is chosen to be 0.25 for the left panel and the sample size n equals 3000 in the right panel. $p$ is varied on the x-axis; see Section~\ref{sec:simul_setup} for remaining details. Standard errors are below 0.006 (left) and 0.004 (right).}
        % Half width plots: design matrix with $n$ i.i.d. rows follows the Gaussian copula model in  \ref{sec:nonlinear_setup} with the auto-correlation coefficient being $0.3$; $Y|X$ follows the nonlinear model described in \ref{sec:nonlinear_setup} with amplitude being $50$ and $30$ non-null variables; number of null copies $K = 500$ for nonlinear fitting algorithms.
        
        \label{fig:vary_p_app_cover}
    \end{figure}
  
\subsubsection{Comparison with \citet{williamson2020unified}}
% \label{app:}
\begin{figure}
\centering
        \includegraphics[width = 0.54\linewidth]{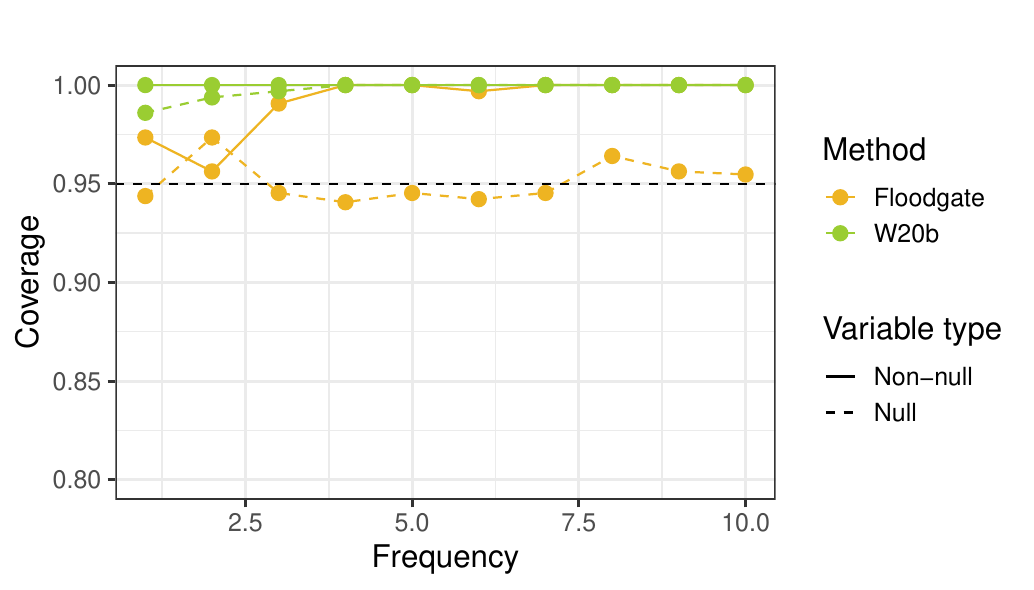}
        % \vspace{-0.5cm}
        \vspace{0.25cm}
        \caption{
        Coverage for floodgate and W20b in the sine function simulation of Section~\ref{sec:comp_vimp}. The frequency $\lambda$ is varied on the x-axis, and the dotted black line in the plot shows the nominal coverage level $1-\alpha$.
        %, while the frequency value $\lambda$ is varied on the x-axis. 
        The results are averaged over 640 independent replicates, and the standard errors are below 0.006.}
             %(left) and 0.01 (right).
            %  \vspace{0.5cm}
        
      \label{fig:comp_freq_cover}
\end{figure}
{
The corresponding coverage plot of Figure \ref{fig:comp_freq} is given in Figure \ref{fig:comp_freq_cover}, where we see both methods have coverages above the nominal level. In addition to the example in Section \ref{sec:comp_vimp}, we also compare floodgate with W20b in the higher-dimensional setting of the left panel of Figure~\ref{fig:vary_p_comb}. 
%It is also interesting to compare both methods when $p$ is 
Due to the computational challenge of running \citet{williamson2020unified}'s method, we only consider the two most efficient algorithms (LASSO and Ridge) among the four described in Appendix \ref{app:sim:algorithms}.
%Looking at the left panel of 
Figure \ref{fig:comp_p} shows W20b to have slightly less consistent coverage than floodgate, but also reinforces the general picture from the lower-dimensional simulation in Section~\ref{sec:comp_vimp} that W20b's LCBs are quite close to zero compared with floodgate's.
}
\begin{figure}[tb]
\centering
\includegraphics[width = 1\linewidth]{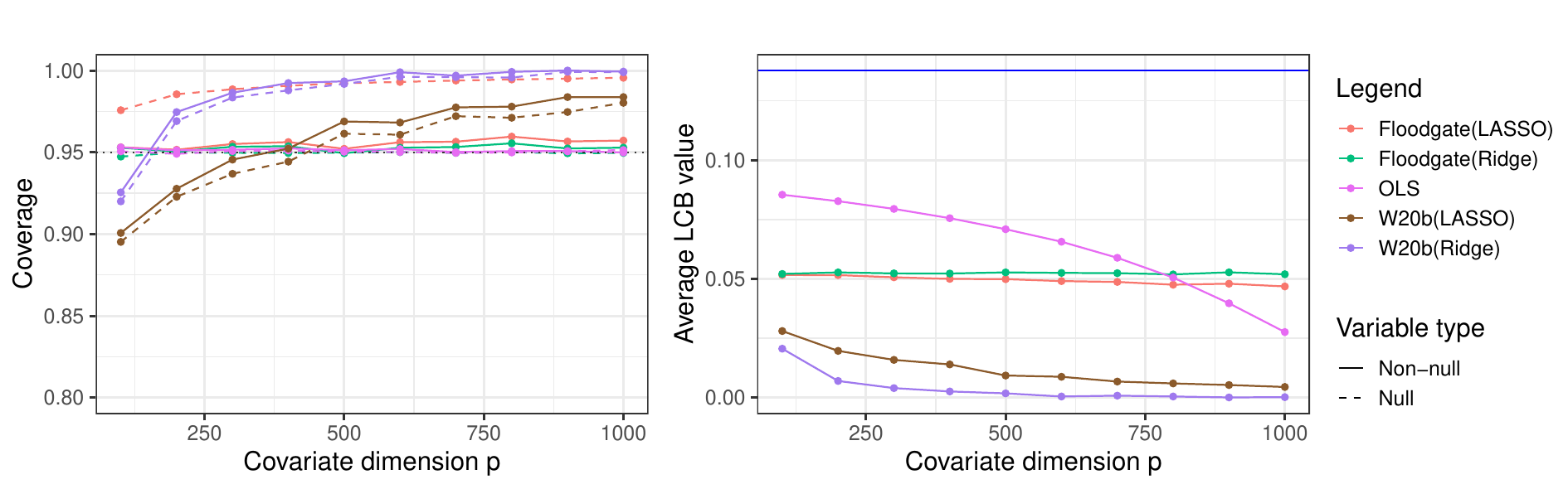}
\caption{Coverage (left) and average LCB values (right) for floodgate, W20b, and OLS (run on the full sample) in the linear-$\mustar$ simulation of Section~\ref{sec:comp_vimp}.
%involving two methods: floodgate and 
%Both methods conduct model fitting using two linear algorithms (LASSO and Ridge). OLS is run on the full sample. 
$p$ is varied on the x-axis, and the solid blue line in the right-hand plot shows the value of $\Ij$; see Section~\ref{sec:simul_setup} for remaining details. The results are averaged over 640 independent replicates, and the standard errors are below 0.012 (left) and 0.004 (right).
}
\label{fig:comp_p}
\end{figure} 

% \subsubsection{Effect of covariate dependence}

% \subsubsection{Effect of sample size}

\subsubsection{Robustness}
\begin{figure}[tb]
\centering
\includegraphics[width =1\linewidth]{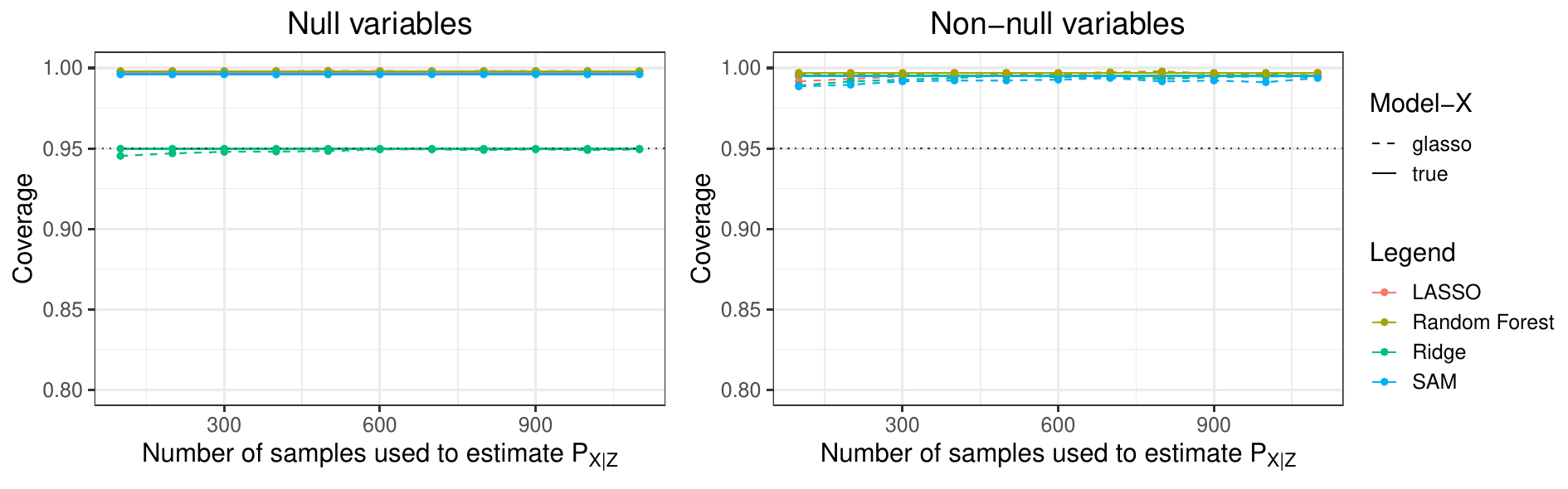}
\caption{Coverage of null (left) and non-null (right) covariates when the covariate distribution is estimated in-sample for the nonlinear-$\mustar$ simulations of Section~\ref{sec:simul_robust}. See Section~\ref{sec:simul_setup} for remaining details. Standard errors are below 0.001 (left) and 0.003 (right).
% Coverage plots for null and non-null variables, with the sample size for model-X part varying. Design matrix with $n$ i.i.d. rows follows the Gaussian copula model in  \ref{sec:nonlinear_setup} with the auto-correlation coefficient being $0.3$; $Y|X$ follows the nonlinear model described in \ref{sec:nonlinear_setup} with amplitude being $50$ and $30$ non-null variables; $n = 1100, p = 1000$; number of null copies $K = 500$ for nonlinear fitting algorithms.
}
\label{fig:robust_nonlinear}
\end{figure}

Figure~\ref{fig:robust_nonlinear} studies the robustness of floodgate for a nonlinear $\mustar$. We see the coverage being rather conservative for the non-null variables, reflecting the coverage-protective gap between $\thetamu$ and $\thetamus=\Ij$.
      \begin{figure}[tb]
        \centering
        \includegraphics[width =1\linewidth]{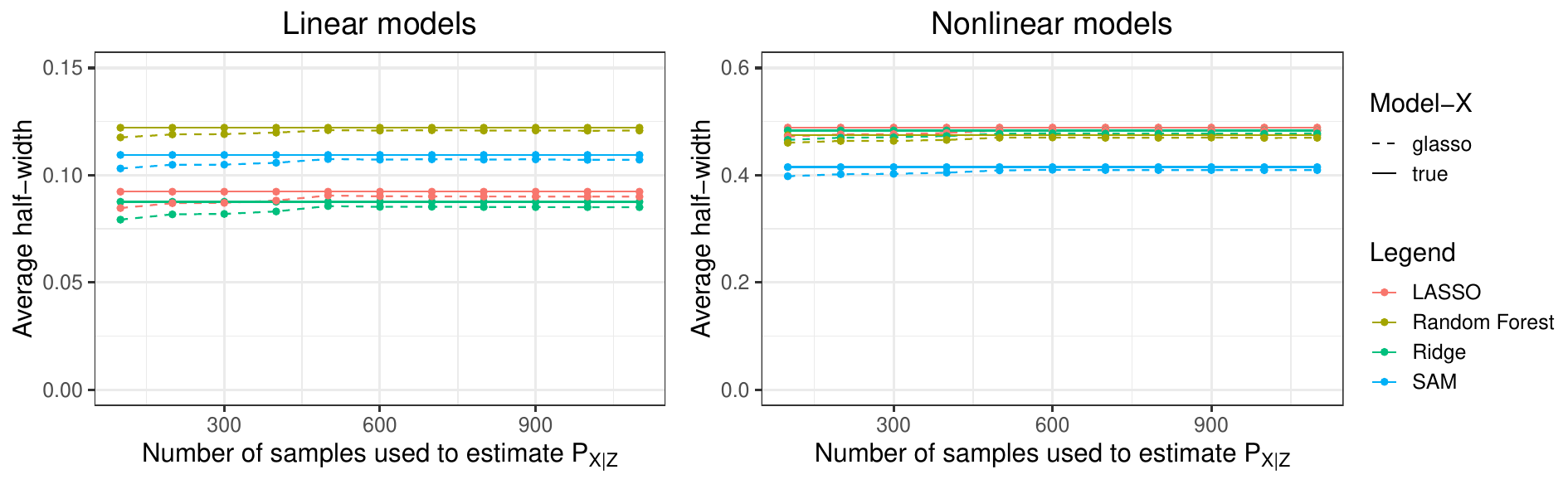}
        \caption{Half-width plot of non-null covariates when the covariate distribution is estimated in-sample for the linear-$\mustar$ (left) and nonlinear-$\mustar$ (right) simulations of Section~\ref{sec:simul_robust}. See Section~\ref{sec:simul_setup} for remaining details. Standard errors are below 0.002 (left) and 0.007 (right).
        % Coverage plots for null and non-null variables, with the sample size for model-X part varying; design matrix with $n$ i.i.d. rows from an AR(1) model with the auto-correlation coefficient being $0.3$; $Y|X \sim \calN(X\beta,1)$, where $\beta$ has non-zero entries with random signs and equal absolute values (which will be the amplitude value divided by $\sqrt{n}$); there are $30$ non-null variables; amplitude equals $5$ ; $n = 1100$; number of null copies $K = 500$ for nonlinear fitting algorithms. Coverage of null (left) and non-null (right) covariates when the covariate distribution is estimated in-sample for the linear-$\mustar$ simulations of Section~\ref{sec:simul_robust}. See Section~\ref{sec:simul_setup} for remaining details.
        }
        \label{fig:robust_app}
    \end{figure}
    {Figure~\ref{fig:robust_app} shows that in the simulations of linear models and nonlinear models, the average half-width of floodgate is robust to estimation error in $P_{X|Z}$.}

\subsubsection{Co-sufficient floodgate}
       {In this section, we demonstrate the performance of co-sufficient floodgate in a linear setting. Figure \ref{fig:relax2} tells a similar story as Figure \ref{fig:sine_relax2} in Section \ref{sec:simul_relax}. 
    %   show that co-sufficient floodgate has satisfying coverage even when the number of batches is small, and has average half-width quite close to the original floodgate procedure which is given the conditional mean of $X\mid Z$ exactly
     Note that despite the linearity of the true model in Figure~\ref{fig:relax2}, the LASSO performs poorly because the true model is quite dense (30 of the 50 covariates are non-null), which also explains why ridge regression performs so well.}
     %In terms of the accuracy performance, its average half-widths (lines with stronger colors) are also close to those of the original floodgate procedure (lines with weaker colors).
    % \label{sec:simul_relax}
    % \begin{figure}[tb]
    %     \centering
    %     \includegraphics[width =1\linewidth]{Pics/Relax/relax.pdf}
    %     \caption{Coverage (left) and average half-widths (right) for co-sufficient floodgate and original floodgate in the linear-$\mustar$ simulations. The number of batches $n_{1}$ is varied over the x-axis. See Section~\ref{sec:simul_setup} and \ref{sec:simul_relax} for remaining details. Standard errors are all below 0.008.
    %     % Coverage plots for null and non-null variables, with the sample size for model-X part varying. Design matrix with $n$ i.i.d. rows follows the Gaussian copula model in  \ref{sec:nonlinear_setup} with the auto-correlation coefficient being $0.3$; $Y|X$ follows the nonlinear model described in \ref{sec:nonlinear_setup} with amplitude being $50$ and $30$ non-null variables; $n = 1100, p = 1000$; number of null copies $K = 500$ for nonlinear fitting algorithms.
    %     }
    %     \label{fig:relax}
    % \end{figure}
    
    \begin{figure}[tb]
        \centering
        \includegraphics[width =1\linewidth]{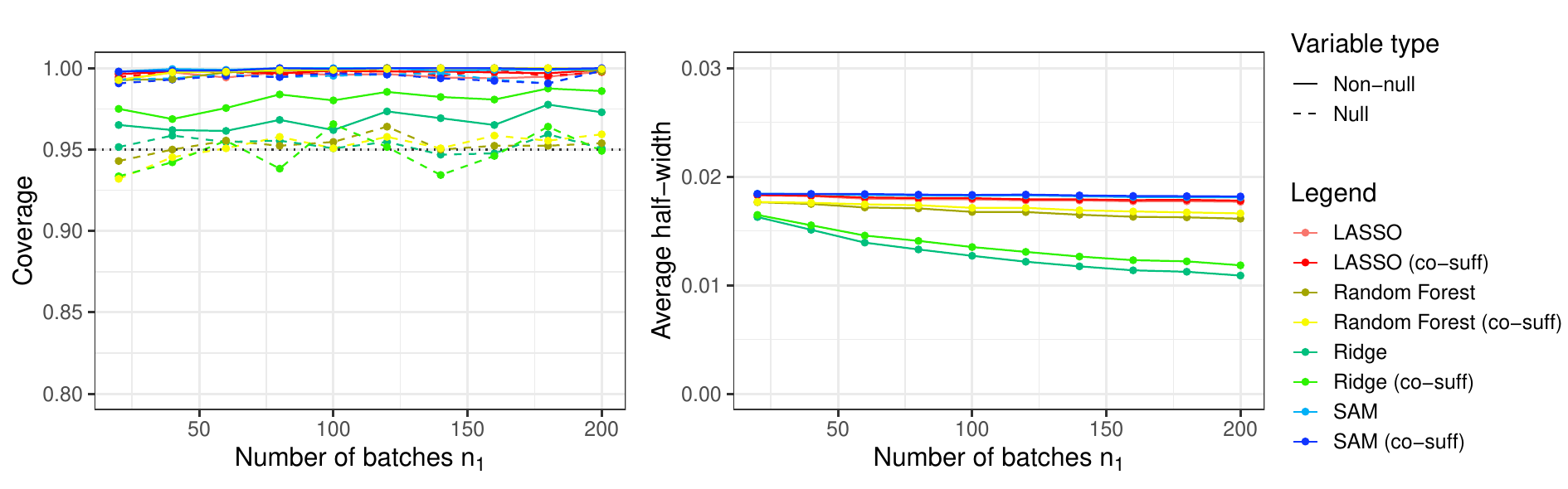}
            \caption{Coverage (left) and average half-widths (right) for co-sufficient floodgate and original floodgate in the linear-$\mustar$ simulations. The number of batches $n_{1}$ is varied over the x-axis. See Section~\ref{sec:simul_setup} and \ref{sec:simul_relax} for remaining details. Standard errors are below 0.008 (left) and 0.001 (right).
        % Coverage plots for null and non-null variables, with the sample size for model-X part varying. Design matrix with $n$ i.i.d. rows follows the Gaussian copula model in  \ref{sec:nonlinear_setup} with the auto-correlation coefficient being $0.3$; $Y|X$ follows the nonlinear model described in \ref{sec:nonlinear_setup} with amplitude being $50$ and $30$ non-null variables; $n = 1100, p = 1000$; number of null copies $K = 500$ for nonlinear fitting algorithms.
        }
        \label{fig:relax2}
    \end{figure}
    
     \subsubsection{Effect of covariate dependence}
    \label{sec:simul_corr}
    {In Figure \ref{fig:corr_comb}, we vary the covariate autocorrelation coefficient and plot the average half-widths of floodgate LCBs of non-null covariates under distributions with the linear (left panel) and the nonlinear (right panel) $\mustar$ described in Section~\ref{sec:simul_setup}, respectively. The left panel of Figure~\ref{fig:corr_comb} also includes a curve for OLS. Since $\Ij$ in a linear model is proportional to $\sqrt{\EE{\Varc{X}{Z}}}$ which varies with the autocorrelation coefficient, we divided the half-widths in Figure~\ref{fig:corr_comb} by this quantity to make it easier to compare values across the x-axis.
    %depends on the coefficient magnitude and the conditional variance of $X\mid Z$, the latter of which varies with the autocorrelation coefficient, we rescaled $\Ij$ in both simulations 
    %so that the estimand remains constant along the x-axis.
    % The coverage never falls below nominal, and these plots can be found in \ljmargin{Appendix~\ref{app:sim:cover}}{fix}. 
    The main takeaway is that the effect of covariate dependence on floodgate is somewhat mild until the dependence gets very large ($>0.5$ correlation). This behavior is intuitive, and indeed we see a parallel trend in the curves for OLS inference in Figure~\ref{fig:corr_comb}.}
    % Coverage plots corresponding to Figure~\ref{fig:corr_comb} and additional plots with different simulation parameters can be found in Appendix~\ref{app:sim:cover}.}
{The corresponding coverage plots of Figure~\ref{fig:corr_comb} are given in Figure~\ref{fig:corr_comb_cover}. Figures~\ref{fig:corr_app} and \ref{fig:corr_app_cover} are additional plots with a different covariate dimension specified in the captions. Figures~\ref{fig:corr_comb_cover} and \ref{fig:corr_app_cover} show that the coverage of floodgate is consistently at or above the nominal 95\% level.}

    \begin{figure}[tb]
        \centering
        \includegraphics[width = 1\linewidth]{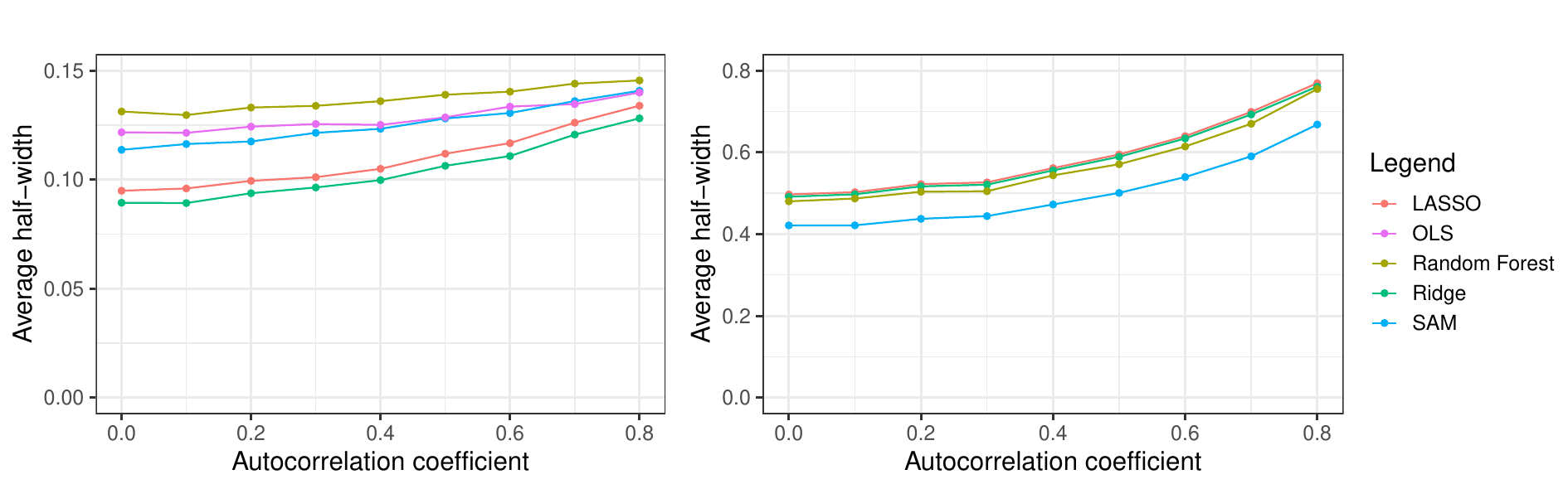}
        \caption{{Average half-widths for the linear-$\mustar$ (left) and nonlinear-$\mustar$ (right) simulations of Section~\ref{sec:simul_corr}. The covariate dimension $p = 1000$ and the covariate autocorrelation coefficient is varied on the x-axis; see Section~\ref{sec:simul_setup} for remaining details. Standard errors are below 0.002 (left) and 0.009 (right).}
        % Half width plots: design matrix with $n$ i.i.d. rows from an AR(1) model; $Y|X \sim \calN(X\beta,1)$, where $\beta$ has non-zero entries with random signs and equal absolute values (which will be the amplitude value divided by $\sqrt{n}$); amplitude equals $5$; there are $30$ non-null variables; $n = 1100$; number of null copies $K = 500$ for nonlinear fitting algorithms.
        }
        \label{fig:corr_comb}
    \end{figure}   

 \begin{figure}[tb]
        \centering
        \includegraphics[width = 1\linewidth]{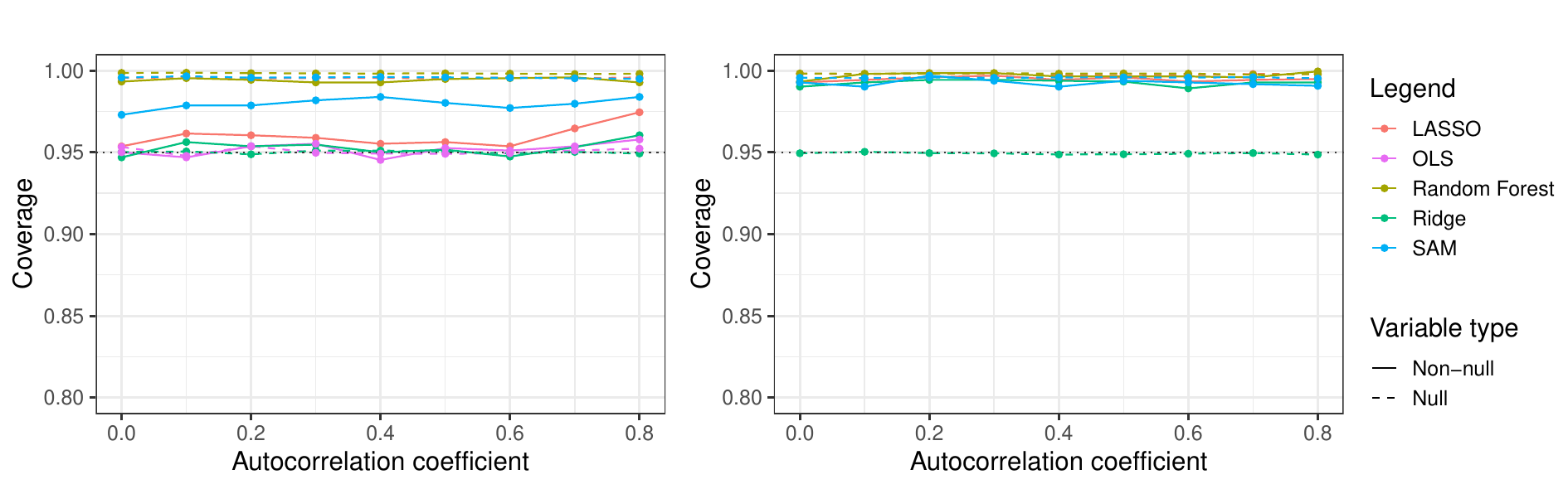}
        \caption{Coverage for the linear-$\mustar$ (left) and nonlinear-$\mustar$ (right) simulations of Section~\ref{sec:simul_corr}. The covariate dimension $p = 1000$ and the covariate autocorrelation coefficient is varied on the x-axis; see Section~\ref{sec:simul_setup} for remaining details. Standard errors are below 0.006 (left) and 0.003 (right).}
        % Half width plots: design matrix with $n$ i.i.d. rows from an AR(1) model; $Y|X \sim \calN(X\beta,1)$, where $\beta$ has non-zero entries with random signs and equal absolute values (which will be the amplitude value divided by $\sqrt{n}$); amplitude equals $5$; there are $30$ non-null variables; $n = 1100$; number of null copies $K = 500$ for nonlinear fitting algorithms.
        
        \label{fig:corr_comb_cover}
    \end{figure}
        \begin{figure}[tb]
        \centering
        \includegraphics[width = 1\linewidth]{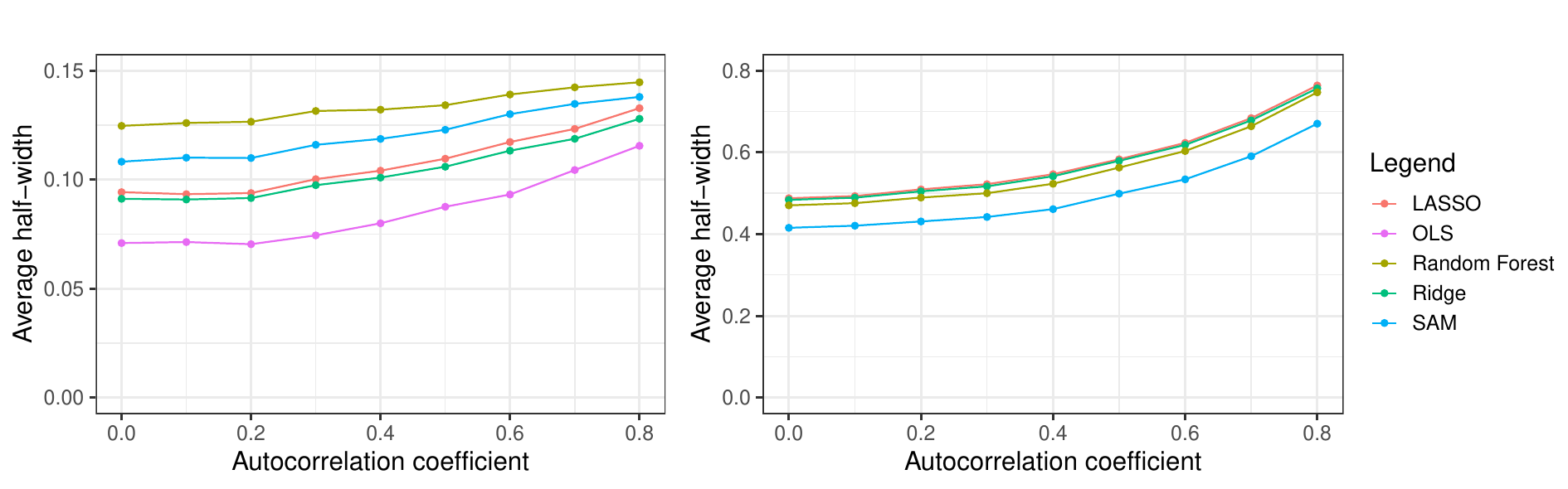}
        \caption{Average half-widths for the linear-$\mustar$ (left) and nonlinear-$\mustar$ (right) simulations of Section~\ref{sec:simul_corr}. The covariate dimension $p = 500$ and the covariate autocorrelation coefficient is varied on the x-axis; see Section~\ref{sec:simul_setup} for remaining details. Standard errors are below 0.002 (left) and 0.01(right).}
        % Half width plots: design matrix with $n$ i.i.d. rows follows the Gaussian copula model in  \ref{sec:nonlinear_setup}; $Y|X$ follows the nonlinear model described in \ref{sec:nonlinear_setup} with amplitude being $50$ and $30$ non-null variables; $n = 1100$; number of null copies $K = 500$ for nonlinear fitting algorithms.
        
        \label{fig:corr_app}
    \end{figure}  
    \begin{figure}[tb]
        \centering
        \includegraphics[width = 1\linewidth]{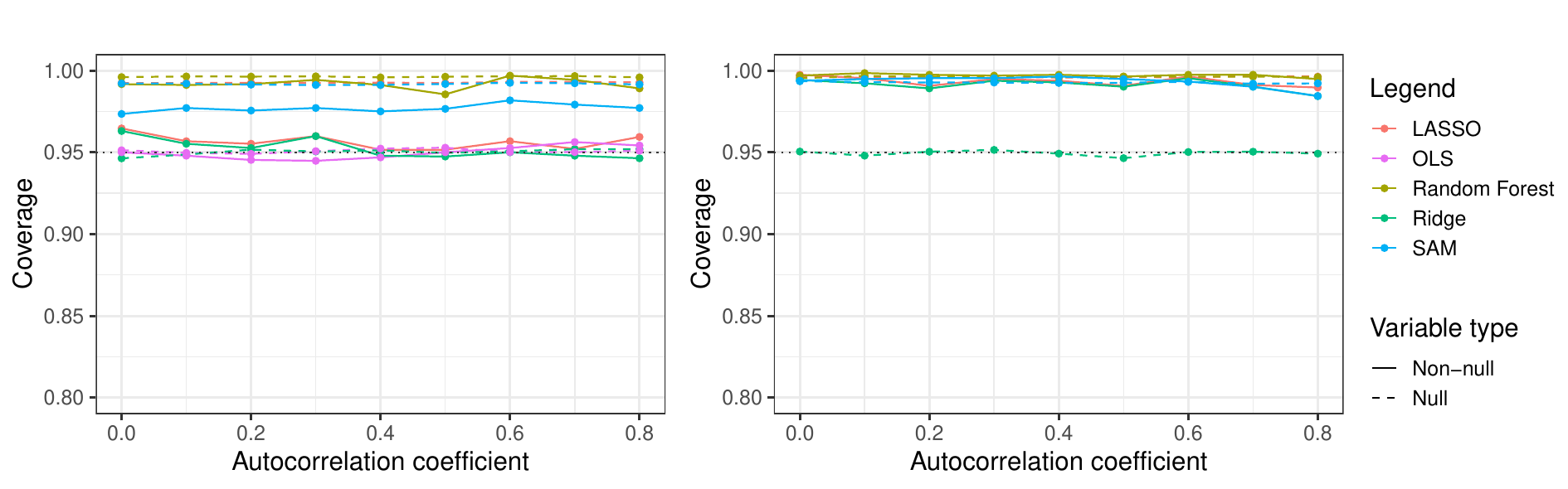}
        \caption{Coverage for the linear-$\mustar$ (left) and nonlinear-$\mustar$ (right) simulations of Section~\ref{sec:simul_corr}. The covariate dimension $p = 500$ and the covariate autocorrelation coefficient is varied on the x-axis; see Section~\ref{sec:simul_setup} for remaining details. Standard errors are below 0.007 (left) and 0.004 (right).}
        % Half width plots: design matrix with $n$ i.i.d. rows follows the Gaussian copula model in  \ref{sec:nonlinear_setup}; $Y|X$ follows the nonlinear model described in \ref{sec:nonlinear_setup} with amplitude being $50$ and $30$ non-null variables; $n = 1100$; number of null copies $K = 500$ for nonlinear fitting algorithms.
        \label{fig:corr_app_cover}
    \end{figure}      
    
    \subsubsection{Effect of sample size}
    In Figures \ref{fig:vary_n_linear} and \ref{fig:vary_n_sine}, we vary the sample size and plot the coverages and average half-widths of floodgate LCBs of non-null covariates under distributions with the linear and the nonlinear $\mustar$ described in Section~\ref{sec:simul_setup}, respectively. 
    % The coverage of null covariates never falls below nominal, and these plots can be found in Appendix~\ref{app:sim:cover}. 
    The main takeaway is that the accuracy of floodgate depends heavily on sample size. Note that in these plots, the signal size is scaled down by the square root of the sample size, so the \emph{selection} problem is roughly getting no easier as the sample size increases, but we still see that floodgate can achieve much more accurate inference for larger sample sizes. 
    % {In additional to the linear setting in Figure \ref{fig:vary_n_linear}, we also conduct simulations for nonlinear $\mustar$ and draw similar conclusions; see Appendix \ref{app:sim:cover} for details.}

    \label{sec:simul_vary_n}
        \begin{figure}[tb]
        \centering
        \includegraphics[width = 1\linewidth]{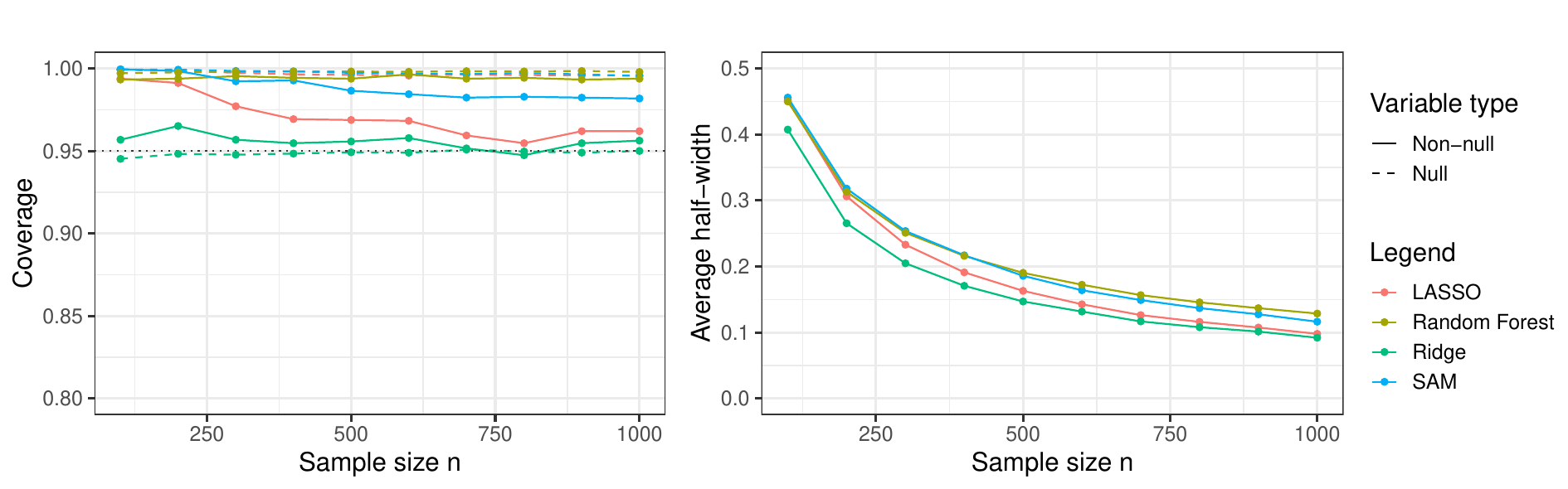}
        \caption{Coverage (left) and average half-widths (right) for the linear-$\mustar$ simulations of Section~\ref{sec:simul_vary_n}. The sample size $n$ is varied on the x-axis; see Section~\ref{sec:simul_setup} for remaining details. Standard errors are below 0.007 (left) and 0.003 (right).
        % Coverage and half width plots: design matrix with $n$ i.i.d. rows from an AR(1) model; $Y|X \sim \calN(X\beta,1)$, where $\beta$ has non-zero entries with random signs and equal absolute values (which will be the amplitude value divided by $\sqrt{n}$); amplitude equals $5$; there are $30$ non-null variables; $p = 1000$; number of null copies $K = 500$ for nonlinear fitting algorithms.
        }
        \label{fig:vary_n_linear}
    \end{figure}  
    % Figure \ref{fig:sine_amp} demonstrates how our method with different fitting algorithms perform when the true conditional model is highly nonlinear and with interaction terms as well. Looking at the precision plot, we find the linear model-based algorithms have inferior performance compared with nonlinear fitting algorithms, which is not surprising. The additive model structure explains a relatively higher precision of \textit{sam} compared with other methods. In the coverage plot, we see all of the methods are always overcoveraging the mMSE gap, regardless of the amplitude value. The explanation is that even as the signals increase, the fitting algorithms can at best hope for the best approximation of the true model in certain function classes, say $\mu^{\star}_{F}$. When the true regression function $\mustar$ is not contained in those classes, there always exists a gap between $\cal{I}_j$ and $\theta_j(\mu^{\star}_{F})$. Hence we see a constant overcoverage pattern in Figure.
    % \subsection{Nonlinear conditional models}
    % \label{sec:simul_nonlinear}

    % \begin{figure}[htbp]
    %     \centering
    %     \includegraphics[width = 1\linewidth]{Pics/Vary_n/sine_n.pdf}
    %     \caption{Coverage and precision plots for non-null variables, with the coefficient amplitude varying.}
    %     \label{fig:sine_amp}
    % \end{figure} 
% \lz{We also study the effect of sample size for nonlinear $\mustar$. Figure \ref{fig:vary_n_sine} tells similar stories as Figure \ref{fig:vary_n_linear} in Section \ref{sec:simul_vary_n}. }
    \begin{figure}[tb]
        \centering
        \includegraphics[width = 1\linewidth]{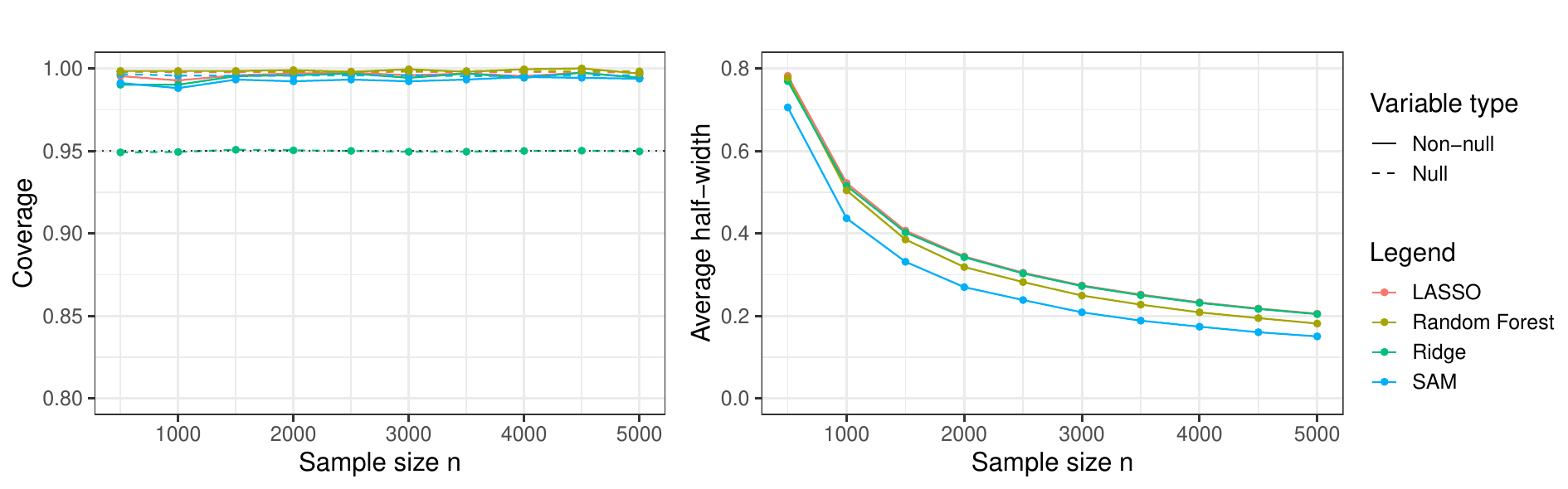}
        \caption{Coverage (left) and average half-widths (right) for the nonlinear-$\mustar$ simulations of Section~\ref{sec:simul_vary_n}. The sample size $n$ is varied on the x-axis; see Section~\ref{sec:simul_setup} for remaining details. Standard errors are below 0.004 (left) and 0.011 (right).
        % Coverage and half width plots: design matrix with $n$ i.i.d. rows follows the Gaussian copula model in  \ref{sec:nonlinear_setup} with the auto-correlation coefficient being $0.3$; $Y|X$ follows the nonlinear model described in \ref{sec:nonlinear_setup} with amplitude being $50$ and $30$ non-null variables; $p = 1000$; number of null copies $K = 500$ for nonlinear fitting algorithms.
        }
        \label{fig:vary_n_sine}
    \end{figure}
\section{Implementation details of genomics application}\label{app:realdata}
As mentioned in Section~\ref{sec:generalizations}, the floodgate approach can be immediately generalized to conduct inference on the importance of a group of variables. This is practically useful in our application to the genomic data, where we group nearby SNPs whose effects are usually found challenging to be distinguished. Specifically, we use the exact same grouping at the same seven resolutions as \citet{sesia2020multi}.

Regarding the genotype modelling, we consider the hidden Markov models (HMM) \citep{scheet2006fast}, as used in \citet{sesia2019gene,sesia2020multi}, which provides a good description of the  linkage disequilibrium (LD) structure. We obtain the fitted HMM parameters from \citet{sesia2020multi} on the UK Biobank data. Since HMM does not offer simple closed form expressions of the conditional quantities in Algorithm ~\ref{alg:MOCK}, we generate null copies of the genotypes and use them for the Monte Carlo analogue of floodgate. Below we simply describe the generating procedure. Under the HMM, we denote the covariates by $W$ (genotypes or haplotypes) and the unobserved hidden states (local ancestries) by $A$, with the joint distribution over $W$ denoted by $P_W$,  the joint distribution over $A$ denoted by $P_A$, which is the latent Markov chain model. For a given contiguous group of variables $g_{j}$, we can sample the null copy of $W_{g_j}$ as follows:
        \begin{enumerate}[(1)]
        \item Marginalize out $W_{g_j}$ and recompute the parameters of the new HMM $P_{\text{-}g_j}$ over $W_{\text{-}g_j}$.
        \item Sample the hidden states $A_{\text{-}g_j}$ by applying the forward-backward algorithm to $W_{\text{-}g_j}$, with the new HMM $P_{\text{-}g_j}$. 
        \item Given $A_{\text{-}g_j}$, sample $A_{g_j}$ according to the latent Markov chain model $P_A$.
        \item Sample $\tilde{W}_{g_j}$ given $A_{g_j}$ according to the emission distribution of the group $g_j$ in the model of $P_W$.
        \end{enumerate}

    To see why the above procedure produces a valid null copy of $W_{g_j}$, consider the following joint distribution, conditioning on $W_{\text{-}g_j}$
$$
P_{\text{joint}}: (W_{g_j}, A_{g_j}, A_{\text{-}g_j})\mid W_{\text{-}g_j}
$$
If we sample $(\tilde{W}_{g_j}, A_{g_j}, A_{\text{-}g_j})$ from the above joint conditional distribution, without looking at $W_{g_j}$ or $Y$, then $\tilde{W}_{g_j} $ has the same conditional distribution as $W_{g_j}$, given $W_{\text{-}g_j}$ and is conditionally independent from $(W_{g_j},Y)$, and thus is a valid null copy of $W_{g_j}$. Regarding how to sample from $P_{\text{joint}}$, we take advantage of the HMM structure and sample $A_{\text{-}g_j}, A_{g_j}, \tilde{W}_{g_j}$ sequentially since 
\begin{eqnarray} \label{eq:sample_A}
A_{g_j}\mid A_{\text{-}g_j}, W_{\text{-}g_j}  \deq  A_{g_j}\mid A_{\text{-}g_j},\\
W_{g_j}\mid A_{g_j}, A_{\text{-}g_j}, W_{\text{-}g_j}  \deq  W_{g_j}\mid A_{g_j}. \label{eq:sample_W}
\end{eqnarray}
 Sampling from $ A_{\text{-}g_j}\mid  W_{\text{-}g_j}$ is feasible since $P_{\text{-}g_j}$ is still a HMM whenever the group $g_j$ is contiguous. Under the HMM with particular parameterization in \citet{scheet2006fast}, the cost of the forward-backward algorithm can be reduced, see \citet{sesia2020multi} for more details. We remark that marginalizing out $W_{g_j}$ only changes the transition structure around the group $g_j$ and the special parameterization over other variables is still beneficial in terms of the computation cost. Sampling of $A_{g_j}$ and $\tilde{W}_{g_j}$ is computationally cheap due to \eqref{eq:sample_A} and \eqref{eq:sample_W}. For a given number of null copies $K$, we will repeat the steps (2)-(4) for $K$ times. But we remark the involving sampling probabilities only have to be computed once.
 
 Regarding the quality control and data prepossessing of the UK Biobank data, we follow the Neale Lab GWAS with application $31063$; details can be found on \url{http://www.nealelab.is/uk-biobank}. A few subjects withdrew consent and are removed from the analysis. Our final data set consisted of $361,128$ unrelated subjects and $591,513$ SNPs along $22$ chromosomes.
 
 For the platelet count phenotype, the analysis by \citet{sesia2020multi} makes several selections over the whole genome at seven different resolution levels. We focus on chromosome $12$ and look at $248$ selected groups from their analysis. For a given group of variables, we generate $K=5$ null copies following the null copy generation procedure described above. 
 
%  In consideration of the robustness to model misspecification, we scale the phenotype variable by subtracting a sample mean estimator $\bar{Y}_{\text{train}}$ which is computed over training samples (used for fitting $\mu$), this will not change our inferential validity due to the independence among samples and the fact that conditioning on the training dataset, $\bar{Y}_{\text{train}}$ can be treated as a constant.

We applied floodgate with a 50-50 data split and fitted $\mu$ to the first half using the cross-validated LASSO as in \citet{sesia2020multi} and included both genotypes (SNPs from chromosomes 1--22) and the non-genetic variables sex, age and squared age. We centered $Y$ by its sample mean from the first half of the data (the half used to fit $\mu$) before applying floodgate. Although this changes nothing in theory, it does improve robustness as small biases in $\mu(X_i,Z_i)-\Ec{\mu(X_i,Z_i)}{Z_i}$ would otherwise get multiplied by $Y_i$'s mean in the computation of $R_i$ in Algorithm~\ref{alg:MOCK}. 

Although our fitting of a linear model in no way changes the validity of floodgate's inference of the completely model-free mMSE gap, it does desensitize the LCB itself to the nonlinearities and interactions that partially motivated $\Ij$ as an object of inference in the first place. Our reasoning is purely pragmatic: as the universe of nonlinearities/interactions is exponentially larger than that of linear models, fitting such models requires either very strong nonlinear/interaction effects or prior knowledge of a curated set of likely nonlinearities/interactions. It is our understanding that nearly all genetic effects, linear and nonlinear/interaction alike, tend to be relatively weak, and the authors are not geneticists by training and thus lack the domain knowledge necessary to leverage the full flexibility of floodgate. Although we were already able to find substantial heritability for many blocks of SNPs with our default choice of the LASSO, it is our sincere hope and expectation that geneticists who specialize in the study of platelet count or similar traits would be able to find even more heritability using floodgate. 
    
    We report LCBs for all blocks simultaneously, although computationally we only actually run floodgate on those selected by \citet{sesia2020multi}. Although their selection used all of the data (including the data we used for floodgate), it does not affect the marginal validity of the LCBs we report, as explained in the last paragraph of Section~\ref{sec:generalizations}.

\end{document}